\documentclass[12pt, oneside, hidelinks]{article}

\usepackage[margin=0.95in]{geometry} 
\usepackage{amsmath,amsthm,amssymb, graphicx, multicol, array,hyperref, bbm, comment, subcaption,  float, booktabs, xpatch, color}
\usepackage{setspace} 

\def\*#1{\mathbf{#1}}
\def\+#1{\mathbb{#1}}

\newcommand{\T}{\top}

\newcommand{\Lp}{\left(}
\newcommand{\Rp}{\right)}

\newcommand{\tvec}{\mathrm{vec}}
\newcommand{\tr}{\mathrm{trace}}

\newcommand{\diag}{\mathrm{diag}}
\newcommand{\trace}{\mathrm{trace}}
\newcommand{\COV}{\mathrm{Cov}}

\newcommand{\RNum}[1]{\uppercase\expandafter{\romannumeral #1\relax}}

\usepackage[T1]{fontenc}
\usepackage[utf8]{inputenc}
\usepackage{lmodern}
\usepackage[english]{babel}
\usepackage[autostyle]{csquotes}
\usepackage{amsmath}
\usepackage{natbib}
\setcitestyle{open={(},close={)}}

\usepackage{hyperref} 
\usepackage{xr-hyper} 
\usepackage{xr}
\makeatletter
\newcommand*{\addFileDependency}[1]{
	\typeout{(#1)}
	\@addtofilelist{#1}
	\IfFileExists{#1}{}{\typeout{No file #1.}}
}
\makeatother

\let\ab\allowbreak

\usepackage{enumitem}
\setlist{itemsep=.01em}
\setlist{topsep=.5em}

\newcommand{\norm}[1]{\left\lVert#1\right\rVert}

\newcommand\und{\underline}

\newtheorem{innercustomgeneric}{\customgenericname}
\providecommand{\customgenericname}{}
\newcommand{\newcustomtheorem}[2]{%
	\newenvironment{#1}[1]
	{%
		\renewcommand\customgenericname{#2}%
		\renewcommand\theinnercustomgeneric{##1}%
		\innercustomgeneric
	}
	{\endinnercustomgeneric}
}

\newcustomtheorem{customthm}{Theorem}
\newcustomtheorem{customlemma}{Lemma}
\newcustomtheorem{customassumption}{Assumption}

\newtheorem{assumption}{Assumption}
\newtheorem{theorem}{Theorem}
\newtheorem{lemma}{Lemma}
\newtheorem{proposition}{Proposition}
\newtheorem{corollary}{Corollary}

\DeclareMathOperator*{\argmin}{arg\,min}
\DeclareMathOperator*{\argmax}{arg\,max}
\let\ab\allowbreak

\newcolumntype{L}[1]{>{\raggedright\arraybackslash}p{#1}}

\usepackage[margin=10pt,labelfont=bf]{caption}
\newcommand\tcaptab[1]{\captionsetup{position=top, font=normalsize, labelfont=bf, textfont=normalfont, justification=centering, margin=0mm, aboveskip=1mm, belowskip=0mm, labelsep=colon, singlelinecheck=false}\caption{#1}}
\newcommand\bnotetab[1]{\captionsetup{position=bottom, font=scriptsize,  textfont=normalfont, margin=1mm,aboveskip=1mm, belowskip=-1mm, justification=justified, singlelinecheck=false}\caption*{#1}}
\newcommand\tcapfig[1]{\captionsetup{position=top, font=normalsize, labelfont=bf, textfont=normalfont, justification=centering, margin=0mm, aboveskip=1mm, belowskip=0mm, labelsep=colon, singlelinecheck=false}\caption{#1}}
\newcommand\bnotefig[1]{\captionsetup{position=bottom, font=scriptsize,  textfont=normalfont, margin=1mm, aboveskip=1mm, belowskip=-1mm, justification=justified, singlelinecheck=false}\caption*{#1}}

\begin{document}
	
	\title{State-Varying Factor Models of Large Dimensions\thanks{\scriptsize We thank Yacine A\"it-Sahalia, Daniele Bianchi, Frank Diebold, Kay Giesecke, Eric Ghysels, Lisa Goldberg and seminar and conference participants at Stanford, the NBER-NSF Time-Series Conference, the Society for Financial Econometrics, the North American Summer Meeting of the Econometric Society, the European Meeting of the Econometric Society, the Western Mathematical Finance Conference, the SIAM Conference on Financial Mathematics \& Engineering and INFORMS for helpful comments.}}
\date{\today}
	\author{
		Markus Pelger\thanks{\scriptsize Stanford University, Department of Management Science \& Engineering, Email: \url{mpelger@stanford.edu}.}
		\and
		Ruoxuan Xiong\thanks{ \scriptsize Stanford University, Department of Management Science \& Engineering, Email: \url{rxiong@stanford.edu}.}
	}

	\begin{onehalfspacing}
		\begin{titlepage}
		\maketitle
                \thispagestyle{empty}
			\vspace{0.5cm}
			
			\begin{abstract}
				This paper develops an inferential theory for state-varying factor models of large dimensions. Unlike constant factor models, loadings are general functions of some recurrent state process. We develop an estimator for the latent factors and state-varying loadings under a large cross-section and time dimension. Our estimator combines nonparametric methods with principal component analysis. We derive the rate of convergence and limiting normal distribution for the factors, loadings and common components. In addition, we develop a statistical test for a change in the factor structure in different states. We apply the estimator to U.S. Treasury yields and S\&P500 stock returns. The systematic factor structure in treasury yields differs in times of booms and recessions as well as in periods of high market volatility. State-varying factors based on the VIX capture significantly more variation and pricing information in individual stocks than constant factor models.
				\vspace{1cm}
				

			\end{abstract}

		\end{titlepage}

		\section{Introduction}
		
		Factor models provide an appealing way to summarize information from large data sets. In factor models, a small number of latent common factors explain a large portion of the co-movements. They have been successfully used in finance, e.g. \cite{Ross1976}, \cite{Chamberlain1983} and \cite{ludvigson2009factor}, and in macro-economics, e.g. \cite{Stock2002} and \cite{jurado2015measuring}. Large dimensional factor models typically assume that the underlying factor structure does not change over time, that is, the factor loadings capturing the exposure to factors are assumed to be constant over time as for example in \cite{bai2002determining}, \cite{bai2003inferential} and \cite{fan2013large}.\footnote{Extensions of the constant loading model include sparse and interpretable latent factors in \cite{pelgerxiong2020}, estimation from incomplete data sets in \cite{xiongpelger2020} and including additional moments to estimate weak factors as in \cite{lettaupelger2019,lettaupelger2018}.} However, since financial and macroeconomic data sets often span a long time period, it can be overly restrictive to assume a constant exposure to factors. Over a long time horizon, domestic and foreign policies change, business cycles occur, technology progresses, and agents' preferences switch \citep{stock2009}. Ignoring these changes can lead to a misspecified model with false inference and prediction \citep{Breitung2011}.
		
		This paper presents an inferential theory for state-varying factor models of large dimensions. Unlike constant-loading factor models, the loadings are general functions of some recurrent state process. We develop an estimator for the latent factors and state-varying loadings for a large number of cross-sectional and time observations. Our estimator combines nonparametric methods with principal component analysis (PCA). We derive the rate of convergence and asymptotic normal distribution for the estimated factors, loadings, and common components. We also develop a statistical test for the change of the loadings in different states.

		Our state-varying model achieves two important goals: First, we can estimate the systematic factor structure conditioned on the outcome of a state variable. For example, we can estimate how the factor structure in asset prices depends on the business cycle.\footnote{\cite{pelger2019} provides empirical evidence for time-variation in latent factor models that is related to recessions.} In particular, we can obtain the loadings as a general function of the state variable and use this insight for building economic models. Second, we can capture time variation in the systematic factor structure. The loadings in our model are time-varying because the state process changes over time. The dynamics of the state process and the functional form of loadings as a function of the state process jointly determine the dynamics of loadings over time. We allow for very general dynamics of the state process that include smooth processes but also discontinuous processes. Hence, our approach allows the loadings to change many times rapidly, but also covers the cases of a small number of large changes or many gradual changes. As we consider a very general functional form for the loading function, the loadings can vary more for particular state outcomes or even be constant for other outcomes of the state process.

		Our approach combines kernel methods with PCA. The underlying idea is to estimate a large-dimensional covariance matrix conditioned on a recurrent state process and to analyze its spectral decomposition. For this purpose, we apply a kernel projection in the time dimension on a particular state value.\footnote{\cite{fan2016projected} model loadings as non-linear functions of time-varying features of the cross-sectional units. Their estimation approach applies PCA to the data matrix that is projected in the cross-section on the subject-specific covariates. We also apply PCA to a projected data matrix, but our projection is applied in the time dimension.}  PCA is then applied to the projected data. Our estimator is easy to use and simple to implement. The inferential theory depends in a complex way on the kernel approximation and the number of cross-sectional and time-series observations.\footnote{In this paper, we consider a scalar state process and leave the extension to multivariate state processes to future research.  We expect multivariate state processes to lead to lower convergence rates due to the ``curse of dimensionality'' inherent in higher-dimensional kernel projections. An additional challenge is that many multivariate state processes do not have the recurrence property.} The theoretical framework for the factor estimation is closely related to \cite{bai2003inferential}'s and \cite{Su2017}'s inferential theory of the PCA estimator. While \cite{bai2003inferential} applies PCA to the unconditional covariance matrix, \cite{Su2017}  use the spot covariance matrix conditioned on a particular point in time. Our approach conditions on a particular realization of the state process. Conditioning on a state with a kernel projection significantly complicates the analysis as it leads to additional bias terms and a complex interplay between the number of time and cross-sectional observations and the kernel bandwidth. We characterize the general conditions that are sufficient for asymptotic consistency and a conditional normal distribution of the loadings, factors, and common components under the assumptions of an approximate factor model that has a similar level of generality as \cite{bai2003inferential}'s framework.\footnote{\cite{wang2019} also study a state-dependent latent factor model. Their focus is the estimation of a large dimensional state-dependent covariance matrix, while we derive the inferential theory for the factors and conditional loadings.}
		
		We develop a novel test for changes in the loadings. Our test statistic allows us to answer the important question in which states loadings are different. The challenge comes from the fact that factor models can only be identified up to invertible linear transformations, and hence we cannot directly compare the loadings estimated for different states with each other. Our test statistic is based on a generalized correlation statistic, which measures how close the two vector spaces spanned by loadings in two states are.\footnote{Generalized correlations or canonical correlations have been studied in \cite{anderson1958introduction}, \cite{yuan2000three}, \cite{bai2006}, \cite{pelger2017} and \cite{andreou2017industrial}} Testing the null hypothesis of the same loadings in two different state realizations turns out to be a ``corner case'' similar to unit root test statistics with a faster convergence rate than under the alternative hypothesis. The test statistic is non-standard and requires a novel bias correction, which we provide. We believe that the novel technique that we develop for our test statistic can also be easily adopted to the gradual change or high-frequency PCA models and will encourage further developments.\footnote{Our test differs from the existing tests, such as those of \cite{Breitung2011}, \cite{Chen2014}, \cite{Han2015}, and \cite{Yamamoto2015}, which check the stability of the moments of factor loadings or common factors, but do not take invertible transformations into account. Our test takes a ``micro'' view to compare loadings in any two states, while \cite{Su2017} takes a ``global'' view to test whether loadings change in the whole time dimension.}

		Our method can estimate the functional relationship between a time-varying state process and structural changes in the loadings, which adds additional economic interpretability to the model.
		We do not require a parametric form for the loadings as a function of the state process, which is potentially misspecified, but allow for a general functional relationship. Our approach allows us to study questions such as how a macroeconomic factor structure depends on the business cycle, while other approaches are more limited to study these types of questions: First, the magnitude of the changes might be too smooth to be captured as a large structural break. Second, there might not be sufficient local observations that are required in the local smoothing framework, which ignores the information that is contained in similar states of the business cycle.

		Our framework generalizes the conventional constant loading factor models and allows for a more parsimonious representation of the data. Under certain assumptions, a factor model with state-dependent loadings can be approximated by a constant loading model with a larger number of latent factors. A more complex functional form of the state-dependent loadings typically requires more basis functions to approximate them and results in a larger number of latent factors in the constant loading approximation. Our state-varying factor model can require significantly fewer factors than a constant loading model to explain the same or more variation in the data. In this sense, our model provides a more parsimonious model. Furthermore, our estimator can be valid even if we observe the state process with noise or omit a relevant state. Our inferential theory is robust to moderate noise contamination in the state process. Even if we miss a relevant state or the noise contamination in the state process is more severe, our estimator can still dominate the conventional factor model approach. We only need to condition on a state process that depends on the source of variation in the loadings. We confirm this result in our simulation and empirical studies.
		
		We show a strong state-dependent time variation in the factor structure of U.S. Treasury yields and S\&P 500 stock returns. The yields of bonds with different maturities are well-explained by the three PCA factors commonly labeled as level, slope, and curvature factors.\footnote{See \cite{diebold2005modeling}, \cite{Diebold2006}, \cite{Cochrane2005} and \cite{Cochrane2009}.} 
		We show that this factor structure depends on state variables such as recession and boom indicators, the stock market's expectation of volatility (VIX), or the unemployment rate. We show that during recessions, in times of high volatility, or in times of a high unemployment rate, the first PCA factor, typically labeled as a level factor, becomes less dominant and shifts to longer-term bonds. However, the second and third PCA factors, labeled as slope and curvature factors, both shift more towards shorter-term bonds. These changes are statistically and economically significant and show that the economic interpretation of ``level'', ''slope'' and ``curvature'' has to be used with caution, as for different economic states, the PCA factors will be different. In the second application on individual stock returns, we show that a state-varying factor model with the VIX as the state variable is more parsimonious in explaining variation and captures more pricing information than the constant loading model. A constant loading model requires five more factors to explain the same amount of variation as our state-dependent factor model. At the same time, an optimal portfolio based on the state-varying factors earns out-of-sample a five times higher risk-adjusted return than the corresponding portfolio based on a constant loading model. Hence, even if we might not capture all time-variation in the loadings with the proposed state variable, we still obtain a model that explains the correlations structure and mean returns better than a constant loading model.
		
Our paper is complementary to the literature on structural breaks and local PCA estimation that pursue a related but different objective.
Our goal is to provide a parsimonious model that allows for time variation due to an observable time-varying state process. 
The literature on structural breaks focuses on detecting and modeling a small number of large breaks in the latent factor structure. It includes \cite{andrews1993tests}, \cite{Chen2014}, \cite{Breitung2011}, \cite{Cheng2016}, \cite{baltagi2020estimating}, \cite{bai2020estimation},
 \cite{ma2018estimation} and \cite{barigozzi2018simultaneous}.\footnote{\cite{Breitung2011} develop three test statistics for structural breaks. \cite{Chen2014} study the detection of large breaks in loadings through a two-stage procedure. \cite{Han2015}  test for structural breaks by studying the stability in second moments. \cite{Yamamoto2015} generalize \cite{Breitung2011}'s test. \cite{Cheng2016} propose a test where both the factor loadings and the number of factors change simultaneously. \cite{baltagi2020estimating} and \cite{bai2020estimation} estimate structural breaks with pseudo factors, \cite{ma2018estimation} propose a three-step approach with local estimates and \cite{barigozzi2018simultaneous} estimate structural breaks with wavelets.} A factor structure, that depends on a state process that jumps, exhibits structural breaks that we can detect reliably. Modeling a one-time large structural break may be inappropriate if changes happen smoothly, for example, policy changes or business cycles can lead to gradual changes. Moreover, the large structural break models are limited to a small number of changes. Re-occurring events, for example, changes in the business cycle and economic conditions, can lead to a large number of changes in the factor structure. The alternative approach is to model changes smoothly with a local estimator. \cite{Su2017} and \cite{eichler2011fitting} use a local kernel estimator in the time dimension to study gradual changes. This approach excludes sudden large changes and exploits only the data in a local neighborhood of a particular time observation. This problem can be overcome by using high-frequency data as in \cite{pelger2017,pelger2019}, \cite{ait2017principal,ait2017using}, \cite{kong2017number,kong2018systematic} and \cite{kong2018testing}, that allows to make general statements about time-variation in large dimensional latent factor models. The idea is similar to the local PCA estimators, but the high-frequency data provides sufficient information to detect more rapid changes. However, appropriate high-frequency data is only available for a limited number of applications, e.g., high-frequency trading data for U.S. equity in the recent past. We show that if we have additional information about a state process, we can reliably model structural breaks or a local time-variation in a uniform framework. In fact, in this case, we can explain more structure in the data than the models based only on structural breaks or a local time-variation that are not taking advantage of this additional information. Importantly, neither the large structural break models nor the gradual change models provide a direct economic link of the change to underlying economic variables.

		\section{Model Overview} \label{model}
		
		\subsection{Setup}\label{model_overview}
		Assume a panel data set of $T$ time-series observations and $N$ cross-sectional observations, denoted as $X \in \+R^{N \times T}$, has a factor structure with $r$ common factors. Let $S_t$ be the value of a state process at time $t$, $X_{it} \in \mathbb{R}$ the cross-sectional observation $i$ at time $t$, $F_t \in \mathbb{R}^{r \times 1}$ the latent factors at time $t$, and $\Lambda_i(S_t) \in \mathbb{R}^{r \times 1}$ the factor loadings of the cross-sectional unit $i$ when the state value is $S_t$:
		
		\[X_{it} = \Lambda_i(S_t)^\top F_t + e_{it} \quad \text{for } i = 1, 2, \cdots, N \text{ and } t = 1, 2, \cdots, T\]
		or in vector notation,
		\[\underbrace{X_t}_{N \times 1}  = \underbrace{\Lambda(S_t)}_{N \times r} \underbrace{F_t}_{r \times 1} + \underbrace{e_t}_{N \times 1} \qquad \text{for $ t = 1, 2, \cdots, T$}.\]

		We observe $X_t$ and $S_t$ and want to estimate $F_t$ and $\Lambda(.)$ in an asymptotic setup where $N$ and $T$ are both large.
		This model generalizes the large dimensional factor model in \cite{bai2002determining} and \cite{bai2003inferential} and allows factor loadings to change over time. The loadings in our model are time-varying because the state process changes over time. The loadings are deterministic general functions of the state process that satisfy some smoothness conditions. The random state process itself has a continuous distribution.\footnote{In contrast to other time-varying factor models, e.g., \cite{Bates2013}, \cite{Cheng2016}, and \cite{Su2017}, our model directly incorporates the driving forces for the changes in loadings. \cite{Park2009} study a similar semiparametric factor model but require cross-sectional variation in the loadings to come from observable covariates, this means they estimate a function $\Lambda(C_{it})$ of observable covariates $C_{it}$ where $\Lambda(.) \in \mathbb R^{r}$ is the same function for all $i$. They apply B-splines to estimate the unknown loading function and estimate the factors with a Newton-Raphson algorithm. Our approach is based on a simple-to-implement PCA method, which allows us to derive an inferential theory.}

		\subsection{Estimation Problem}
		We want to estimate the factor model conditioned on the state outcome $S_t=s$. Before providing formal arguments, we will describe the intuition behind our approach. 
		If the idiosyncratic component is conditionally uncorrelated with the factors, then the conditional second moment matrix equals
		\begin{align*}
		\+E[X_t X_t^\T |S_t=s] = \Lambda(s) \+E[F_t F_t^\T|S_t=s]\Lambda(s)^{\top} + \COV(e_t|S_t=s).
		\end{align*}
		We will assume that the factors are systematic in the sense that they affect many cross-sectional units captured by a full rank assumption of $\frac{1}{N}\Lambda(s)^{\top}\Lambda(s)$ for $N \rightarrow \infty$. Furthermore, the idiosyncratic component has conditionally only a weak dependency structure modeled by a sparsity assumption on the conditional residual covariance matrix. Hence, the largest eigenvalues of $\+E[X_t X_t^{\top}|S_t=s]$ should come from the systematic part and the corresponding eigenvectors will be linked to the loadings $\Lambda(s)$. This motivates the application of PCA to the conditional second moment matrix to estimate the factor and loadings.\footnote{If $X$ has a conditional mean of zero, PCA is applied to the conditional covariance matrix.} The essential identification condition is the full rank of the conditional second moment matrix of the factors $\+E[F_t F_t^\T|S_t=s]$ and of the limit loading matrix $\lim_{N\rightarrow \infty} \frac{1}{N}\Lambda(s)^{\top}\Lambda(s)$.

		The conditional second moment matrix is estimated by a kernel projection of the data that puts higher weights on time observations where the state process takes values in a neighborhood of $s$,
		i.e. we analyze $\frac{1}{NT} K_s^{1/2}X^{\top}X K_s^{1/2}$ with an appropriate $T \times T$ diagonal matrix $K_s$ of kernel weights $\frac{1}{h}K\left(\frac{S_t-s}{h} \right)$ and bandwidth $h$.\footnote{The PCA estimation can either be applied to $\frac{1}{NT} K_s^{1/2}X^{\top}X K_s^{1/2}$, in which case the eigenvectors are related to the transformed factors, or to $\frac{1}{NT} X K_s X^{\top}$, which relates the eigenvectors to the transformed loadings.}
		The analysis is inherently complicated by the bias arising in any kernel estimation from using observations from nearby states. In more detail, the observations can be written as
		\begin{align*}
		X_{it}=\Lambda_i(s)^{\top}F_t + (\Lambda_i(S_t)-\Lambda_i(s))^{\top}F_t + e_{it}= \Lambda_i(s)^{\top}F_t + \omega_{it}+e_{it}.
		\end{align*}
		The bias term $\omega_{it}$ behaves like an additional error term that requires a different treatment than the idiosyncratic error.

		\subsection{Identification Assumption}\label{sec:Identification}

		Unlike the conventional constant loading factor model there are two sources of time-variation in our model: the time-series of the factors and of the state process. This poses the inherent identification problem what is a factor and what is a state? We impose the identification assumption that the second factor moment does not depend on the state process to separate the state process from the latent factors. We denote the conditional and unconditional second factor moment by $\Sigma_{F|s}=\+E[F_t F_t^\top |S_t=s]$ and $\Sigma_F=\+E[F_t F_t^\top]$ respectively. $\Sigma_{F|S_t}=\+E[F_t F_t^\top |S_t]$ is the implied stochastic process.
		\begin{assumption}{Identification Assumption:}\label{Ass:Ident}
			For all $s$ in the support of $S_t$ we assume that
			the conditional second factor moment does not depend on the state process and is positive definite: $\Sigma_{F|s}:=\+E[F_t F_t^{\top} |S_t=s] = \+E[F_t F_t^{\top}]=: \Sigma_F$. \label{Ass:Ident1}
		\end{assumption}
		Note that the identification assumption is not limiting the generality of our model. Replacing the factors $F_t$ by $\tilde F_t = F_t \Sigma_{F|S_t}^{-1/2} \Sigma^{1/2}_F$ and loadings by $\tilde \Lambda (S_t) = \Lambda(S_t) \Sigma_{F|S_t}^{1/2} \Sigma_F^{-1/2}$ we can ensure that for any conditional factor model the assumption is satisfied. It is possible to relax the full-rank assumption, i.e. the number of systematic factors in different states can differ.

		As in any PCA estimation problem the factors and state-varying loadings are only identified up to an invertible matrix $H(S_t)$ which in our case can be state-varying, i.e. $X_t =\Lambda(S_t)F_t + e_t = \Lambda(S_t)H(S_t) H(S_t)^{-1} F_t + e_t$. We estimate the factors as the eigenvectors of the conditional second moment matrix. Our estimates coincide with the factors if $\Sigma_F=I_r$ and $\lim_{N \rightarrow \infty} \frac{1}{N}\Lambda(s)^{\top} \Lambda(s)$ is diagonal. In the general case, the conditional eigenvectors $\hat F_t$ estimate $H(S_t) F_t$ where $H(S_t)$ is uniquely determined. This is a generalization of the standard assumption in the unconditional factor models that the estimated factors are the eigenvectors of the unconditional second moment and the estimated loadings are orthogonal.

		We illustrate how Assumption \ref{Ass:Ident} identifies the factors with a number of examples. For simplicity we consider in these examples only a one-factor model, but the extensions to more factors are straightforward. Assume that the cross-section is modeled by
		\begin{align*}
		X_t = \beta_1 S_t \tilde F_t + e_t.    
		\end{align*}
		We can either view this as a constant loading model ($\Lambda= \beta_1 , F_t=S_t \tilde F_t$) or a state-varying factor model.
Both formulations are equally valid and will explain the same amount of variation. 
		With the identification Assumption \ref{Ass:Ident} we obtain the model ($\Lambda= \beta_1  S_t \sqrt{\frac{\+E[\tilde F^2|S_t]}{\+E[\tilde F^2]}} , F_t=\tilde F_t \sqrt{\frac{\+E[\tilde F^2]}{\+E[\tilde F^2|S_t]}}$). It simplifies to ($\Lambda= \beta_1  S_t, F_t=\tilde F_t$) if $\tilde F$ is independent of $S_t$.

		If the systematic component of $X$ is a function of $F_t$ and $S_t$, we can formulate it as
		\begin{align*}
		X_t = g(\tilde F_t,S_t) + \epsilon_t
		\end{align*}
		for some function $g$. Assume that $g$ can be approximated well by a second-order Taylor approximation, i.e. we will model $g$ as a second-order polynomial function:
		\begin{align*}
		X_t =&  \beta_0 + \beta_1 S_t + \beta_2 S_t^2 + \beta_3 \tilde F_t + \beta_4 \tilde F_t^2 + \beta_5 \tilde F_s S_t  + \epsilon_t\\
		=& \underbrace{\left(\begin{pmatrix} \beta_0 + \beta_1 S_t + \beta_2 S_t^2 \\ \beta_3 + \beta_5 S_t \\ \beta_4  \end{pmatrix}^{\top} \Sigma^{1/2}_{F|S_t} \Sigma_F^{-1/2}\right)}_{\Lambda(S_t)}  \underbrace{\left( \Sigma^{1/2}_F \Sigma^{-1/2}_{F|S_t} \begin{pmatrix}  1 \\ \tilde F_2 \\ \tilde F_t^2  \end{pmatrix} \right)}_{F_t} + \epsilon_t.
		\end{align*}
		Thus, we can either formulate the model as a constant loading model with 6 factors or a state-varying loading model with 3 factors. Both formulations are equivalent in terms of explaining variation, but the state-varying model is more parsimonious. Assumption \ref{Ass:Ident} separates the state process from the latent factor in the above formulation. 
		
		Next, we consider the same model but assume that $g$ is a third-order polynomial function based on a Taylor expansion:
		\begin{align*}
		X_t 
		=& \left(  \beta_0 + \beta_1 S_t + \beta_2 S_t^2 + \beta_3 S_t^3  \right) 1 + \left( \beta_4 + \beta_7 S_t + \beta_8 S_t^2\right) \tilde F_t  + \left( \beta_5 + \beta_9 S_t \right)\tilde F_t^2  + \beta_6 \tilde F_t^3   + \epsilon_t.
		\end{align*}
		Assumption \ref{Ass:Ident} identifies a state-varying four-factor model which can equivalently be written as a 10 factor model with constant loadings. Hence, for a more complex functional relationship, the state-varying model is more parsimonious compared to the constant loading model. Note, that in the examples the loadings are a linear function of a finite number of transformations of the state process. Our focus is on the relevant model where we have a continuum of state outcomes $S_t$ and a non-linear loading function $\Lambda(S_t)$ that requires a large number of basis functions for its approximation. In this case there exists in general no multi-factor representation with constant loadings.%
		
		\subsection{Robustness to Misspecification}\label{sec:robust}
		Our framework allows us to study conditional latent factors. The purpose can be to understand how the systematic dependency structure changes with some specific state variable or to have a parsimonious factor model that allows for time variation. In the second case, it is important to find the state variable that is the source of the time variation in the factor model. In practice, we might not know which state process attributes to changes in the loadings. In the Internet Appendix, we show that our results are still valid if we use a noisy approximation of the source of change:
		\[X_{it} = (\Lambda_i(S_t) + \varepsilon_{it})^\top F_t + e_{it} \quad i = 1, 2, \cdots, N \text{ and } t = 1, 2, \cdots, T.\]
		The term $\epsilon_{it}$ is the time-varying component of the loading coefficient that cannot be explained by the state process $S_t$. It can, for example, be due to a measurement error in the state process $S_t$ or an omitted additional state process that affects only a small number of loadings. We show that under mild assumptions, the term $\varepsilon_{it}^{\top} F_t$ can be treated like an additional non-systematic error term that will not affect our results. In this sense, our model is robust to mild misspecification.

		More generally, missing a relevant systematic state can be accounted for by including more latent factors. In this sense our model is also robust to more severe misspecification as we illustrate with the following example. Assume there are two state processes $S_t$ and $\tilde S_t$ that have a systematic effect, i.e. $X_t =  g(\tilde F,S_t,\tilde S_t) + \epsilon_t$.
		If $g$ is a second order polynomial we have
		\begin{align*}
		X_t =& \left(  \beta_0 + \beta_1 S_t + \beta_2 S_t^2 \right) + \left( \beta_3 + \beta_7 S_t \right) \tilde S_t + \left( \beta_5 + \beta_8 S_t  \right) \tilde F_t + \beta_4 \tilde S_t^2 \\ &+ \beta_6\tilde  F_t^2 + \beta_9 \tilde S_t \tilde F_t + \epsilon_t,
		\end{align*}
		which has equivalent representations as a 10-factor constant loading model, a 7-factor model conditioned on $S_t$ and a 3-factor model conditional on $S_t$ and $\tilde S_t$. Hence, even if we do not condition on all relevant states or use a noisy approximation, the state-varying factor model can provide a more parsimonious representation than the constant loading version. We formalize this idea in the Internet Appendix.

		\section{Estimation} \label{estimation}

		We estimate the factor model conditional on the realization of the state process $S_t = s$. Our approach generalizes the conventional PCA estimator by using projected data. We first apply a kernel projection to calculate the second moment matrix conditioned on the state outcome $s$. Second, we use PCA on the conditional second moment matrix to obtain the estimated factors and loadings for the state outcome $s$. The estimated factors are the eigenvectors of the conditional second moment matrix. Loadings in state $s$ are the regression coefficients of the projected data on the estimated factors. 
		
		We estimate factors and loadings by minimizing the following criterion function
		\begin{align}\label{loss}
		\hat{F}^s, \hat{\Lambda}(s)  &= \argmin_{F, \Lambda(s)} \underbrace{\frac{1}{NT(s)} \sum_{i = 1}^N \sum_{t = 1}^T K_s(S_t)(X_{it} - \Lambda_i(s)^\top F_t)^2}_{V_s} \nonumber,
		\end{align}
		where $T(s) = \sum_{t=1}^T K_s(S_t)$ and $K_s(S_t) = \frac{1}{h}K\left(\frac{S_t - s}{h}\right)$. $T(s)$ can be interpreted as the effective number of time observations used to estimate the factor structure conditioned on a state value. $K(\cdot)$ denotes a kernel function. $h$ is a bandwidth parameter, which depends on how much information we want to use from the neighboring states and our prior knowledge about the smoothness of the loadings as a the function of the state process. 
		We can think of $V_s$ as the average loss function conditioned on the state outcome $s$. We reformulate the problem as a conventional least squares problem by projecting the data, factors and idiosyncratic components in the time dimension on the kernel: $X_{it}^s = K_s^{1/2}(S_t) X_{it}$, $F^s_t =  K_s^{1/2}(S_t) F_t$ and $e_{it}^s = K_s^{1/2}(S_t) e_{it}$. The objective function can then be expressed as
		\begin{align*}
		V_s &= \frac{1}{NT(s)} \sum_{i = 1}^N \sum_{t = 1}^T (X^s_{it} - \Lambda_i(s)^\T F^s_t)^2 
		= \frac{1}{NT(s)} \trace \{(X^s - \Lambda(s)(F^s)^\top) (X^s - \Lambda(s)(F^s)^\top)^\top  \},
		\end{align*}
		where $K_s = \diag(K_s(S_1), K_s(S_2), \cdots, K_s(S_T))$, $K_s^{1/2} = \diag(K_s^{1/2}(S_1), K_s^{1/2}(S_2), \cdots, K_s^{1/2}(S_T))$,
		\begin{eqnarray*}
			X^s &=& X K_s^{1/2} =  \begin{bmatrix}
				X^s_{1} & X^s_{2} & \cdots & X^s_{T}
			\end{bmatrix} \in \mathbb{R}^{N \times T} \\
			(F^s)^\top &=& F^\T K_s^{1/2} = \begin{bmatrix}
				F^s_{1} & F^s_{2} & \cdots & F^s_{T}
			\end{bmatrix} \in \mathbb{R}^{r \times T} \\
			\Lambda(s)^\top &=& \begin{bmatrix}
				\Lambda_1(s) & \Lambda_2(s) & \cdots & \Lambda_N(s)
			\end{bmatrix} \in \mathbb{R}^{r \times N}.
		\end{eqnarray*}

		$V_s$ is a quadratic and convex loss function. Factors and loadings can be estimated up to some invertible rotation.\footnote{If $\hat{\Lambda}(s)$ and $\hat{F}^s$ are a solution minimizing $V_s$, then $\forall G \neq I_r \in \mathbb{R}^{r \times r}$ which are invertible, $\hat{\Lambda}(s) G$ and $\hat{F}^s G^{-1}$ also minimize $V_s$.} 
		After normalizing $(\hat{F}^s)^\T   \hat{F}^s/T(s) = I_r$ and concentrating out the loadings, the objective function becomes a conventional PCA problem: 
		\begin{align*}
		\hat{F}^s = \underset{F^s}{\argmax} \,\, \trace \left\lbrace (F^s)^\top \left( \frac{1}{NT(s)} (X^s)^\T X^s \right) F^s \right\rbrace.
		\end{align*}
		
		The estimator $\hat{F}^s$ equals $\sqrt{T(s)}$ times the eigenvectors of the $r$ largest eigenvalues of the matrix $\frac{1}{NT(s)} (X^s)^\T X^s$. $V^s_r$ is the diagonal matrix with diagonal elements equal to the $r$ largest eigenvalues in decreasing order of the matrix $\frac{1}{NT(s)} (X^s)^\T X^s$. The conditional loadings are estimated as $\hat{\Lambda}(s) = X^s \hat{F}^s/T(s)$, and the unconditional factors can be estimated for each state outcome as $\hat{F} = K^{-\frac{1}{2}}_s \hat{F}^s$.

		\section{Assumptions} \label{assumption}
		Let $M<\infty$ denote a generic constant. The matrix norm below is the Frobenius norm $\norm{A} = \trace(A^\top A)^{1/2}$. We condition on the state outcome $s$ which is in the support of the distribution of the state process $S_t$.

		\begin{assumption} \label{ass_state}
			State and kernel function: 
			\begin{enumerate}
				\item The state process $S_t$ at time $t$ is observed and positive recurrent. $\pi(s)$ is the stationary probability density function (PDF) of $S_t$. $\pi(s)$ is continuous and has first order bounded derivative.
				\item The kernel function $K(\cdot)$ is a symmetric, continuously differentiable and nonnegative function that has a compact support and $\int u^4 k(u) du$ exists.\footnote{Many common kernels satisfy this Assumption: 1. Gaussian kernel $K(u) = \frac{1}{\sqrt{2\pi}}exp(-\frac{u^2}{2})$. 2. Uniform kernel $K(u) = \frac{1}{2} \mathbbm{1}(|u| \leq 1)$. 3. Epanechnikov kernel $k(u) = \frac{3}{4}(1-u^2)\mathbbm{1}(|u| \leq 1)$. 4. Biweight kernel ($k(u) = \frac{15}{16} (1-u^2)^2 \mathbbm{1} (|u| \leq 1) $). 5. Triweight kernel ($k(u) = \frac{35}{32} (1-u^2)^3 \mathbbm{1} (|u| \leq 1) $)). 6. Many higher order kernels obtained by multiplying them by a higher order polynomial in $u^2$.}
			\end{enumerate}
			
		\end{assumption}

		Under Assumption \ref{ass_state}, 
		$S_t$ is positive recurrent, implying the existence of a stationary distribution. The assumptions that $\pi(s)$ is continuous and has first-order bounded derivative, implies a continuous state space for $S_t$. For simplicity we assume that $S_t$ follows its stationary distribution for all $t$, but it is straightforward to relax this assumption. It is sufficient that there exists a $T_0$ with $T_0/T \rightarrow 0$ such that $S_{T_0} \sim \pi$. As a result it holds for all $t \geq T_0$ that $S_t \sim \pi$.		
		 In the remainder of this paper, we estimate the factor model in the state with a stationary density greater than zero. Intuitively, this means that a neighborhood of any state can be visited infinitely many times in an infinite time horizon. Under this assumption, we show consistency for $N$ and $T$ jointly going to infinity. If the stationary distribution $\pi(s)$ is continuous and has bounded first-order derivative, together with the kernel's property, we can estimate the stationary distribution nonparametrically, which is $\hat{\pi}(s) = T(s)/T = \frac{1}{T}\sum_{t = 1}^T K_s(S_t) \xrightarrow{p} \pi(s)$ as $h \rightarrow 0$ and $Th \rightarrow \infty$.
		
		This assumption implies that the state process can take infinitely many values, but it does not make an assumption about the dynamics of the state process.\footnote{Many relevant state processes in economics and finance can be modeled with a continuous state space, e.g., inflation rates, growth rates, volatility processes or return processes. In this case, the state process takes infinitely many values.} The state process can be a slowly changing process, or it can be an abruptly changing process, even with many jumps.  
		
		On the other hand, if the state process takes only finitely many values, we can separate the data by state values and estimate a factor model for each state with stationary probability greater than 0. This would be a special case of estimating the factor model from the data projected by a kernel, which has an indicator term, such as the uniform kernel, and picks an appropriate bandwidth $h$.  With some modifications of the proofs, the theorems hold. 
		
		Because the kernel function is symmetric and continuously differentiable, the bias in the nonparametric density estimator is of order $O_p(h^2)$, which is smaller than $O_p(h)$ obtained from a non-symmetric kernel. In addition, the existence of a fourth moment of the kernel implies that the tail in the kernel cannot be too heavy. When we estimate the factor model in $s$, we also use data in other states weighted by the kernel, which results in some biases in the estimator. Assumption \ref{ass_state}.2 ensures that the bias can be controlled by the kernel weight and will not dominate in the limiting distribution of the estimators.

		\begin{assumption} \label{ass_factor}
			Conditional Factors: 
				 It holds $\max_t \, \+E[\norm{F_t}^4]  \leq \bar{F} < \infty$, $\max_t \, \+E[\norm{F_t}^4 |\mathcal{F}_S ]  \leq \bar{F} < \infty$ $\forall s $, where $\mathcal{F}_S$ is the filtration of the state process, and $ \frac{1}{T(s)} \sum_{t = 1}^T K_s(S_t) F_t F_t^\T  \xrightarrow{p} \Sigma_{F|s} $ as $T \rightarrow \infty$ for some positive definite matrix $\Sigma_{F|s}$.
		\end{assumption}
		
		Assumption \ref{ass_factor} is the state-conditional version of assumption A in \cite{bai2003inferential}. The fourth moment of the factor is bounded both without and with the state filtration, which makes it possible to have asymptotic results for both unconditional and conditional estimated factors. This assumption implies that given any realization of the state process, the fourth moment of the factor cannot explode. The conditional covariance matrix of the factors $\Sigma_{F|s} $ needs to be positive definite, which implies that no factor is degenerated after being projected on a particular state outcome.

		\begin{assumption} \label{ass_loading}
			Factor Loadings: 
			\begin{enumerate}
				\item Factor loadings are deterministic functions of the state process. Furthermore, $\forall s$ and $\forall i$, $\norm{\Lambda(s)^\T \Lambda(s)/N - \Sigma_{\Lambda(s)}} \rightarrow 0$ for some positive definite matrix $\Sigma_{\Lambda(s)}$. 
				\item  $\Lambda(s)$ is {deterministic and} Lipschitz continuous in $s$: There exist some constant $C$, $\norm{\Lambda_i(s + \Delta s) - \Lambda_i(s)} \leq C |\Delta s|$, $\forall s, \Delta s$ and $i$. 
			\end{enumerate}
		\end{assumption}
		
		Assumption \ref{ass_loading} ensures that, in every state, each factor has a nontrivial contribution to the variance of the data. 
		{The loadings are deterministic functions of the state process, which is a stochastic process. Therefore, the unconditional loadings are random, but conditional on the outcome of the state process, they are deterministic.}
		The loadings are Lipschitz continuous with respect to the state, e.g., a differentiable function of the state {with} bounded first-order derivative. If the {state process} is bounded, most {differentiable} loading functions can satisfy this assumption. This assumption, together with the kernel assumption, guarantee that the bias generated from using data in neighboring states is not a leading term in the limiting distribution of the estimators.

		\begin{assumption} \label{ass_err}
			Time and Cross Sectional Dependence and Heteroskedasticity:\\
			There exists a positive constant $M < \infty$ such that for all $N$ and $T$:
			\begin{enumerate}
				\item \label{eandSindep} \sloppy $\+E\left[ e_{it} \right] = 0$. $\+E\left[ e_{it}^8 \right] \leq M$. $e$ and $S$ are independent. 
				\item Weak time-series dependence: \sloppy $\+E\left[e_t^\T e_u/N \right] = \+E\left[\frac{1}{N} \sum_{i=1}^N e_{it}e_{iu}\right] = \gamma_N(t,u)$. $|\gamma_N(t,t)| \leq M, \, \forall t$, $|\gamma_N(t,u)| \leq M, \, \forall t,u$, and $\sum_{u=1}^T |\gamma_N(t,u)| \leq M, \, \forall t$. 
				\item Weak cross-sectional dependence: \sloppy $\+E\left[e_{it} e_{lt} \right] = \tau_{il, t}$, with $|\tau_{il, t}| \leq |\tau_{il}|$ for some $\tau_{il}$ and $ \sum_{l = 1}^N |\tau_{il}| \leq M$ for all $i$. 

				\item Weak total dependence: $\+E\left[ e_{it} e_{lu} \right] = \tau_{il, tu}$ and $\frac{1}{NT} \sum_{i = 1}^N \sum_{l = 1}^N \sum_{t = 1}^T \sum_{u = 1}^T  |\tau_{il,tu}| \leq M$. 

				\item Bounded cross-sectional fourth moment correlation: \\ For every (t,u), $\+E |N^{-1/2} \sum_{i=1}^N [e_{it}e_{iu} - \+E(e_{it}e_{iu})]|^4 \leq M$. 
				\item Weak dependence between factors and idiosyncratic components: 
				\begin{enumerate}
					\item Define $\+E\left[ F_u e_u^\T  e_t/N \right] = \gamma_{N, F}(t,u)$. Then $\norm{\gamma_{N, F}(t,t) }  \leq M$ $\forall t$, and \\$\sum_{u=1}^T \norm{\gamma_{N, F}(t,u)} \leq M$ and $\sum_{t=1}^T \norm{\gamma_{N, F}(t,u)} \leq M$ $\forall t, \, \forall u$. 
					\item Define $\+E\left[ F_u e_u^\T  e_t/N | S_t, S_u \right] = \gamma^s_{N, F}(t,u)$. Then $\norm{\gamma^s_{N, F}(t,t) }  \leq M$ $\forall t$ and \\$\sum_{u=1}^T \norm{\gamma^s_{N, F}(t,u)} \leq M$ and $\sum_{t=1}^T \norm{\gamma^s_{N, F}(t,u)} \leq M$ $\forall t, \, \forall u$.
				\end{enumerate}
				
			\end{enumerate}
		\end{assumption}
		
		Assumption \ref{ass_err} allows the unconditional idiosyncratic components to have weak time-series and cross-sectional dependences. Our model is an approximate static factor model similar to \cite{bai2002determining} and \cite{bai2003inferential}. For simplicity, we also assume the state process is independent of the idiosyncratic components. This assumption, together with weak time-series and cross-sectional unconditional dependence in the idiosyncratic components, implies that the idiosyncratic components can have weak time-series and cross-sectional dependence conditional on the state process. The last part assumes weak unconditional and conditional correlation between factors and idiosyncratic components. We state them separately because the factors and state process may be dependent.

		\begin{assumption} \label{ass_mom}
			Moments and Central Limit Theorem (CLT): There exists an $M \leq \infty$,  such that for any $s$ and for all $N$, $T$, and $h$ that satisfy $h \rightarrow 0$, $Nh \rightarrow \infty$,
			$Th \rightarrow \infty$, $Nh^2 \rightarrow 0$ and $Th^3 \rightarrow 0$,
			\begin{enumerate}
				\item Projected factors and idiosyncratic components: \\
$\max_t \frac{1}{N T}  \sum_{i = 1}^N \sum_{j = 1}^{N} \left(\sum_{u = 1}^{T} \frac{1}{\pi(s)} \gamma^{s,t}_{F,e}(i,j,u,u)  +  h \sum_{p \neq u} \gamma^{s,t}_{F,e}(i,j,u,p) \right) \leq M$,	\\			
				where $\gamma^{s,t}_{F,e}(i,j,u,p)  = \left \| \+E [F_{u} F_{p}^{\top}  (e_{iu}e_{it}- \+E[e_{iu}e_{it}]) (e_{jp}e_{jt}- \+E[e_{jp}e_{jt}]) |S_u = s, S_{p} = s ] \right \|^2$.

				\item Projected factors, loadings and idiosyncratic components:\\
				$\frac{1}{N T}  \sum_{i = 1}^N \sum_{j = 1}^{N} \left(\sum_{u = 1}^{T} \frac{1}{\pi(s)} \gamma^{s,t}_{F,\Lambda,e}(i,j,u,u)  +  h \sum_{p \neq u} \gamma^{s,t}_{F,\Lambda,e}(i,j,u,p) \right) \leq M$,  \\
				where $\gamma^{s,t}_{F, \Lambda, e}(i,j,u,p)  = \left \| \+E \left[ F_u \Lambda_i(s)^{\top}  F_p \Lambda_j(s)^{\top}      e_{iu} e_{jp}  |S_u = s, S_{p} = s \right] \right \|^2$.
				\item CLT for loadings and idiosyncratic components: 
				$$ \frac{1}{\sqrt{N}}\sum_{i = 1}^N \Lambda_i(s) e_{it} \ab  \xrightarrow{d} N(0, \Gamma_{t}^s),$$
				where $\Gamma_{t}^s = \lim_{N \rightarrow \infty} \sum_{i = 1}^N\sum_{l = 1}^N \Lambda_i(s) \Lambda_l(s)' \+E[e_{it} e_{lt}]$.
				\item CLT for projected factors and idiosyncratic components: 
				$$\frac{\sqrt{Th}}{T(s)} \sum_{t = 1}^T K_s(S_t) F_t e_{it} \xrightarrow{d} N(0, \Phi_{i}^s),$$
				where
				$\Phi_i^s = \lim_{T \rightarrow \infty} \left( \frac{1}{T} \sum_{t = 1}^T \left( \frac{R_K}{\pi(s)} \gamma_{FF}^s(t, t)   + h \sum_{t \neq u} \gamma_{FF}^s(t, u)  \right) \right)$, $R_K = \int K^2(u) du$\\ $\gamma_{FF}^s(t, t) = \+E\left[ F_tF_t^\T e_{it}^2 | S_t = s \right]$ and $\gamma_{FF}^s(t, u) = \+E\left[ F_tF_u^\T  e_{it} e_{iu} | S_t = s, S_u = s \right].$
				\item Loadings and projected idiosyncratic components: \\ 
				$\max_t \frac{1}{N T}   \sum_{i = 1}^N \sum_{j = 1}^{N} \left(\sum_{u = 1}^{T} \frac{1}{\pi(s)} \gamma^{s,t}_{\Lambda,e}(i,j,u,u)  +  h \sum_{p \neq u} \gamma^{s,t}_{\Lambda,e}(i,j,u,p) \right) \leq M$,
				where\\ $\gamma^{s,t}_{\Lambda,e}(i,j,u,p)  = \left \| \+E \left[\Lambda_{i}(s) \Lambda_{j}(s)^{\top}   (e_{iu}e_{tu}- \+E[e_{iu}e_{tu}]) (e_{jp}e_{tp}- \+E[e_{ip}e_{tp}]) |S_u = s, S_{p} = s \right] \right \|^2$.
			\end{enumerate}
		\end{assumption}

		Assumption \ref{ass_mom} are moment conditions and central limit theorems, which are satisfied by mixing processes of factors, loadings, and idiosyncratic components projected by the kernel function of the state process. This assumption is required only for asymptotic distribution results. The CLT for loadings and idiosyncratic components will be used to show the limiting distribution for estimated factors and common components. The CLT for projected factors and idiosyncratic components will be used to show the limiting distribution for estimated loadings and common components. Assumption \ref{ass_mom}.1, \ref{ass_mom}.2 and \ref{ass_mom}.5 allow for a weak conditional dependency between the idiosyncratic component and the factors and loadings. It is trivially satisfied if the idiosyncratic components are independent of the factors and loadings. The bandwidth parameter $h$ appears in the assumptions, which is used to balance the $h$ term in the denominator of the second moment of the kernel function. Intuitively, for a smaller bandwidth $h$, we are using information from a smaller portion of the data, which lowers the convergence rate. We will discuss the rate conditions on $h$, $N$, and $T$ in more detail in the next section.
		
		\begin{assumption} \label{ass_eigen}
			The eigenvalues of the $r \times r$ matrix $\Sigma_{\Lambda(s)} \Sigma_{F|s}$ are distinct.
		\end{assumption}

		Factors and loadings can be estimated up to some rotation matrix, and this matrix can be uniquely determined by Assumption \ref{ass_eigen}.\footnote{In general, our approach allows the relative importance of factors to switch in the support of the state process. $\Sigma_{\Lambda(s)} \Sigma_{F|s}$ and its eigenvalues are continuous in $s$. Thus, Assumption \ref{ass_eigen} does not allow factors to switch their relative importance in the neighborhood of the state that we condition on. In this case, there exists some state value for which $\Sigma_{\Lambda(s)} \Sigma_{F|s}$ has repeated eigenvalues. The individual factors are not identified for this particular state outcome. Nevertheless, the common component is still identified.} 

		\section{Asymptotic Results} \label{results}
		Under appropriate rate conditions, we can consistently estimate the factors, loadings, and common components and obtain an asymptotic normal distribution. The rate conditions are similar to the results in \cite{bai2003inferential}, but replace $T$ by the effective number of time observations $Th$. However, the kernel bias term requires additional rate restrictions.
		
		We assume that the number of factors $r$ has been consistently estimated. A possible consistent estimator for the number of factors is proposed in \cite{bai2002determining} and based on an information criterion. 
In the Internet Appendix, we generalize this estimator to our setup, which allows us to estimate the number of factors conditional on a specific state outcome. This estimator selects the number of factors by trading off the amount variation explained in a specific state and a penalty function based on an information criterion. This penalty function essentially replaces the rate $T$ by $Th$ in the estimator of \cite{bai2002determining}. In the Internet Appendix, we also discuss how to use a cross-validation approach to determine the number of factors.

		\begin{theorem}\label{thm_consistency}
			Consistency of Estimated Factors: \\
			Under Assumptions \ref{Ass:Ident}-\ref{ass_err}, $N, Th \rightarrow \infty$, $\delta_{NT, h} = min(\sqrt{N}, \sqrt{Th})$ and $\delta_{NT, h}h \rightarrow 0$, we have 
			\begin{eqnarray}
			\delta_{NT, h}^2 \left( \frac{1}{T} \sum_{t=1}^T \norm{\hat{F}^s_t - (H^s)^\T  F^s_t}^2\right) = O_p(1)
			\end{eqnarray}
			and 
			\begin{eqnarray}
			\delta_{NT, h}^2 \left( \frac{1}{N} \sum_{i=1}^N \norm{\hat{\Lambda}_i(s) - (H^s)^{-1} \Lambda_i(s)}^2\right) = O_p(1)
			\end{eqnarray}
			with $H^s = \frac{\Lambda(s)^\T  \Lambda(s)}{N} \frac{(F^s)^\T  \hat{F}^s}{T(s)}  (V^s_r)^{-1}$ and $V^s_r$ is the diagonal matrix consisting of the $r$ largest eigenvalues of  $\frac{1}{NT(s)} (X^s)^\T  X^s$.
		\end{theorem}
		
		$\hat{F}^s_t$ are estimates of the projected factors $F^s_t$. The factors $F_t$ can be identified up to some rotation matrix, $H^s$, which depends on the state.
		The convergence rate is the smaller of $N$ and $Th$, denoted as $\delta^2_{NT, h}$, as $N, Th \rightarrow \infty$ and $h \rightarrow 0$. As expected, the bandwidth parameter interacts with $T$, but not with $N$, as the kernel projection is equivalent to weighting data differently in the time dimension. 
		
		The additional restriction $\delta_{NT, h}h \rightarrow 0$ is due to the kernel bias. In more detail, the projected observations can be written as
		\begin{eqnarray}\label{eqn:decom-X_t}
		X^s_t = \Lambda(S_t) F^s_t + e^s_t = \underbrace{\Lambda(s) F^s_t + e^s_t}_{\bar X^s_t}  + \underbrace{(\Lambda(S_t) - \Lambda(s))F^s_t}_{\Delta X^s_t} . 
		\end{eqnarray}
		Factors are estimated as eigenvectors from the projected data, i.e.  
		\begin{eqnarray}\label{eqn:factor-est}
		\left(\frac{1}{NT(s)} (X^s)^\T X^s \right) \hat{F}^s = \hat{F}^s V^s_r.
		\end{eqnarray}
		Plugging Equation (\ref{eqn:decom-X_t}) into Equation (\ref{eqn:factor-est}) we obtain
		\begin{eqnarray*}
			\frac{1}{NT(s)} [ F^s \Lambda(s)^\T  \Lambda(s) (F^s)^\T  \hat{F}^s + F^s \Lambda(s)^\T  e^s \hat{F}^s + (e^s)^\T  \Lambda(s) (F^s)^\T  \hat{F}^s  && \\
			+ (e^s)^\T  e^s \hat{F}^s + (\Delta X^s)^\T \bar{X}^s \hat{F}^s + (\bar{X}^s)^\T  \Delta X^s\hat{F}^s  + (\Delta X^s)^\T \Delta X^s \hat{F}^s]   &=& \hat{F}^s V^s_r,
		\end{eqnarray*}
		where $\bar X^s = [\bar X^s_1, \cdots, \bar X^s_T]$ and $\Delta X^s = [\Delta X^s_1, \cdots, \Delta X^s_T]$.
		The three terms, $\frac{1}{NT(s)} (\Delta X^s)^\T \bar{X}^s \hat{F}^s$, $\frac{1}{NT(s)} (\bar{X}^s)^\T  \Delta X^s\hat{F}^s$, $\frac{1}{NT(s)} (\Delta X^s)^\T \Delta X^s \hat{F}^s$ are bias terms from using observations in nearby states.\footnote{Although \cite{Su2017} also uses a PCA estimator with a kernel method to estimate a time-varying factor model, they do not take these bias terms into consideration.} The bias terms are controlled by the Lipschitz condition in Assumption \ref{ass_loading} and the assumptions on the kernel function. We show that the bias terms are negligible for $\delta_{NT, h} h \rightarrow 0$.  
		{In particular, a candidate bandwidth to satisfy the rate assumptions is $h = 1/\sqrt{T}$.}

		\begin{theorem}\label{thm_factor}
			Limiting Distribution of Estimated Factors: \\
			Under Assumptions \ref{Ass:Ident}-\ref{ass_eigen}, if $\sqrt{Nh}/(Th) \rightarrow 0$, $Nh \rightarrow \infty$ and $Nh^2 \rightarrow 0$, we have for the time $t$ conditioned on $S_t=s$:
			\begin{align}
			\sqrt{N}\left( \hat{F}_t - (H^s)^\T  F_t \right) \xrightarrow{d} N(0, (V^s)^{-1}Q^s \Gamma_t^s (Q^s)^\T    (V^s)^{-1}),
			\end{align}
			\sloppy $Q^s$ is the limit $ \frac{(\hat{F}^s)^\T   F^s}{T(s)} \xrightarrow{p} Q^s$ and $V^s = \diag(v_1^s, v_2^s, \cdots, v_r^s)$, $v_1^s > v_2^s > \cdots > v_r^s > 0$ are the eigenvalues of $\Sigma_{\Lambda(s)}^{1/2}\Sigma_{F|s} \Sigma_{\Lambda(s)}^{1/2}$ with $V^s_r \xrightarrow{p} V^s$.
			
		\end{theorem}
		
		This theorem shows the asymptotic normality of the estimated factors $\hat{F}^s_t/K^{1/2}_s(S_t)$ up to the some of rotation of true factors $F_t$ for the times when the state process is equal to the target outcome $s$.\footnote{Instead of conditioning only on the times $t$ when $S_t=s$ we could allow for the times $t$ when $S_t$ satisfies $\frac{1}{\sqrt{N}} \sum_{i = 1}^N \norm{\Lambda_i(S_t)-\Lambda_i(s)}  = o_p(1)$.}
Since $\hat{F}^s_t$ is an estimate of the projected $F^s_t$, dividing both sides by $K^{1/2}_s(S_t)$, $\hat{F}_t = \hat{F}^s_t/K^{1/2}_s(S_t)$ results in an estimate of $F_t$.
A valid bandwidth for the case $N \asymp T$ is $h = 1/T^{1/2+\epsilon}$ for some small $\epsilon >0$.

		Note that the convergence rate $\sqrt{N}$ is the same as in the constant loading factor model (see Theorem 1 in \cite{bai2003inferential}). The variance is equal to that of an ordinary least square regression (OLS) of the panel on the unknown population loadings. 
		
		\begin{theorem} \label{thm_loading}
			Limiting Distribution of Estimated Factor Loadings: \\ 
			Under Assumptions \ref{Ass:Ident}-\ref{ass_eigen}, if $\sqrt{Th}/N \rightarrow 0$, $Th \rightarrow \infty$, and $Th^3 \rightarrow 0$, then for each $i$, 
			\begin{eqnarray}
			\sqrt{Th}(\hat{\Lambda}_i(s) - (H^s)^{-1} \Lambda_i(s)) \xrightarrow{d} N(0, ((Q^s)^\top)^{-1} \Phi^s_i (Q^s)^{-1}).
			\end{eqnarray}
		\end{theorem}

		This theorem shows the asymptotic normality of the estimated conditional loadings up to some rotation. $\hat{\Lambda}_i(s) - (H^s)^{-1} \Lambda_i(s)$ has some error and bias terms, including a leading bias term for the time average of $\Delta X^s_i = [(\Lambda_i(S_t) - \Lambda_i(s))F^s_t]_{t = 1, 2, \cdots, T}$. We show in the appendix that $\frac{1}{T(s)} \sum_{t=1}^T \Delta X^s_{it} = O_p(h)$. Therefore, when $Th^3 \rightarrow 0$, the bias terms are sufficiently small relative to the error term. A candidate bandwidth to satisfy the assumptions is $h = 1/\sqrt{T}$ when $\sqrt[4]{T}/N \rightarrow 0$.
		
		As expected the convergence rate is $\sqrt{Th}$, which is slower than the convergence rate $\sqrt{T}$ in the constant loading factor model (see Theorem 2 in \cite{bai2003inferential}). The variance is equal to an OLS regression of the projected data on the projected unknown population factors. The smaller the bandwidth $h$, the slower the convergence rate and the smaller the bias. The variance of $\hat{\Lambda}_i(s) - (H^s)^{-1} \Lambda_i(s)$ is $O_p\left(\frac{1}{Th} \right)$ and the bias of $\hat{\Lambda}_i(s) - (H^s)^{-1} \Lambda_i(s)$ is $O_p(h)$. The optimal bandwidth to balance variance and bias and satisfy the assumptions in the asymptotic distribution is $h \asymp 1/\sqrt[3+\epsilon]{T}$ for some small $\epsilon > 0$.

		We denote the common component by $C_{it,s} = F_t^\top \Lambda_i(s)$ and its estimator by $\hat{C}_{it,s} = \hat{F}_t \hat{\Lambda}_i(s) = \left(\frac{\hat{F}^s_t}{K^{1/2}_s(S_t)}\right)^\T  \hat{\Lambda}_i(s)$.
		\begin{theorem}\label{thm_common}
			Limiting Distribution of Common Components: \\ 
			Under Assumptions \ref{Ass:Ident}-\ref{ass_eigen} as $Nh \rightarrow \infty$,
			$Th \rightarrow \infty$, $Nh^2 \rightarrow 0$ and $Th^3 \rightarrow 0$, we have for each $i$ and the time $t$ conditioned on $S_t=s$: 
			\begin{gather}
			\left(\frac{1}{N} V_{it,s} + \frac{1}{Th} W_{it,s}\right)^{-1/2} \left( \hat{C}_{it,s} - C_{it,s}\right) \xrightarrow{d} N(0, 1),
			\end{gather}
			where $V_{it,s} = \Lambda_i(s)^\T   \Sigma_{\Lambda(s)}^{-1} \Gamma_{t}^s \Sigma_{\Lambda(s)}^{-1}  \Lambda_i(s)$ and $W_{it,s} = F_t^\T  \Sigma_{F|s}^{-1} \Phi_{i} ^s\Sigma_{F|s}^{-1} F_t$. 
		\end{theorem}
		
		The estimated common components converge to an asymptotic normal distribution that combines the results of the previous two theorems. 
		Note that the systematic part is identified without a rotation. The variance in the asymptotic distribution is determined by two components, factor and loading distributions. The first component $V_{it,s}$ is from the asymptotic distribution of estimated factors $\hat{F}^s_tK^{-1/2}_s(S_t)$. The second component $W_{it,s}$ comes from the asymptotic distribution of estimated loadings $\hat{\Lambda}_i(s)$. It depends on the relationship between $N$ and $Th$, which one dominates.  If $N$ and $Th$ have similar scales, both $V_{it,s}$, and $W_{it,s}$ play a role in the variance of the asymptotic distribution. However, if $Th/N=o(1)$, the asymptotic distribution of the loadings dominates (which allows us to drop the additional assumption on the times $t$), while if $N/(Th)=o(1)$, the factor distribution dominates.

		Lemma \ref{lemma:estimator-thm2-4} in the Internet Appendix provides consistent estimators for the asymptotic covariance matrices in Theorems \ref{thm_factor} to \ref{thm_common}. Our feasible estimators allow for a sparse correlation and autocorrelation structure for the residual terms.

		\section{Generalized Correlation Test for Change in Loadings} \label{test}

		We derive a test statistic to detect if and for which states loadings are different. This is distinct from a ``global'' test if loadings change at some time without guidance when the change actually happens. We provide an answer to the relevant economic question for which specific times and states loadings are different.\footnote{ \cite{Su2017} provide a ``global'' test for the constancy of factor loadings over time. Similar arguments could be applied to our framework. The proof would go through with some modification about controlling the bias from using data in other states. In a similar spirit, \cite{kong2018systematic} provide a global test in a high-frequency setup.} \cite{pelger2019} illustrates in an empirical study that it is important to identify when and how time-varying loadings change as this can provide valuable economic insights.  
		
		Since the loadings can be estimated up to some rotation matrix, the test statistic needs to be invariant to invertible linear transformations. A candidate measure is the total generalized correlation, which measures how close the two vector spaces spanned by loading vectors in two states are. The total generalized correlation ranges from 0 to the number of factors $r$. 0 means that two spaces are orthogonal, while $r$ represents that two spaces are the same. 
		
		It is worth noting that it is insufficient to test if the loading vectors for individual factors are different in different states. For example, it is possible that the first factor explains less variation in another state and switches with the second factor. In this case, measuring the correlation of the loadings of the first factor for different state outcomes would indicate a change in loadings, while the factor structure itself actually does not change. Thus, it is crucial to study the harder problem if the span of all factor loadings changes with the state.
		
		We consider the two state outcomes $s_1$ and $s_2$ with the corresponding loadings $\Lambda(s_1)$ and $\Lambda(s_2)$. Note that our state process $S_t$ still has a continuous support. Testing the constancy of factor loadings for the particular state realizations $s_1$ and $s_2$ is equivalent to testing whether there exists an invertible matrix $G$, such that $\Lambda(s_1) = \Lambda(s_2)G$. We use a slightly modified estimator for the loadings and estimators that will simplify the notation. Instead of normalizing the projected factors to be orthonormal, we apply this normalization to the loadings. This means we use $\bar{\Lambda}(s_l) = \hat{\Lambda}(s_l) (V^{s_l}_r)^{-1/2}$ and $\bar{F}^{s_l} = \hat{F}^{s_l} (V^{s_l}_r)^{1/2}$. All results are valid for the modified estimator under the same assumptions as for the estimators introduced in the previous section. $\bar{F}^{s_l}$ has the same asymptotic distribution as $\hat{F}^{s_l}$ except that it replaces the asymptotic variance by that of $\hat{F}^{s_l}$ multiplied by $(V^{s_l}_r)^{-1/2}$ on the left and and on the right. Similarly, $\bar{\Lambda}(s_l)$ has the same asymptotic distribution as $\hat{\Lambda}(s_l)$, except that the asymptotic variance is multiplied $(V^{s_l}_r)^{1/2}$ on the left and on the right.\footnote{In order to study $\bar{\Lambda}(s_1)$ and $\bar{\Lambda}(s_2)$, we need to redefine $H^{s_l} = \frac{(F^{s_l})^\T  F^{s_l}}{T(s_l)} \frac{\Lambda(s_l)^\T \bar{\Lambda}(s_l)}{N} (V^{s_l}_r)^{-1}$, $H^{s_l} \xrightarrow{p} (Q^{s_l})^{-1}$, where $Q^{s_l} = V^{s_l} (\Upsilon^{s_l})^\T \Sigma_{F|s_l}^{-1/2}$, and $V^{s_l}$ \footnote{$V^{s_l}$ here is the same as the $V^{s_l}$ in theorem \ref{ass_factor}, since the eigenvalues of $\Sigma_{\Lambda(s)}^{1/2}\Sigma_{F|s} \Sigma_{\Lambda(s)}^{1/2}$ are the same as those of $\Sigma_{F|s_l}^{1/2}\Sigma_{\Lambda(s_l)}\Sigma_{F|s_l}^{1/2}$} are eigenvalues of $\Sigma_{F|s_l}^{1/2}\Sigma_{\Lambda(s_l)}\Sigma_{F|s_l}^{1/2}$, $\Upsilon^{s_l}$ is the corresponding eigenvector matrix such that $(\Upsilon^{s_l})^{T}\Upsilon^{s_l} = I_r$. Under the same assumption as Theorem \ref{ass_loading}, the asymptotic distribution of $\bar{\Lambda}_i(s_l)$ is $\sqrt{Th}\left(\bar{\Lambda}_i(s_l) - (H^{s_l})^\T  \Lambda_i(s_l) \right) \xrightarrow{d} N(0, (V^{s_l})^{-1} Q^{s_l} \Phi^s_i (Q^{s_l})^{T} (V^{s_l})^{-1})$, where $\Phi^s_i$ is the same as the $\Phi^s_i$ in Theorem \ref{ass_loading}. Let $\lambda_{li} = \Lambda_i(s_l) $ and $v_{li} = (H^{s_l})^\T  \frac{\sqrt{Th}}{T(s_l)} \left(\frac{1}{N} \sum_{k = 1}^N \lambda_{lk} \lambda^\T_{lk} \right) \left((F^{s_l})^\T  e_i^{s_l} \right)$, then we have $\sqrt{Th}\left(\bar{\Lambda}_i(s_l) - (H^{s_l})^\T  \Lambda_i(s_l) \right)  = v_{li} + o_p(1)$ under the same assumptions as in Theorem \ref{ass_loading}.}

		The generalized correlation test statistic requires some mildly stronger assumptions.
		
		\begin{assumption} \label{doublesum} Moments and Central Limit Theorem: There exists an $M\leq \infty$, such that $\forall$ $k$ and $i$, for any $l, l' = 1, 2$
			
			\begin{enumerate}
				\item Double-sum factors, loadings and projected idiosyncratic components in two states: \label{veeminueeelam}\\
				$E\norm{\frac{Th}{NT^2(s_l)} \sum_{i=1}^N \sum_{k=1}^N (F^{s_l})^\T  e_k^{s_l} \lambda^\T_{l'i}  \sum_{t=1}^T [e^{s_l}_{it}e^{s_l}_{kt} - \+E(e^{s_l}_{it}e^{s_l}_{kt})]  }^2 \leq M.$ 
				\item Double-sum loadings and projected idiosyncratic components in two states:  \label{lameeminuslam} 
				
				$E\norm{\frac{\sqrt{Th}}{NT(s_l)} \sum_{i=1}^N \sum_{k=1}^N \lambda_{li} \lambda^\T_{l'i}  \sum_{t=1}^T [e^{s_l}_{it}e^{s_l}_{kt} - \+E(e^{s_l}_{it}e^{s_l}_{kt})]  }^2 \leq M.$ 
				\item Projected factors, loadings and idiosyncratic components in two states:  \label{vlamfelam}
				
				$ \+E \norm{\frac{\sqrt{Th}}{\sqrt{N} T(s_l)}\sum_{i=1}^N (F^{s_l})^\T  e_i^{s_l}\lambda^\T_{l'i} }^2 \leq M.$
				\item Define $\mu_{l,l'} = \frac{1}{N T(s_l)}\sum_{i=1}^N \sum_{j=1}^{T} K_{s_l}(S_t) F_t e_{it}\lambda^\T_{l'i} $ and let $B = \begin{bmatrix}
				\tvec\left( \mu_{1,1} \right) \\ \tvec\left( \mu_{1,2} \right) \\ \tvec\left( \mu_{2,1} \right) \\ \tvec\left( \mu_{2,2} \right) 
				\end{bmatrix} $. It holds\footnote{Here we denote by $\tvec(.)$ the vectorization operator. Inevitably the matrix $\Sigma_{B,B}$ is singular due to the symmetric nature of the covariance and a proper formulation uses vech operators and elimination matrices.}
				\begin{eqnarray}
				\sqrt{NTh} (B-0) \xrightarrow{d} N(0, \Sigma_{B, B}).
				\end{eqnarray} 
				\label{CLT}
			\end{enumerate}
		\end{assumption}
		\vspace{-0.8cm}
		
		Assumption \ref{doublesum} is closely related to Assumption \ref{ass_mom}, but Assumption \ref{doublesum} involves loadings in two states, $s_l$ and $s_{l'}$. Assumptions \ref{doublesum}.\ref{veeminueeelam} and \ref{doublesum}.\ref{lameeminuslam} are similar to Assumption \ref{ass_mom}.5, but these two assumptions are averaged twice in the cross-sectional dimension. 
		Assumption \ref{doublesum}.\ref{vlamfelam} generalizes Assumption \ref{ass_mom}.2 and it is identical to Assumption \ref{ass_mom}.2 when $l = l'$. 
		Assumption \ref{doublesum}.\ref{CLT} is a joint central limit theorem for the cross-sectional and time-series average of the residuals. 
		
		In order to simplify notation, we denote $\Lambda_l = \Lambda(s_l)$ and $\bar{\Lambda}_l = \bar{\Lambda}(s_l)$. We define the estimated total generalized correlation as 
		\begin{align*} \hat{\rho} = \tr \left\lbrace \left( \frac{1}{N}\bar{\Lambda}^\T_1\bar{\Lambda}_1 \right)^{-1} \left( \frac{1}{N}\bar{\Lambda}^\T_1\bar{\Lambda}_2 \right) \left( \frac{1}{N}\bar{\Lambda}^\T_2\bar{\Lambda}_2 \right)^{-1} \left( \frac{1}{N}\bar{\Lambda}^\T_2\bar{\Lambda}_1 \right) \right\rbrace 
		\end{align*}
		and the population counterpart as $\rho = \tr \left\lbrace \left( \frac{1}{N}\Lambda^\T_1 \Lambda_1 \right)^{-1} \left( \frac{1}{N}\Lambda^\T_1 \Lambda_2 \right) \left( \frac{1}{N}\Lambda^\T_2 \Lambda_2 \right)^{-1} \left( \frac{1}{N} \Lambda^\T_2 \Lambda_1 \right) \right\rbrace$.

		Testing if $\Lambda_1$ is some linear rotation of $\Lambda_2$ is equivalent to 
		\begin{eqnarray*}
			\mathcal{H}_0&:& \Lambda_1 =  \Lambda_2 G \text{ for some full rank square matrix } G \\
			\mathcal{H}_1&:& \Lambda_1 \neq  \Lambda_2 G \text{ for any square matrix } G \in \mathbbm R^{r \times r}. 
		\end{eqnarray*}
		If we multiple any full rank square matrix $G$ to the right of $\bar{\Lambda}_1$ or $\bar{\Lambda}_2$, $\hat{\rho}$ does not change and the same holds for $\rho$. Note that if $\Lambda_1 =  \Lambda_2 G$, then it holds $\rho = \tr(I_r) = r$, where $I_r\in \mathbb{R}^{r \times r}$ is an identity matrix.  Hence, it is equivalent to test\footnote{Here we use the following result: \begin{lemma} \label{lemma_leqr}
				Let $\Lambda_1 \in \mathbb{R}^{N \times k_1}$ and $\Lambda_2 \in \mathbb{R}^{N \times k_2}$. Assume $N \geq \max(k_1, k_2)$, $rank(\Lambda_1) = k_1$ and $rank(\Lambda_2) = k_2$, let $k = \min(k_1, k_2)$, then we have $\rho \leq k$.
		\end{lemma}} 
		\begin{align*}
			\mathcal{H}_0: \rho = r \qquad
			\mathcal{H}_1: \rho < r.
		\end{align*}
		

		
		Theorem \ref{thm_rho} provides the inferential statistic for a one-sided test of the null hypothesis $\rho=r$.\footnote{Note that our generalized correlation test statistic would also work when the dimensions of the loading spaces change with the state.}
		
		\begin{theorem} \label{thm_rho}
			Under Assumptions \ref{Ass:Ident}-\ref{doublesum} and under the null hypothesis $\rho = r$, if $Nh \rightarrow \infty$, $Th \rightarrow \infty$, $\sqrt{N}/(Th) \rightarrow 0$, $\sqrt{Th}/N \rightarrow 0$,  $Nh^2 \rightarrow 0$ and  $NTh^3 \rightarrow 0$, then 
			\begin{eqnarray} \label{gen_corr_thm}
			\sqrt{NTh} (\hat{\rho} - r - \xi^\T  b) \xrightarrow{d} N(0, \xi^\T  D \Sigma_{B, B} D^\T \xi).
			\end{eqnarray}
			The matrix $D$ and a consistent plug-in estimator $\hat D$ are given in the Internet Appendix. $\xi^\T  b$ is a bias correction term. 
			
			Let $\Sigma_{e_T} = \+E[e^\T  e/N]$ and $\Sigma_{e_N} = \+E[e e^\T /T]$. Assume there are only finitely many non-zero elements in each row of $\Sigma_{e_T}$ and $\Sigma_{e_N}$  and we know the sets $\Omega_{e_T}$ and $\Omega_{e_N}$ of nonzero indices. 
			A consistent estimator of the bias correction term is \\
			$\hat b = \begin{bmatrix}
			\tvec\left( \hat x_{1, 1} + \hat y_{1, 1} \right) \\ \tvec\left( \hat x_{1, 2} + \hat y_{1, 2} \right) \\ \tvec\left(\hat x_{2, 1} + \hat y_{2, 1}\right) \\ \tvec\left(  \hat x_{2, 2} + \hat y_{2, 2} \right)
			\end{bmatrix}$ and $\hat \xi = \begin{bmatrix}
			\tvec\left( -(\hat G_1^{-1} \hat G_2 \hat G_4^{-1} \hat G_3 \hat G_1^{-1})^{\top} \right) \\ \tvec\left( \hat G_1^{-1} \hat G_2 \hat G_4^{-1} \right) \\ \tvec\left( \hat G_4^{-1} \hat G_3 \hat G_1^{-1} \right) \\ \tvec\left( -(\hat G_4^{-1} \hat G_3 \hat G_1^{-1} \hat G_2 \hat G_4^{-1} )^{\top} \right) \end{bmatrix},$   \\
			where $\hat G_1= \frac{1}{N}\bar \Lambda^{\top}_1 \bar \Lambda_1, \hat G_2= \frac{1}{N} \bar \Lambda_1^{\top} \bar \Lambda_2, \hat G_3= \frac{1}{N} \bar \Lambda_2^{\top} \bar \Lambda_1, \hat G_4= \frac{1}{N}\bar \Lambda_2^{\top} \bar \Lambda_2$ and $\hat x_{l,l'} = \hat x_{l,l',l,l'} + \hat x_{l,l,l,l'} + \hat x_{l,l',l',l'}$ and $\hat y_{l,l'}=\hat z_{l,l'}+\hat z_{l',l}$ with components
			\begin{align*}
			\hat x_{u,v,p,w} =& (\bar{V}_r^{s_p})^{-1}  \left(\frac{1}{N}\sum_{i=1}^N   \bar \lambda_{pi} \bar \Lambda^\T_{ui}  \right) \left( \frac{1}{N T(s_u) T(s_{v})} \sum_{(t_1, t_2) \in  \Omega_{e_T}} \bar{F}^{s_u}_{t_1} (\bar{F}^{s_{v}}_{t_2})^\T  (\bar{e}_{t_1}^{s_u})^\T  \bar{e}_{t_2}^{s_{v}} \right) \\
			& \left(\frac{1}{N}\sum_{i=1}^N   \bar \lambda_{vi} \bar \Lambda^\T_{wi}  \right)(\bar{V}_r^{s_{w}})^{-1},\\
			\hat z_{p,w} =& (\bar{V}_r^{s_p})^{-1} \frac{1}{N^2 T(s_p)} \sum_{(i,j) \in \Omega_{e_N}} \bar \lambda_{pi}  (\bar{\und{e}}^{s_p}_i)^\T \bar{\und{e}}^{s_p}_j  \bar \Lambda^\T_{wj},
			\end{align*}
			where $\bar e^{s_l}_t = X^{s_l}_t - \bar \Lambda(s_l)\bar F^{s_l}_t$ and $\bar{\und{e}}^{s_l}_i = X^{s_l}_i - \bar F^{s_l} \bar \lambda_{li}$. The feasible test statistic
			\begin{align*}
			\sqrt{NTh} \frac{(\hat{\rho} - r - \hat \xi^\T  \hat b)}{\sqrt{\hat \xi^\T  \hat D \hat \Sigma_{B, B} \hat D^\T \hat \xi} } 
			\end{align*}
			is asymptotically $N(0,1)$ distributed under $\mathcal{H}_0$ and diverges to $-\infty$ with probability 1 under $\mathcal{H}_1$.
		\end{theorem}

		There are two surprising results. First, the test statistic for the null hypothesis $\rho=r$ is super-consistent, i.e. converges at the higher rate $\sqrt{NTh}$. Under the assumption $\rho <r$, a simple delta-method argument applied to the trace shows that the convergence rate is slower at $\sqrt{N}$ as stated in Lemma \ref{lemma:rho} in the Internet Appendix.
		Second, the special case of $\rho=r$ requires a bias correction in contrast to $\rho <r$ where the bias can be ignored. The bias arises because the higher rate of convergence does not allow us to ignore certain higher-order terms in the asymptotic expansion of $\hat \rho$.  
		Note that by construction (see Lemma \ref{lemma_leqr}), we have $\hat \rho  \leq r$.  Theorem \ref{thm_rho} shows that under the null hypothesis, $\hat{\rho}$ is distributed asymptotically normal around $r + \xi^\T  b $ which implies that the bias term is negative. 
		
		Let $h = 1/T^{1/2 + \varepsilon}$. All rate conditions in Theorem \ref{thm_rho} can be reduced to $N/T^{1/2 + \varepsilon} \rightarrow \infty$, $N/T^{1-2 \varepsilon}\rightarrow 0$, $N/T^{1/2 + 3\varepsilon} \rightarrow 0$. If $0 < \varepsilon < 1/6$ (equivalent to $1/T^{1/2} < h < 1/T^{3/4}$), there exists combinations of $N$ and $T$ that satisfy the rate conditions. For example, if $\varepsilon = 1/8$, then the rate conditions can be reduced to $N/T^{3/4} \rightarrow 0$ and $T^{5/8}/N \rightarrow 0$. The rate conditions are more stringent than Theorem \ref{thm_consistency}-\ref{thm_common}, because $\hat \rho$ converges at the faster rate $\sqrt{NTh}$. The strong condition $\sqrt{NTh} \cdot h \rightarrow 0$, equivalent to $NTh^3 \rightarrow 0$ is needed to neglect the bias term.
Simulations suggest that the distribution result is still a good approximation even if the rate conditions are not satisfied.

		In order to obtain a consistent estimator of the bias term, we assume that the residual covariance matrix is sparse similar to \cite{fan2013large}. Our sparsity assumption imposes that there are only finitely many nonzero elements in each row of the covariance matrix of the errors $\Sigma_{e_N} = \+E[e e^\T /T]$ and similarly in the autocovariance matrix $\Sigma_{e_T} = \+E[e^\T  e/N]$. For simplicity, we assume that we know the set of nonzero indices. This assumption could be relaxed, and we could estimate the nonzero elements with a thresholding approach similar to \cite{fan2013large} under additional assumptions.

		\section{Simulation} \label{simulation}
		We study the finite sample properties of our estimators with Monte-Carlo simulations. First, we show that the simulated distributions of the estimated loadings, factors, and common components converge to the asymptotic distributions. Second, we show that the functional form of the loadings as a function of the state can be reliably recovered. Third, we verify the good size and power properties of our test statistic. Fourth, we test the performance of our estimator for a misspecified state process. The Internet Appendix contains a validation study for selecting the number of factors and bandwidth and shows the good performance of our estimator relative to existing estimation approaches based on structural breaks or local PCA estimation.
		
		\subsection{Asymptotic Distribution Theory of Estimators} \label{asy_est}
		In the baseline model, we generate data from a one-factor model $X_{it} = \Lambda_i(S_t) F_t + e_{it}$, where $F_t \sim N(0,1)$. The state process is an Ornstein-Uhlenbeck (OU) process which is a mean-reverting process with stationary distribution. In more detail, we simulate the state process as $S_t = \theta (\mu - S_t) d_t + \sigma dW_t$, where $\theta =1$, $\mu = 0.2$, and $\sigma =1$ and its stationary distribution has mean $\mu = 0.2$ and variance $\sigma^2/(2\theta) = 1/2$. The OU process is popular for modeling stochastic volatility in financial data, which is aligned with the volatility index as state process in our empirical applications. The loadings are cubic functions of the state process,   $\Lambda_i(S_t) = \Lambda_{0i} + \frac{1}{2} S_t \Lambda_{1i} + \frac{1}{4} S_t^2 \Lambda_{2i} + \frac{1}{8} S_t^3 \Lambda_{3i}$, where  $\Lambda_{0i}, \Lambda_{1i}, \Lambda_{2i}, \Lambda_{3i} \sim N(0,1)$. 
		The functional form of the loading function is motivated by our empirical findings. The loadings as a function of volatility change non-linearly, and the changes are larger for state values that deviate more from its mean. The coefficients in the cubic, quadratic and linear terms are chosen to guarantee that loadings will not be completely dominated by the state realizations with the largest absolute values, which is again in line with our empirical findings.
		We generate three different idiosyncratic processes: (1) i.i.d. $e_{it} \sim \text{IID } N(0,1)$, (2) heteroskedastic $e_{it} = \sigma_i v_{it}, \sigma_i \sim \text{IID } U(0.5, 1.5), v_{it} \sim N(0,1)$ and (3) cross sectional dependent $e_t \sim N(0, \Sigma_e)$, $\Sigma_e = (c_{ij})_{i,j = 1, 2, \cdots, N}$ with $c_{ij} = 0.5^{|i-j|}$.

%
%

		\begin{figure}[t!]
			\tcapfig{Histograms of Estimated Loadings}
			\centering
			\begin{subfigure}{.25\textwidth}
				\centering
				\includegraphics[width=1\linewidth]{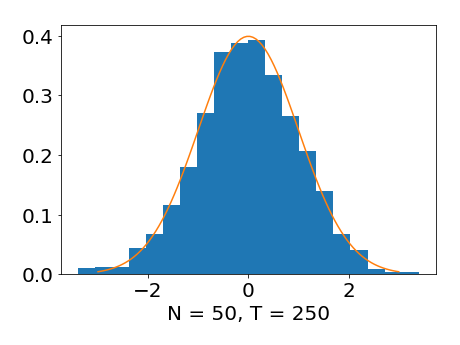}
			\end{subfigure}%
			\begin{subfigure}{.25\textwidth}
				\centering
				\includegraphics[width=1\linewidth]{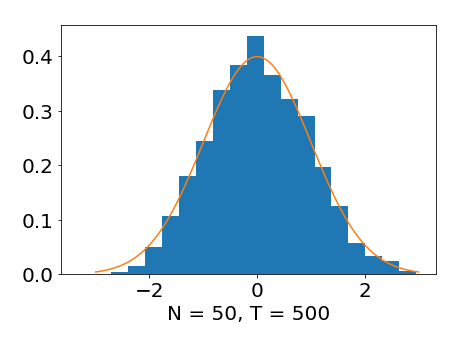}
			\end{subfigure}%
			\begin{subfigure}{.25\textwidth}
				\centering
				\includegraphics[width=1\linewidth]{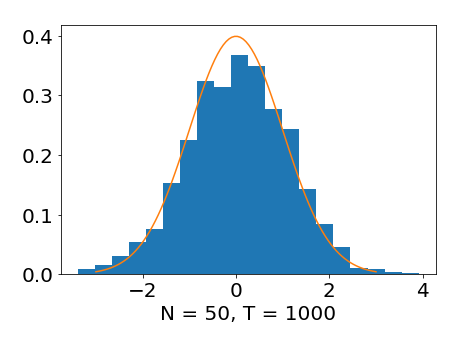}
			\end{subfigure}
			\begin{subfigure}{.25\textwidth}
				\centering
				\includegraphics[width=1\linewidth]{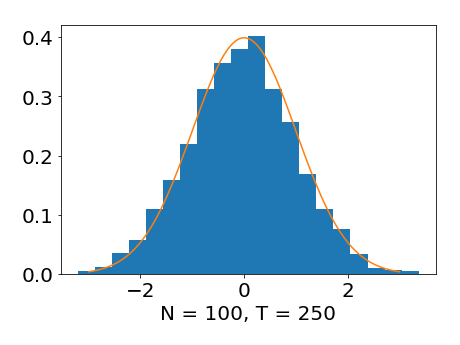}
			\end{subfigure}%
			\begin{subfigure}{.25\textwidth}
				\centering
				\includegraphics[width=1\linewidth]{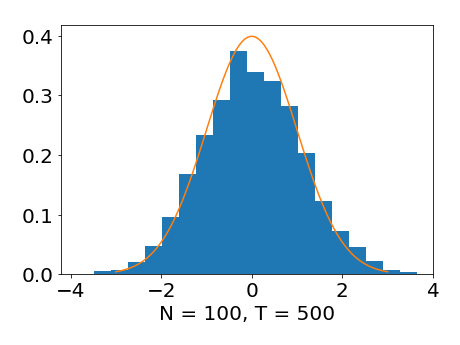}
			\end{subfigure}%
			\begin{subfigure}{.25\textwidth}
				\centering
				\includegraphics[width=1\linewidth]{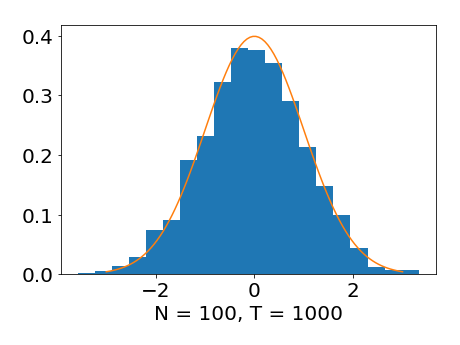}
			\end{subfigure}
			\begin{subfigure}{.25\textwidth}
				\centering
				\includegraphics[width=1\linewidth]{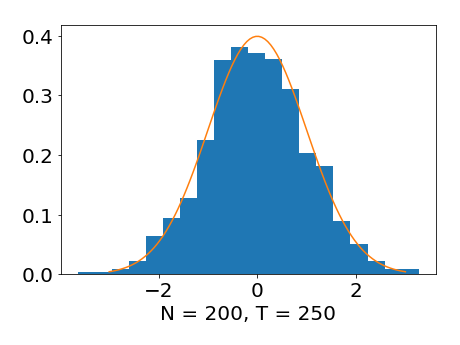}
			\end{subfigure}%
			\begin{subfigure}{.25\textwidth}
				\centering
				\includegraphics[width=1\linewidth]{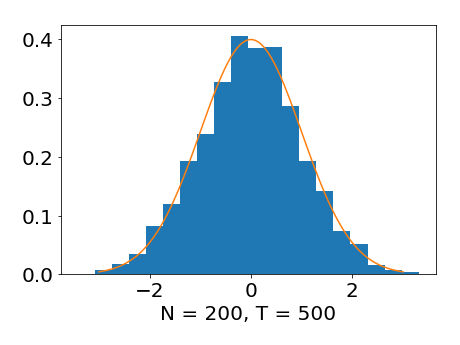}
			\end{subfigure}%
			\begin{subfigure}{.25\textwidth}
				\centering
				\includegraphics[width=1\linewidth]{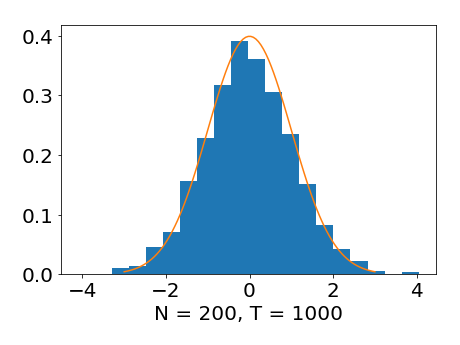}
			\end{subfigure}
			\bnotefig{Histograms of estimated loadings ($N=50, 100, 200$; $T=250, 500, 1000$; $h=0.3$) for i.i.d. errors. The normal density function is superimposed on the histograms. The histograms are based on 2,000 Monte Carlo simulations.}
			\label{hist_loading}
		\end{figure}

		Figure \ref{hist_loading} shows histograms of the standardized estimated loadings for different $N$ and $T$. The estimates are centered and standardized using consistent estimates of the theoretical mean and standard deviation. We set the state outcome to $s = 0.5$ and bandwidth to $h=0.3$ to balance the bias and variance inherited in the nonparametric method.\footnote{The squared error of the nonparametric method is $O_p\left( \max \left(  \frac{1}{N}, \frac{1}{Th}, h^2\right)  \right)$. In order for the results in Section 5 and 6 to hold, we have $Nh \rightarrow \infty$, $Th \rightarrow \infty$, $Nh^2 \rightarrow 0$ and $Th^3 \rightarrow 0$. This gives us a guideline for selecting the bandwidth $h$ in the simulation and empirical studies and suggests range of 0.1 to 0.5. The Internet Appendix collects the results for various bandwidth and shows that our findings are robust to the choice of $h$.}
	The Internet Appendix collects the results for the estimated factors and common components and includes the cases of heteroskedastic and cross-sectionally dependent errors. The results are virtually identical, and we find that the simulated data is very well approximated by the theoretically implied normal distribution. Thus, our results are robust to heteroskedastic or cross-sectionally dependent errors.

		
		We can estimate well the functional form of the loadings depending on the state. Figure \ref{loading_state} compares the estimated functional form with the true functional form of the loadings of four randomly selected cross-section units. The factor model is estimated in every possible state between -3 and 3. The estimated functional form of the loadings matches the true functional form very well.\footnote{In Figure \ref{fig:loadingSW} in the Internet Appendix, we compare the estimation results of our state-varying factor model with the local time-varying model of \cite{Su2017} under the same simulation setup. Our state-varying factor model can recover the correct functional form while the local window estimator fails.}

		\begin{figure}[t!]
			\tcapfig{Estimated Functional Form of Loading versus the State Variable}
			\centering
			\begin{subfigure}{.3\textwidth}
				\centering
				\includegraphics[width=1\linewidth]{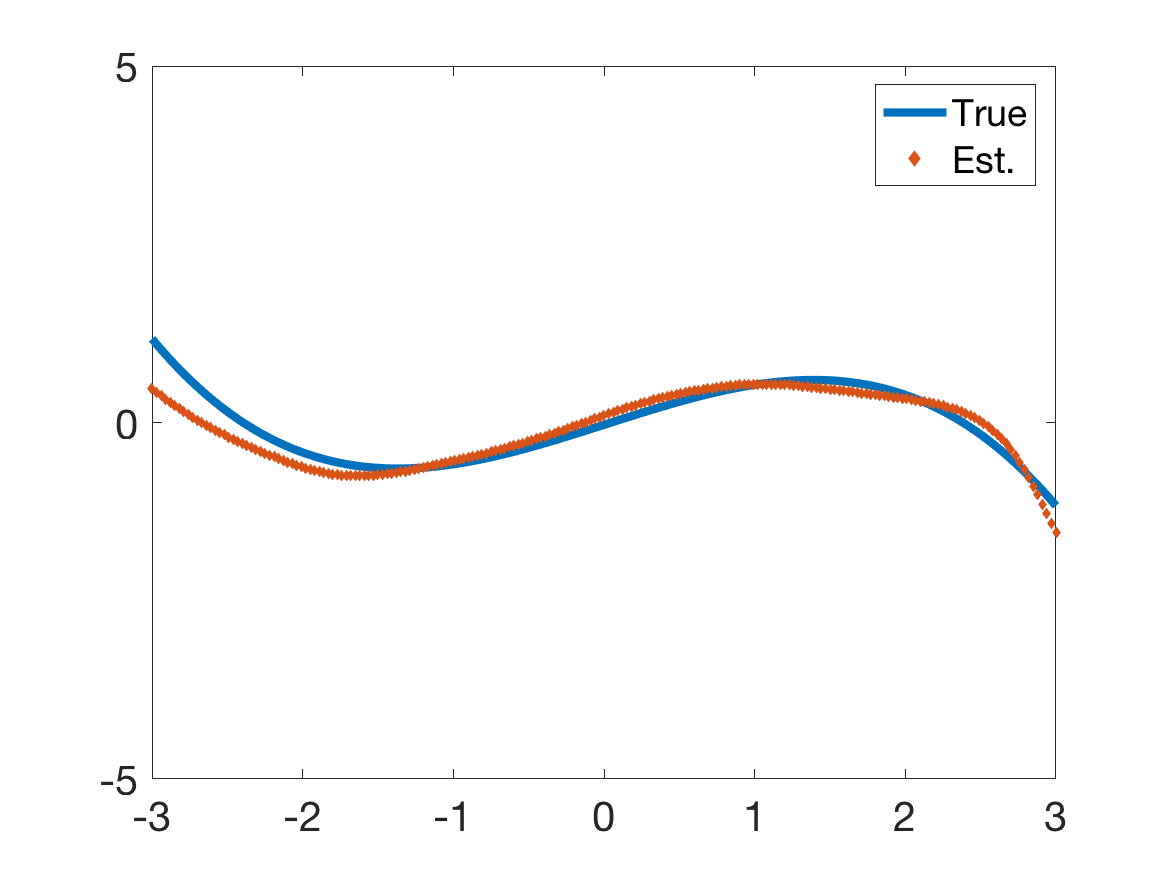}
			\end{subfigure}%
			\begin{subfigure}{.3\textwidth}
				\centering
				\includegraphics[width=1\linewidth]{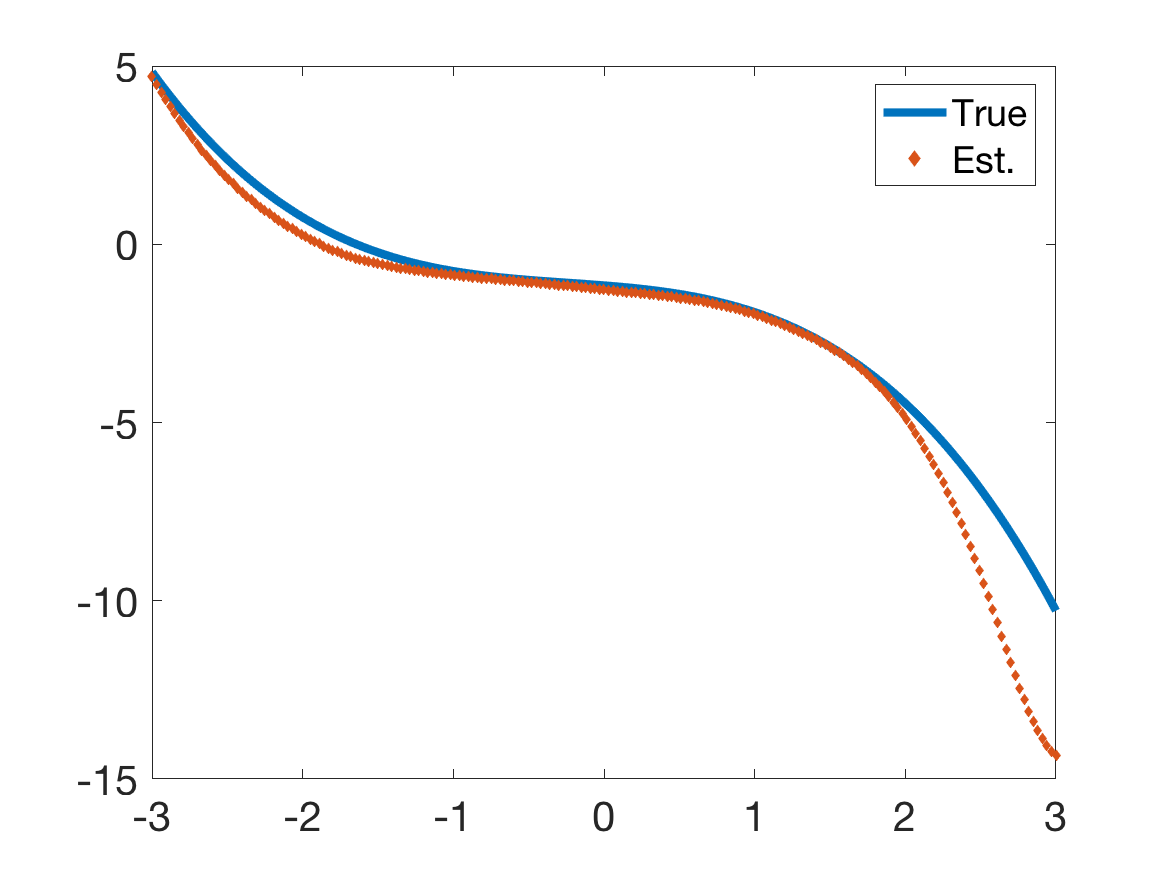}
			\end{subfigure}
			\begin{subfigure}{.3\textwidth}
				\centering
				\includegraphics[width=1\linewidth]{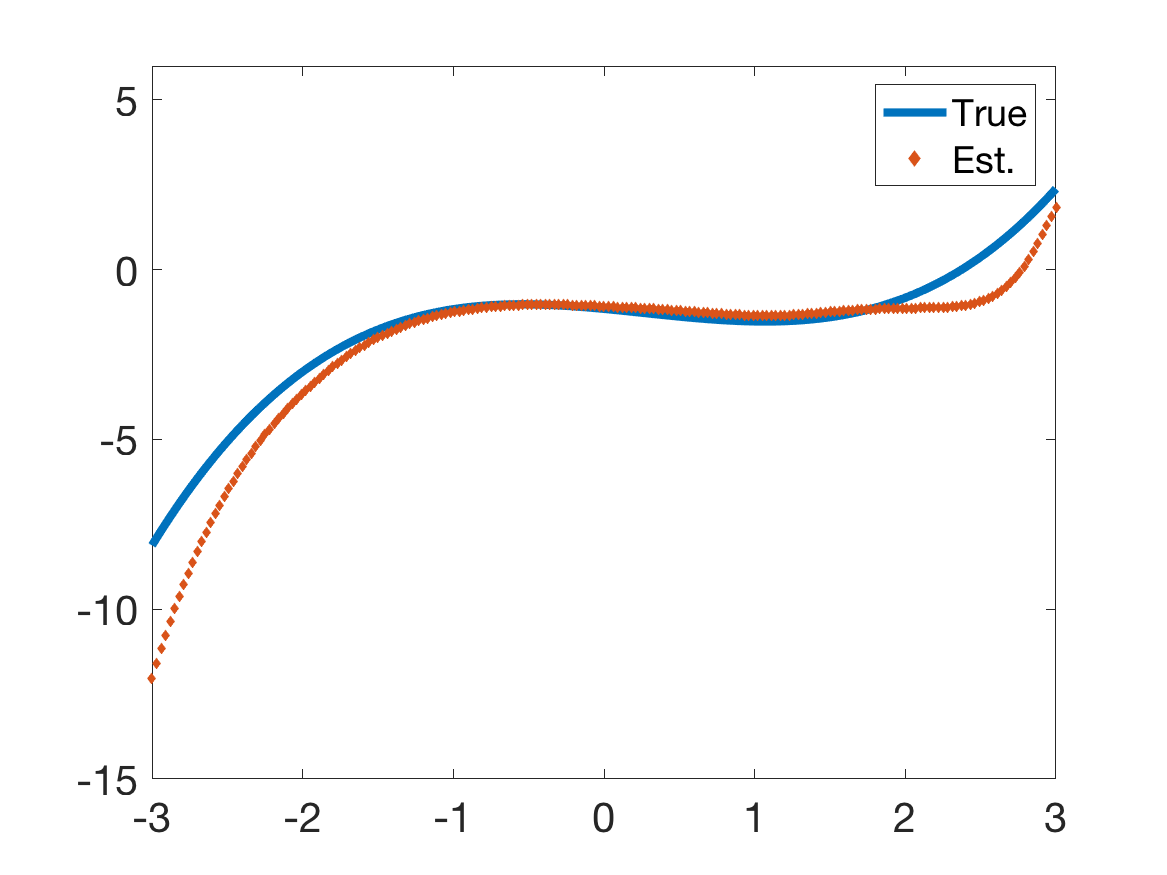}
			\end{subfigure}%
			\begin{subfigure}{.3\textwidth}
				\centering
				\includegraphics[width=1\linewidth]{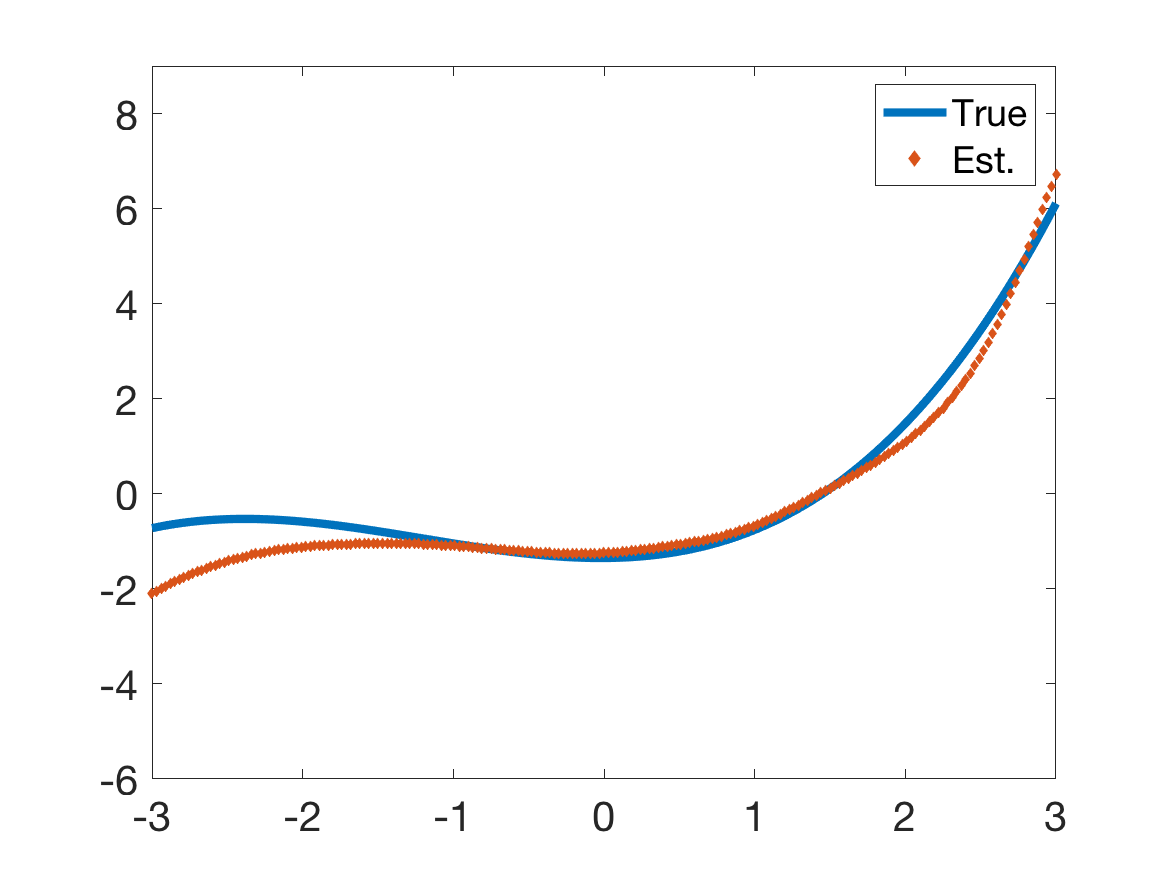}
			\end{subfigure}
			\bnotefig{Estimated functional form of loading versus the state variable ($N = 100, T = 500, h = 0.5$). The true functional form is superimposed on the estimated function.}
			\label{loading_state}
		\end{figure}

		\subsection{Generalized Correlation Test}
		
		The data generating process is similar to the data generating process in Section \ref{asy_est}, except that we use constant loadings to generate the data for all states. Figure \ref{hist_rho} is generated by keeping the realization of the single factor, loadings, and state fixed and simulating the i.i.d. errors. The histograms for heteroscedastic errors and cross-sectionally dependent errors are in the Internet Appendix. Without loss of generality, we select the two-state outcomes $s_1 = 0.4$ and $s_2 = 0.6$ to calculate the generalized correlation $\hat{\rho}$. We compare the empirical distribution of $\hat{\rho}$ standardized by the consistent estimators of its theoretical mean and deviation with a standard normal distribution. Figure \ref{hist_rho} shows that the standardized generalized correlation is very well approximated by a normal distribution.\footnote{Although we correct for the bias, the empirical distribution is still slightly shifted to the left. Our bias correction term only takes into account the dominant bias term. We believe that correcting for higher-order bias terms can correct the remaining minor bias. Note that the remaining minor bias makes our test statistic more conservative, i.e., we are more likely to reject the null hypothesis.}

		\begin{figure}[t!]
			\tcapfig{Histograms of Generalized Correlation Test Statistic}
			\centering
			\begin{subfigure}{.25\textwidth}
				\centering
				\includegraphics[width=1\linewidth]{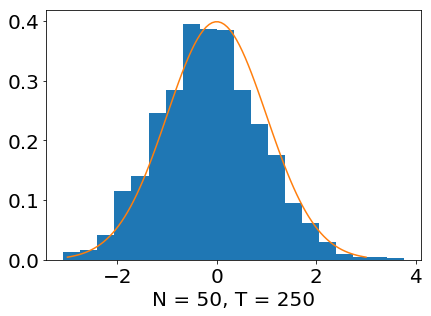}
			\end{subfigure}%
			\begin{subfigure}{.25\textwidth}
				\centering
				\includegraphics[width=1\linewidth]{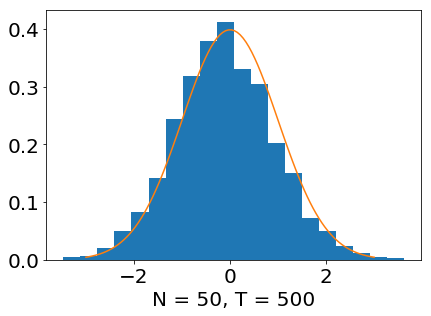}
			\end{subfigure}%
			\begin{subfigure}{.25\textwidth}
				\centering
				\includegraphics[width=1\linewidth]{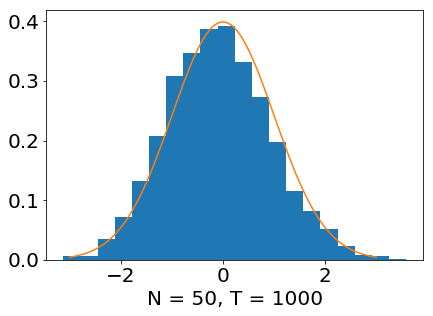}
			\end{subfigure}
			\begin{subfigure}{.25\textwidth}
				\centering
				\includegraphics[width=1\linewidth]{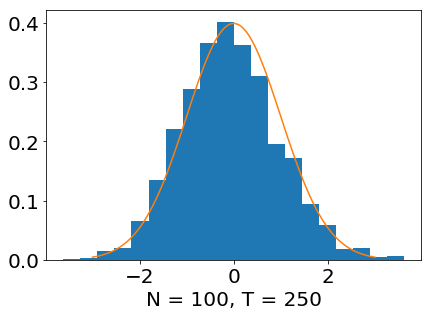}
			\end{subfigure}%
			\begin{subfigure}{.25\textwidth}
				\centering
				\includegraphics[width=1\linewidth]{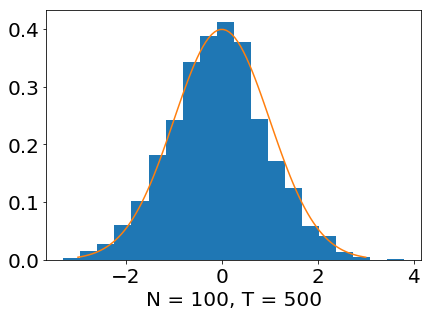}
			\end{subfigure}%
			\begin{subfigure}{.25\textwidth}
				\centering
				\includegraphics[width=1\linewidth]{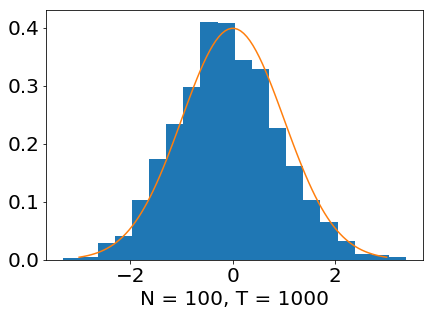}
			\end{subfigure}
			\begin{subfigure}{.25\textwidth}
				\centering
				\includegraphics[width=1\linewidth]{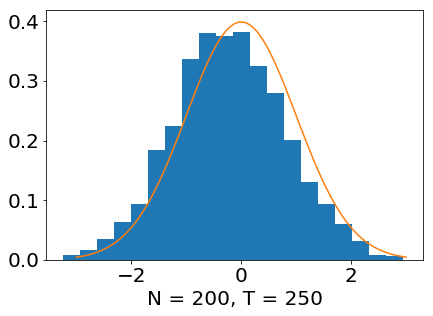}
			\end{subfigure}%
			\begin{subfigure}{.25\textwidth}
				\centering
				\includegraphics[width=1\linewidth]{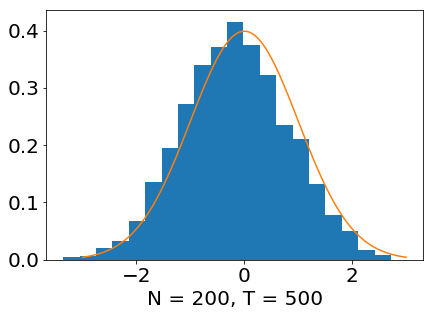}
			\end{subfigure}%
			\begin{subfigure}{.25\textwidth}
				\centering
				\includegraphics[width=1\linewidth]{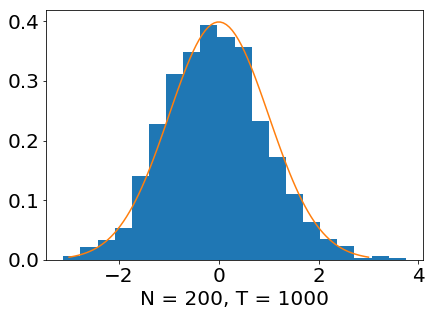}
			\end{subfigure}
			\bnotefig{Histograms of estimated standardized and bias-corrected generalized correlation test statistic. ($N = 50, 100, 200; T = 250, 500, 1000; h = 0.3$). The normal density function is superimposed on the histograms. Each subplot is based on 2,000 Monte-Carlo simulations.}
			\label{hist_rho}
		\end{figure}

			\begin{figure}[t!]
			\tcapfig{Generalized Correlation Test of Estimated Loadings in Any Paired States}
			\centering
			\begin{subfigure}{.4\textwidth}
				\centering
				\includegraphics[width=0.9\linewidth]{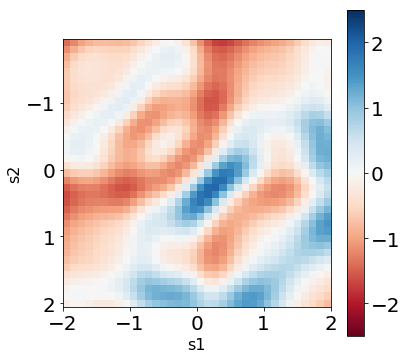}
				\caption{t-value}
			\end{subfigure}%
			\begin{subfigure}{.43\textwidth}
				\centering
				\includegraphics[width=0.9\linewidth]{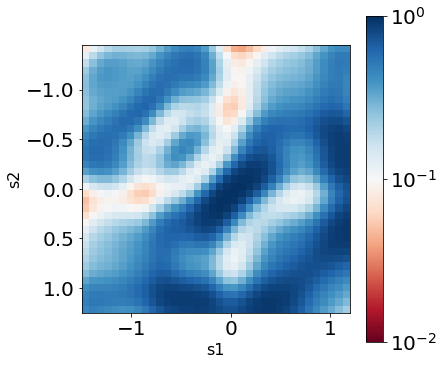}
				\caption{p-value}
			\end{subfigure}
			\bnotefig{Generalized correlation test of estimated loadings in any paired states ($N = 100$, $T = 500$ and $h = 0.3$. $\mathcal{H}_0$: there exists a full rank matrix $G$, $\Lambda_{s_2} = \Lambda_{s_1} G$, $\mathcal{H}_1$: for any full rank matrix $H$, $\Lambda_{s_2} \neq \Lambda_{s_1} G$). $x$-axis and $y$-axis are both state values.  The value at point ($s_1, s_2$) in figure (a) represents the normalized generalized correlation (t-value) of $\bar{\Lambda}_{s_1}$ and $\bar{\Lambda}_{s_2}$. The value at point ($s_1, s_2$) in Figure (b) represents the p-value corresponding to the t-value in Figure (a).}
			\label{gen_corr_heatmap}
		\end{figure}

		Figure \ref{gen_corr_heatmap} shows the p-values and t-values of any paired state outcomes when the loadings are constant. From the subplot of p-values, we would conclude that the loadings are constant for almost all paired loadings. As we face a multiple testing problem, there exists, as expected, a small number of false rejections for a given significance level.

		Simulations show the good power properties of the generalized correlation test. We assume the true underlying model has constant loadings in one interval and state-varying loadings in another interval. More specifically, data is generated such that loadings are constant in $s\in [0.3, 1]$ and linearly or quadratically depend on the state in $s \in [0, 0.3)$.  Table \ref{power} shows the acceptance probability for the null hypothesis for a 95\% significance level. When loadings in two states are different, the power of the generalized correlation test increases as $N$ or $T$ increases. The power is close to $1$ when the data size is at least $(N, T)=(100, 500)$.


		\begin{table}[t!]
			\tcaptab{Proportion of the Generalized Correlation Accepting the Null Hypothesis}
			{\small\begin{tabular}{@{}p{0.14\textwidth}*{6}{L{\dimexpr0.15\textwidth-2\tabcolsep\relax}}@{}}
				\toprule
				& \multicolumn{3}{c}{Loading linear in state} &
				\multicolumn{3}{c}{Loading quadratic in state} \\
				\cmidrule(r{4pt}){2-4} \cmidrule(l){5-7}
				$(N, T) \backslash (s_1, s_2)$  & (0.1, 0.9)  & (0.25, 0.75) & (0.90, 0.95) & (0.1, 0.9)  & (0.25, 0.75) & (0.90, 0.95) \\
				\midrule
				(50, 250) & 0.328  & 0.424 & 0.942 & 0.128 & 0.220 & 0.918 \\
				(50, 500) & 0.014  & 0.044 & 0.938 & 0.000  & 0.002 & 0.932 \\
				(50, 1000) & 0.002  & 0.000 & 0.952 & 0.000  & 0.000 & 0.970 \\
				(100, 250) &  0.084 & 0.124 & 0.948 & 0.022  & 0.024 & 0.934 \\
				(100, 500) & 0.000  & 0.002 & 0.954 & 0.002  & 0.002 & 0.938 \\
				(100, 1000) & 0.000  & 0.000 & 0.954 & 0.000  & 0.000 & 0.954 \\
				(200, 250) & 0.014  & 0.014 & 0.942 & 0.002  & 0.000 & 0.940 \\
				(200, 500) & 0.000  & 0.000 & 0.934 & 0.000  & 0.000 &  0.964 \\
				(200, 1000) & 0.000  & 0.000 & 0.946 & 0.000  & 0.000 &  0.946 \\
				\bottomrule
			\end{tabular}}
			\bnotetab{This table shows the proportion of Standardized Generalized Correlation $\rho$ of $\hat{\Lambda}(s_1)$ and $\hat{\Lambda}(s_2)$ that is within $[-1.65, +\infty)$. The state follows $S \sim U(0,1)$. Loading linear in state: $\Lambda(s) = \Lambda_1 + \mathbbm{1}(s \leq 0.3) (s-0.3) \Lambda_2$; Loadings quadratic in state: $\Lambda(s) = \Lambda_1 + \mathbbm{1}(s \leq 0.3) (s-0.3) \Lambda_2 + \mathbbm{1}(s \leq 0.3) (s-0.3)^2 \Lambda_3$). Among the loadings in the three pairs of states that we compare, the true loadings are different in $(s_1, s_2) = (0.1, 0.9)$ and $(s_1, s_2) =(0.25, 0.75)$, but the same in $(s_1, s_2) = (0.9, 0.95)$. Since the estimated loadings are smooth in $s$, more trials are rejected when we test loadings in $(s_1, s_2) = (0.1, 0.9)$ compared to $(s_1, s_2) =(0.25, 0.75)$. When testing loadings in $(s_1, s_2) = (0.9, 0.95)$, nearly 95\% of the trials are accepted, which is aligned with the asymptotic distribution under the null hypothesis. We run 500 Monte-Carlo simulations. The generalized correlation $\hat{\rho}$ of estimated loadings in two states is standardized by estimates of the mean and bias correction term and the standard deviation from Theorem \ref{gen_corr_thm}.}
			\label{power}
		\end{table}
		
		\subsection{Variation Explained by Factor Models}\label{sec:sim-est-err}
		We compare the amount of explained variation for the constant and state-varying factor model under misspecification. We consider a state observed with noise and a missing relevant state. Our simulation results confirm that our estimator is robust to noise in the observed state process and provides a more parsimonious model than a constant loading model as long as we condition on a process that is related to the underlying state process.
		
		We compare the in- and out-of-sample explained variation of $X$ and the common component for different factor estimators. The explained variation labeled as $R_X^2$ and $R_C^2$ is defined as
		\begin{align*}
		R^2_X= 1 - \frac{\sum_{i=1}^N \sum_{t=1}^T (X_{it} -\hat C_{it})^2}{\sum_{i=1}^N \sum_{t=1}^T X_{it}^2}  \qquad     R^2_C= 1 - \frac{\sum_{i=1}^N \sum_{t=1}^T (C_{it} -\hat C_{it})^2}{\sum_{i=1}^N \sum_{t=1}^T C_{it}^2} ,
		\end{align*}
		where the common component is either based on a state-varying or constant loading model. The out-of-sample common component projects the loading functions estimated in-sample on the out-of-sample observations, i.e. $\hat{C}_t = \hat{\Lambda}_t^\T  \Lp \hat{\Lambda}_t^\T \hat{\Lambda}_t \Rp^{-1} \hat{\Lambda}_t^\T X_t$. For the out-of-sample results we use the first $T/2$ time-series observations to estimate the loadings and test the model out-of-sample on the second $T/2$ observations.

		Table \ref{tab:sim-rsq} reports the explained variation for a noisy state process. This model can also be interpreted as a missing state process. Even when the noise has the same magnitude as the state process, the explained variation is very close to the case of using the true state. In contrast, the constant loading model explains one third less of the variation with the same number of factors.
		
		\begin{table}[t!]
			\tcaptab{In-Sample and Out-of-Sample $R^2$ Conditioned on Noisy State Process}
			\centering
			{\small
			\begin{tabular}{lcccc}
				\toprule
				& \multicolumn{2}{c}{In-sample} &
				\multicolumn{2}{c}{Out-of-sample} \\
				\cmidrule(r{4pt}){2-3} \cmidrule(l){4-5}
				& $R_X^2$  & $R_C^2$  & $R_X^2$ & $R_C^2$ \\
				\midrule
				State-Varying Model: $G = S$         & 0.677 & 0.987 & 0.643 & 0.982 \\
				State-Varying Model: $G = S + 0.1 v$ & 0.676 & 0.985 & 0.642 & 0.980 \\
				State-Varying Model: $G = S + 0.5 v$ & 0.653 & 0.952 & 0.611 & 0.934 \\
				State-Varying Model: $G = S + v$     & 0.616 & 0.894 & 0.559 & 0.856 \\
				State-Varying Model: $G = S + 2 v$   & 0.569 & 0.818 & 0.490 & 0.749 \\
				Constant Loading Model & 0.442 & 0.650 & 0.427 & 0.653 \\
				\bottomrule
			\end{tabular}
			}
			\bnotetab{In-sample and out-of-sample $R^2$ conditioned on noisy state process $G$ (true loadings depend only on $S_t$: $\Lambda_i(S_t) = \Lambda_{0i} + \frac{1}{2} S_t \Lambda_{1i} + \frac{1}{4} S_t^2 \Lambda_{2i} + \frac{1}{8} S_t^3 \Lambda_{3i}$): $N = 100$, $T = 500$; $S_t$ follows the same distribution as in Section \ref{asy_est}; $v_t \sim N(0,1)$ is the noise in the state process.}
		\end{table}\label{tab:sim-rsq}
		
				Figure \ref{fig:SimEV} considers missing a systematically relevant state in a non-linear state function. In this case, both the state-varying and constant loading model are misspecified. The loading function is modeled as $\Lambda_i(S_{1,t},S_{2,t}) = \exp(\Lambda_{1,i} S_{1,t} + \Lambda_{2,i} S_{2,t} )$ where the two independent states follow the same distribution as in Section \ref{asy_est}. We condition only on one state process and calculate the explained variation in $X$ and the common component out-of-sample. As before, we estimate the model on the first half of the data to obtain the out-of-sample fit on the second half. Conditioning on both state variables should yield a model with a high explained variation with only one factor. Both the state-dependent model with one state and the constant loading model do not explain a large amount of variation with one factor. 
		However, the state-varying loading model with two factors can almost perfectly explain the variation out-of-sample. In contrast, the constant loading model requires eight factors to capture the same amount of variation as a misspecified state-dependent model with two factors. This is exactly the same pattern that we observe in our empirical analysis of stock returns.

		\begin{figure}[t!]
			\tcapfig{Out-of-Sample $R^2$ for $X$ and $C$ for Misspecified Model}
			\centering
			\includegraphics[width=0.9\linewidth]{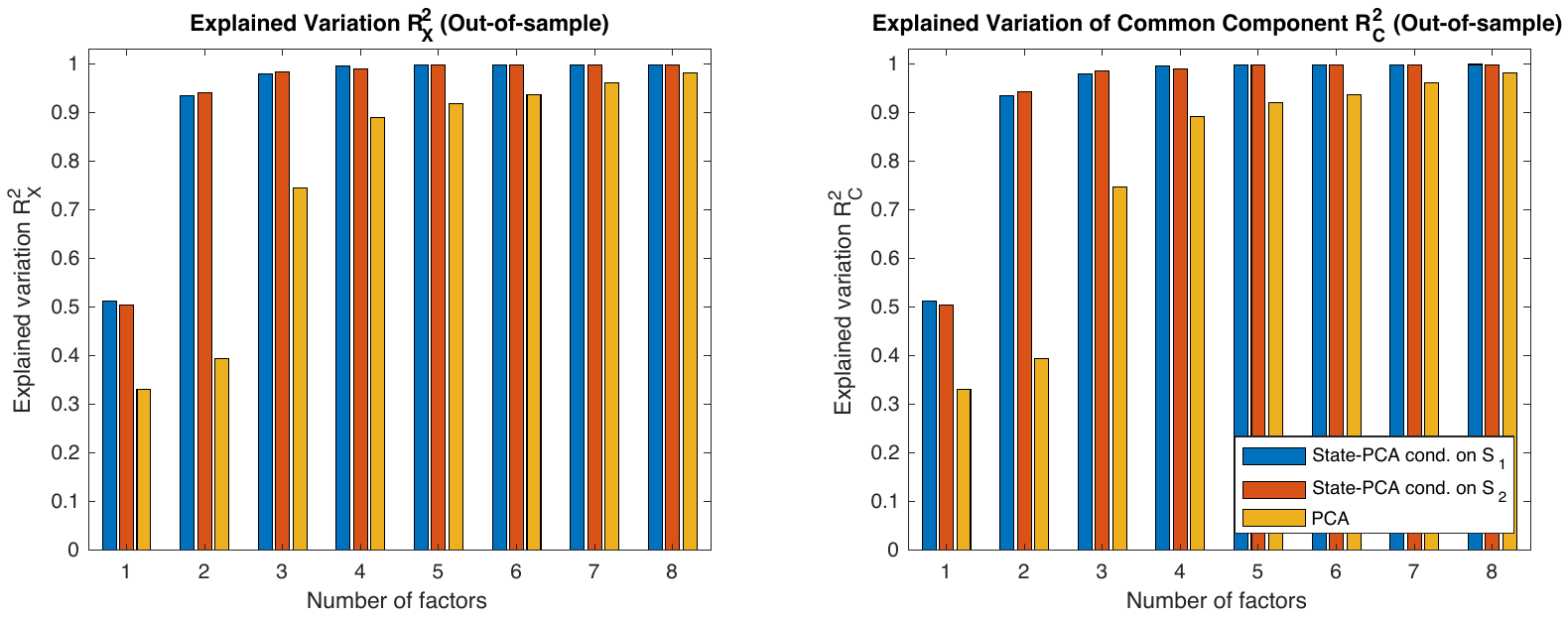}
			\bnotefig{Out-of-Sample $R^2$ for $X$ and $C$ for misspecified model (true loadings depend on two states and we condition only on one (State-PCA) or use constant loadings (PCA)): $N = 100$, $T = 500$; State-varying model: $\Lambda_i(S_{1,t},S_{2,t}) = \exp(\Lambda_{1,i} S_{1,t} + \Lambda_{2,i} S_{2,t} )$, where $S_1$ and $S_2$ are independent and follow OU processes. $\Lambda_{1,i},\Lambda_{2,i} \overset{iid}{\sim}N(0,1).$}
			\label{fig:SimEV}
		\end{figure}

		\section{Empirical Application to U.S. Treasury Securities} \label{empirical}
		
		We apply our approach to the treasury securities market and show that the factor structure changes with economic conditions. The U.S. Treasury yield structure has been shown to be well explained by the first three principal components.\footnote{See \cite{diebold2005modeling}, \cite{Diebold2006}, \cite{Cochrane2005} and \cite{Cochrane2009}.} The first three PCA factors are commonly referred to as the level (the long rate), slope (a long minus short rate), and curvature factor (a short and long rate average minus a mid-maturity) and can characterize the yield curves for different maturity bonds.
		We analyze how these three factors are influenced by three different macro-economic state variables. First, we use an NBER-based boom and recession indicator as a discrete state process. 
Second, we condition on the CBOE Volatility Index (VIX). Third, we model macro-economic conditions using the U.S. unemployment rate. Our findings strongly support a time-varying factor structure.
		
		The data set is daily data of the U.S. Treasury Securities Yields from 07/31/2001 to 12/01/2016. The terms range from 1, 3, 6 months to 1, 2, 3, 5, 7, 10, 20, 30 years. We first separate the data into booms and recessions based on NBER-based recession indicators and estimate a factor model for each state. Figure \ref{yield_recess} shows the loadings for the first three factors.  The level, slope, and curvature patterns of loadings versus bond terms persist in the loadings in the boom and in the loadings in the recession. However, there are differences in the values of the loadings, or the composition weights in the factors in the two different states.
		
		\begin{figure}[t!]
			\tcapfig{Treasury Security Data: Factor Loadings Conditioned on Boom and Recession}
			\centering
			\begin{subfigure}{.33\textwidth}
				\centering
				\includegraphics[width=1\linewidth]{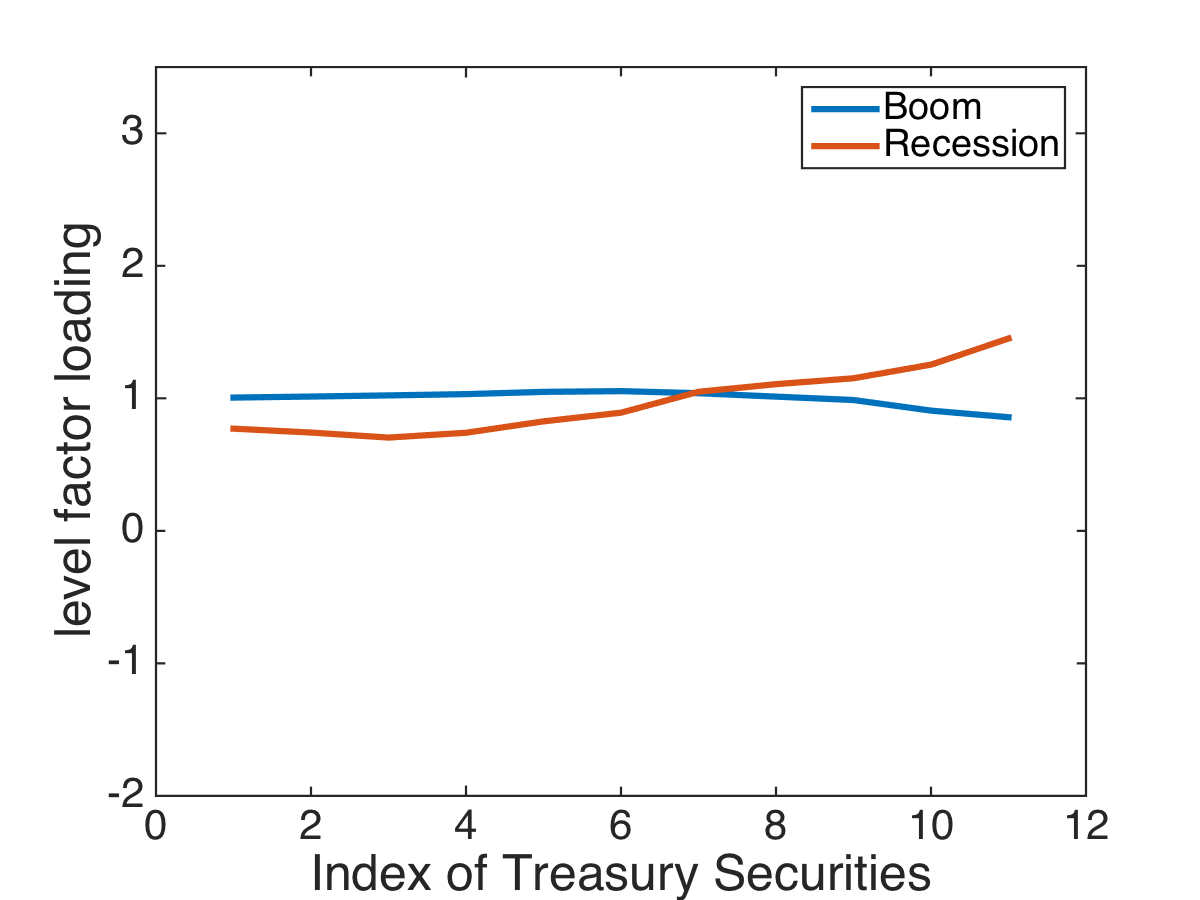}
				\caption{Level Factor}
			\end{subfigure}%
			\begin{subfigure}{.33\textwidth}
				\centering
				\includegraphics[width=1\linewidth]{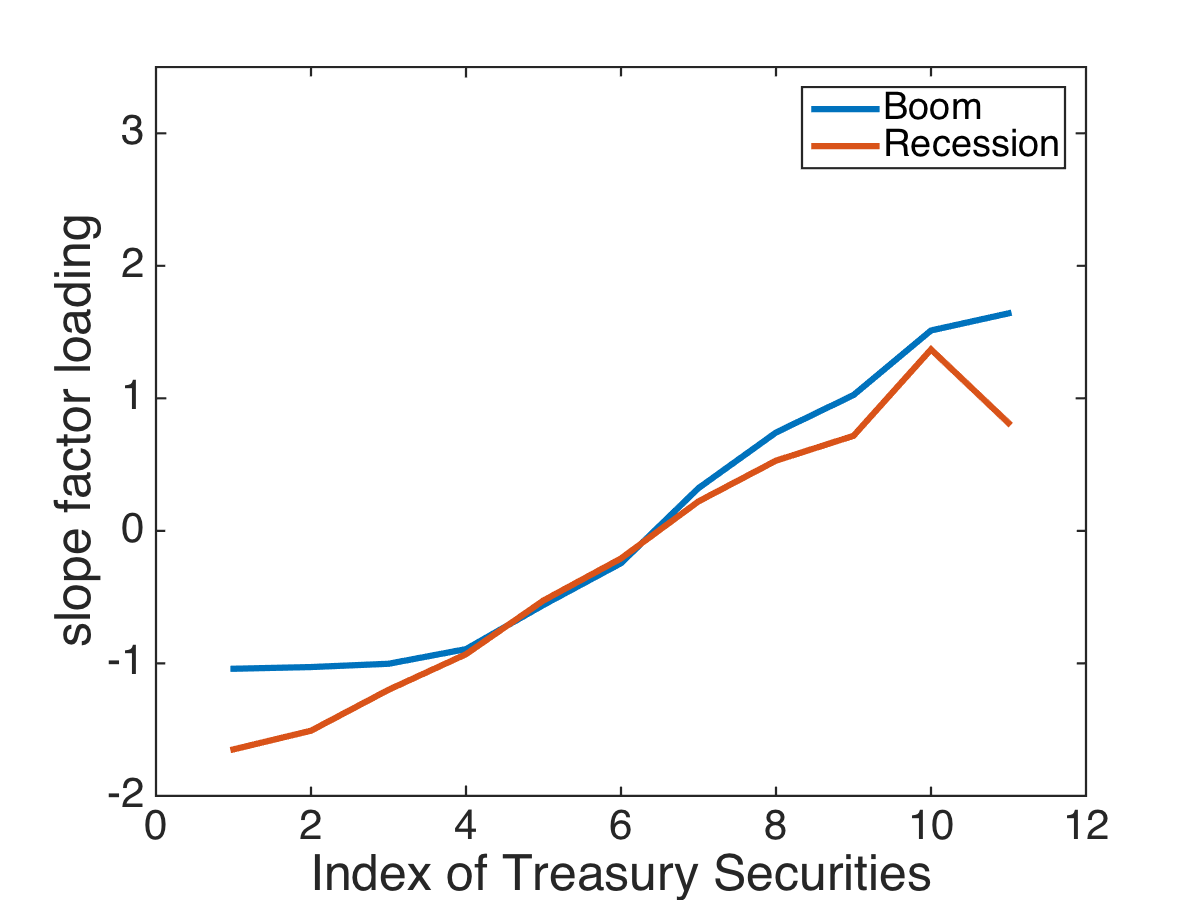}
				\caption{Slope Factor}
			\end{subfigure}%
			\begin{subfigure}{.33\textwidth}
				\centering
				\includegraphics[width=1\linewidth]{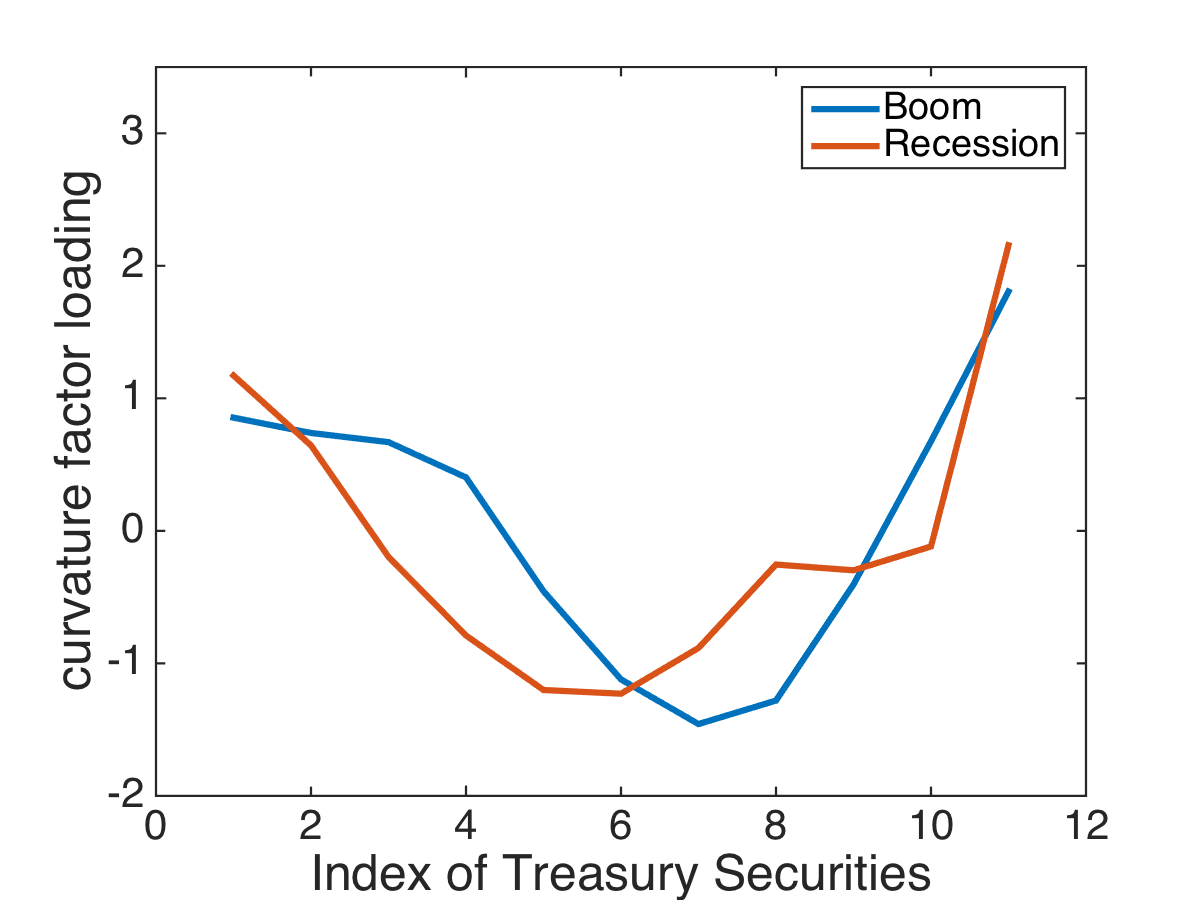}
				\caption{Curvature Factor}
			\end{subfigure}
			\bnotefig{First three latent factor loadings for treasury securities conditioned on boom and recession states. The x-axis is the index of Treasury Securities. The larger the index, the longer the bond term.}
			\label{yield_recess}
		\end{figure}
		
		It is coarse to characterize macro-economic conditions by only two state outcomes. The volatility index VIX and the unemployment rate can be viewed as continuous state processes that will provide a more refined analysis of the state-dependency. The CBOE Volatility Index (VIX) is a measure of the implied volatility of S\&P 500 index options. A higher VIX indicates a more volatile market. Typically a recession coincides with a high VIX. We use the standard convention of logarithmic VIX values to account for its heavy tails on the right. Figure \ref{log-vix} shows that the log-normalized VIX seems to be recurrent, as required by our methodology. 
		
		We estimate a factor model conditional on every possible log-normalized VIX value. We choose the bandwidth $h=0.1$  and confirm in the Internet Appendix that our results are robust to this choice. Figure \ref{prop-var} shows the variance explained by the first three factors. The level factor becomes less dominant as VIX goes up. Meanwhile, the slope factor becomes more important as the VIX increases. In a more volatile market, more yield movements are explained by the long minus short rate changes. 
		
		Figure \ref{yield_vix} shows how the loadings change with the VIX. The color bar indicates the log-normalized VIX value. Green curves represent loadings in low VIX states. Purple and red curves represent loadings in high VIX states. Even though the level, slope, and curvature patterns persist in all states, changes in the state variables lead to shifts in the curves of loadings versus bond terms, implying the changes of the compositions of level, slope, and curvature factors. For the level factor, longer-term bonds increase in weights as the VIX goes up. The loadings of the second factor, the slope factor, shift towards shorter maturities in more volatile markets. The curvature factor has a clear parallel shift to the left with increasing VIX. These results are consistent with the factor model conditioned on booms and recessions.

		\begin{figure}[t!]
			\tcapfig{Treasury Security Data: Log-normalized VIX and Variance Explained}
			\centering
			\begin{subfigure}{.5\textwidth}
				\centering
				\includegraphics[width=0.8\linewidth]{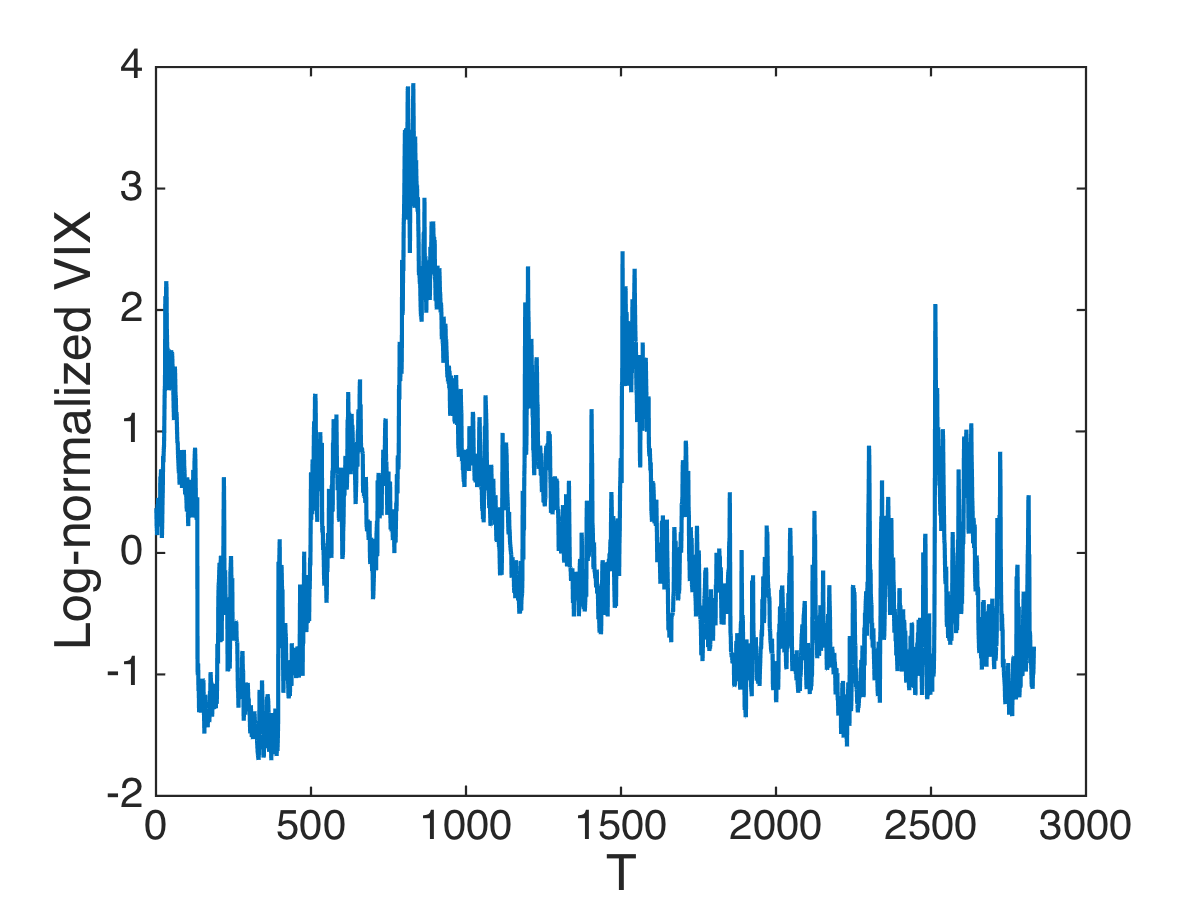}
				\caption{Log-normalized VIX}
				\label{log-vix}
			\end{subfigure}%
			\begin{subfigure}{.5\textwidth}
				\centering
				\includegraphics[width=0.8\linewidth]{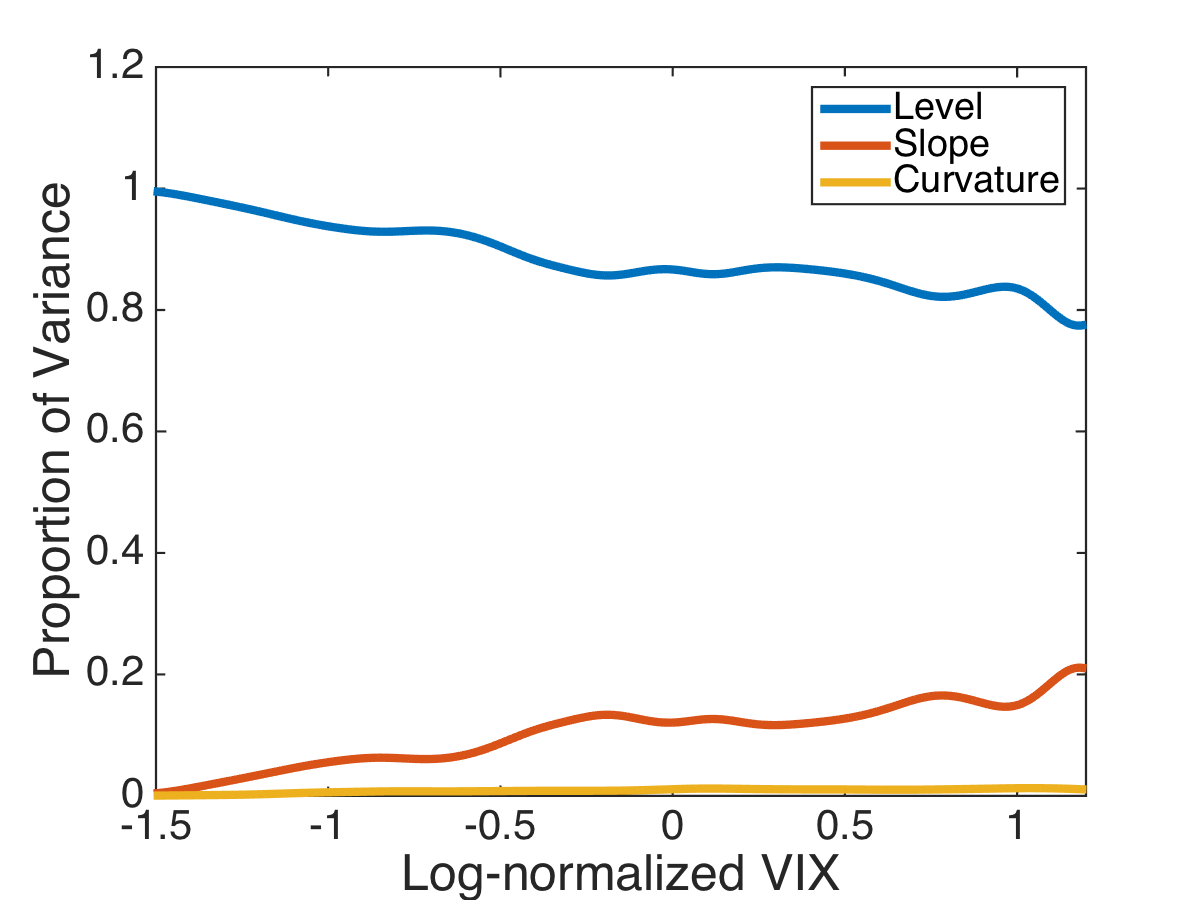}
				\caption{Proportion of variance explained}
				\label{prop-var}
			\end{subfigure}
			\bnotefig{Log-normalized VIX from 07/31/2001 to 12/01/2016 (index represents the number of trading days from 07/31/2001) and proportion of variance explained by the first three factors in different log-normalized VIX}
		\end{figure}

%
%
%
		
		\begin{figure}[t!]
			\tcapfig{Treasury Security Data: Factor Loadings Conditioned on VIX}
			\centering
			\begin{subfigure}{.33\textwidth}
				\centering
				\includegraphics[width=0.95\linewidth]{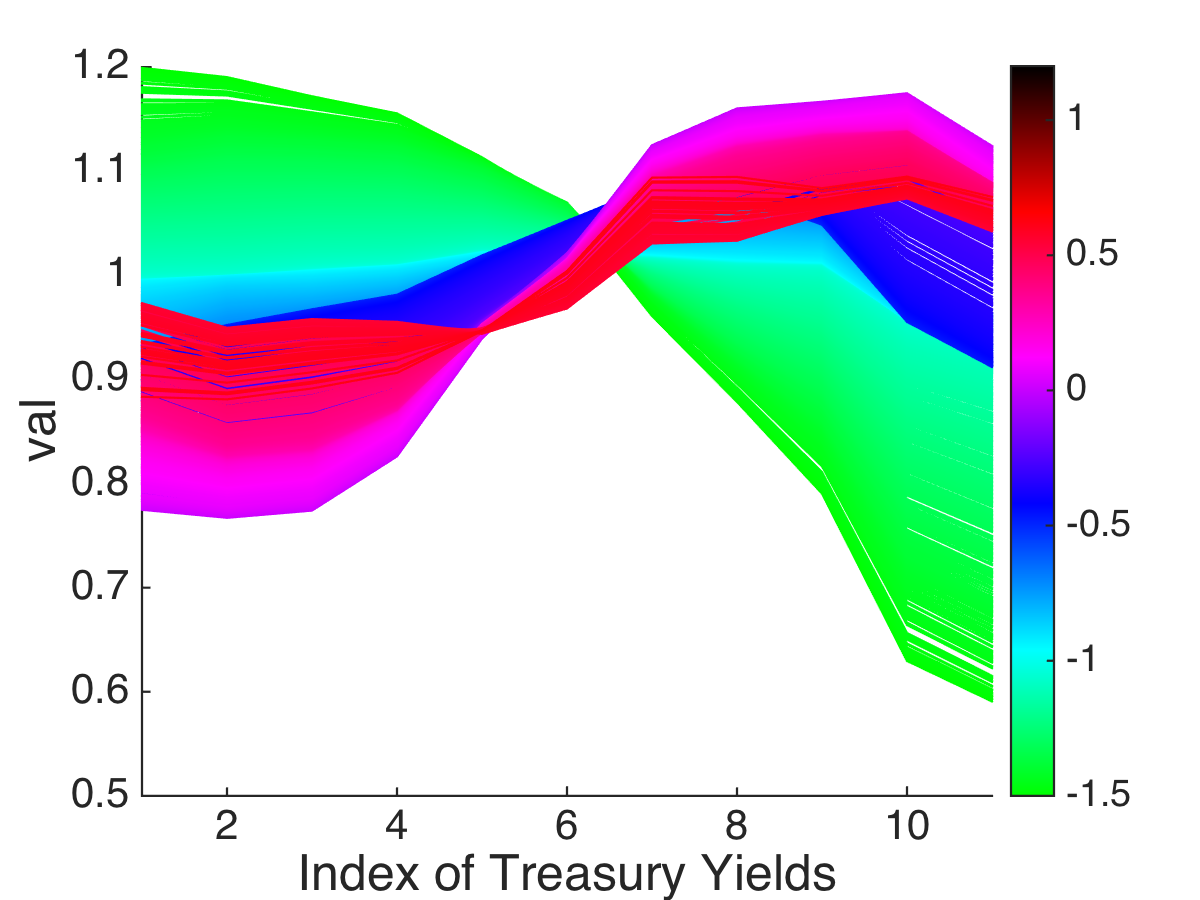}
				\caption{Level Factor}
			\end{subfigure}%
			\begin{subfigure}{.33\textwidth}
				\centering
				\includegraphics[width=0.95\linewidth]{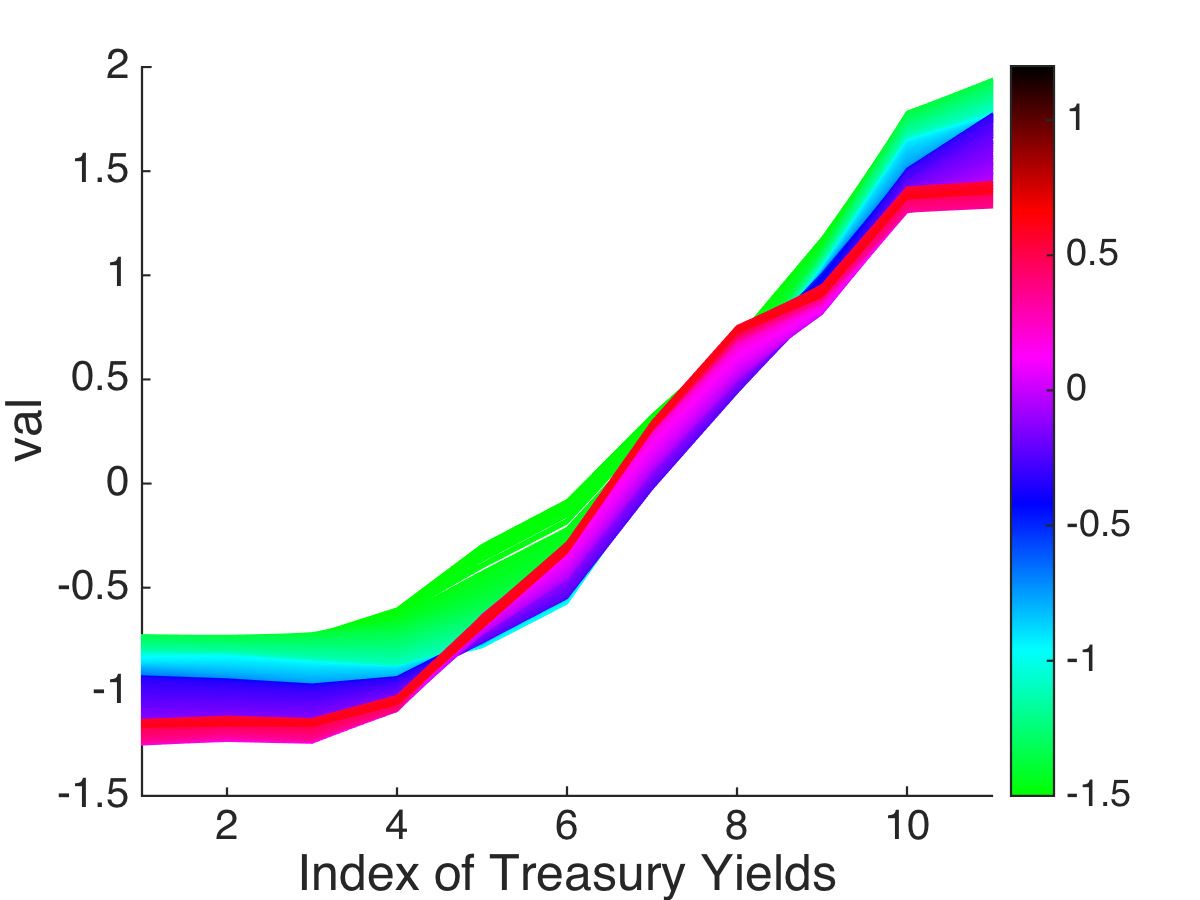}
				\caption{Slope Factor}
			\end{subfigure}%
			\begin{subfigure}{.33\textwidth}
				\centering
				\includegraphics[width=0.95\linewidth]{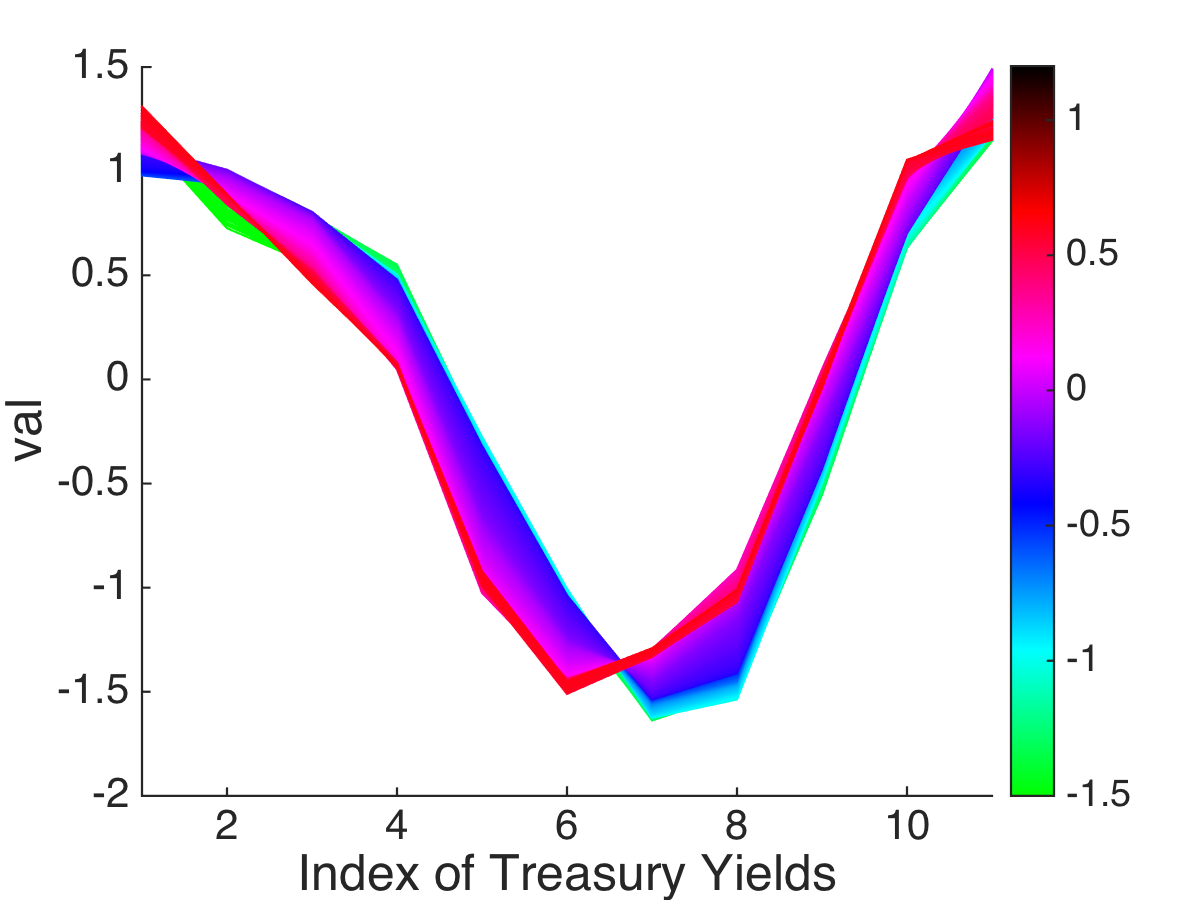}
				\caption{Curvature Factor}
			\end{subfigure}
			\bnotefig{First three latent factor loadings for treasury securities conditioned on VIX (the color bar indicates log-normalized VIX value). The x-axis is the same as Figure \ref{yield_recess}, indicating the maturity of the bonds.}
			\label{yield_vix}
		\end{figure}

		We use the generalized correlation approach to test for which states the factor structure changes. The previous results indicate that each of the first three eigenvectors changes with the VIX. However, it could be possible that the span of the eigenvectors does not change, i.e. it is possible that the factor structure is stable over time. Figure \ref{yield_vix_gen_corr} shows the results for the generalized correlation test statistic and its p-values for any combination of two states $s_1$ and $s_2$. As expected, the diagonal values take the largest values implying that the factor structure is very ``close'' for these paired states. The red regions represent changes in the factor structure. Apparently, the loading space is different in states with positive values (high VIX) from states with negative values (low VIX).\footnote{Our test-statistic is not a global test for changes in the factor structure, but aims at comparing two specific states. In order to use our results for a global test the p-values would need to be adjusted to account for multiple hypothesis testing.}

		\begin{figure}[t!]
			\tcapfig{Treasury Security Data: Generalized Correlation Test  in Any Paired States}
			\centering
			\begin{subfigure}{.4\textwidth}
				\centering
				\includegraphics[width=0.9\linewidth]{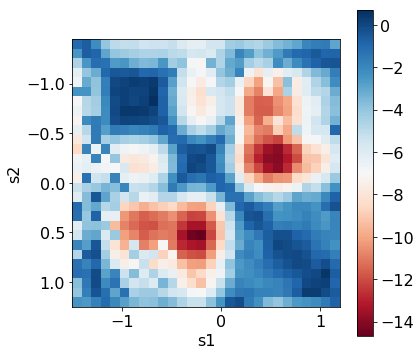}
				\caption{t-value}
			\end{subfigure}%
			\begin{subfigure}{.4\textwidth}
				\centering
				\includegraphics[width=0.9\linewidth]{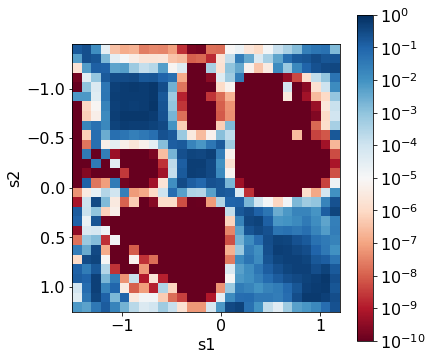}
				\caption{p-value}
			\end{subfigure}
			\bnotefig{Generalized correlation test of estimated loadings in any paired states in US Treasury securities data using log normalized VIX as state variable ($\mathcal{H}_0$: there exists a full rank matrix $G$, $\Lambda_{s_2} = \Lambda_{s_1} G$, $\mathcal{H}_1$: for any full rank matrix $H$, $\Lambda_{s_2} \neq \Lambda_{s_1} G$). $x$-axis and $y$-axis are both log-normalized VIX.  The value at point ($s_1, s_2$) in Figure (a) represents the normalized generalized correlation (t-value) of $\bar{\Lambda}_{s_1}$ and $\bar{\Lambda}_{s_2}$. The value at point ($s_1, s_2$) in Figure (b) represents the p-value corresponding to the t-values.}
			\label{yield_vix_gen_corr}
		\end{figure}
		
		Furthermore, we use the U.S. unemployment rate as the third state variable. The results are very similar to the VIX as the state variable and delegated to the Internet Appendix. In the Internet Appendix, we also compare the amount of variation explained by different factor models.  Treasury yields are somewhat special in the sense that their variation can almost perfectly be explained by three factors. The state-varying factor model with three factors explains slightly more variation, comparable to a four-factor model with constant loadings. However, if the goal is to explain variation, both a time-varying and a constant three-factor models perform well. The takeaway from this empirical application is to understand that the economic interpretation of ``level'', ``slope'' and ``curvature'' has to be used with caution. Depending on the economic conditions, the first PCAs are different.
The next application to individual stock returns shows that in other asset classes, the state-varying model can actually explain a significant larger amount of variation than its constant counterpart.

		\section{Empirical Application to Stock Returns}\label{sec:stocks}
		We estimate the latent factor structure in individual stock returns and show that the state-varying factor model is more parsimonious in explaining variation and captures more pricing information than the constant loading model. Our data set is the same as in \cite{pelger2019} and consists of the daily stock returns for the balanced panel of S\&P 500 stocks from January 1st 2004 to December 31st 2016. We include only stocks with returns available for the full-time horizon, which leaves us with a panel of $N=332$ and $T=3253$. We supplement the data with the daily risk-free rate from Kenneth French's website. As before, we condition on the log normalized VIX. We study the amount of variation explained by different factor models, the loadings for different states, and the optimal portfolio strategies implied by the factor models.
		
		\begin{figure}[t!]
			\tcapfig{S\&P500 Stock Return Data: In-Sample and Out-of-Sample Variance Explained}
			\centering
				\begin{subfigure}{.5\textwidth}
		\centering
		\includegraphics[width=1\linewidth]{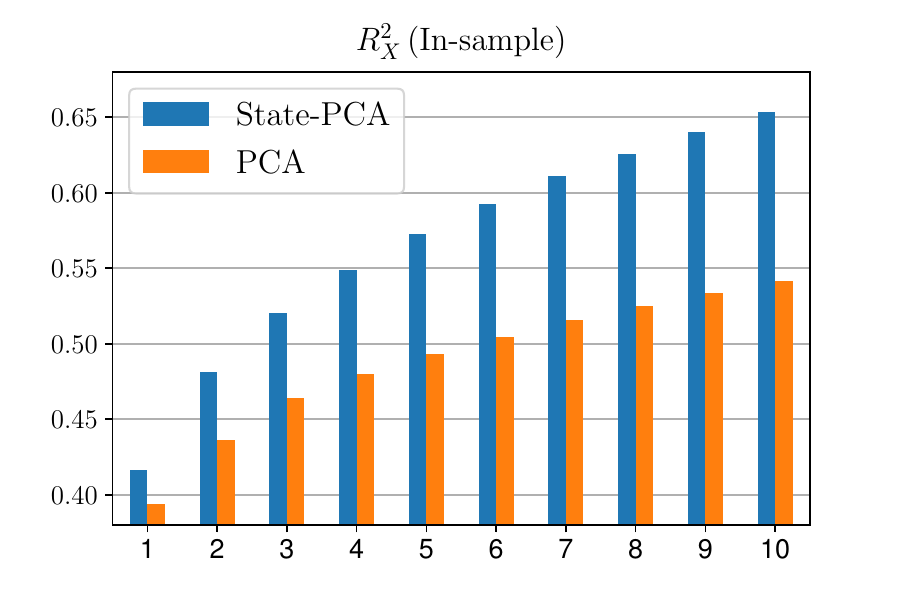}
	\end{subfigure}%
	\begin{subfigure}{.5\textwidth}
		\centering
		\includegraphics[width=1\linewidth]{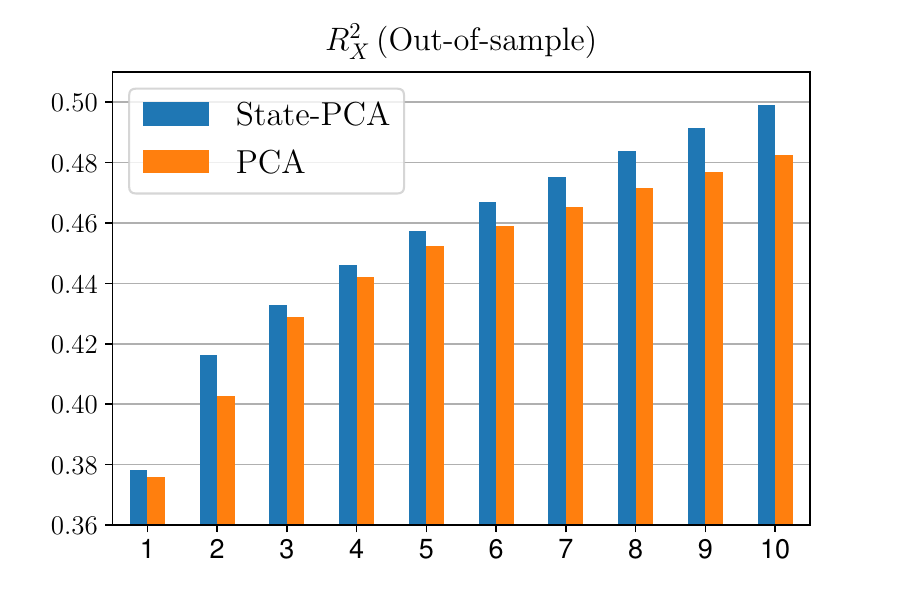}
	\end{subfigure}
	\bnotefig{S\&P500 stocks: Variation explained by the state-varying and constant loading model. State is log-normalized VIX. The constant loading model needs roughly five more factors to explain the same in-sample variation and two to three more factors to explain the same out-of-sample variation as the state-varying model. The bandwidth is chosen optimally. The Internet Appendix explains the choice of bandwidth and includes robustness results.}
			\label{fig:sp500-rsq}
		\end{figure}
		
				\begin{figure}[t!]
			\tcapfig{S\&P 500 Stock Return Data: Generalized Correlation Test in Any Paired States}
			\centering
			\begin{subfigure}[c]{.3\textwidth}
				\centering
				\includegraphics[width=1\linewidth]{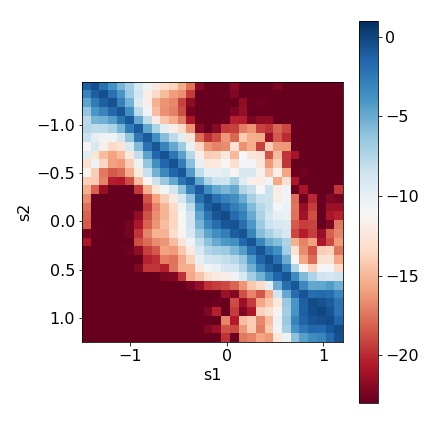}
				\caption{t-value}
			\end{subfigure}%
			\begin{subfigure}[c]{.3\textwidth}
				\centering
				\includegraphics[width=1\linewidth]{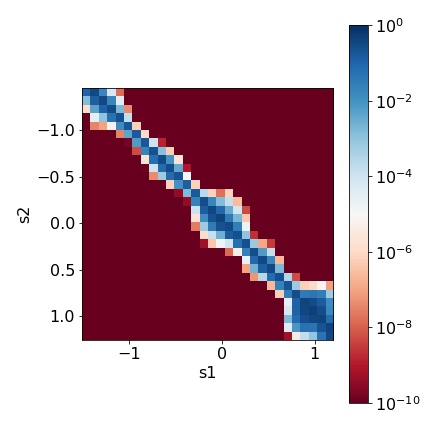}
				\caption{p-value}
			\end{subfigure}
			\bnotefig{Generalized correlation test for S\&P500 returns with log-normalized VIX as state variable for 5 factors. (a) standardized generalized correlations (t-values) and (b) corresponding p-values.}
			\label{gen_corr_sp500}
		\end{figure}
		
		Figure \ref{fig:sp500-rsq} reports the explained variation in- and out-of-sample for the state-varying and constant loading model. For the out-of-sample results, we first estimate the loadings on the first three years of data and then update the loadings estimates on an expanding window to obtain the out-of-sample systematic component for the next ten years.  Obviously, the state-varying factor model explains more variation than the constant loading model in- and out-of-sample for the same number of factors. Therefore, conditioning on the VIX results in a more parsimonious factor model to explain the co-movement in stock returns. Our results do not depend on a prior on the number of factors. In particular, it implies that stock returns do not follow a constant loading model and that the VIX is related to the source of time-variation. We do not require that the VIX explains all the time-variation in the loadings, but we show the conditional model provides a better description of the data than the unconditional one.
		
			Figure \ref{gen_corr_sp500} shows the test results for the generalized correlation test for the combination of any two state outcomes of the VIX. We use a five-factor model motivated by the five factors of \cite{fama2015five} and \cite{lettaupelger2018}. The span of the loadings drastically changes with the realization of the VIX, which confirms the previous results. The Internet Appendix shows that, even in a one-factor model, the span of the state-varying loadings is different from a constant factor model and studies the portfolio implications of the time-varying factors. 

		\section{Conclusion} \label{conclusion}
		The exposure of financial or macro-economic variables to factors may change with policies, macroeconomic environment, and technology innovation. Failing to correctly model the exposure may result in misspecifying factors and potentially inflating the number of factors identified in the model. We model these driving forces as a state process to build a state-varying factor model. We combine a nonparametric kernel projection with PCA to estimate the factor model in a particular state. Our model allows for general time-variation in the loadings for a given state. Asymptotic properties of estimated factors, loadings, and common components are presented. We develop a test for detecting changes in loadings at different states based on a generalized correlation statistic. Simulations show the good finite sample properties of our estimators and test statistic. 
		
		The analytical analysis is challenging for both the kernel estimator of the loadings, factors, and the test statistic because we have to take into account bias terms. We show under which conditions these bias terms can be neglected or how to estimate and correct for this bias. In turns out that our test statistic for changes in the loadings is non-standard with a super-consistent rate.  
		
		We believe that the proposed estimator and test statistic have wide applications in macro-economics and finance. In two empirical studies, we apply our estimators to U.S. Treasury securities and individual stock returns. In the first case, we use a recession indicator, the VIX, and the unemployment rate as state variables. In recessions, times of high volatility or times of high unemployment rate, the level factor explains less variance in the data and becomes less important, while the slope factor gains importance. In particular, the composition of the slope and curvature factors is shifted to shorter maturities in bad or volatile times. Based on our generalized correlation test, we identify the states for which the factor structure changes. The takeaway is that the economic interpretation of ``level'', ''slope'' and ``curvature'' has to be used with caution, as for different economic states, the PCA factors will be different. In the second application on individual stock returns, we show that the state-varying factor model with the VIX as state variable is more parsimonious in explaining variation and captures more pricing information than the constant loading model. Hence, even if we do not capture all time-variation in the loadings with the proposed state variable, we still obtain a model that explains the correlations structure and mean returns better than a constant loading model.

		
	\end{onehalfspacing}

%
	
		\singlespacing
\bibliographystyle{econometrica}
\let\oldbibliography\thebibliography
\renewcommand{\thebibliography}[1]{%
  \oldbibliography{#1}%
  \setlength{\itemsep}{3pt}%
}
	{\footnotesize
		\bibliography{reference}
	}
	\onehalfspacing

\newpage

	\begin{titlepage}
	
	\begin{center}
	{\Huge Internet Appendix for\\State-Varying Factor Models of Large Dimensions}
	\end{center}

	 \thispagestyle{empty}
		\vspace{0.5cm}
		
		\begin{abstract}
			The Internet Appendix collects the proofs and additional results that support the main text. The additional theoretical results include a detailed description of special cases and related models and an extension to noisy and misspecified state processes. We also provide an estimator for the number of factors. The additional empirical results consider alternative state processes and discuss the choice of tuning parameters. We also study a portfolio application of our state-varying factors. The extensive simulation section compares the performance relative to alternative latent factor models that allow for time-variation and studies the choice of bandwidth and number of factors with cross-validation arguments. Lastly, we collect the detailed proofs for all the theoretical statements. 
			\vspace{1cm}
			
			\noindent\textbf{Keywords:} Factor Analysis, Principal Components, State-Varying, Nonparametric, Kernel-Regression, Large-Dimensional Panel Data, Large $N$ and $T$

			\noindent\textbf{JEL classification:} C14, C38, C55, G12
		\end{abstract}
	\end{titlepage}

	\setcounter{page}{1}
	
	\setcounter{section}{0}
	\setcounter{subsection}{0}
	
	\renewcommand{\thesection}{IA.\Alph{section}}
	\renewcommand{\thesubsection}{\thesection.\arabic{subsection}}
	
	\setcounter{equation}{0}
	\renewcommand{\theequation}{\thesection.\arabic{equation}}
	

	
	\renewcommand{\theequation}{IA.\arabic{equation}}%
	\renewcommand{\thefigure}{IA.\arabic{figure}} \setcounter{figure}{0}
	\renewcommand{\thetable}{IA.\Roman{table}} \setcounter{table}{0}


	
	

\begin{small}
	\section{Overview}

	The Internet Appendix collects the proofs and additional results that support the main text. Section \ref{sec:theory} provides additional theoretical results. Section \ref{sec:examples} shows that our state-varying factor model nests several relevant models as special cases. Section \ref{sec:alternatives} shows how our model is related to alternative models in the literature. In Section \ref{sec:noisy}, we relax our model and assume that we only use a noisy approximation of the underlying state process. Section \ref{sec:misspecified} shows that our model dominates a constant loading model even if the state process is misspecified. In Section \ref{sec:number} we generalize the information criterion based estimator for the number of factors of \cite{bai2002determining} to our setup. Section \ref{sec:disrete} discusses the case of a discrete state process.    
	
	The additional empirical results in Section \ref{sec:EmpApp} consider alternative state processes and discuss the choice of tuning parameters. We also show that our results are robust to the choice of bandwidth and study a portfolio application of our state-varying factors. The extensive simulations in Section \ref{sec:SimApp} compare the performance relative to alternative latent factor models that allow for time-variation and discuss the choice of bandwidth and number of factors with cross-validation arguments. In Section \ref{sec:comparison}, we show our estimator has a better performance than general purpose estimators for structural breaks or local time variation if we have the additional information about time-varying state processes that we can exploit. In Section \ref{sec:comparison}, we illustrate that a local window estimator is inferior when state processes change fast. In Section \ref{sec:tuning}, we illustrate how to optimally choose the number of factors and the bandwidth of the kernel projection. Lastly, in Section \ref{sec:proofs}, we collect the detailed proofs for all the theoretical statements.

	%
	%
	%
	%
	%
	%
	%
	%
	%

	\section{Additional Theoretical Results}\label{sec:theory}
	
	\subsection{Special Cases}\label{sec:examples}
	
	Our state-varying factor model nests several models as a special case. For simplicity we consider here the case of a one factor model, i.e. $r=1$.
	\begin{enumerate}
		\item {\bf Linear functional form:} If loadings are modeled as an affine function of the state process, i.e. $\Lambda_i(S_t) = \Lambda_{i,1} + \Lambda_{i,2} S_t$, we can rewrite the one factor model as a two factor model:
		\begin{align*}
		X_{it} &= \Lambda_{i,1} \underbrace{F_t}_{F_{t,1}}  + \Lambda_{i,2} \underbrace{\left( S_t F_t  \right)}_{F_{t,2}}+ e_{it}.
		\end{align*}
		\item {\bf Polynomial functional form:} If loadings are a polynomial of degree $q$, i.e. $\Lambda_i(S_t) = \Lambda_1 + ...+\Lambda_{q+1} S_t^q$, we can rewrite the one factor model as a $q+1$ factor model:
		\begin{align*}
		X_{it} &= \Lambda_{i,1} \underbrace{F_t}_{F_{t,1}}  +...+ \Lambda_{i,q+1} \underbrace{\left( S_t^q F_t  \right)}_{F_{t,q+1}}+ e_{it}.
		\end{align*}
		\item{\bf Discrete state space:} In the case where the loadings are non-linear functions of the state but the state process is discrete (we assume for simplicity here that there are only two state outcomes), the one factor model can again be formulated as a two factor model:
		\begin{align*}
		X_{it} &= \underbrace{g_i(s_1)}_{\Lambda_{i,1}} \underbrace{\mathbbm 1_{\{S_t=s_1 \}} F_t}_{F_{t,1}} + \underbrace{g_i(s_2)}_{\Lambda_{i,2}} \underbrace{\mathbbm 1_{\{S_t=s_2 \}} F_t}_{F_{t,2}} + e_{it}.
		\end{align*}
		\item{\bf Smooth time-variation in loadings:} The slowly changing loading model of \cite{Su2017} can be interpreted as a deterministic state model with $S_t=t$. Only time observations in a neighborhood of the target value $t_0$ are used for the estimation.
	\end{enumerate}
	In the first two special cases our state-varying factor model is equivalent to a constant loading model but more parsimonious. The third special case of a discrete state space model is useful to provide the intuition behind our estimator: Conditioning on a specific state outcome corresponds to selecting only those time observations where the discrete state process takes the target value. Using boom and recession indicators as a discrete state model in our empirical analysis, we illustrate that loadings change over time. However, the second case rules out state processes which can take many different values. The fourth special case uses only information in a local neighborhood. If the time-variation in the loadings has a cyclical component, previous observations that are not in a local neighborhood contain information that can be used in the estimation.  
	
	We consider the relevant model $\Lambda_i(S_t)=g_i(S_t)$ where we have a continuum of state outcomes for $S_t$ and a non-linear loading function $g(.)$ that requires a large number of basis functions to approximate. In this case, there exists, in general, no multi-factor representation. Hence, neither forecasting nor economic interpretation is possible in a constant loading model. This type of model seems to be supported by our empirical examples. 
	
	\subsection{Related Models}\label{sec:alternatives}
	Following the setup of \cite{bai2020estimation} we can formulate the high-dimensional factor model with a structural break at $t=\tau$ as 
	\[
	X_{it}= \begin{cases} 
	\Lambda_{i1}^\top F_t + e_{it} \qquad &\text{for $t=1,2,...,\tau$}\\
	\Lambda_{i2}^\top F_t + e_{it} \qquad &\text{for $t=\tau+1,...,T$.}
	\end{cases}
	\]
	This type of model can be extended to have multiple breaks points, which are usually not known and have to be estimated. It is limited to a finite number of breaks that are sufficiently far away from each other. The structural break model can be embedded into our model if for example $\Lambda(S_t)$ and $S_t=s_1$ for $t \leq \tau$ and $S_t=s_2$ for $t >\tau$. As we use the additional information of the state process in our framework, we can deal with a continuum of breakpoints as long as they are due to changes in the state process.

	The noisy state process model $X_{it} = (\Lambda_i(S_t) + \varepsilon_{it})^\top F_t + e_{it}$ can be interpreted as a random coefficient model. The loadings $\Lambda_i(S_t) + \varepsilon_{it}$ are random and time-varying because of the randomness in the state process and the noise component $\varepsilon_{it}$. However, there are two major differences to a conventional random coefficient model. First, we estimate the model conditional on a particular realization of the state process, i.e., we implicitly take out the randomness in the state process. Second, we show that under mild assumptions on the loading noise component $\varepsilon_{it}$, which limits its cross-sectional and time-series dependence, it becomes part of the latent residual component $e_{it}$. Thus, all the conditional PCA estimation results of the model without the loading noise component continue to hold.

	An alternative approach to include information from a known state process is a factor-augmented regression studied in \cite{bai2006}. Here the state process (or a finite number of transformations of the state process) is added as an observable factor to the model:
	\begin{align*}
	X_{it}=\beta_i S_t + \Lambda_i^\top F_t + e_{it}.
	\end{align*}
	First, the residuals of a regression on the state process are calculated, and second, PCA is applied to the covariance matrix of those residuals to estimate the constant loading model. This framework is different from ours as it does not allow the impact of factors to depend on the state process. For example, it could not capture an asset pricing model in which the effect of the market factor on asset returns changes during the business cycle.

	\subsection{Extension to Noisy State Process}\label{sec:noisy}
	
	Our state-varying factor model requires the knowledge of the state process driving the loading variation, which can be restrictive in some cases. A natural relaxation is to assume that we only use a noisy approximation of the underlying state process: 
	\[X_{it} = (\Lambda_i(S_t) + \varepsilon_{it})^\top F_t + e_{it} \quad i = 1, 2, \cdots, N \text{ and } t = 1, 2, \cdots, T\]
	or in vector notation,
	\[\underbrace{X_t}_{N \times 1}  = \underbrace{\Lambda(S_t)}_{N \times r} \underbrace{F_t}_{r \times 1} + \underbrace{\mathcal{E}_t}_{N \times r} \underbrace{F_t}_{r \times 1}  + \underbrace{e_t}_{N \times 1} = \Lambda(S_t)F_t + \psi_t + e_t   \qquad \text{$ t = 1, 2, \cdots, T$}.\]
	The term $\epsilon_{it}$ is the time-varying component of the loading coefficient that cannot be explained by the state process $S_t$. It can, for example, be due to a measurement error in the state process $S_t$ or an omitted additional state process. Without loss of generality, we can assume that $\epsilon_{it}$ has a time-series mean of zero as the non-zero mean can be captured by the latent loading function $\Lambda_i$.
	In the following, we will argue that the additional term $\psi_t=\mathcal{E}_t F_t$ can be treated like an additional error term that will not affect our previous results. In this sense, our model is robust to model miss-specification. 
	
	Our approach is related to \cite{fan2016projected}. They model loadings as non-linear functions of time-varying features of the cross-sectional units. Their estimation approach applies PCA to the data matrix that is projected in the cross-section on the subject-specific covariates. In addition to the covariates, they allow for a subject specific orthogonal residual component in the loadings. In contrast, our projection is applied in the time dimension. We also allow for an additional component independent of the state process to capture additional variation. Our assumptions on this noise component are similar to their setup.
	
	Defining $\widetilde e_{it} = \varepsilon^\top_{it} F_t + e_{it}$ we reformulate our model as $X_{it}=\Lambda_i(S_t) F_t + \widetilde e_{it}$. The noise term in the loadings $\varepsilon_{it}$ needs to satisfy essentially the same assumptions as $e_{it}$ for Theorem \ref{thm_consistency} to \ref{thm_rho} to hold. In particular, $\varepsilon_{it}$ can only have weak cross-sectional and time-series correlation.
	
	\setcounter{assumption}{8}	
	
	\begin{assumption}\label{assumption_noisy}Weak noise dependency:
		\begin{enumerate}
			\item Assume $\varepsilon_{it}$ is independent of $S_u$, $F_u$ and $e_{ju}$ for all $i$, $j$, $t$ and $u$. Furthermore, Assumptions \ref{ass_err}.1-5 hold with $e_{it}$ replaced by $\varepsilon_{it}$
			\footnote{The error $\varepsilon_{it}$ is a $r$-dimensional vector and hence the assumptions are formulated for each element of the vector.}
			, $\+E[\norm{F_t}^8] \leq M < \infty$ and $\max_t . \+E[\norm{F_t}^8|\mathcal{F}_S] \leq M < \infty$.
			\item Assumptions \ref{ass_mom} and \ref{doublesum} hold with $e_{it}$ replaced by $\varepsilon^\top_{it}F_t$.
		\end{enumerate}
	\end{assumption}
	
	
	Assumption \ref{assumption_noisy}.1 imposes the same weak correlation structure on $\varepsilon_{it}$ as on $e_{it}$.  Assumption \ref{assumption_noisy}.2 imposes the same dependency structure between $F_t$, $\Lambda(s)$ with $F_t \varepsilon_{it}$ as between $F_t$, $\Lambda(s)$ with $e_{it}$ in Assumptions \ref{ass_mom} and \ref{doublesum}. 
	Assumption \ref{assumption_noisy}.2 is only slighter stronger as we now essentially limit the dependency between $F_t F_t^{\top}$ and $\varepsilon_{it}$ instead of the $F_t$ and $e_{it}$. Under Assumption \ref{assumption_noisy} the new error term $\epsilon_{it}$ satisfies the same assumptions as the previous error term $e_{it}$:
	\begin{corollary}\label{cor_noisy_ass_hold}
		Assume the noisy state-varying factor model holds. 
		\begin{enumerate}
			\item Under Assumption \ref{assumption_noisy}.1, Assumption \ref{ass_err} holds with $e_{it}$ replaced by $\widetilde e_{it}$.
			\item Under Assumption \ref{assumption_noisy}.2, Assumptions \ref{ass_mom} and \ref{doublesum} hold with $e_{it}$ replaced by $\widetilde e_{it}$.
		\end{enumerate}
	\end{corollary}
	As a result all our previous theorems are still valid.
	\begin{corollary}\label{cor_noise}
		Assume the noisy state-varying factor model holds. 
		\begin{enumerate}
			\item Under the assumptions of Theorem \ref{thm_consistency} and Assumption \ref{assumption_noisy}.1, the results of Theorem \ref{thm_consistency} continue to hold.
			\item Under Assumption \ref{assumption_noisy} and the assumptions of Theorem \ref{thm_factor}, \ref{thm_loading} or respectively \ref{thm_common}, the results of Theorem \ref{thm_factor}, \ref{thm_loading} or respectively \ref{thm_common} continue to hold.
			\item Let $\Sigma_{\tilde e_T} = \+E[\widetilde e^\top \widetilde e/N]$, $\Sigma_{\tilde e_N} = \+E[\widetilde e \widetilde e^\top/T]$ and $\Sigma_{\widetilde e} = \+E[\tvec(\tilde e)\tvec(\tilde e)^\top]$. Assume there are only finitely many non-zero elements in each row of $\Sigma_{\tilde e_T}$, $\Sigma_{\tilde e_N}$ and $\Sigma_{\tilde e}$ and we know the sets of nonzero indices, $\Omega_{\tilde e_T}$, $\Omega_{\tilde e_N}$ and $\Omega_{\tilde e}$.
			Under Assumption \ref{assumption_noisy} and the assumptions of Theorem \ref{thm_rho} or respectively Lemma \ref{consistentcovariance}, the results of Theorem \ref{thm_rho} or respectively Lemma \ref{consistentcovariance} continue to hold.    
		\end{enumerate}
	\end{corollary}
	Corollary \ref{cor_noise} implies that we do not need to know the exact source of changes in the loadings, but it is sufficient to use a state process that is correlated with the cause of change in the structure. Furthermore, the asymptotic standard errors and the generalized correlation test for state-dependency provide valid results even if we do not implicitly capture all the variation in the loadings with our choice of state process.
	
	\subsection{Misspecified State Process}\label{sec:misspecified}
	
	Our estimator can dominate a constant loading model, even if the state process is misspecified. In the last section, we have shown that our estimator is robust to moderate misspecification of the state process. However, we have ruled out a ``systematic and non-diversifiable'' miss-specification, which could arise if, for example, a systematic state process is omitted in the estimation. Here, we illustrate that even if we omit a relevant state process, the state-varying loading estimator can outperform the constant loading estimator for the same number of factors. We only require our candidate state process to be dependent on the true population state process. Our arguments are based on the population model to illustrate the key ideas.
	
	Assume that that population model follows $X_{it} = \Lambda_i(S_t)^\T F_t + e_{it}$ where $S_t$ is the true state process. We compare PCA applied to the second moment conditioned on the true state, a wrong state and without conditioning. We denote by $\Sigma_{S_t=s}=E[X_t X_t^{\top} |S_t=s]$ the second moment conditioned on the true state. Given the stationary density of the state process $\pi_{S_t}(s)$ the unconditional second moment equals
	\[\Sigma =E[X_t X_t^{\top}]= \int \Sigma_{S_t = s} \pi_{S_t}(s) d s. \]
	We denote by $G_t$ another state process whose stationary density is $\pi_{G_t}(g)$. 
	The conditional second moment on $G_t = g$ equals
	\[\Sigma_{G_t = g} =E[X_t X_t^{\top} |G_t=g]= \int \Sigma_{S_t = s} \pi_{S_t|g_t = g}(s) ds,\]
	where $\pi_{S_t|g_t = g}(s)$ is the conditional density of $S_t$ given $G_t = g$. Note, that 
	\[\Sigma = \int \Sigma_{G_t = g} \pi_{G_t}(g) dg. \]
	We compare the variation explained from PCA-based factors estimated from different second moment population matrices. The constant loading model will always explain less variation than the state-varying model conditioned on the correct state:
	\begin{eqnarray}
	\nonumber \max_{\Lambda: \Lambda^\T \Lambda/N = I_r}  tr \Lp \Lambda^\T \Sigma \Lambda \Rp &=& \max_{\Lambda: \Lambda^\T \Lambda/N = I_r} tr \Lp \Lambda^\T \Lp \int \Sigma_{S_t = s} \pi_{S_t}(s) ds \Rp \Lambda  \Rp \\
	&\leq&  \int tr  \Lp \max_{\Lambda_s: \Lambda_s^\T \Lambda_s/N = I_r} \Lambda_s^\T  \Sigma_{S_t = s} \Lambda_s  \Rp  \pi_{S_t}(s) ds. \label{eqn:exact-spec}
	\end{eqnarray}
	which follows from the convexity of the maximum operator. The inequality is strict if $rank(\Sigma) > r$. Similarly the state-varying model with a potentially wrong state explains at least as much variation as the constant loading model
	\begin{eqnarray}
	\nonumber \max_{\Lambda: \Lambda^\T \Lambda/N = I_r} tr \Lp\Lambda^\T \Sigma \Lambda \Rp &=& \max_{\Lambda: \Lambda^\T \Lambda/N = I_r} tr \Lp \Lambda^\T \Lp  \int \Sigma_{G_t = g} \pi_{G_t}(g) ds \Rp \Lambda \Rp  \\
	&\leq&  \int tr  \Lp \max_{\Lambda_g: \Lambda_g^\T \Lambda_g/N = I_r} \Lambda_g^\T   \Sigma_{G_t = g}  \Lambda_g  \Rp  \pi_{G_t}(g) dg. \label{eqn:noisy-spec-1}
	\end{eqnarray}

	Let $d_j$ and $d_{G_t=g, j}$ be the $j$-th largest eigenvalues of $\Sigma$ and $\Sigma_{G_t = g}$ respectively. 
	Inequality (\ref{eqn:noisy-spec-1}) is strict if and only if 
	\[\int \Lp \sum_{j = 1}^r d_{G_t=g, j} \Rp \pi_{G_t}(g) ds  > \sum_{j=1}^r d_j. \]
	Note that if $S_t$ and $G_t$ are independent, it holds that $\pi_{S_t|g_t = g}(s) = \pi_{S_t}(s)$ and thus $\Sigma_{G_t = g} = \Sigma$ and $d_j = d_{G_t = g, j}$ for all $g$. In this case, Inequality (\ref{eqn:noisy-spec-1}) becomes an equality. 
	This inequality can become a strict inequality only if $G_t$ and $S_t$ are dependent, but $G_t$ does not need to be equal to $S_t$ in order to explain more variation than the constant loading model for the same number of factors.

	This discussion is based on the population model, and for the estimated model, we need to take the estimation error into account, which adds another layer of complexity. The goal of this discussion is to formalize the idea of Section \ref{sec:robust}, namely that the state-varying model provides a more parsimonious representation. In our empirical applications, we show that we need considerably more factors in the constant loading model to explain the same amount of variation as in the state-varying model. This confirms that the states that we condition on are relevant and cannot be independent of the true underlying state process.

	\subsection{Asymptotic Covariance Matrix of Generalized Correlation Statistic}
	
	The matrix $D$ in Theorem 5 is given by
	$D^{\top} = \begin{bmatrix}
	C_{1,1}^{\top} & C_{1,2}^{\top} & C_{2,1}^{\top} & C_{2,2}^{\top}
	\end{bmatrix}$, 
	\begin{eqnarray*}
		C_{1,1} &=& \begin{bmatrix}
			M_{1,1,2}^\T  \otimes M_{1,1,1} + M_{1,1,3} \otimes M_{1,1,4}^\T  + M_{1,1,5} \otimes M_{1,1,6}^\T  + M_{1,1,8}^\T  \otimes M_{1,1,7} & 0 & 0 & 0   
		\end{bmatrix} \\
		C_{1,2} &=& \begin{bmatrix}
			M_{1,2,3} \otimes M_{1,2,4}^\T  & M_{1,2,2}^\T  \otimes M_{1,2,1} & M_{1,2,5} \otimes M_{1,2,6}^\T  & M_{1,2,8}^\T  \otimes M_{1,2,7}
		\end{bmatrix} \\
		C_{2,1} &=& \begin{bmatrix}
			M_{2,1,8}^\T  \otimes M_{2,1,7} & M_{2,1,5} \otimes M_{2,1,6}^\T  & M_{2,1,2}^\T  \otimes M_{2,1,1} & M_{2,1,3} \otimes M_{2,1,4}^\T 
		\end{bmatrix} \\
		C_{2,2} &=& \begin{bmatrix}
			0 & 0 & 0& M_{2,2,2}^\T  \otimes M_{2,2,1} + M_{2,2,3} \otimes M_{2,2,4}^\T  + M_{2,2,5} \otimes M_{2,2,6}^\T  + M_{2,2,8}^\T  \otimes M_{2,2,7}
		\end{bmatrix}
	\end{eqnarray*}
	$M_{l,l',1} =  \left(V^{s_l} \right)^{-1}  ((Q^{s_l})^\T )^{-1} \Sigma_{\lambda_l,\lambda_l} $, $M_{l,l',2} = (Q^{s_{l'}} )^{-1}$, $M_{l,l',3} = \left(V^{s_l} \right)^{-1} ((Q^{s_l})^\T )^{-1} $, $M_{l,l',4} = \Sigma_{\lambda_l,\lambda_{l'}} (Q^{s_{l'}} )^{-1} $, $M_{l,l',5} = ((Q^{s_l})^\T )^{-1}$, $M_{l,l',6} = \Sigma_{\lambda_{l'},\lambda_{l'}} (Q^{s_{l'}})^{-1} \left( V^{s_{l'}} \right)^{-1}  $, $M_{l,l',7} = ((Q^{s_l})^\T )^{-1} \Sigma_{\lambda_l,\lambda_{l'}} $, and $M_{l,l',8} = (Q^{s_{l'}})^{-1} \left( V^{s_{l'}} \right)^{-1}  $

	A plug-in estimator $\hat D$ is $\hat{M}_{l,l',1} =  \left(\bar{V}^{s_l} \right)^{-1}  \frac{1}{N} \sum_{i=1}^N \bar{\lambda}_{li}\bar{\lambda}_{li}^\T $, $\hat{M}_{l,l',2} = I_r$, $\hat{M}_{l,l',3} =\left(\bar{V}^{s_l} \right)^{-1} $, $\hat{M}_{l,l',4} = \frac{1}{N} \sum_{i=1}^N \bar{\lambda}_{li}\bar{\lambda}_{l'i}^\T $, $\hat{M}_{l,l',5} = I_r$, $\hat{M}_{l,l',6} = \frac{1}{N} \sum_{i=1}^N \bar{\lambda}_{l'i}\bar{\lambda}_{l'i}^\T $, $\hat{M}_{l,l',7}= \frac{1}{N} \sum_{i=1}^N \bar{\lambda}_{li}\bar{\lambda}_{l'i}^\T $, $\hat{M}_{l,l',8} = \left(V^{s_{l'}} \right)^{-1}$. 
	Lemma \ref{consistentcovariance} discusses conditions for consistency of the plug-in estimator.

	\subsection{Estimator for the Number of Factors}\label{sec:number}
	
	We generalize the information criterion based estimator for the number of factors of \cite{bai2002determining} to our setup. As the derivations for the asymptotic distribution provide the counterpart of the upper bounds derived
	in \cite{bai2002determining} for our setup, it is relatively straightforward to obtain an estimator for the number of factors.

	For a given number of candidate factors $k$, we define the loss function weighted by the kernel as
	\begin{align}
	V^s(k,F^{s,k}) =& \frac{1}{N T(s)} \sum_{i = 1}^{N} \sum_{t = 1}^{T}  (X^s_{it} - (\Lambda_i(s)^k)^\T F^{s,k}_t)^2, \label{eqn:loss-function}
	\end{align}
	where $F^{s,k} = K_s^{1/2} F^k$, $F^k$ is a matrix of $k$ factors, $\Lambda(s)^k$ is a matrix of $k$ conditional loadings, and $X^s = X K_s^{1/2}$. The sum of the squared residuals where $k$ factors are estimated is denoted as
	\[V^s(k,\hat{F}^{s,k}) = \frac{1}{N T(s)} \sum_{i = 1}^{N} \sum_{t = 1}^{T} (X^s_{it} - (\Lambda_i(s)^k)^\T \hat{F}^{s,k}_t)^2. \]
	Similar to \cite{bai2002determining}, we consider the following objective function to estimate the number of factors:
	\[\mathcal{L}^s(k) = V^s(k,\hat{F}^{s,k})  + g(N,T,h),  \]
	where $g(N,T,h)$ is a penalty function. The following lemma generalizes the \cite{bai2002determining} estimator for the number of factors.
	\setcounter{lemma}{1}
	\begin{lemma}\label{lemma:num-factors}
		Under Assumptions \ref{Ass:Ident}-\ref{ass_err}, $N \rightarrow \infty$, $Th \rightarrow \infty$,  $Nh^2 \rightarrow 0$ and $T h^3 \rightarrow 0$ and the estimated factors minimize \eqref{eqn:loss-function}. Let $\hat{k} = \argmin_{k \leq k_{\max}}  \mathcal{L}^s(k,\hat{F}^s) $. Then $\lim_{N, T \rightarrow \infty} P(\hat{k} = r ) = 1$ if (i) $g(N, T, h) \rightarrow 0$ and (ii) $\delta_{NT,h}^2 g(N, T, h) \rightarrow \infty$ as $N, T \rightarrow \infty$, where $ \delta_{NT,h}^2 = \min(N, Th)$.
	\end{lemma}
	The proof of Lemma \ref{lemma:num-factors} is a straightforward extension of the proof of Theorem 2 in \cite{bai2002determining}. The derivation of Lemma \ref{lemma:num-factors} is based on the following lemma.
	\begin{lemma}\label{lemma:num-factors-prep}
		Under Assumptions \ref{Ass:Ident}-\ref{ass_err}, $N \rightarrow \infty$, $Th \rightarrow \infty$,  $Nh^2 \rightarrow 0$ and $T h^3 \rightarrow 0$, the following statements hold. $H^{s,k}$ denotes a rotation matrix as defined in Theorem \ref{thm_consistency}.
		\begin{enumerate}
			\item For any $k$ with $1 \leq k \leq r$,  we have $V^s(k,\hat{F}^{s,k}) - V^s(k,F^s H^{s,k} ) = O(\delta_{NT,h}^{-1})$.
			\item For each $k$ with $k < r$, there exists a $\tau_k$ such that $\mathrm{plim} \inf_{N, T\rightarrow \infty}  V^s(k,F^s H^{s,k} ) -  V^s(r,F^s  ) = \tau_k $.
			\item For any fixed $k$ with $k \geq r$, $V^s(k,\hat{F}^{s,k}) - V^s(r,\hat{F}^{s,r}) = O_p(\delta_{NT,h}^{-2}) $.
		\end{enumerate}
	\end{lemma}
	The major difference between Lemma \ref{lemma:num-factors-prep} and Lemmas 2, 3, 4 in \cite{bai2002determining} is the convergence rate. Our convergence rate is the smaller of $N$ and $Th$ because the kernel projection implicitly uses a subset of the time-series observations. Moreover, $Nh^2 \rightarrow 0$ and $T h^3 \rightarrow 0$ imply $\delta_{NT, h} h \rightarrow 0$ and  guarantee that the kernel bias terms are negligible in the difference of $V^s(k,F^k) $ evaluated for different $k$ and $F^{s,k}$. We can show this property using a similar proof as for Theorem \ref{thm_consistency}. Hence, the difference mainly depends on the convergence rate of $ \frac{1}{T} \sum_{t=1}^T \norm{\hat{F}^s_t - (H^s)^\T  F^s_t}^2$, which is $\delta_{NT,h}^2$. 
	
	In this Internet Appendix, we also show how to select the number of factors and bandwidth parameter based on cross-validation arguments. In more detail, the number of factors and the bandwidth can be viewed as tuning parameters that can be selected on a validation data set to maximize the amount of explained variation, while the model itself is estimated on the training data. Then, the model can be evaluated out-of-sample on the test data. We explore this idea in simulation and empirical applications. We confirm that the number of factors and bandwidth chosen optimally on the validation data also maximize the out-of-sample $R^2$ on the test data.

	\subsection{Generalized Correlation Test Under the Alternative Hypothesis}\label{sec:testalternative}

	\begin{lemma}\label{lemma:rho}
		Under Assumptions \ref{Ass:Ident}-\ref{doublesum} and under the alternative hypothesis, if $N/\sqrt{Th} \rightarrow 0$, $Nh \rightarrow \infty$, $Th \rightarrow \infty$, $\sqrt{Th}/N \rightarrow 0$, $Nh^2 \rightarrow 0$, $Th^3 \rightarrow 0$, and $NTh^3 \rightarrow 0$, 
		\begin{eqnarray} \label{ass_lamjoint}
		\sqrt{N} \left(  \begin{bmatrix}
		\tvec\left(\frac{1}{N} \Lambda_1^\T \Lambda_1 \right) \\
		\tvec\left(\frac{1}{N} \Lambda_1^\T \Lambda_2 \right) \\
		\tvec\left(\frac{1}{N} \Lambda_2^\T  \Lambda_1 \right) \\
		\tvec\left(\frac{1}{N} \Lambda_2^\T  \Lambda_2 \right) \\
		\end{bmatrix} - \begin{bmatrix}
		\Sigma_{\Lambda_1, \Lambda_1} \\
		\Sigma_{\Lambda_1, \Lambda_2} \\
		\Sigma_{\Lambda_2, \Lambda_1} \\
		\Sigma_{\Lambda_2, \Lambda_2} \\
		\end{bmatrix}  \right)
		\xrightarrow{d} 
		N(0, \Pi)
		\end{eqnarray}
		and the eigenvalues of $\Sigma_{\Lambda_1, \Lambda_1}$, $\Sigma_{\Lambda_2, \Lambda_2}$ are bounded away from 0, then we have 
		\begin{eqnarray*}
			\sqrt{N} (\hat{\rho} - \bar{\rho} ) \rightarrow N(0, \xi^\T  \Pi \xi)
		\end{eqnarray*}
		where $\xi = \begin{bmatrix}
		\tvec\left( -(G_1^{-1} G_2 G_4^{-1} G_3 G_1^{-1})^{\top} \right) \\ \tvec\left( G_1^{-1} G_2 G_4^{-1} \right) \\ \tvec\left( G_4^{-1} G_3 G_1^{-1} \right) \\ \tvec\left( -(G_4^{-1} G_3 G_1^{-1} G_2 G_4^{-1} )^{\top} \right) \end{bmatrix}$, $G_1 = \Sigma_{\Lambda_1, \Lambda_1}$, $G_2 = \Sigma_{\Lambda_1, \Lambda_2}$, $G_3 = \Sigma_{\Lambda_2, \Lambda_1}$, $G_4 = \Sigma_{\Lambda_2, \Lambda_2}$ and $\bar{\rho} = tr \left( \Sigma_{\Lambda_1, \Lambda_1}^{-1} \Sigma_{\Lambda_1, \Lambda_2} \Sigma_{\Lambda_2, \Lambda_2}^{-1} \Sigma_{\Lambda_2, \Lambda_1} \right)$. 
	\end{lemma}

	\subsection{Discrete State Space}\label{sec:disrete}	
	The focus of this paper is to estimate latent factor models conditioned on a continuous state space. Assumption \ref{ass_state} assumes that the state space is continuous.	
	If the state space is discrete, for example, the recession indicator that we consider in our empirical application for the yield data, we can, in principle, also use the kernel method to estimate the factor model conditional on a discrete state outcome. If we use a uniform kernel $k(u) = \frac{1}{2} \mathbbm{1}(|u| \leq 1)$ with a very small bandwidth, then observations in the time periods with a different state value have zero kernel weight. For example, if we want to estimate a factor model in recessions, using a uniform kernel with a small bandwidth is equivalent to using the data only in recessions to estimate the factor model. However, our proofs, that are based on a kernel applied to a continuous state-space model, would require slight modifications.

The case of a discrete state space is actually a simpler problem. Estimating a latent model conditioned on a discrete state outcome, is simply PCA applied to only those time periods where we observe the specific state realizations. We can directly extend the theoretical arguments in \cite{bai2003inferential} to show the asymptotic normality of the estimated loadings $\hat{\Lambda}(s)$ in state $s$ and estimated factors $\hat{F}_t$ for $t$ that satisfies $S_t = s$. The convergence rate for $\hat{\Lambda}(s)$ is $T_s$, where $T_s = \frac{1}{T} \sum_{t = 1}^T \mathbbm{1}(S_t = s)$ is the number of time periods in state $s$. The convergence rate for $\hat{F}_t$ is $\sqrt{N}$, the same as in \cite{bai2003inferential}. 
	
\newpage

	\section{Additional Empirical Results}\label{sec:EmpApp}
	
	\subsection{Term Structure Factors Conditioned on Unemployment Rates}
	
	We use the U.S. unemployment rate as the third state variable. Unemployment rates, released monthly by the Bureau of Labor Statistics, rise or fall in the wake of changing economic conditions. Although it is lagged, it can roughly indicate how well the economy is doing. The monthly unemployment rates from 07/2001 to 12/2016 are shown in Figure \ref{unemploy}. The unemployment rate reached its peak at 10\% during the financial crisis and dropped afterward. 
	
	Similar to using the VIX as the state variable, we estimate a factor model conditional on every possible value of the unemployment rate. Due to the wider range of values of the unemployment rate compared to the log-normalized VIX, we choose a larger bandwidth $h=1$.\footnote{{Results using different bandwidths are similar and available upon request.}} The proportion of variance explained by the first three factors is shown in Figure \ref{prop-var-unemploy}. The first factor becomes less important as the unemployment rate rises, while the second factor gains importance, even close to the level of the first factor.  
	
	\begin{figure}[H]
		\centering
		\begin{subfigure}{.35\textwidth}
			\centering
			\includegraphics[width=1\linewidth]{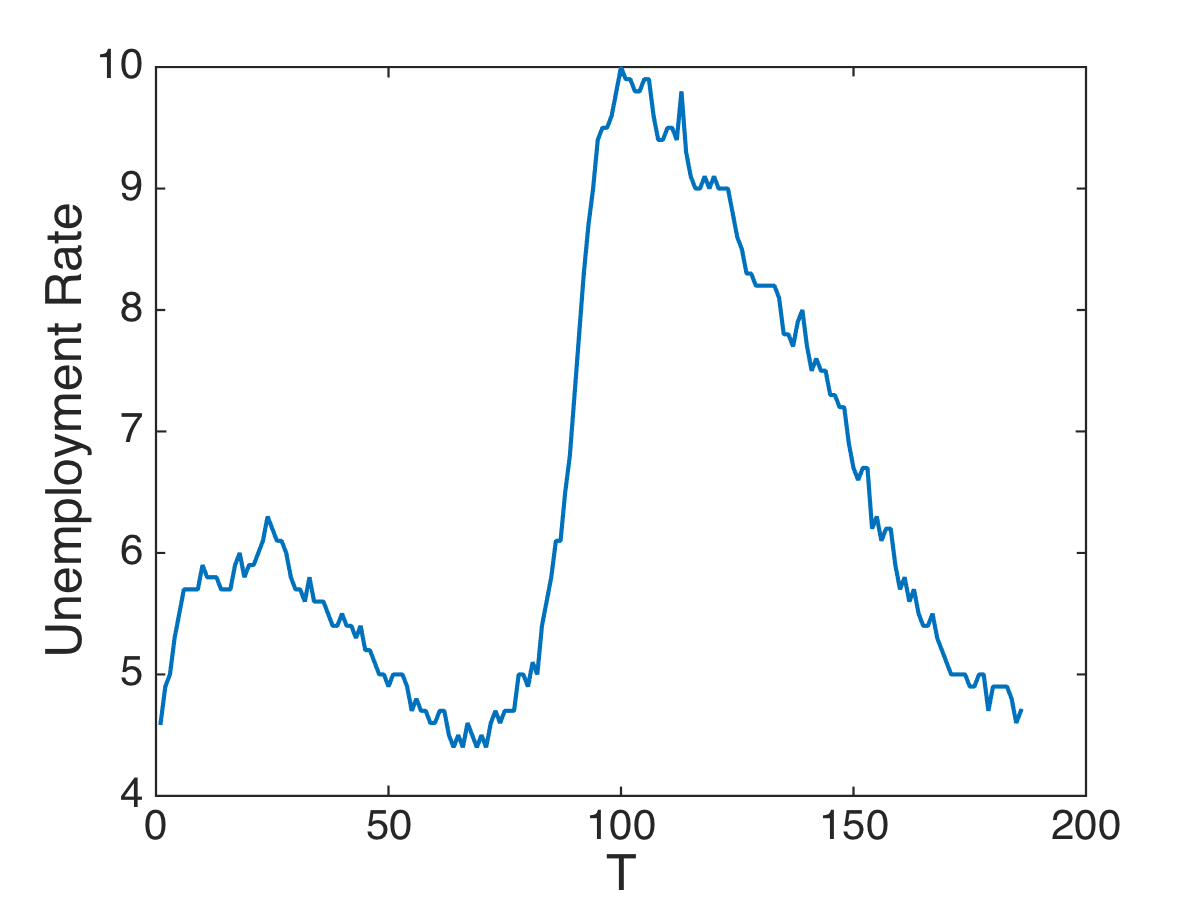}
			\caption{\\ Unemployment Rate}
			\label{unemploy}
		\end{subfigure}%
		\begin{subfigure}{.35\textwidth}
			\centering
			\includegraphics[width=1\linewidth]{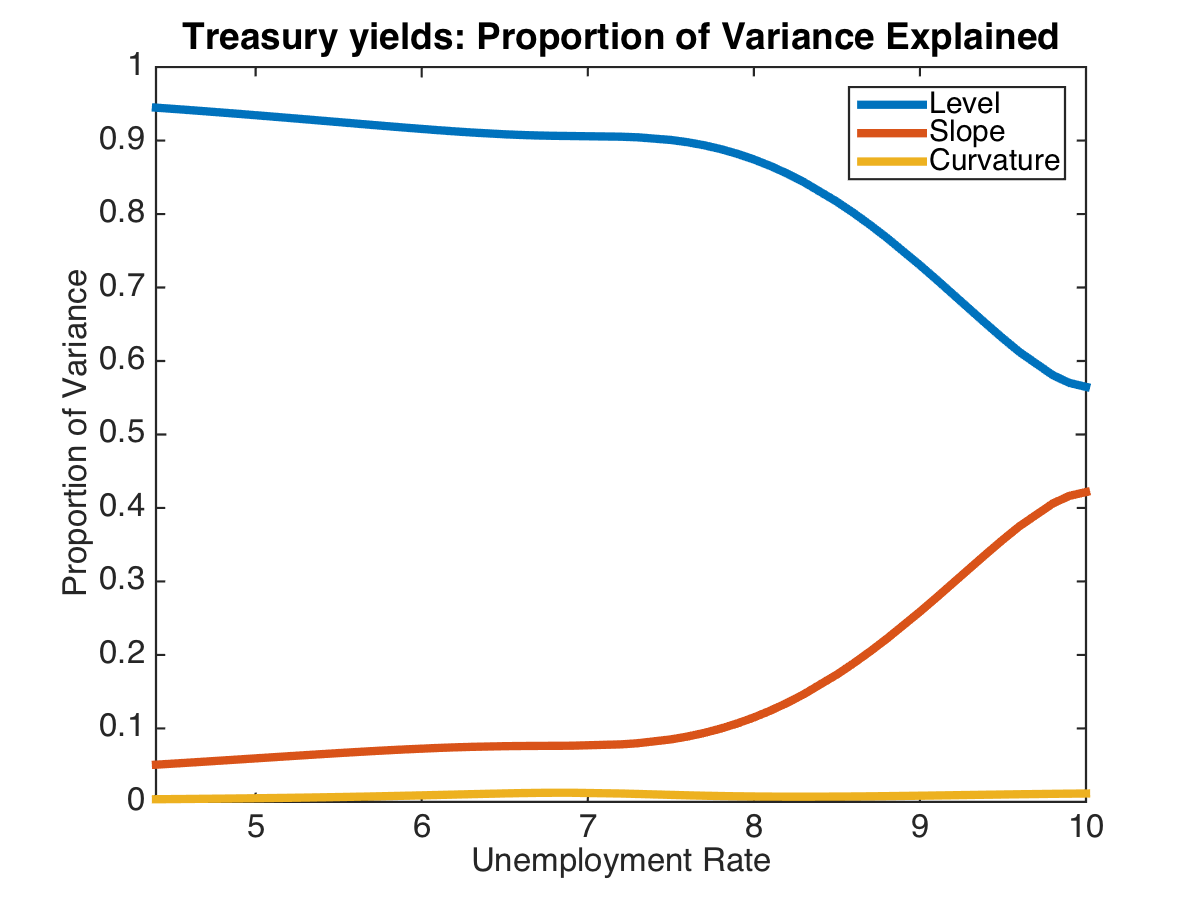}
			\caption{\\ Proportion of variance explained}
			\label{prop-var-unemploy}
		\end{subfigure}%
		\begin{subfigure}{.35\textwidth}
			\centering
			\includegraphics[width=1\linewidth]{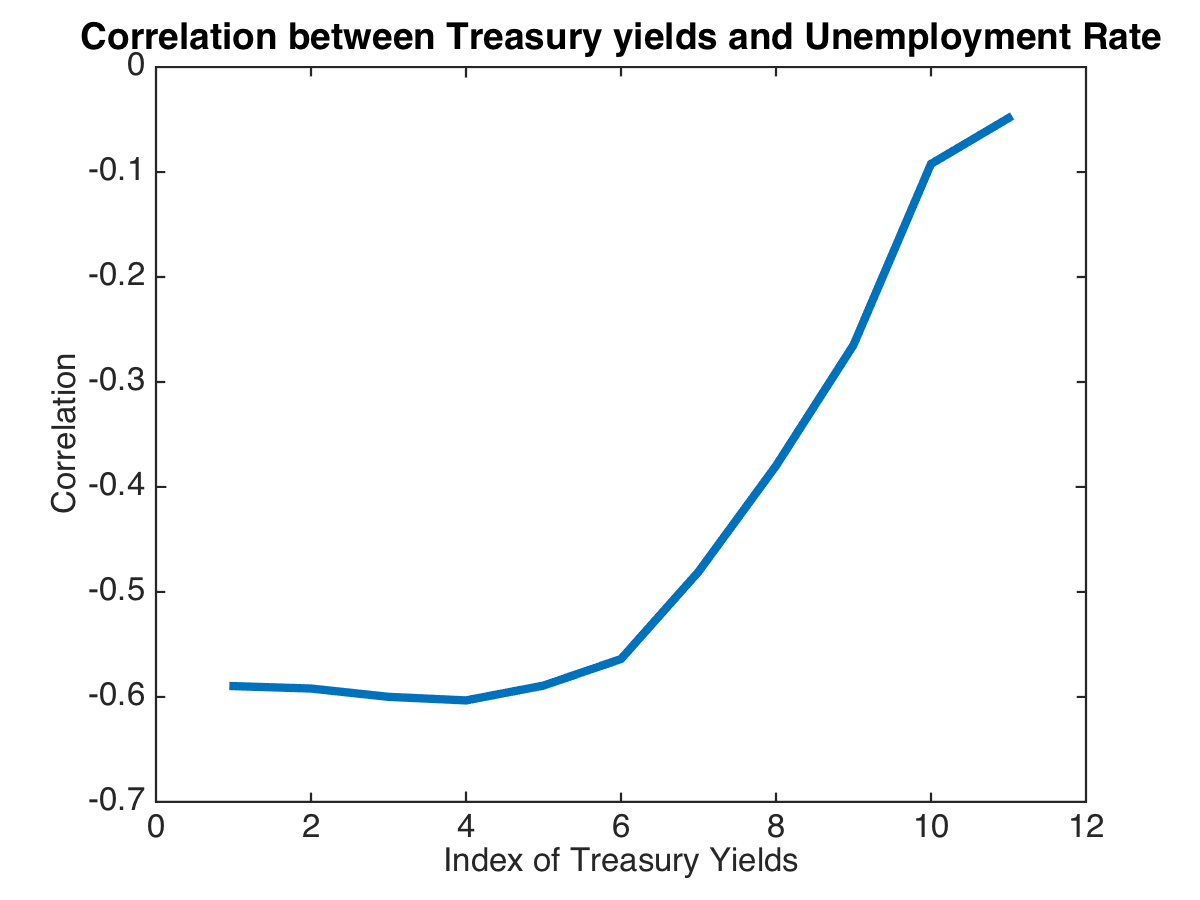}
			\caption{\\ Correlation}
			\label{corr}
		\end{subfigure}
		\caption{Unemployment rate from 07/31/2001 to 12/01/2016 (index represents the number of months from 07/31/2001), proportion of variance explained by the first three factors in different unemployment rates, and the correlation between yields and unemployment rate for different maturity bonds}
	\end{figure}

	Figure \ref{corr} shows that yields of short-term bonds are strongly negatively correlated with the unemployment rate, while yields of long-term bonds are almost uncorrelated with the unemployment rate. This correlation is captured by the variation of factor models in different unemployment rates. Figure \ref{yield_unemploy} shows how loadings of different maturity bonds change with unemployment rates. The weights of long-term bonds (10-, 20-, 30-year bonds) in the level factor drop significantly, to even close to 0, as the unemployment rate increases to its maximum. The level factor's value mainly depends on yields of short-term (< 3-year) bonds. A recession is usually associated with a high VIX and a high unemployment rate. The results conditional on unemployment rates are consistent with the previous two results. The relative shifts in the three factors have a similar pattern as those for the VIX and the recession indicator. The patterns are stronger in terms of the magnitude of changes in factor loadings when using the unemployment rate as the state variable.

	Last but not least, we test for which unemployment rates the factor model changes. The result is shown in Figure \ref{yield_unemploy_gen_corr}. The values of the test statistic and the p-values exhibit a weaker pattern than in the case of the VIX. However, the general finding that the factor structure is different for state values far away from each other persists. We suspect that Figure \ref{yield_unemploy_gen_corr} is more scattered than Figure \ref{yield_vix_gen_corr} because our choice of bandwidth leads to a less smooth fit of the loading function than in the VIX case.

	\begin{figure}[H]
		\centering
		\begin{subfigure}{.35\textwidth}
			\centering
			\includegraphics[width=1\linewidth]{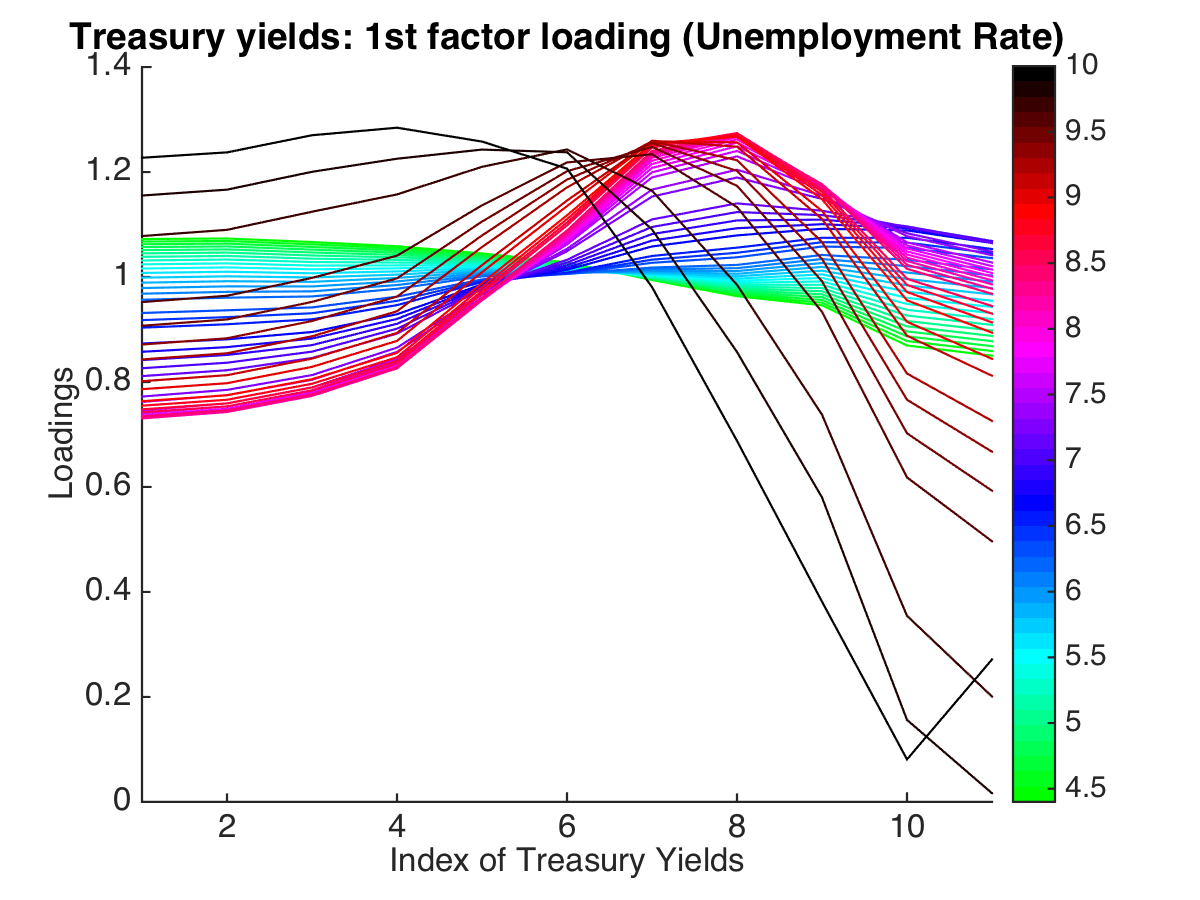}
			\caption{Level Factor}
		\end{subfigure}%
		\begin{subfigure}{.35\textwidth}
			\centering
			\includegraphics[width=1\linewidth]{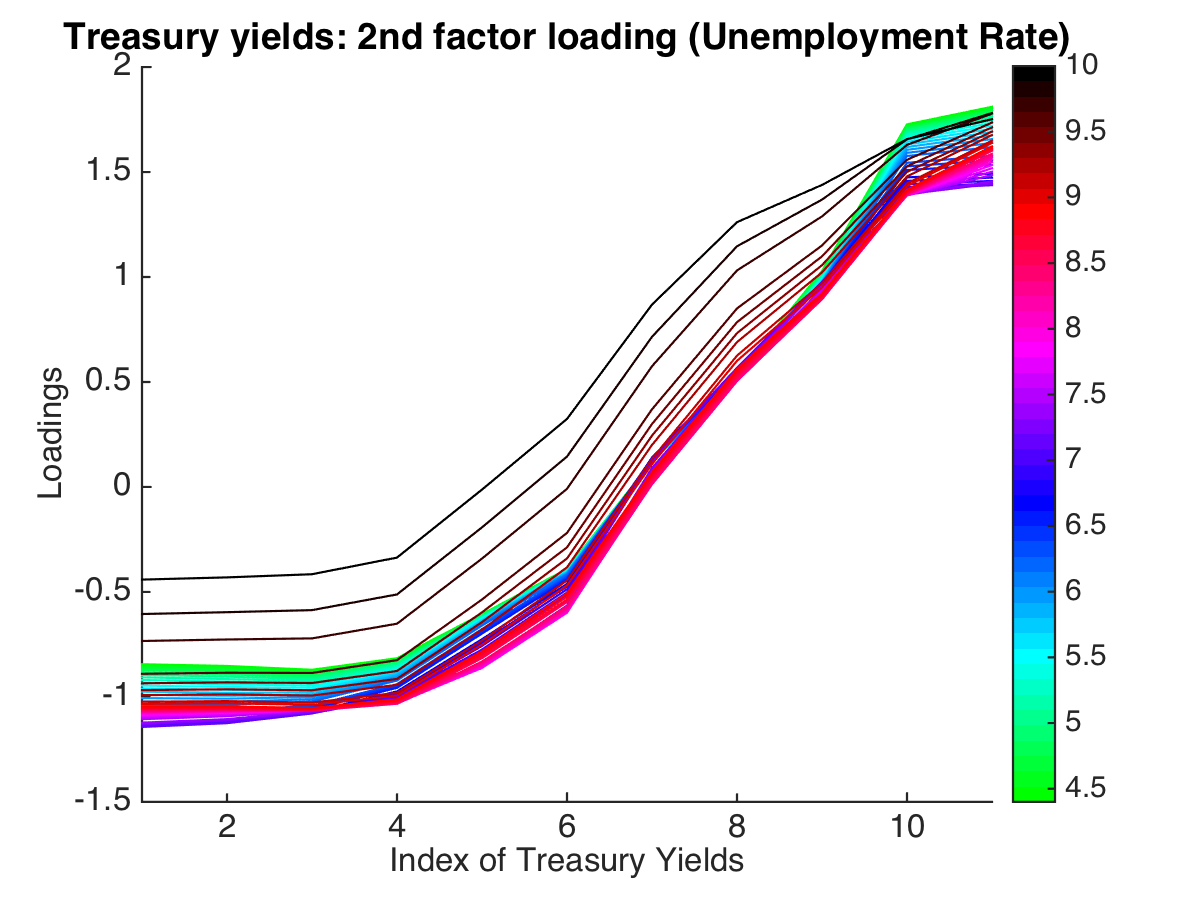}
			\caption{Slope Factor}
		\end{subfigure}%
		\begin{subfigure}{.35\textwidth}
			\centering
			\includegraphics[width=1\linewidth]{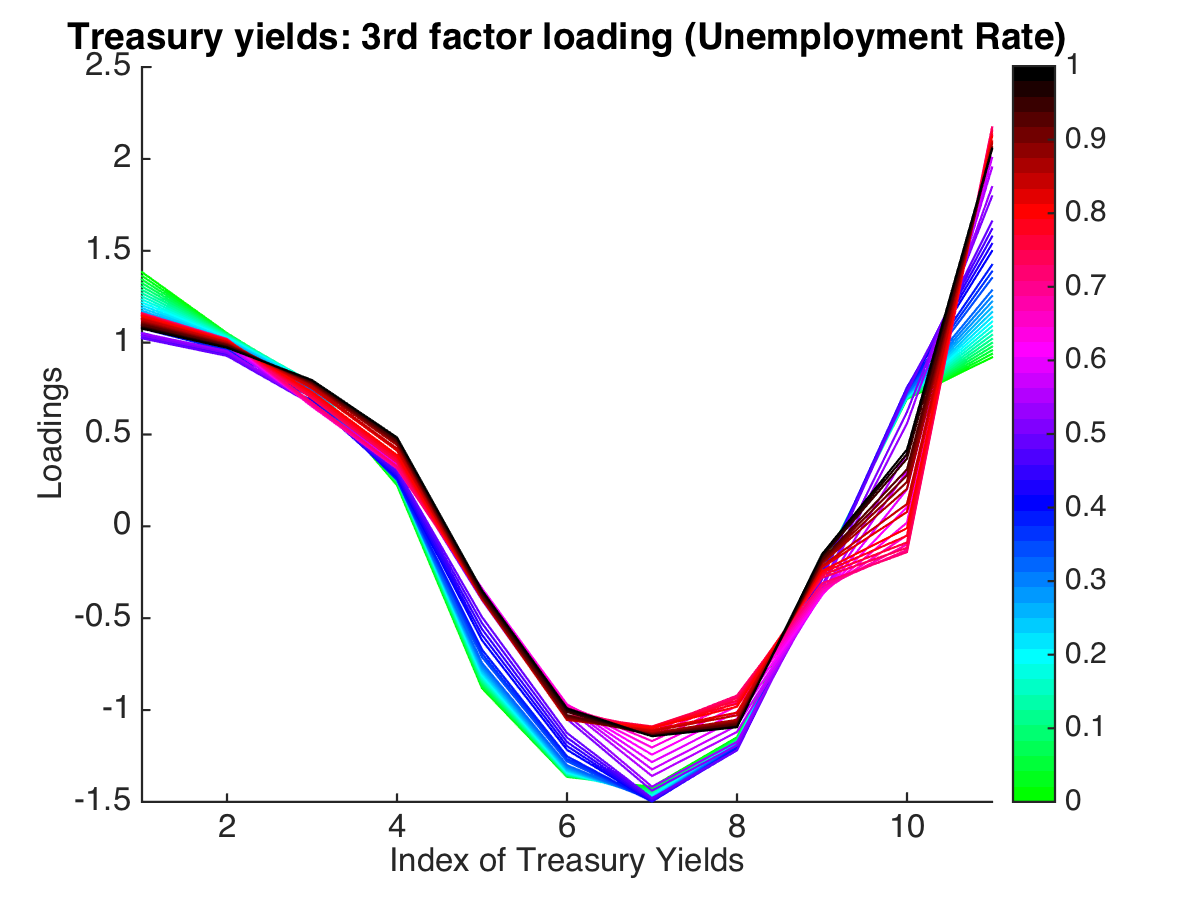}
			\caption{Curvature Factor}
		\end{subfigure}
		\caption{First three latent factor loadings for treasury securities conditioned on unemployment rate  (The color bar indicates unemployment rate).}
		\label{yield_unemploy}
	\end{figure}

	\begin{figure}[H]
		\centering
		\begin{subfigure}{.4\textwidth}
			\centering
			\includegraphics[width=1\linewidth]{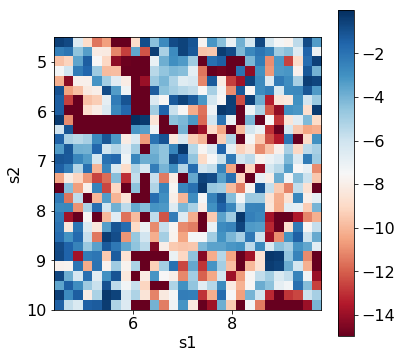}
			\caption{t-value}
		\end{subfigure}%
		\begin{subfigure}{.44\textwidth}
			\centering
			\includegraphics[width=1\linewidth]{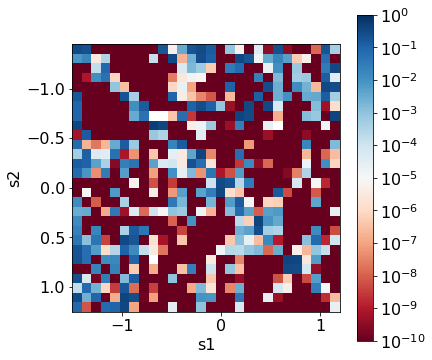}
			\caption{p-value}
		\end{subfigure}
		\caption{Generalized correlation test of estimated loadings in any paired states in US Treasury Securities data using unemployment rate as the state variable ($\mathcal{H}_0$: there exists a full rank matrix $G$, $\Lambda_{s_2} = \Lambda_{s_1} G$, $\mathcal{H}_1$: for any full rank matrix $G$, $\Lambda_{s_2} \neq \Lambda_{s_1} G$). $x$-axis and $y$-axis are both log-normalized VIX.  The value at point ($s_1, s_2$) in figure (a) represents the standardized generalized correlation (t-value) of $\bar{\Lambda}_{s_1}$ and $\bar{\Lambda}_{s_2}$. The value at point ($s_1, s_2$) in Figure (b) represents the p-value corresponding to the t-value in Figure (a)}
		\label{yield_unemploy_gen_corr}
	\end{figure}

	\subsection{Stock Returns Conditioned on VIX}

	\begin{figure}[H]
		\tcapfig{S\&P 500 Stock Return Data: Generalized Correlation Test in Any Paired States}
		\centering
		\begin{subfigure}[c]{.3\textwidth}
			\centering
			\includegraphics[width=1\linewidth]{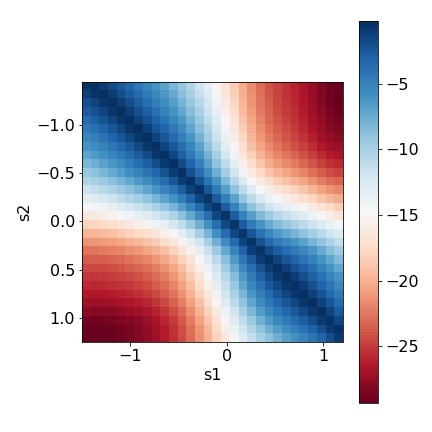}
			\caption{t-value (1 factor)}
		\end{subfigure}%
		\begin{subfigure}[c]{.3\textwidth}
			\centering
			\includegraphics[width=1\linewidth]{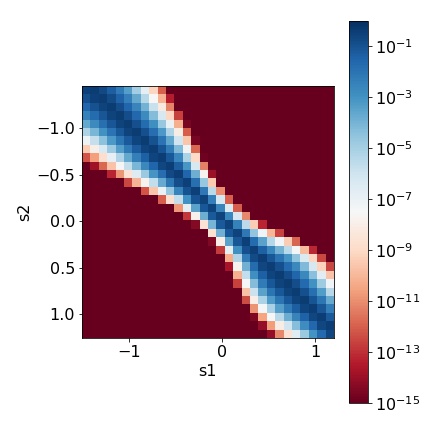}
			\caption{p-value (1 factor)}
		\end{subfigure}
		\begin{subfigure}[c]{.3\textwidth}
			\centering
			\includegraphics[width=1\linewidth]{Figures/stock_gen_corr_t_value_5_factors.jpg}
			\caption{t-value (5 factors)}
		\end{subfigure}%
		\begin{subfigure}[c]{.3\textwidth}
			\centering
			\includegraphics[width=1\linewidth]{Figures/stock_gen_corr_p_value_5_factors.jpg}
			\caption{p-value (5 factors)}
		\end{subfigure}
		\bnotefig{Generalized correlation test for S\&P500 returns with log-normalized VIX as state variable for 1 and 5 factors. (a) and (c) standardized generalized correlations (t-values) and (b) and (d) corresponding p-values.}
		\label{gen_corr_sp500}
	\end{figure}

	Figure \ref{gen_corr_sp500} shows the test results for the generalized correlation test for the combination of any two state outcomes of the VIX. We use a five-factor model motivated by the five factors of \cite{fama2015five} and \cite{lettaupelger2018}. The span of the loadings drastically changes with the realization of the VIX, which confirms the previous results. Even in a one-factor model, the span of the state-varying loadings is different from a constant factor model. This finding will be confirmed in the portfolio application.

	\begin{figure}[H]
		\tcapfig{S\&P 500 Stock Return Data: Out-of-Sample Sharpe Ratio of Mean-Variance Efficient Portfolio}
		\centering
		\includegraphics[width=0.7\linewidth]{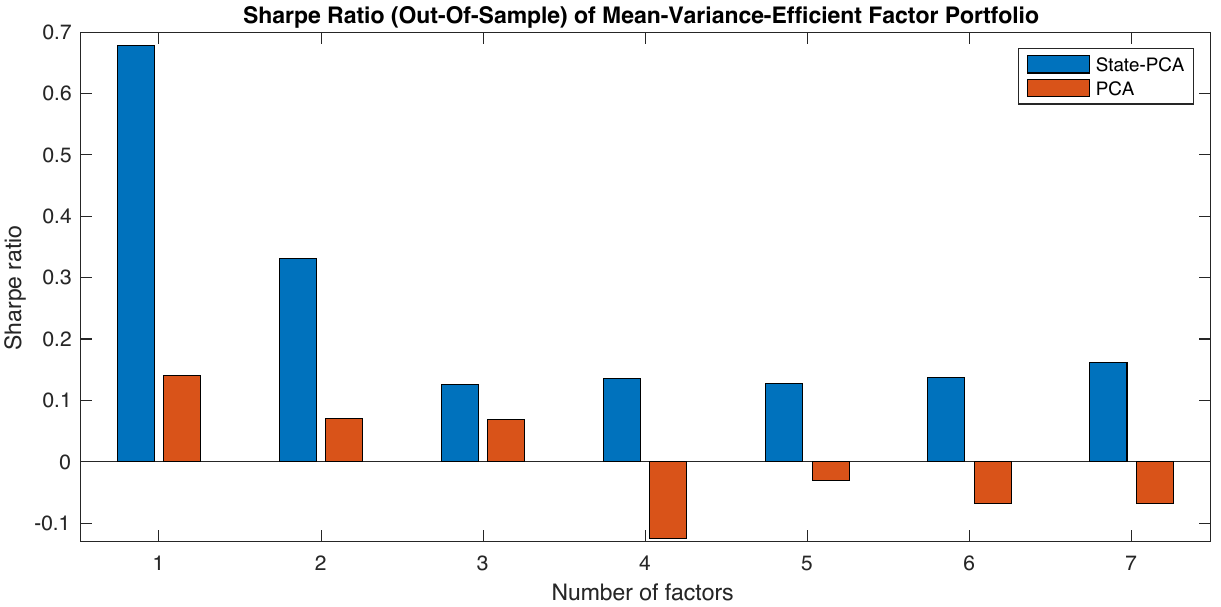}
		\bnotefig{S\&P500 stocks: Out-of-sample Sharpe ratio of mean-variance efficient portfolio based on the latent factors of the state-varying and constant loading model. State is log-normalized VIX. We use the first 3 years (2004/01/01-2006/12/31) for training and update out-of-sample results on an expanding window for the next 10 years (2007/01/01-2016/12/31).}
		\label{fig:SRsp500}
	\end{figure}
	
	Last but not least, we study the portfolio implications of the estimated factor. Arbitrage pricing theory implies that only systematic risk earns a risk premium, and hence the pricing kernel is spanned by the systematic factors. In other words, the mean-variance-efficient portfolio should only be composed of systematic factors. The mean-variance optimization problem based on latent factors has the appealing feature that it is independent of the rotational indeterminacy of the latent factors, i.e., it does not depend on a specific choice of $H(S_t)$. Note that although a risk premium can only be earned by a systematic factor, not every systematic factor is necessarily compensated for risk. This implies that, for example, in a five-factor model, an optimal portfolio invests only in these five factors, but some of the factors can have a weight of zero in the portfolio.\footnote{See \cite{lettaupelger2018} for a discussion.} Figure \ref{fig:SRsp500} plots the annualized Sharpe ratios\footnote{The Sharpe ratio is the expected return of an asset in excess of the risk-free rate normalized by its standard deviation. A higher Sharpe ratio corresponds to a higher average return for the same amount of risk measured by the standard deviation. The mean-variance efficient portfolio has the highest Sharpe ratio.} out-of-sample for different numbers of factors. As in-sample results are known for over-fitting,\footnote{See \cite{lettaupelger2018}.} we report only the out-of-sample results. Note that our investment strategy is an actually investable portfolio, as we use the VIX value at the market opening to calculate the portfolio returns from the opening to the closing of the market.\footnote{The conditional mean-variance efficient portfolio weights are a function of the state variable $S_t$, which is known at the time of the investment.} The mean-variance efficient portfolio based on state-varying factors has a significantly higher Sharpe ratio than the constant loading model. In fact, the constant loading factors can result in negative Sharpe ratios, which indicate that they are missing a crucial time-variation, which is captured by our model. It seems that a one-factor model with state-varying loadings captures most of the pricing information while adding more factors distorts the model.

	\subsection{Choice of Tuning Parameters}
	
	We show how to select the number of factors and the bandwidth parameter based on cross-validation arguments. In more detail, the number of factors and the bandwidth can be viewed as tuning parameters that can be selected on a validation data set to maximize the amount of explained variation, while the model itself is estimated on the training data. Then, the model can be evaluated out-of-sample on the test data. We confirm that the number of factors and bandwidth chosen optimally on the validation data also maximize the out-of-sample $R^2$ on the test data.

	We split the data into training, validation, and test sets. The first 25\% time periods constitute the training set, the following 25\% time periods are the validation set, and the remaining 50\% time periods represent the test data. We use the loadings and factor weights estimated on the training to choose the number of factors and bandwidth that maximize the $R^2$ on the validation data. Given the estimated model and tuning parameters, we evaluate the $R^2$ on the test data. We plot the $R^2_{X,\, \mathrm{val}}$ and $R^2_{X,\, \mathrm{test}}$ as a function of number of factors $k$ and bandwidth $h$.
	
	Figures \ref{fig:stock_vix_varying_k_h_out_of_sample}, \ref{fig:yield_vix_varying_k_h_out_of_sample} and \ref{fig:yield_unemploy_varying_k_h_out_of_sample} show that the optimal tuning parameters on the validation data yield the best out-of-sample results. In particular, the estimation results confirm that the estimation results are relatively robust to the choice of the bandwidth.
	
	We also report the results of using all the data for the in-sample estimation and an expanding window estimate for out-of-sample in Figures \ref{fig:stock_vix_compare}, \ref{fig:yield_vix_compare} and \ref{fig:yield_unemploy_compare}. For example, Figure \ref{fig:stock_vix_compare} reports the explained variation in- and out-of-sample for the state-varying and constant loading model for stock returns. For the out-of-sample results, we first estimate the loadings on the first three years of data and then update the loadings estimates on an expanding window to obtain the out-of-sample systematic component for the next ten years. Obviously, the state-varying factor model explains more variation than the constant loading model in- and out-of-sample for the same number of factors. Therefore, conditioning on the VIX results in a more parsimonious factor model to explain the co-movement in stock returns. Our results do not depend on a prior on the number of factors. In particular, it implies that stock returns do not follow a constant loading model and that the VIX is related to the source of time-variation. We do not require that the VIX explains all the time-variation in the loadings, but we show the conditional model provides a better description of the data than the unconditional one. We use the optimal bandwidth as suggested by the cross-validation, but as shown in Figures \ref{fig:stock_vix_varying_k_h_out_of_sample} and \ref{fig:stock_vix_varying_k_h_in_sample} the results are relatively robust to the choice of bandwidth.

In Figures \ref{fig:yield_vix_compare} and \ref{fig:yield_unemploy_compare}, we compare the amount of variation explained by different factor models for the yield data. Treasury yields are somewhat special in the sense that their variation can almost perfectly be explained by three factors. The state-varying factor model with three factors explains slightly more variation, comparable to a four-factor model with constant loadings. However, if the goal is to explain variation, a time-varying, and a constant three-factor model both perform well. The takeaway from this empirical application is to understand that the economic interpretation of ``level'', ``slope'' and ``curvature'' has to be used with caution. Depending on the economic conditions, the first PCAs are different, and for example, the first PCA factor can move from a simple average to a combination of a long-short and average portfolio. 	
	

	\subsubsection{Tuning Parameters for Stock Returns}

	\begin{figure}[H]
		\centering
		\begin{subfigure}{\textwidth}
			\centering
			\includegraphics[width=1\linewidth]{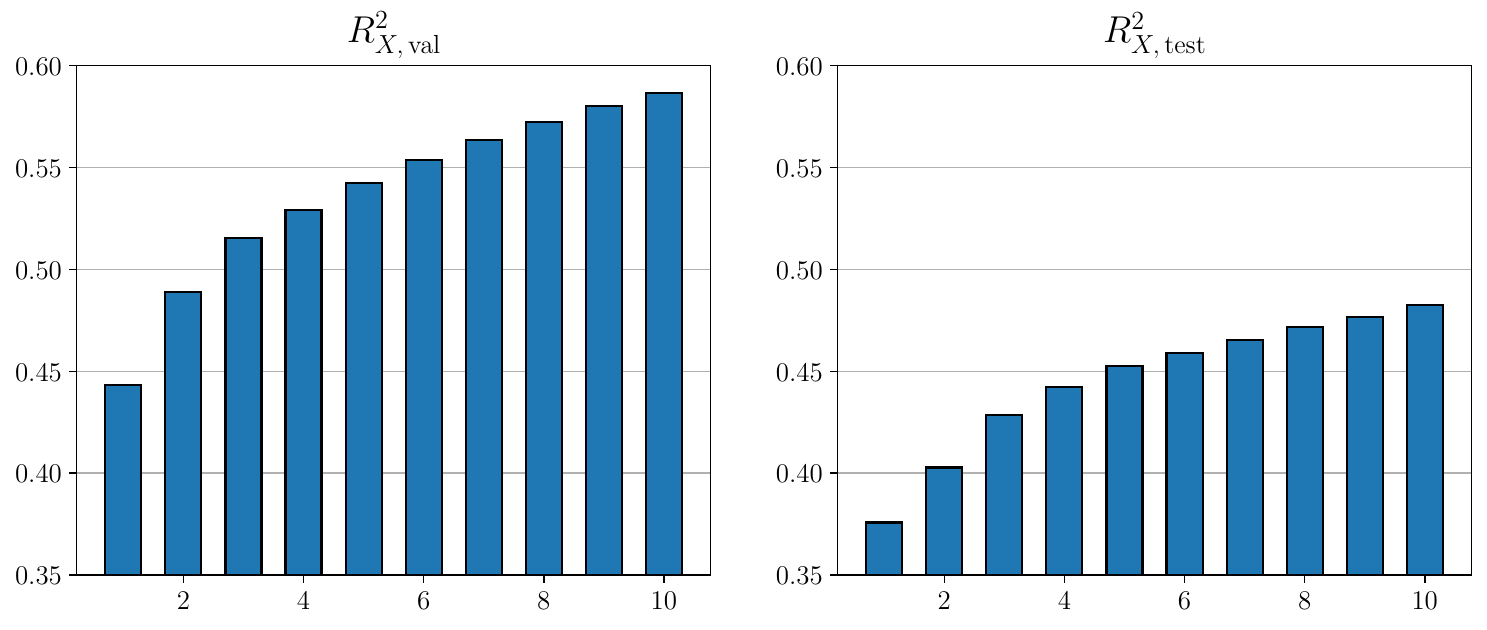}
			\caption{Constant loading factor model}
		\end{subfigure}
		\begin{subfigure}{\textwidth}
			\centering
			\includegraphics[width=1\linewidth]{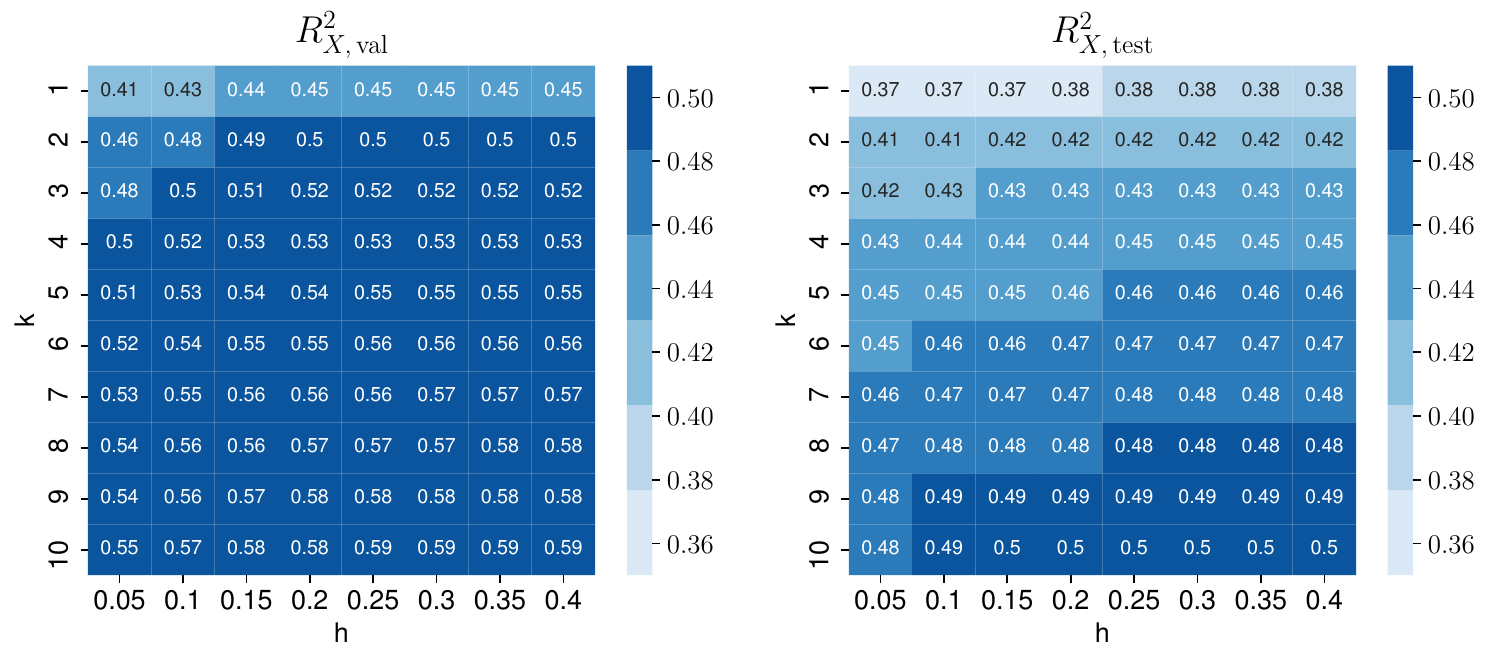}
			\caption{State-varying factor model}
		\end{subfigure}
		\caption{Optimal tuning parameter selection for S\&P 500 stock returns with the log-normalized VIX as state process. The figure plots $R^2_X$ on the validation and test data as function of the number of factors $k$ and the bandwidth $h$. The training data is 01/01/2004 to 31/12/2006, the validation data is 01/01/2007 to 12/31/2009 and the test data is 01/01/2010 to 12/31/2016.}
		\label{fig:stock_vix_varying_k_h_out_of_sample}
	\end{figure}

	\begin{figure}[H]
		\centering
		\begin{subfigure}{.5\textwidth}
			\centering
			\includegraphics[width=1\linewidth]{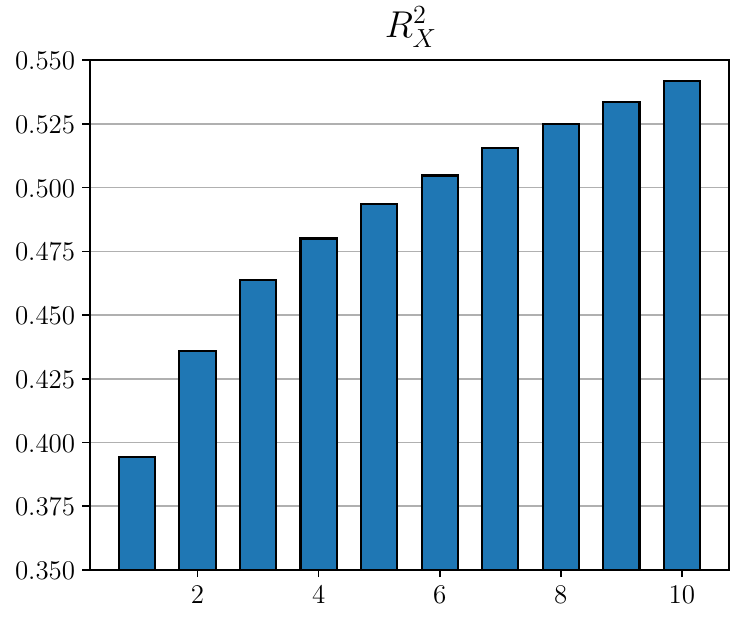}
			\caption{Constant loading factor model}
		\end{subfigure}%
		\begin{subfigure}{.5\textwidth}
			\centering
			\includegraphics[width=1\linewidth]{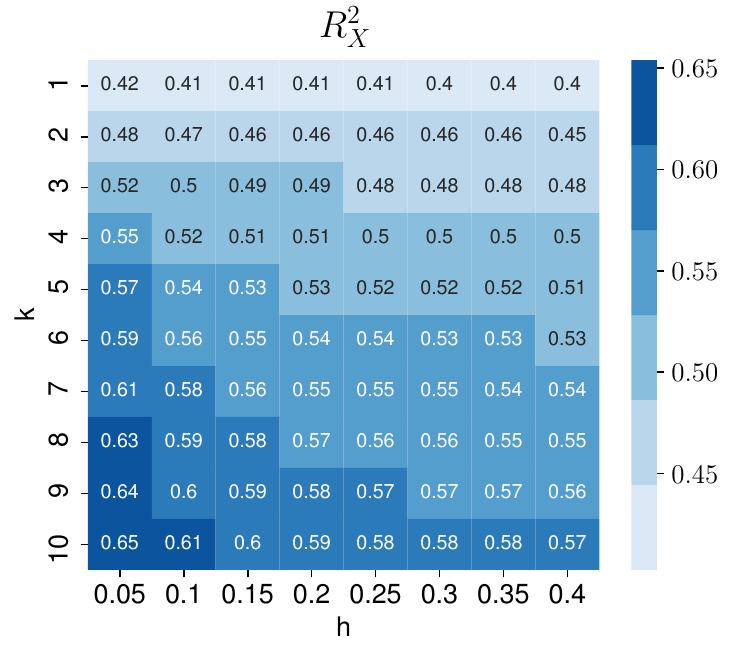}
			\caption{State-varying factor model}
		\end{subfigure}
		\caption{Optimal tuning parameter selection for S\&P 500 stock returns with the log-normalized VIX as state process. The figure plots $R^2_X$ on the in-sample training data as function of the number of factors $k$ and the bandwidth $h$. The in-sample results use the full data set.}	
		\label{fig:stock_vix_varying_k_h_in_sample}
	\end{figure}

	\begin{figure}[H]
		\centering
		\begin{subfigure}{.5\textwidth}
			\centering
			\includegraphics[width=1\linewidth]{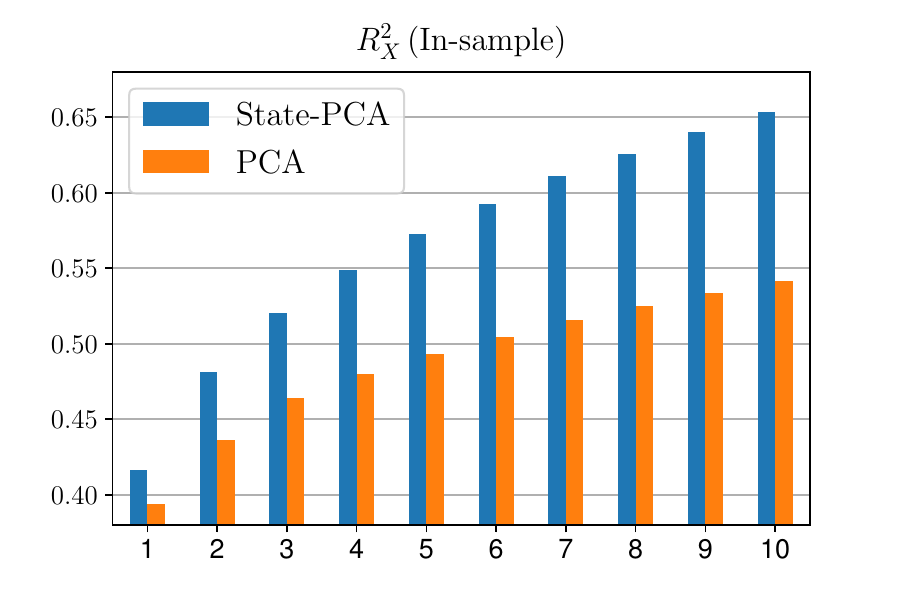}
		\end{subfigure}%
		\begin{subfigure}{.5\textwidth}
			\centering
			\includegraphics[width=1\linewidth]{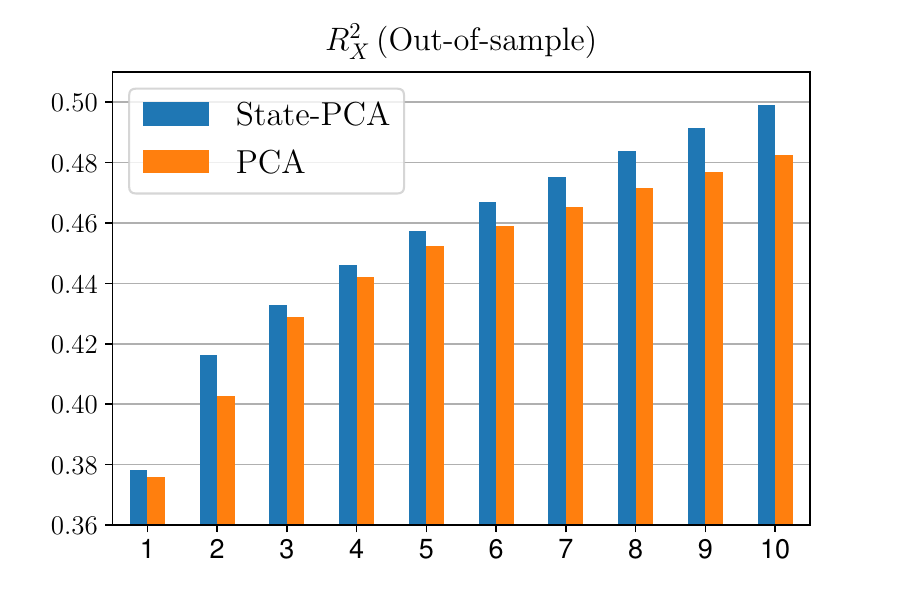}
		\end{subfigure}
		\caption{In-sample (training) and out-of-sample (test) $R^2_X$ for S\&P 500 stock returns with the log-normalized VIX as state process. The bandwidth is chosen optimally.}
		\label{fig:stock_vix_compare}
	\end{figure}
	
	\subsubsection{Tuning Parameters for Treasury Securities}

	\begin{figure}[H]
		\centering
		\begin{subfigure}{\textwidth}
			\centering
			\includegraphics[width=1\linewidth]{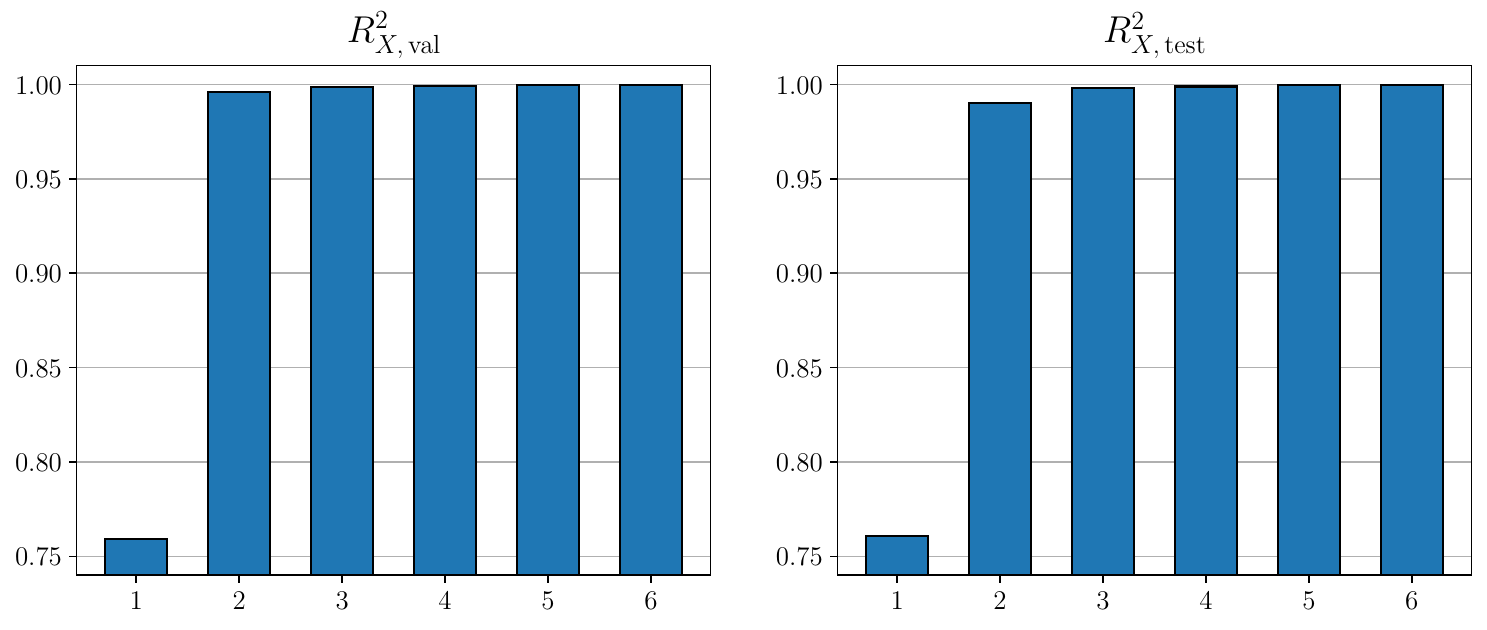}
			\caption{Constant loading factor model}
		\end{subfigure}
		\begin{subfigure}{\textwidth}
			\centering
			\includegraphics[width=1\linewidth]{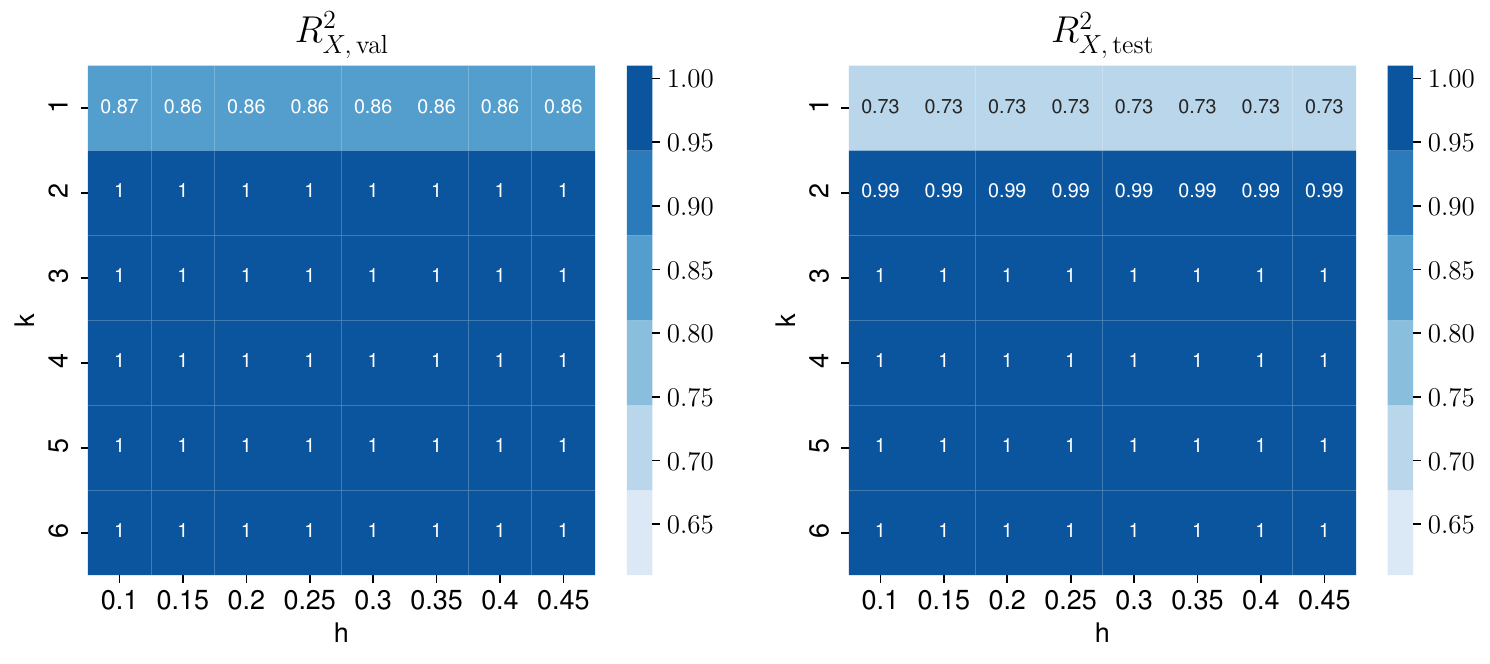}
			\caption{State-varying factor model}
		\end{subfigure}
		\caption{Optimal tuning parameter selection for U.S. Treasury Securities with the log-normalized VIX as state process. The figure plots $R^2_X$ on the validation and test data as function of the number of factors $k$ and the bandwidth $h$. The first 25\% time observations are the training data, the following 25\% time observations are the validation data, and the remaining 50\% time are the out-of-sample test data}
		\label{fig:yield_vix_varying_k_h_out_of_sample}
	\end{figure}

	\begin{figure}[H]
		\centering
		\begin{subfigure}{.5\textwidth}
			\centering
			\includegraphics[width=1\linewidth]{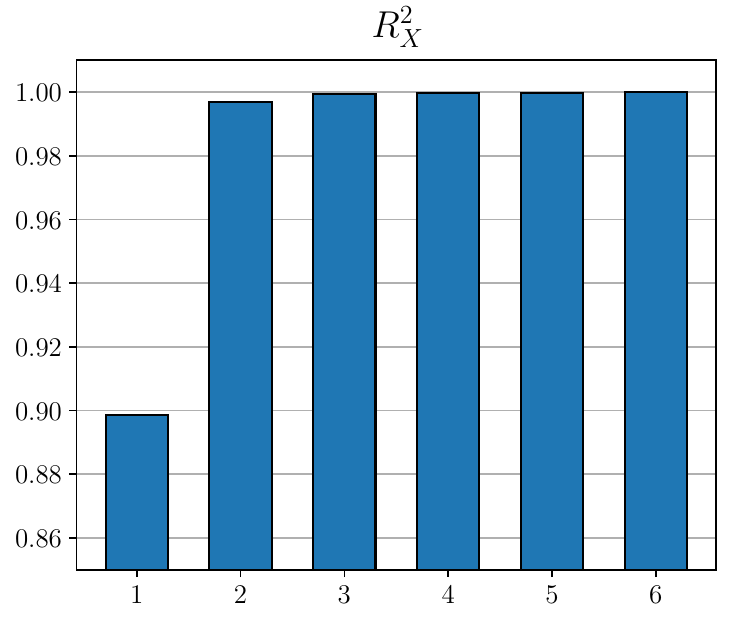}
			\caption{Constant loading factor model}
		\end{subfigure}%
		\begin{subfigure}{.5\textwidth}
			\centering
			\includegraphics[width=1\linewidth]{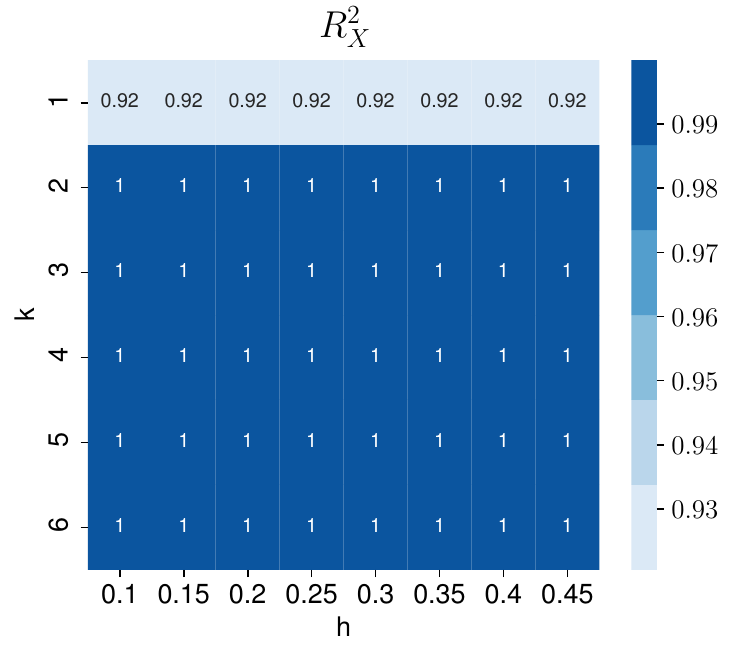}
			\caption{State-varying factor model}
		\end{subfigure}
		\caption{Optimal tuning parameter selection for U.S. Treasury Securities with the log-normalized VIX as state process. The figure plots $R^2_X$ on the in-sample training data as function of the number of factors $k$ and the bandwidth $h$. The in-sample results use the full data set.}
		\label{fig:yield_vix_varying_k_h_in_sample}
	\end{figure}

		\begin{figure}[H]
		\centering
		\begin{subfigure}{.5\textwidth}
			\centering
			\includegraphics[width=1\linewidth]{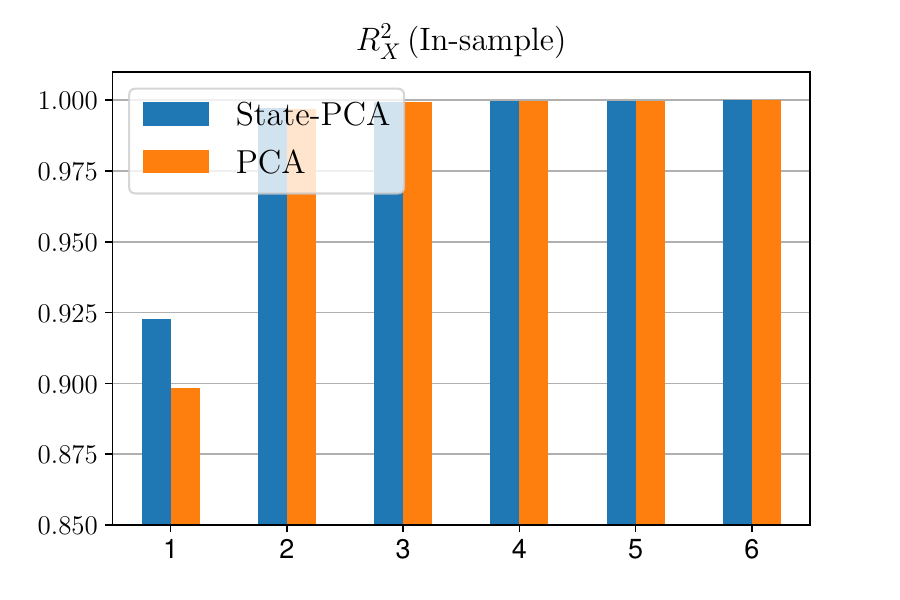}
		\end{subfigure}%
		\begin{subfigure}{.5\textwidth}
			\centering
			\includegraphics[width=1\linewidth]{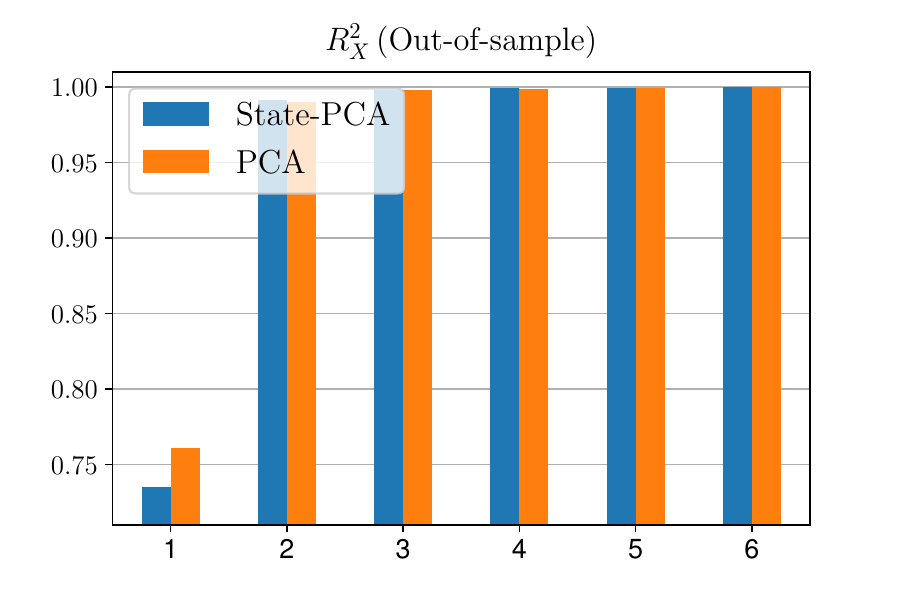}
		\end{subfigure}
		\caption{In-sample (training) and out-of-sample (test) $R^2_X$ for US Treasury Securities with the log-normalized VIX as state process. The bandwidth is chosen optimally. We use the first half of the data for the estimating the loadings and update them on an expanding window to obtain the out-of-sample common component.}
		\label{fig:yield_vix_compare}
	\end{figure}

	\begin{figure}[H]
		\centering
		\begin{subfigure}{\textwidth}
			\centering
			\includegraphics[width=1\linewidth]{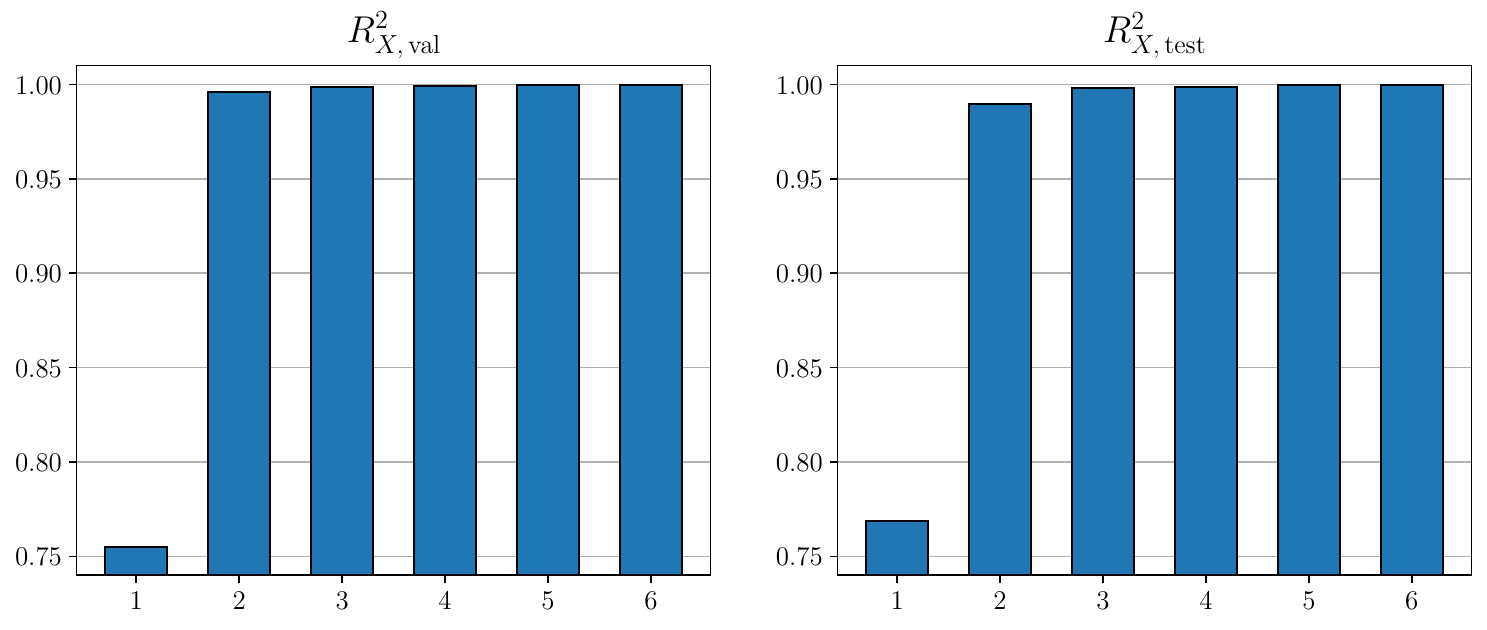}
			\caption{Constant loading factor model}
		\end{subfigure}
		\begin{subfigure}{\textwidth}
			\centering
			\includegraphics[width=1\linewidth]{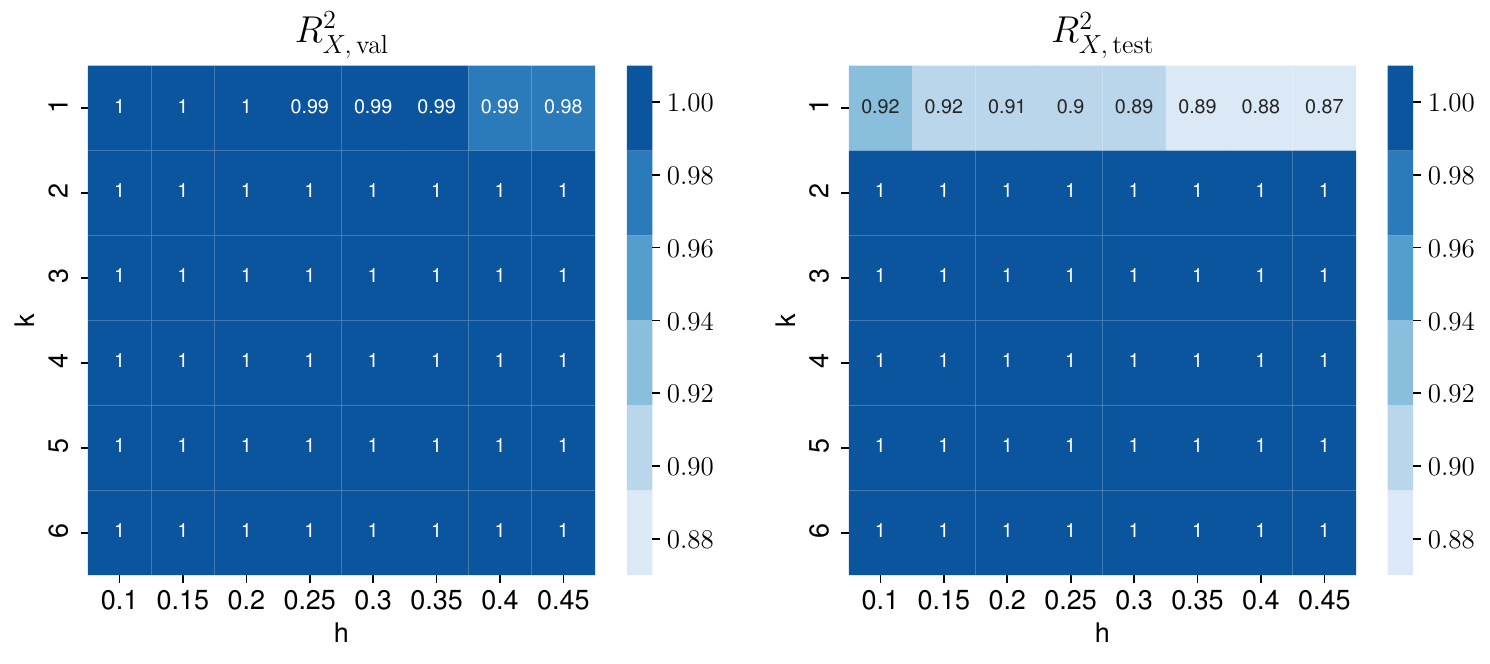}
			\caption{State-varying factor model}
		\end{subfigure}
		\caption{Optimal tuning parameter selection for U.S. Treasury Securities with the unemployment rate as state process. The figure plots $R^2_X$ on the validation and test data as function of the number of factors $k$ and the bandwidth $h$. The first 25\% time observations are the training data, the following 25\% time observations are the validation data, and the remaining 50\% time are the out-of-sample test data}
		\label{fig:yield_unemploy_varying_k_h_out_of_sample}
	\end{figure}
	
	\begin{figure}[H]
		\centering
		\begin{subfigure}{.5\textwidth}
			\centering
			\includegraphics[width=1\linewidth]{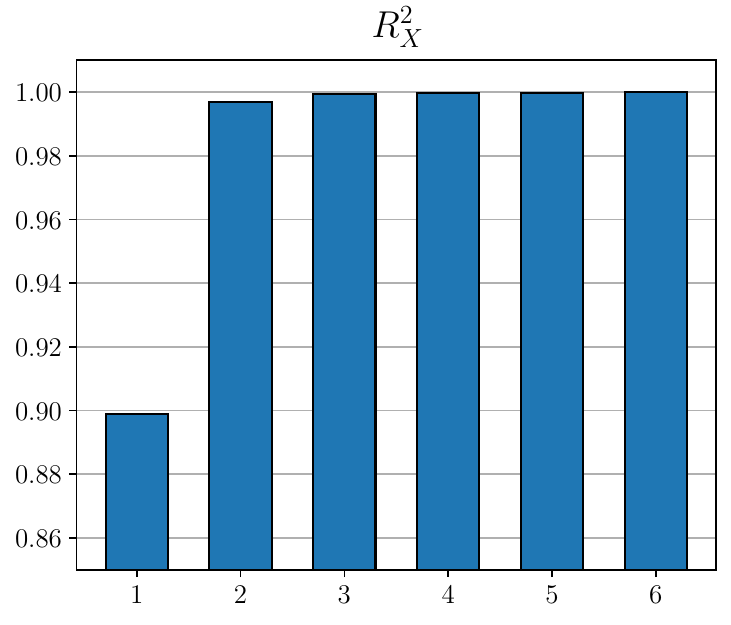}
			\caption{Constant loading factor model}
		\end{subfigure}%
		\begin{subfigure}{.5\textwidth}
			\centering
			\includegraphics[width=1\linewidth]{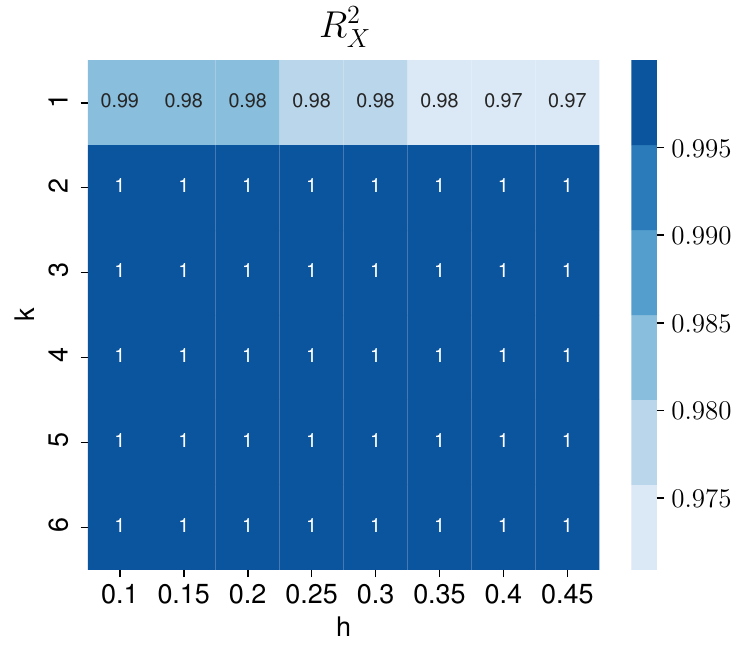}
			\caption{State-varying factor model}
		\end{subfigure}
		\caption{Optimal tuning parameter selection for U.S. Treasury Securities with the unemployment rate as state process. The figure plots $R^2_X$ on the in-sample training data as function of the number of factors $k$ and the bandwidth $h$. The first 25\% time observations are the training data, the following 25\% time observations are the validation data, and the remaining 50\% time are the out-of-sample test data}	\label{fig:yield_unemploy_varying_k_h_in_sample}
	\end{figure}

		\begin{figure}[H]
		\centering
		\begin{subfigure}{.5\textwidth}
			\centering
			\includegraphics[width=1\linewidth]{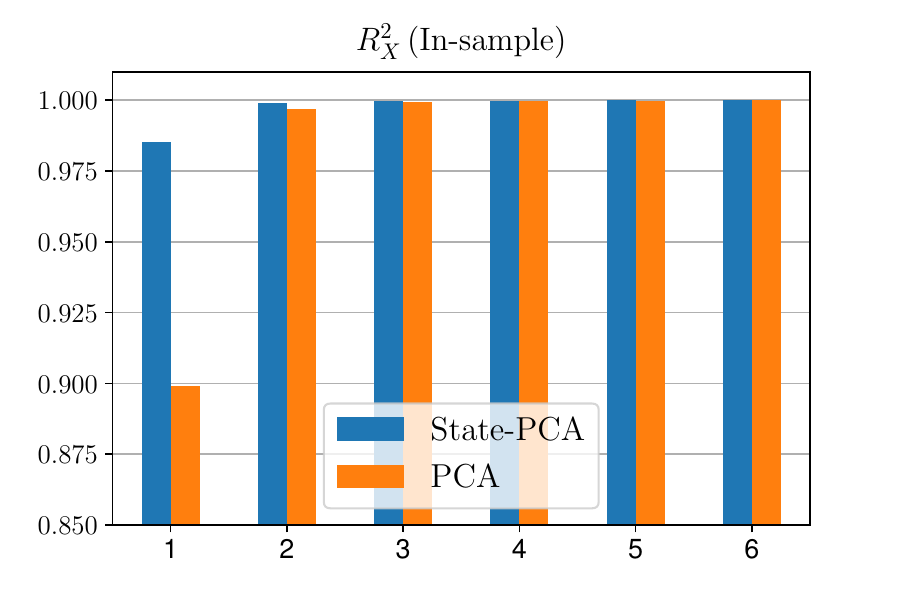}
		\end{subfigure}%
		\begin{subfigure}{.5\textwidth}
			\centering
			\includegraphics[width=1\linewidth]{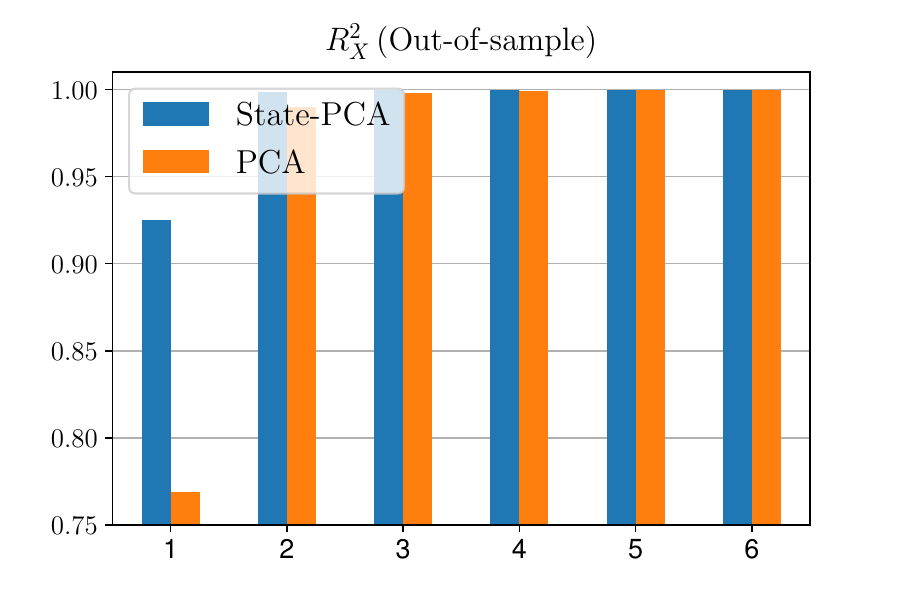}
		\end{subfigure}
		\caption{In-sample (training) and out-of-sample (test) $R^2_X$ for US Treasury Securities with the unemployment rate as state process. The bandwidth is chosen optimally. We use the first half of the data for the estimating the loadings and update them on an expanding window to obtain the out-of-sample common component.}
		\label{fig:yield_unemploy_compare}
	\end{figure}

	\newpage

	\section{Additional Simulation Results}\label{sec:SimApp}

	\subsection{Comparison with Alternative Latent Factors Models}\label{sec:comparison}
	
	In this section we compare the amount of explained variation for different latent factor models that allow for time-variation in the factor structure. In the baseline model, we generate data from a one-factor model
	\begin{align*}
	X_{it} = \Lambda_i(S_t) F_t + e_{it},
	\end{align*}
	where $F_t \sim N(0,1)$ and $e_{it} \sim N(0,1)$. The loadings are cubic functions in the state process,   $\Lambda_i(S_t) = \Lambda_{0i} + \frac{1}{2} S_t \Lambda_{1i} + \frac{1}{4} S_t^2 \Lambda_{2i} + \frac{1}{8} S_t^3 \Lambda_{3i}$, where  $\Lambda_{0i}, \Lambda_{1i}, \Lambda_{2i}, \Lambda_{3i} \sim N(0,1)$. All random variables are independent. We consider three types of state processes:
	\begin{enumerate}
		\item One structural break: The state process is piecewise linear with one jump at half-time:
		\begin{align*}
		S_t = \begin{cases}
		-1 & t \leq \frac{T}{2}\\
		1 & t > \frac{T}{2}
		\end{cases}
		\end{align*}
		\item Three structural breaks: The state process is piecewise linear with three jumps at $\frac{T}{4}, \frac{T}{2} $ and $ \frac{3T}{4}$:
		\begin{align*}
		S_t = \begin{cases}
		-2 & t \leq \frac{T}{4}\\
		-1 & \frac{T}{4}\ < t \leq \frac{T}{2} \\
		1 & \frac{T}{2} < t \leq \frac{3T}{4}\\
		2 & \frac{3T}{4} < t
		\end{cases}
		\end{align*}
		\item Ornstein Uhlenbeck (OU) process: We simulate the state process as $S_t = \theta (\mu - S_t) d_t + \sigma dW_t$, where $\theta =1$, $\mu = 0.2$, and $\sigma =1$. (the same parameters as in Section 7.1 in the main text)
	\end{enumerate}
	Following the suggestions of referees we compare the amount of explained variation $R_X^2$ and the proximity to the common component $R_C^2$ for the following reference approaches:
	
	\begin{enumerate}
		\item {\bf Constant loading model:} The conventional estimator of \cite{bai2003inferential} based on PCA.
		\item {\bf Time-varying loading model:} The time-varying factor model of \cite{Su2017} uses a local kernel estimator in time.
		\item {\bf Structural breaks with pseudo factors (number of factors):} \cite{baltagi2020estimating} use the property that the factor model with changes in loadings is equivalent to a factor model with stable loadings but pseudo factors and detect multiple structural breaks using the sample covariance matrix of the estimated pseudo factors. The number of breaks is assumed to be finite and related to the number of factors which is estimated by the method of \cite{bai2002determining}. Breaks cannot be too close to each other.
		\item {\bf Structural breaks with pseudo factors (regression based):} The approach of \cite{bai2020estimation} uses the same insight as \cite{baltagi2020estimating}. A model with structural breaks can be represented as a constant loading model with more factors. \cite{bai2020estimation} use a regression based approach to estimate the loadings of factors and ``pseudo'' factors which allows them to detect large and small breaks. The paper assumes one structural break.
		\item {\bf Time-varying states with splines:} \cite{Park2009} study a semiparametric factor model. They apply B-splines to estimate the unknown loading function and estimate the factors with a Newton-Raphson algorithm. Their estimator does not estimate latent loading functions that are cross-sectionally different. The loadings can only differ if the loadings are a function of observed cross-section specific variables.
	\end{enumerate}

	\begin{table}[H]
		\centering
		\begin{tabular}{l|rr|rr|rr}
			\toprule
			& \multicolumn{2}{c|}{One Jump} & \multicolumn{2}{c|}{Three Jumps} & \multicolumn{2}{c}{OU} \\
			{} & $R^2_X$ &   $R^2_C$ &  $R^2_X$ &  $R^2_C$ &  $R^2_X$ & $R^2_C$ \\
			\midrule
			State-varying &         0.575 &         0.989 &            0.733 &             0.993 &   0.623 &   0.935 \\
			Constant &         0.461 &         0.791 &            0.411 &             0.555 &   0.422 &   0.661 \\
			Time-Varying  &         0.565 &         0.966 &            0.710 &             0.961 &   0.439 &   0.681 \\
			BNS   &         0.575 &         0.989 &            0.684 &             0.928 &   0.437 &   0.682 \\
			BKW   &         0.575 &         0.987 &            0.702 &             0.950 &   0.450 &   0.698 \\
			PMHB  &         0.010 &         0.003 &            0.010 &             0.006 &   0.011 &   0.005 \\
			\bottomrule
		\end{tabular}
		\caption{This table compares the explained variation of $X$ and the common component $C$ based on 100 Monte Carlo simulations for one factor with the state-varying factor model, constant loading factor model \citep{bai2002determining,bai2003inferential}, time-varying factor model \citep{Su2017}, BNS \citep{baltagi2020estimating}, BKW \citep{baltagi2020estimating} and PMHB \citep{Park2009}. $N = 100$ and $T = 500$ and the bandwidth is set to $h=0.1$ for both the state-varying and time-varying factor model.}
		\label{table:comparison}
	\end{table}

	Table \ref{table:comparison} shows that in all three scenarios, our state-varying factor model can explain the most variation and approximates the unknown factor structure the best. Each of the methods has naturally a setup in which it performs well. Our approach performs well if we know that the time-variation is driven by an observed state process, but we have no prior knowledge about the functional relationship and the latent factor structure. If this state process captures structural breaks in the factor structure, our approach seems to perform better than estimation methods that neglect the additional information captured by the state process. The OU example illustrates that the local time-varying estimator \citep{Su2017} does not perform much better than a constant loading model if the loadings change quickly. In summary, the general purpose estimators for structural breaks or local time-variation do not work better than our approach if we have an additional structure that we can exploit. As expected, the estimation approach of \cite{Park2009} will not work well in this setup as it requires observed cross-section specific variables to obtain loadings that vary in the cross-section.

	\subsection{Comparison with Local Time-Varying Factor Model}\label{sec:timevarying}
	
	We compare the estimation results of our state-varying factor model with the local time-varying model of \cite{Su2017}. The data is generated as in Section \ref{asy_est} in the main text. Our state-varying factor model can recover the correct functional form while the local window estimator fails.  
	\begin{figure}[H]
		\centering
		\begin{subfigure}{.3\textwidth}
			\centering
			\includegraphics[width=1\linewidth]{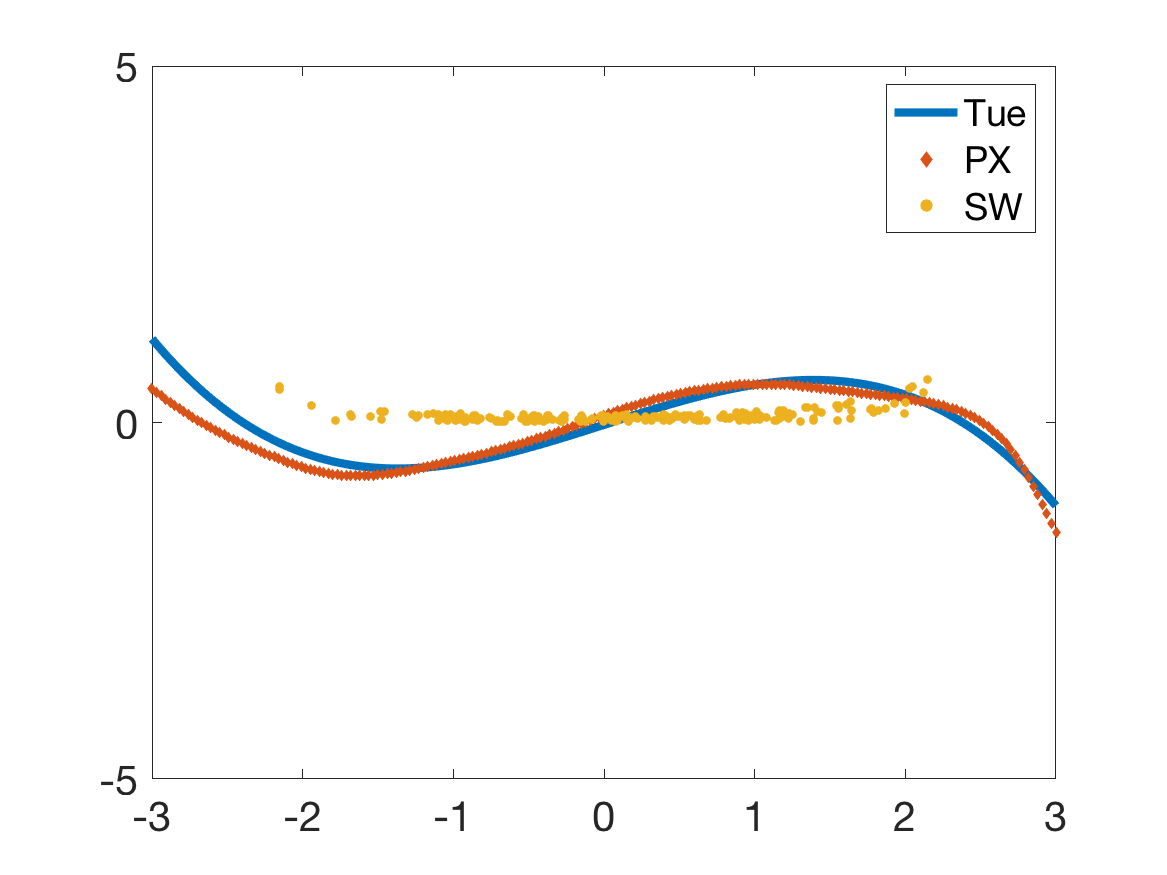}
		\end{subfigure}%
		\begin{subfigure}{.3\textwidth}
			\centering
			\includegraphics[width=1\linewidth]{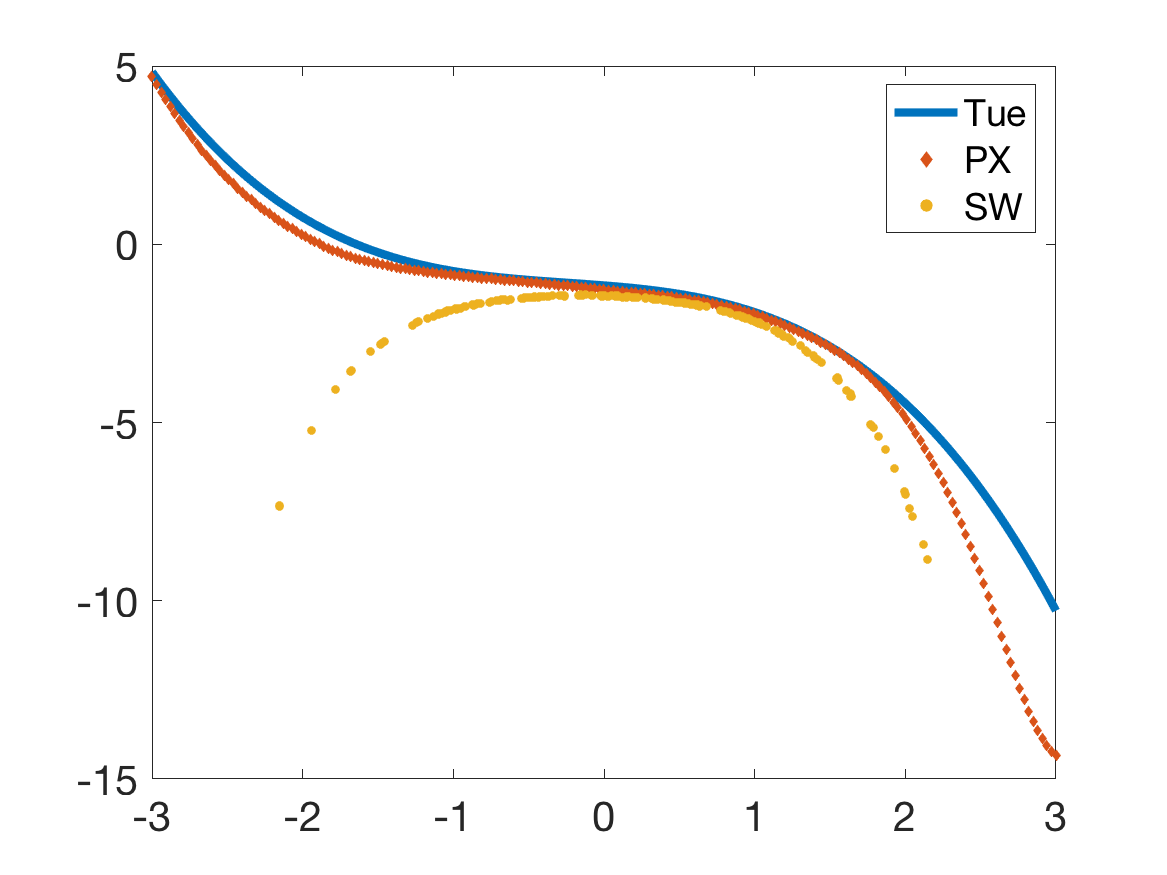}
		\end{subfigure}
		\begin{subfigure}{.3\textwidth}
			\centering
			\includegraphics[width=1\linewidth]{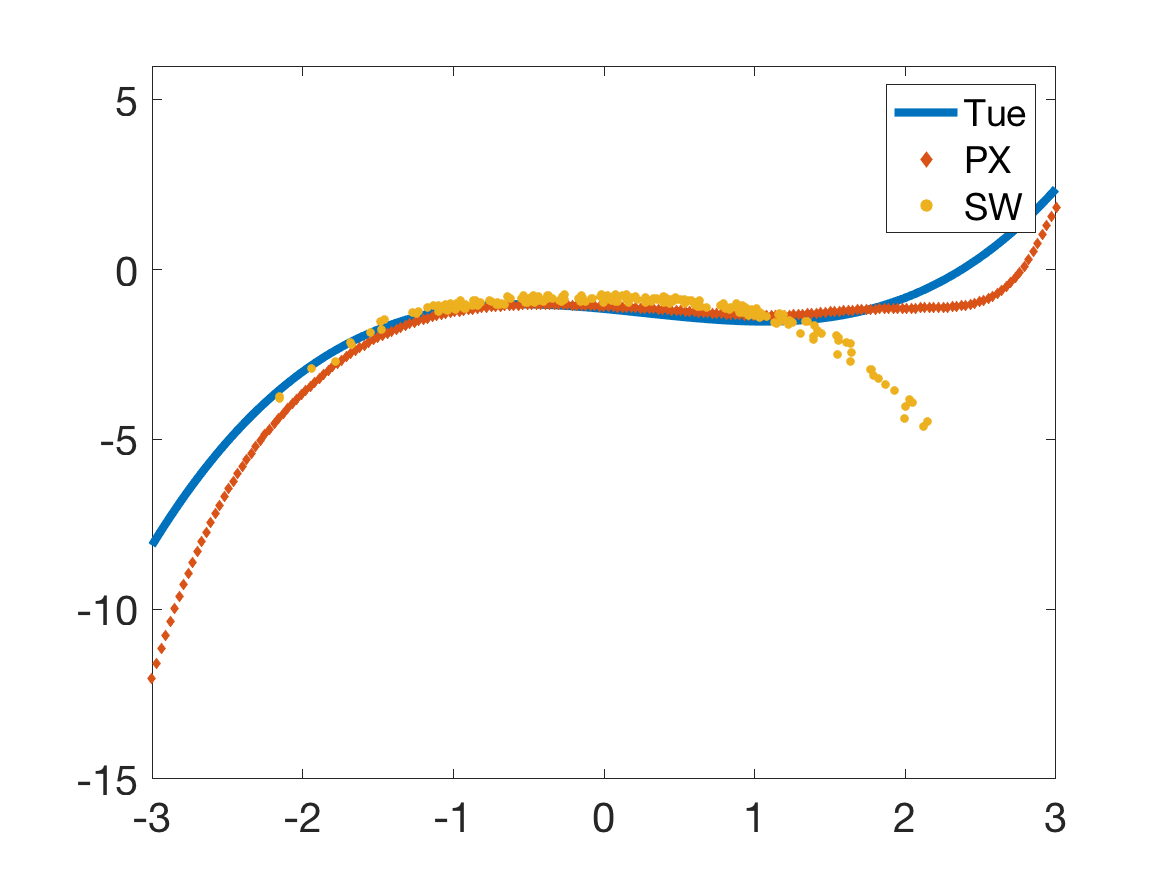}
		\end{subfigure}%
		\begin{subfigure}{.3\textwidth}
			\centering
			\includegraphics[width=1\linewidth]{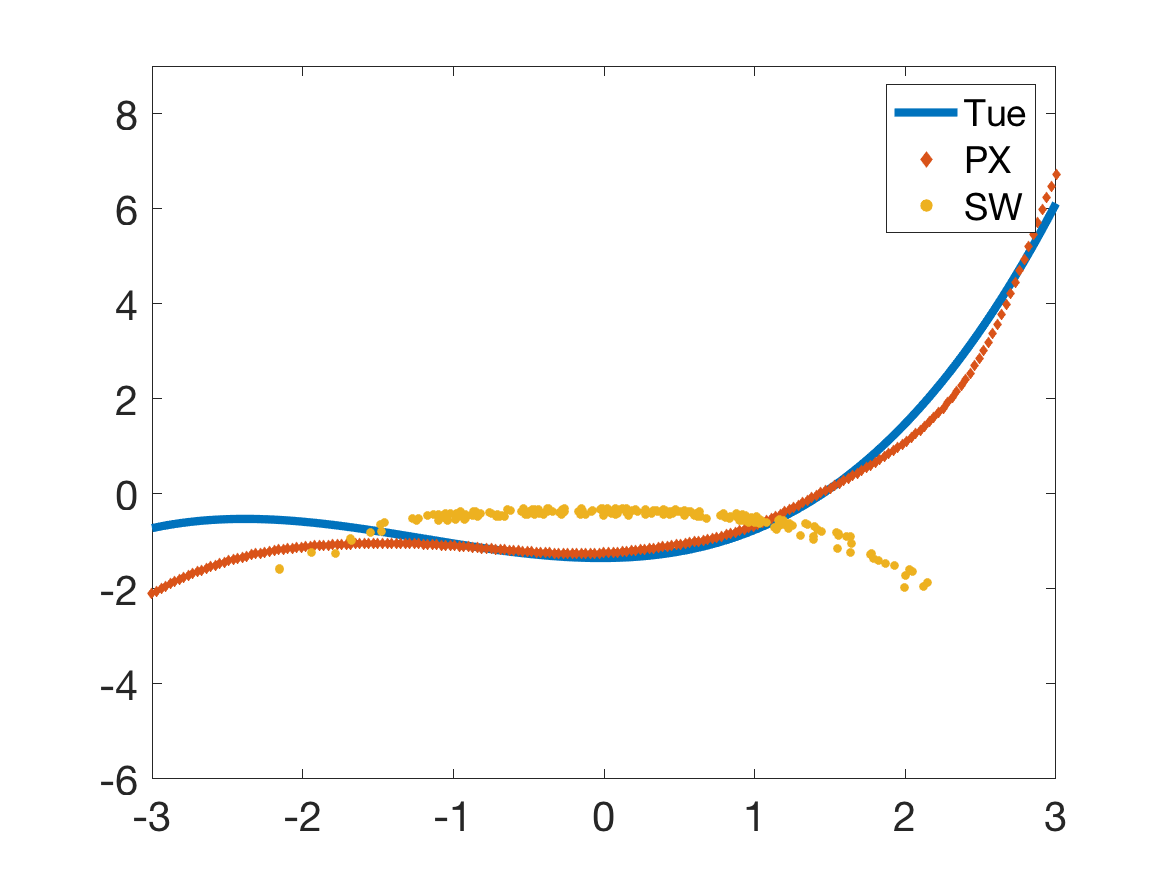}
		\end{subfigure}
		\caption{Estimated functional form of loading versus the state variable ($N = 100, T = 500, h = 0.5$). Loadings estimated from state-varying factor model (denoted as ``PX'' in the subplots) is compared with loadings estimated from time-varying factor model \citep{Su2017} (denoted as ``SW'' in the subplots). The true functional form is superimposed on the estimated form.
			\label{fig:loadingSW}}
	\end{figure}

	\subsection{Choice of Tuning Parameters}\label{sec:tuning}
	
	We illustrate in a simulation setup how to optimally choose the number of factors and the bandwidth of the kernel projection.
	In the baseline model we generate data from a three-factor model
	\begin{align*}
	X_{it} = \Lambda_i(S_t)^\top F_t + e_{it},
	\end{align*}
	where $F_t \sim N(0,I_3)$ and $e_{it} \sim N(0,1)$. The loadings are cubic functions of the state process, that is $\Lambda_i(S_t) = \Lambda_{0i} + \frac{1}{2} S_t \Lambda_{1i} + \frac{1}{4} S_t^2 \Lambda_{2i} + \frac{1}{8} S_t^3 \Lambda_{3i}$, where $\Lambda_{0i}, \Lambda_{1i}, \Lambda_{2i}, \Lambda_{3i} \sim N(0,I_3)$. All processes are independent. We simulate the state process as an OU process, $S_t = \theta (\mu - S_t) d_t + \sigma dW_t$, where $\theta =1$, $\mu = 0.2$, and $\sigma =1$ (these are the same parameters as in Section 7.1 in the main text).

	We show how to select the number of factors and the bandwidth parameter based on cross-validation arguments. In more detail, the number of factors and the bandwidth can be viewed as tuning parameters that can be selected on a validation data set to maximize the amount of explained variation, while the model itself is estimated on the training data. Then, the model can be evaluated out-of-sample on the test data. We confirm that the number of factors and bandwidth chosen optimally on the validation data also maximize the out-of-sample $R^2$ on the test data.
	We split the data into training, validation and test sets. The first 60\% time periods constitute the training set, the following 20\% time periods are the validation set, and the remaining 20\% time periods represent the test data. We use the loadings and factor weights estimated on the training to choose the number of factors and bandwidth that maximize the $R^2$ on the validation data. Given the estimated model and tuning parameters, we evaluate the $R^2$ on the test data. 
	
	Figure \ref{fig:varying_k_h} plots the $R^2_{X,\, \mathrm{train}}$, $R^2_{X,\, \mathrm{val}}$, $R^2_{X,\, \mathrm{test}}$, $R^2_{C,\, \mathrm{train}}$, $R^2_{C,\, \mathrm{val}}$, and $R^2_{C,\, \mathrm{test}}$ as a function of the number of factors $k$ and bandwidth $h$. The natural selection criterion is to choose the smallest $k$ and largest $h$ to achieve a $R^2_{X,\mathrm{val}}$ that is close to the maximum value. The reason why it is usually preferred to opt for the best parsimonious model on the validation data and not simply the best validation model, is that the criterion function is also estimated with some noise. The validation data suggests to choose $k=3$ and $h=0.4$ which achieves 99\% of the best validation value. For these values we obtain $R^2_{X,\, \mathrm{test}} = 0.78$ and $R^2_{C,\, \mathrm{test}} = 0.94$ out-of-sample. Note that the validation and test results are relatively robust to the choice of bandwidth.

	\begin{figure}[H]
		\centering
		\includegraphics[width=1\linewidth]{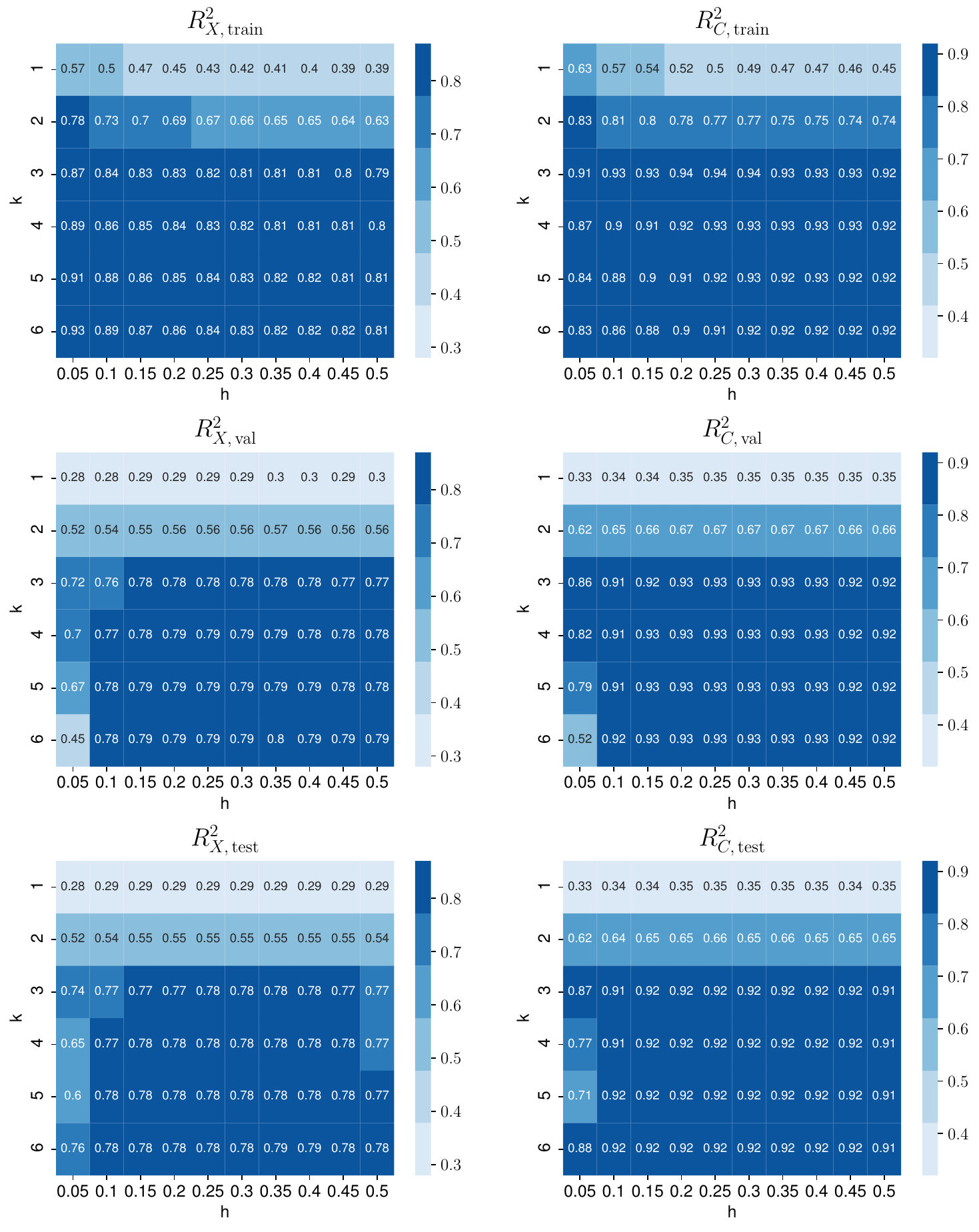}
		\caption{Optimal tuning parameter selection in the simulation. The figure plots $R^2_X$ and $R^2_C$ on the training, validation and test data as function of the number of factors $k$ and the bandwidth $h$. The first 60\% time observations are the training data, the following 20\% time observations are the validation data, and the remaining 20\% time are the out-of-sample test data. $N = 100$ and $T = 500$.}
		\label{fig:varying_k_h}
	\end{figure}
	
	%
	
	\subsection{Asymptotic Distribution}
	\subsubsection{IID Errors}

	\begin{figure}[H]
		\centering
		\begin{subfigure}{.25\textwidth}
			\centering
			\includegraphics[width=1\linewidth]{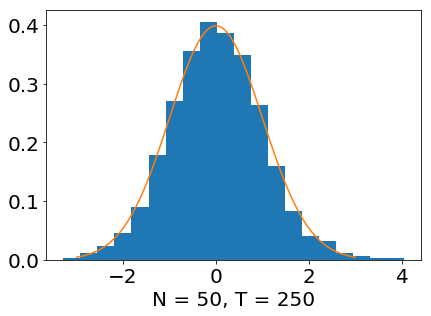}
		\end{subfigure}%
		\begin{subfigure}{.25\textwidth}
			\centering
			\includegraphics[width=1\linewidth]{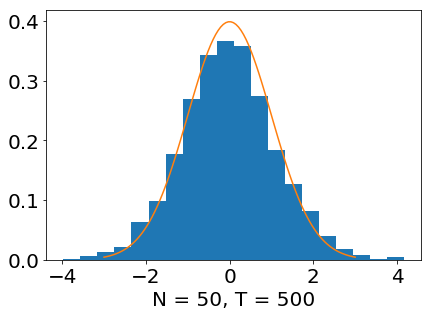}
		\end{subfigure}%
		\begin{subfigure}{.25\textwidth}
			\centering
			\includegraphics[width=1\linewidth]{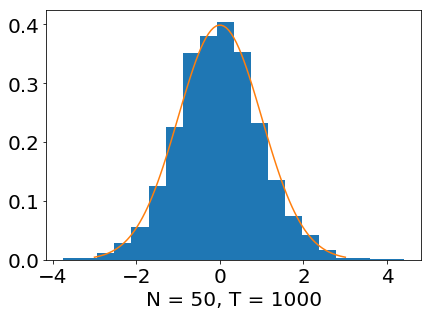}
		\end{subfigure}
		\begin{subfigure}{.25\textwidth}
			\centering
			\includegraphics[width=1\linewidth]{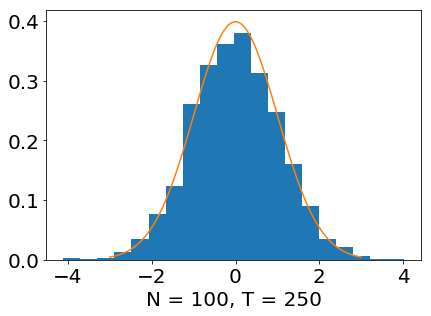}
		\end{subfigure}%
		\begin{subfigure}{.25\textwidth}
			\centering
			\includegraphics[width=1\linewidth]{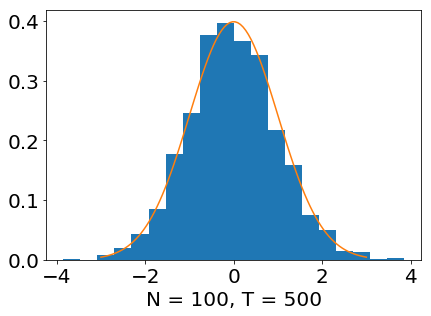}
		\end{subfigure}%
		\begin{subfigure}{.25\textwidth}
			\centering
			\includegraphics[width=1\linewidth]{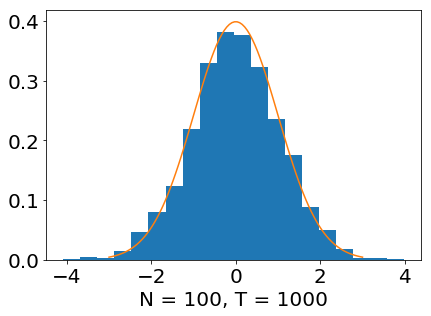}
		\end{subfigure}
		\begin{subfigure}{.25\textwidth}
			\centering
			\includegraphics[width=1\linewidth]{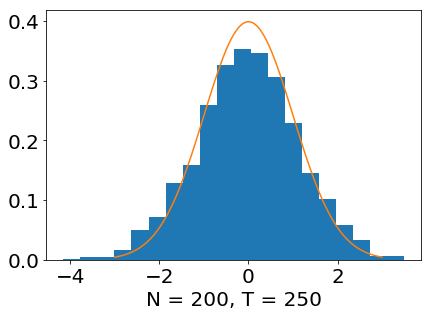}
		\end{subfigure}%
		\begin{subfigure}{.25\textwidth}
			\centering
			\includegraphics[width=1\linewidth]{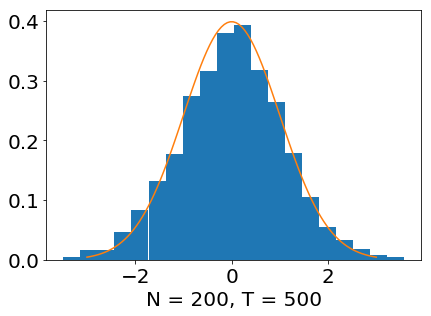}
		\end{subfigure}%
		\begin{subfigure}{.25\textwidth}
			\centering
			\includegraphics[width=1\linewidth]{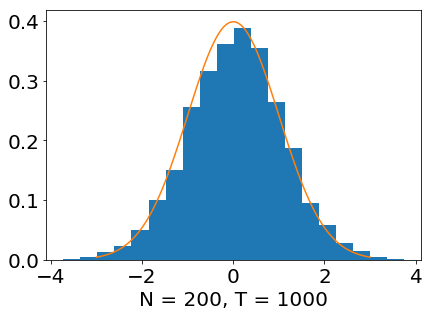}
		\end{subfigure}
		\caption{Histograms of estimated standardized factors ($N=$50, 100, 200; $T$=250, 500, 1000; $h$=0.3) for IID errors. The normal density function is superimposed on the histograms.}
		\label{hist_factor}
	\end{figure}

	\begin{figure}[H]
		\centering
		\begin{subfigure}{.25\textwidth}
			\centering
			\includegraphics[width=1\linewidth]{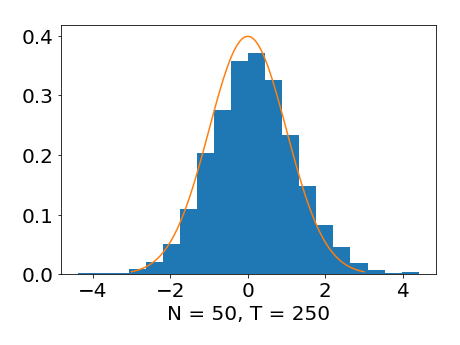}
		\end{subfigure}%
		\begin{subfigure}{.25\textwidth}
			\centering
			\includegraphics[width=1\linewidth]{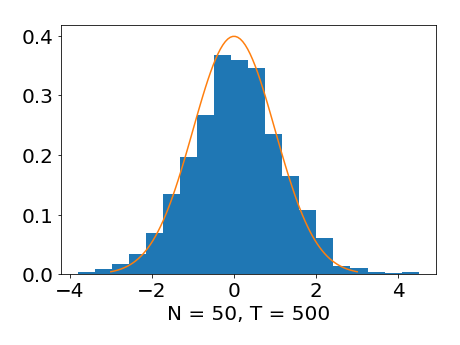}
		\end{subfigure}%
		\begin{subfigure}{.25\textwidth}
			\centering
			\includegraphics[width=1\linewidth]{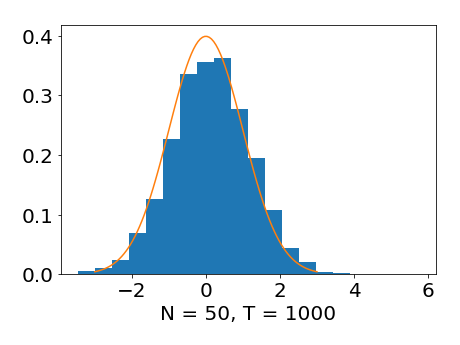}
		\end{subfigure}
		\begin{subfigure}{.25\textwidth}
			\centering
			\includegraphics[width=1\linewidth]{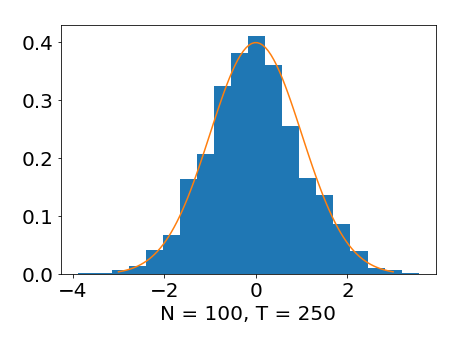}
		\end{subfigure}%
		\begin{subfigure}{.25\textwidth}
			\centering
			\includegraphics[width=1\linewidth]{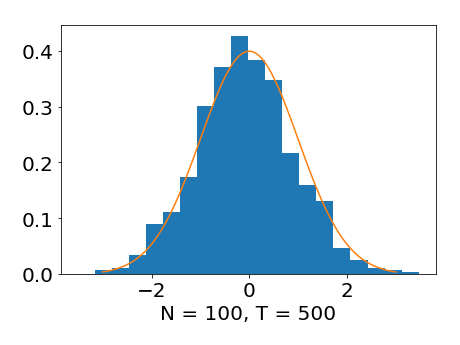}
		\end{subfigure}%
		\begin{subfigure}{.25\textwidth}
			\centering
			\includegraphics[width=1\linewidth]{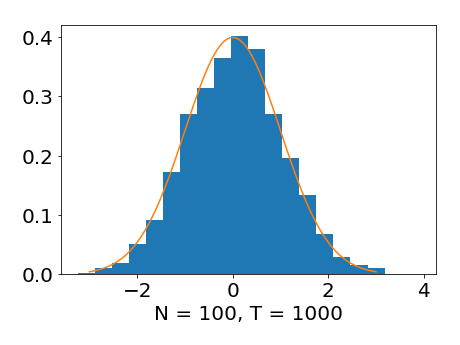}
		\end{subfigure}
		\begin{subfigure}{.25\textwidth}
			\centering
			\includegraphics[width=1\linewidth]{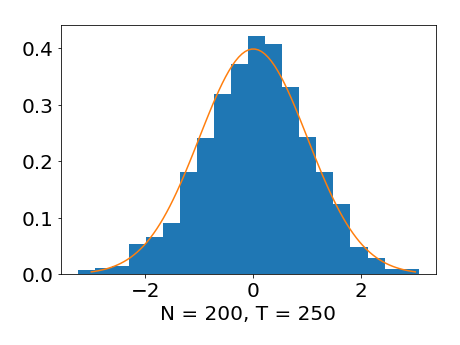}
		\end{subfigure}%
		\begin{subfigure}{.25\textwidth}
			\centering
			\includegraphics[width=1\linewidth]{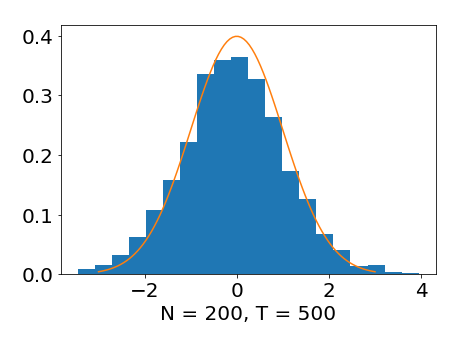}
		\end{subfigure}%
		\begin{subfigure}{.25\textwidth}
			\centering
			\includegraphics[width=1\linewidth]{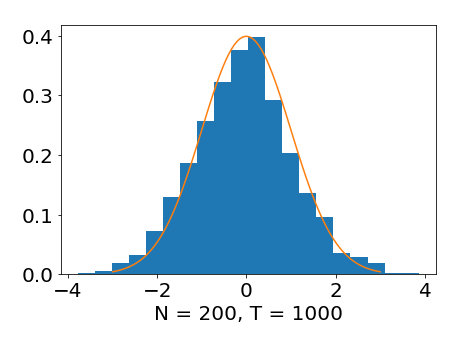}
		\end{subfigure}
		\caption{Histograms of estimated standardized common components ($N=50, 100, 200$; $T=250, 500, 1000$; $h=0.3$) for IID errors. The normal density function is superimposed on the histograms.}
		\label{hist_common}
	\end{figure}

	\subsubsection{Heteroskedastic Errors}
	\begin{figure}[H]
		\centering
		\begin{subfigure}{.22\textwidth}
			\centering
			\includegraphics[width=1\linewidth]{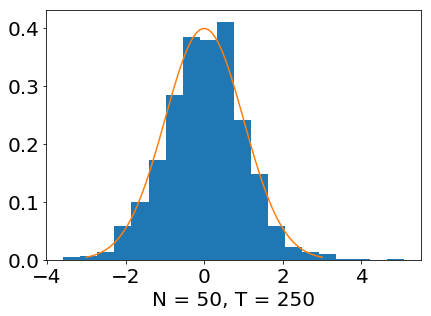}
		\end{subfigure}%
		\begin{subfigure}{.22\textwidth}
			\centering
			\includegraphics[width=1\linewidth]{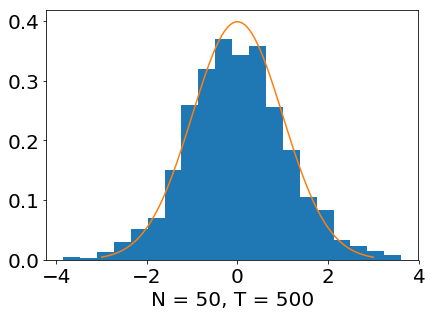}
		\end{subfigure}%
		\begin{subfigure}{.22\textwidth}
			\centering
			\includegraphics[width=1\linewidth]{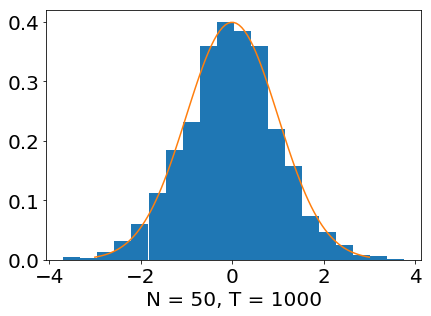}
		\end{subfigure}
		\begin{subfigure}{.22\textwidth}
			\centering
			\includegraphics[width=1\linewidth]{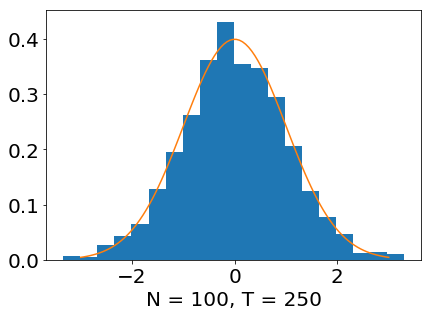}
		\end{subfigure}%
		\begin{subfigure}{.22\textwidth}
			\centering
			\includegraphics[width=1\linewidth]{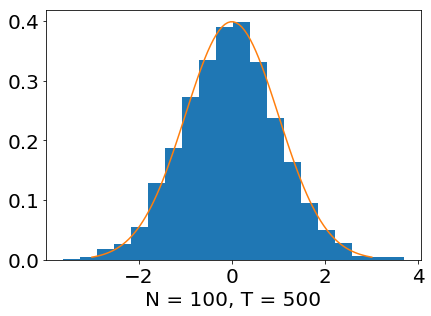}
		\end{subfigure}%
		\begin{subfigure}{.22\textwidth}
			\centering
			\includegraphics[width=1\linewidth]{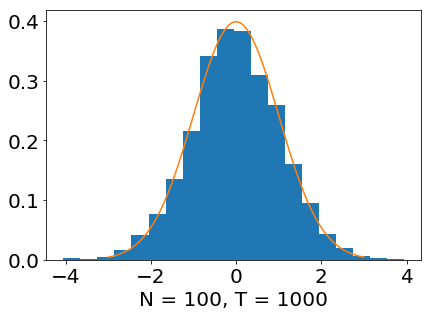}
		\end{subfigure}
		\begin{subfigure}{.22\textwidth}
			\centering
			\includegraphics[width=1\linewidth]{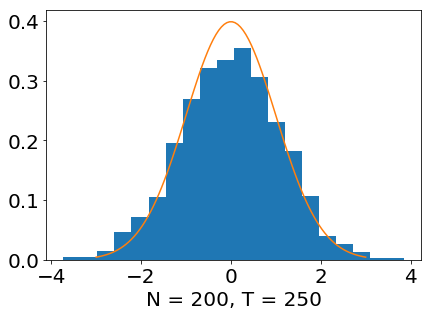}
		\end{subfigure}%
		\begin{subfigure}{.22\textwidth}
			\centering
			\includegraphics[width=1\linewidth]{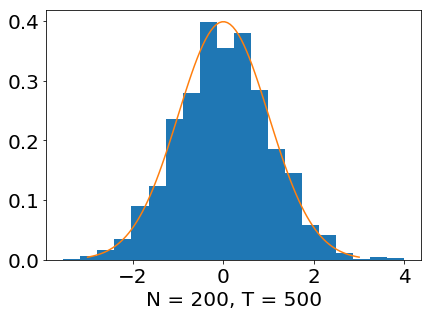}
		\end{subfigure}%
		\begin{subfigure}{.22\textwidth}
			\centering
			\includegraphics[width=1\linewidth]{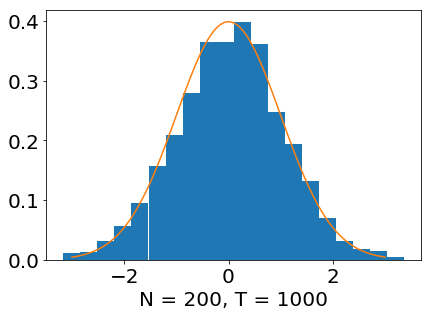}
		\end{subfigure}
		\caption{Histograms of estimated standardized factors ($N =50, 100, 200; T=250, 500, 1000; h=0.3$). The normal density function is superimposed on the histograms. (Heteroskedastic errors)}
		\label{hist_factor_heter}
	\end{figure}

	\begin{figure}[H]
		\centering
		\begin{subfigure}{.22\textwidth}
			\centering
			\includegraphics[width=1\linewidth]{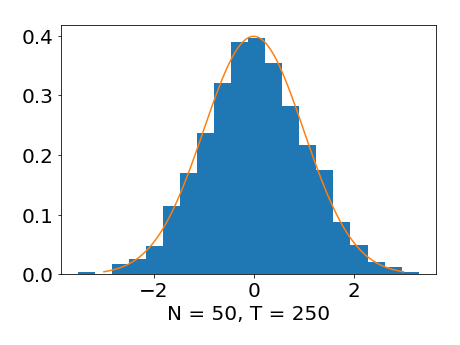}
		\end{subfigure}%
		\begin{subfigure}{.22\textwidth}
			\centering
			\includegraphics[width=1\linewidth]{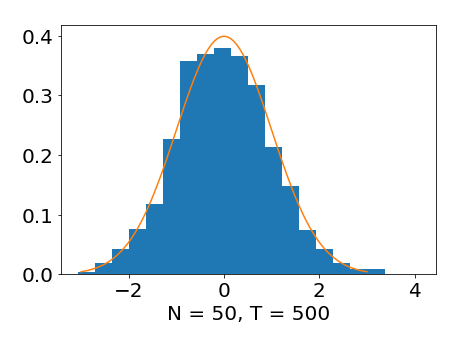}
		\end{subfigure}%
		\begin{subfigure}{.22\textwidth}
			\centering
			\includegraphics[width=1\linewidth]{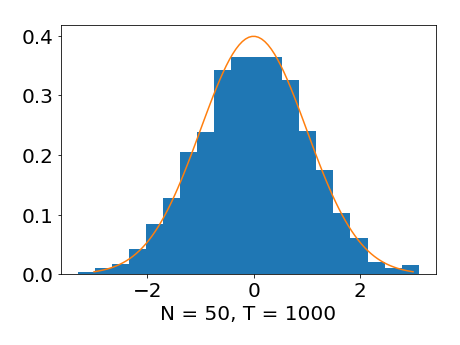}
		\end{subfigure}
		\begin{subfigure}{.22\textwidth}
			\centering
			\includegraphics[width=1\linewidth]{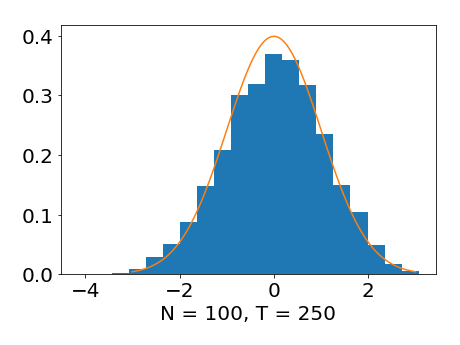}
		\end{subfigure}%
		\begin{subfigure}{.22\textwidth}
			\centering
			\includegraphics[width=1\linewidth]{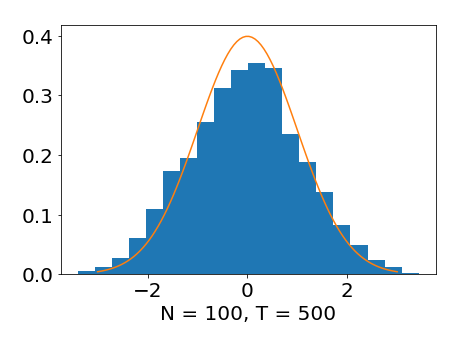}
		\end{subfigure}%
		\begin{subfigure}{.22\textwidth}
			\centering
			\includegraphics[width=1\linewidth]{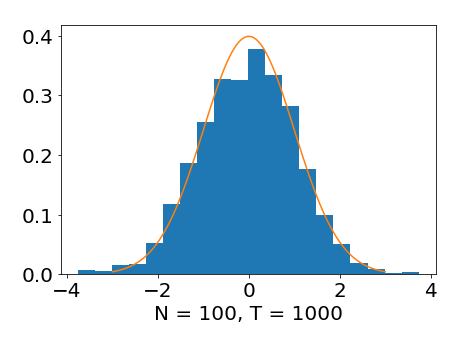}
		\end{subfigure}
		\begin{subfigure}{.22\textwidth}
			\centering
			\includegraphics[width=1\linewidth]{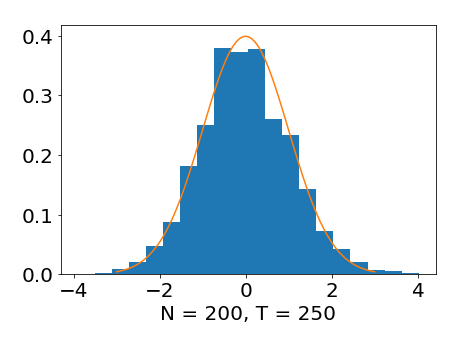}
		\end{subfigure}%
		\begin{subfigure}{.22\textwidth}
			\centering
			\includegraphics[width=1\linewidth]{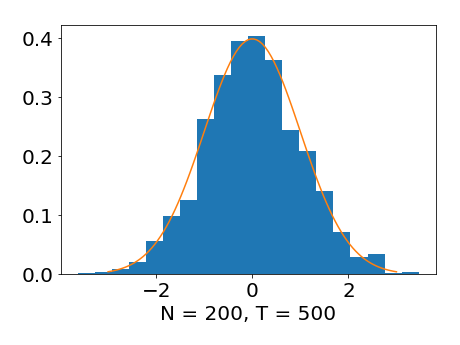}
		\end{subfigure}%
		\begin{subfigure}{.22\textwidth}
			\centering
			\includegraphics[width=1\linewidth]{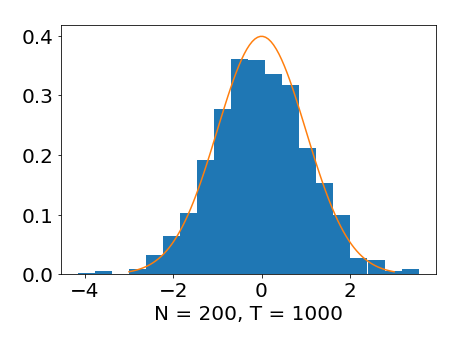}
		\end{subfigure}
		\caption{Histograms of estimated standardized loadings. ($N =50, 100, 200; T=250, 500, 1000; h=0.3$). The normal density function is superimposed on the histograms (Heteroskedastic errors)}
		\label{hist_loadings_heter}
	\end{figure}

	\begin{figure}[H]
		\centering
		\begin{subfigure}{.22\textwidth}
			\centering
			\includegraphics[width=1\linewidth]{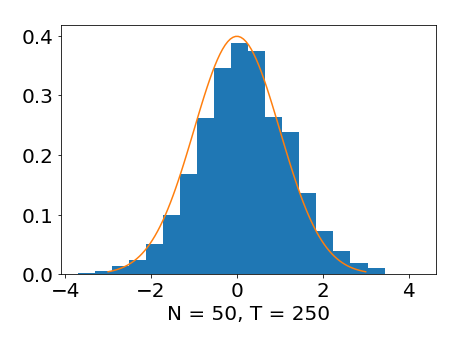}
		\end{subfigure}%
		\begin{subfigure}{.22\textwidth}
			\centering
			\includegraphics[width=1\linewidth]{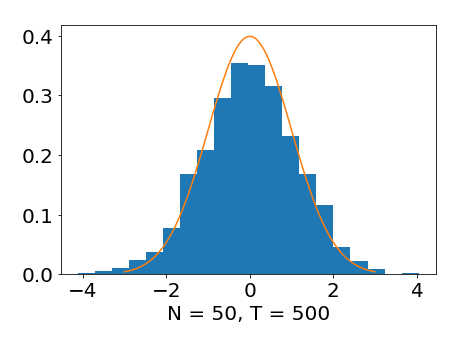}
		\end{subfigure}%
		\begin{subfigure}{.22\textwidth}
			\centering
			\includegraphics[width=1\linewidth]{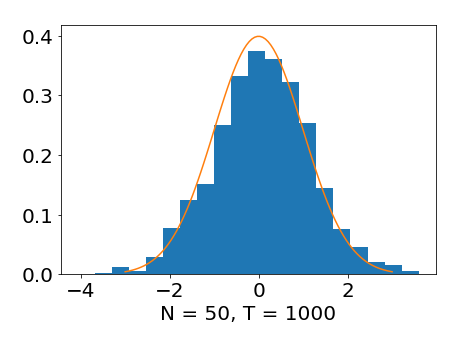}
		\end{subfigure}
		\begin{subfigure}{.22\textwidth}
			\centering
			\includegraphics[width=1\linewidth]{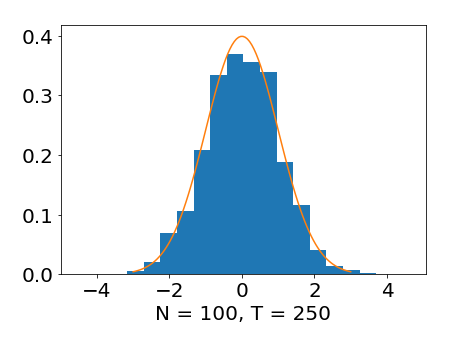}
		\end{subfigure}%
		\begin{subfigure}{.22\textwidth}
			\centering
			\includegraphics[width=1\linewidth]{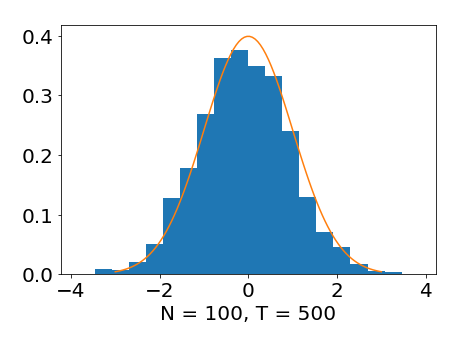}
		\end{subfigure}%
		\begin{subfigure}{.22\textwidth}
			\centering
			\includegraphics[width=1\linewidth]{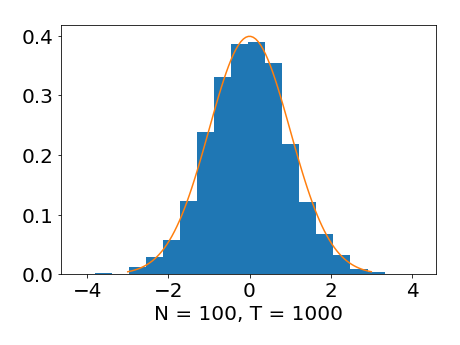}
		\end{subfigure}
		\begin{subfigure}{.22\textwidth}
			\centering
			\includegraphics[width=1\linewidth]{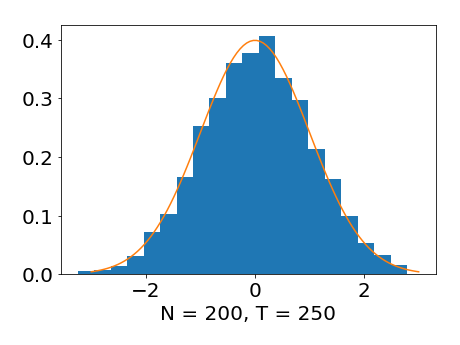}
		\end{subfigure}%
		\begin{subfigure}{.22\textwidth}
			\centering
			\includegraphics[width=1\linewidth]{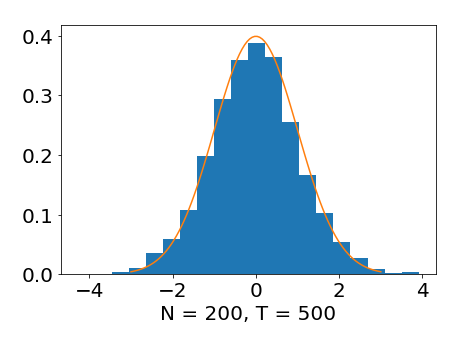}
		\end{subfigure}%
		\begin{subfigure}{.22\textwidth}
			\centering
			\includegraphics[width=1\linewidth]{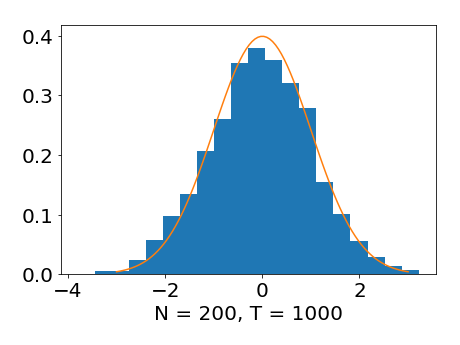}
		\end{subfigure}
		\caption{Histograms of estimated standardized common components. ($N =50, 100, 200; T=250, 500, 1000; h=0.3$). The normal density function is superimposed on the histograms (Heteroskedastic errors)}
		\label{hist_comp_heter}
	\end{figure}
	
	\begin{figure}[H]
		\centering
		\begin{subfigure}{.22\textwidth}
			\centering
			\includegraphics[width=1\linewidth]{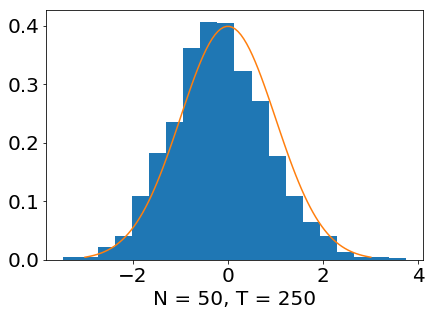}
		\end{subfigure}%
		\begin{subfigure}{.22\textwidth}
			\centering
			\includegraphics[width=1\linewidth]{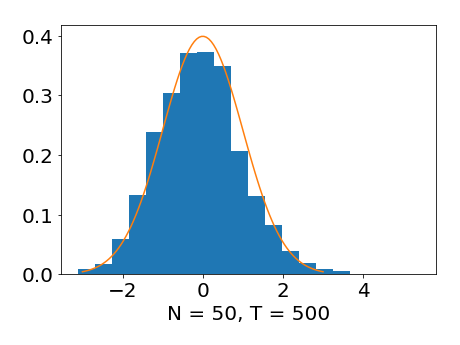}
		\end{subfigure}%
		\begin{subfigure}{.22\textwidth}
			\centering
			\includegraphics[width=1\linewidth]{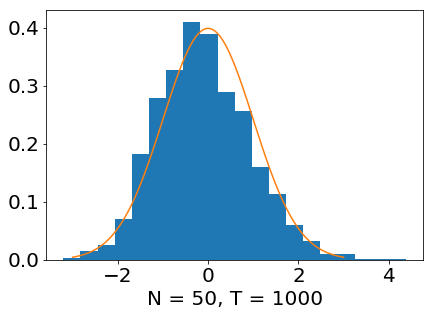}
		\end{subfigure}
		\begin{subfigure}{.22\textwidth}
			\centering
			\includegraphics[width=1\linewidth]{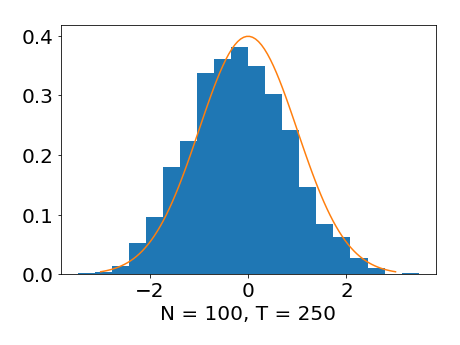}
		\end{subfigure}%
		\begin{subfigure}{.22\textwidth}
			\centering
			\includegraphics[width=1\linewidth]{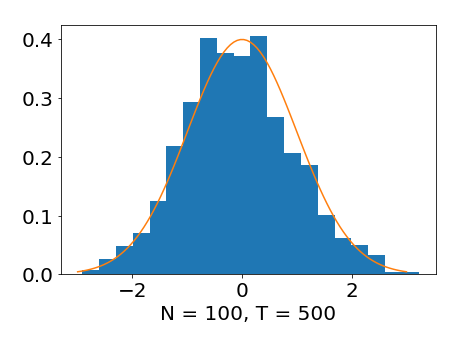}
		\end{subfigure}%
		\begin{subfigure}{.22\textwidth}
			\centering
			\includegraphics[width=1\linewidth]{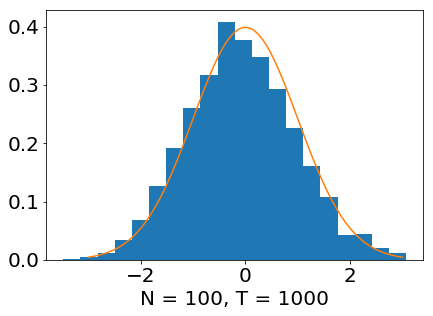}
		\end{subfigure}
		\begin{subfigure}{.22\textwidth}
			\centering
			\includegraphics[width=1\linewidth]{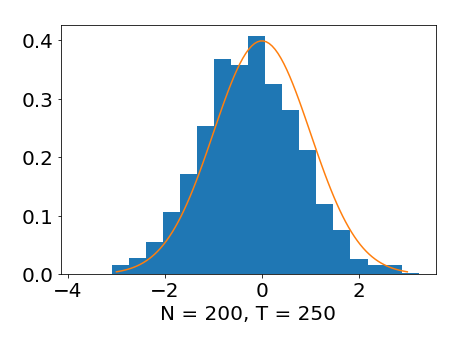}
		\end{subfigure}%
		\begin{subfigure}{.22\textwidth}
			\centering
			\includegraphics[width=1\linewidth]{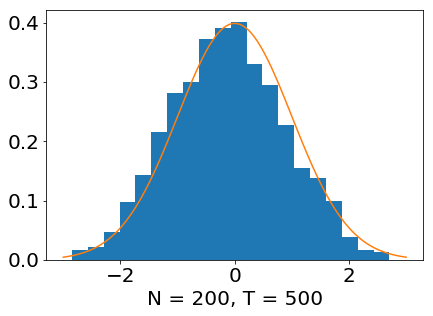}
		\end{subfigure}%
		\begin{subfigure}{.22\textwidth}
			\centering
			\includegraphics[width=1\linewidth]{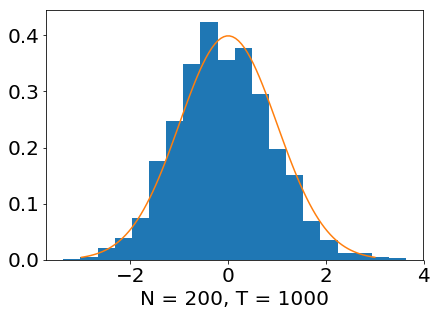}
		\end{subfigure}
		\caption{Histograms of estimated standardized and bias-corrected generalized correlation test statistic. ($N =50, 100, 200; T=250, 500, 1000; h=0.3$). The normal density function is superimposed on the histograms (Heteroskedastic errors)}
		\label{hist_gen_heter}
	\end{figure}
	
	\subsubsection{Cross-Sectionally Dependent Errors}
	\begin{figure}[H]
		\centering
		\begin{subfigure}{.22\textwidth}
			\centering
			\includegraphics[width=1\linewidth]{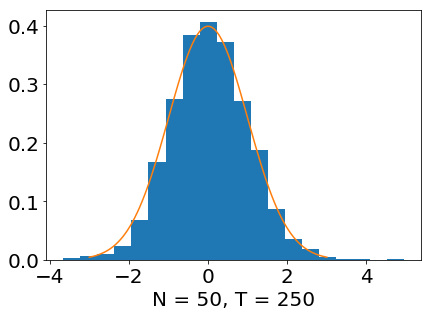}
		\end{subfigure}%
		\begin{subfigure}{.22\textwidth}
			\centering
			\includegraphics[width=1\linewidth]{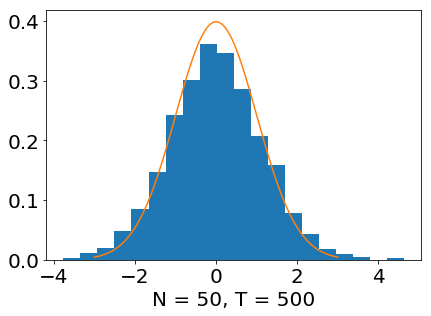}
		\end{subfigure}%
		\begin{subfigure}{.22\textwidth}
			\centering
			\includegraphics[width=1\linewidth]{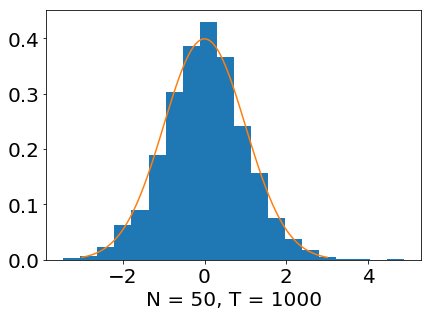}
		\end{subfigure}
		\begin{subfigure}{.22\textwidth}
			\centering
			\includegraphics[width=1\linewidth]{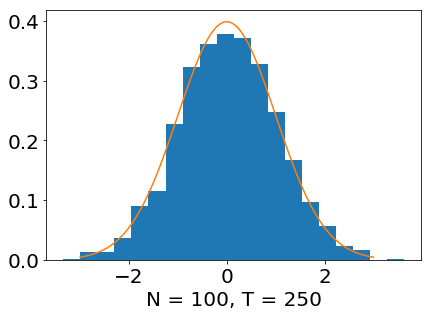}
		\end{subfigure}%
		\begin{subfigure}{.22\textwidth}
			\centering
			\includegraphics[width=1\linewidth]{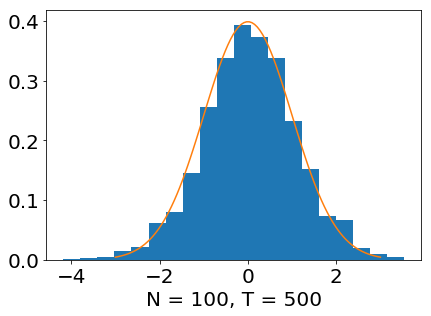}
		\end{subfigure}%
		\begin{subfigure}{.22\textwidth}
			\centering
			\includegraphics[width=1\linewidth]{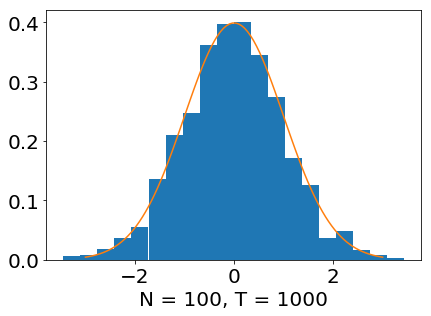}
		\end{subfigure}
		\begin{subfigure}{.22\textwidth}
			\centering
			\includegraphics[width=1\linewidth]{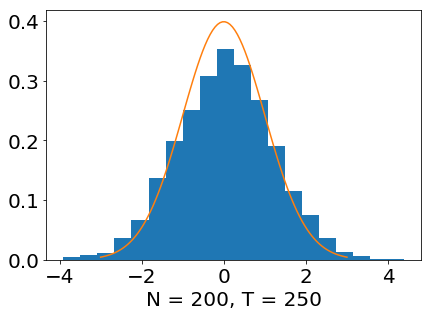}
		\end{subfigure}%
		\begin{subfigure}{.22\textwidth}
			\centering
			\includegraphics[width=1\linewidth]{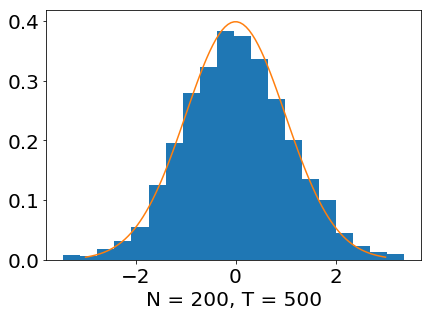}
		\end{subfigure}%
		\begin{subfigure}{.22\textwidth}
			\centering
			\includegraphics[width=1\linewidth]{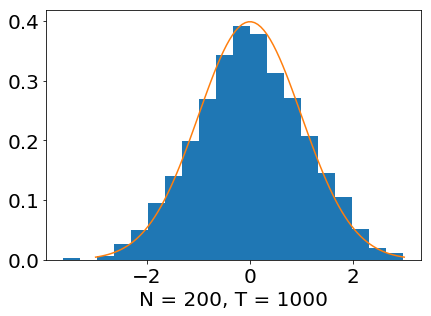}
		\end{subfigure}
		\caption{Histograms of estimated standardized factors. ($N =50, 100, 200; T=250, 500, 1000; h=0.3$).  The normal density function is superimposed on the histograms (Cross-sectionally dependent errors)}
		\label{hist_factor_cross}
	\end{figure}

	\begin{figure}[H]
		\centering
		\begin{subfigure}{.22\textwidth}
			\centering
			\includegraphics[width=1\linewidth]{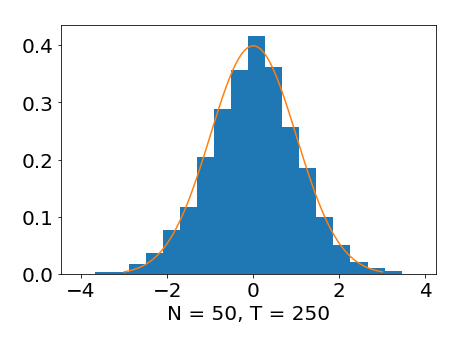}
		\end{subfigure}%
		\begin{subfigure}{.22\textwidth}
			\centering
			\includegraphics[width=1\linewidth]{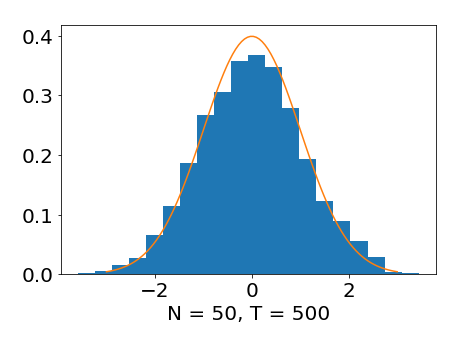}
		\end{subfigure}%
		\begin{subfigure}{.22\textwidth}
			\centering
			\includegraphics[width=1\linewidth]{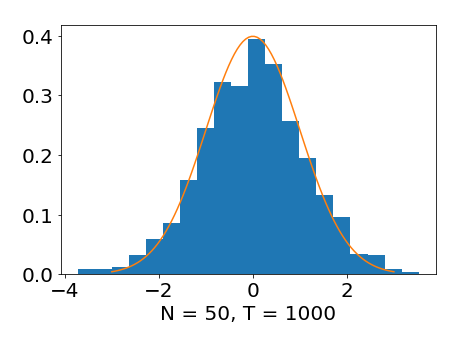}
		\end{subfigure}
		\begin{subfigure}{.22\textwidth}
			\centering
			\includegraphics[width=1\linewidth]{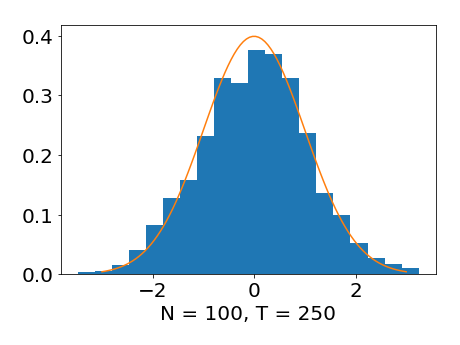}
		\end{subfigure}%
		\begin{subfigure}{.22\textwidth}
			\centering
			\includegraphics[width=1\linewidth]{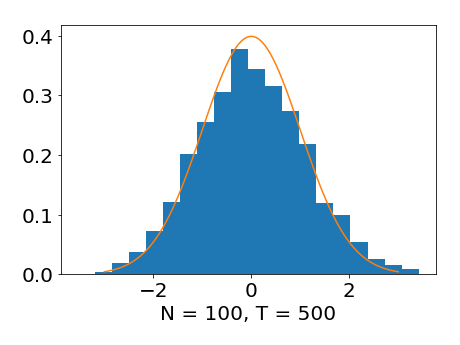}
		\end{subfigure}%
		\begin{subfigure}{.22\textwidth}
			\centering
			\includegraphics[width=1\linewidth]{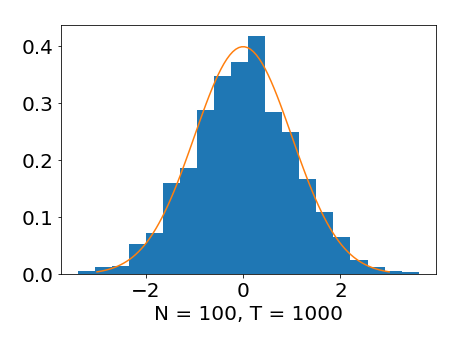}
		\end{subfigure}
		\begin{subfigure}{.22\textwidth}
			\centering
			\includegraphics[width=1\linewidth]{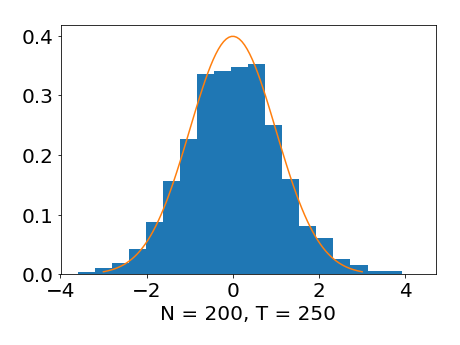}
		\end{subfigure}%
		\begin{subfigure}{.22\textwidth}
			\centering
			\includegraphics[width=1\linewidth]{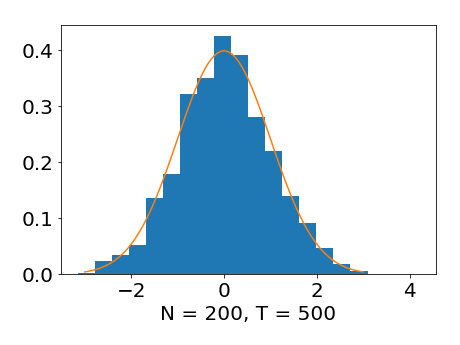}
		\end{subfigure}%
		\begin{subfigure}{.22\textwidth}
			\centering
			\includegraphics[width=1\linewidth]{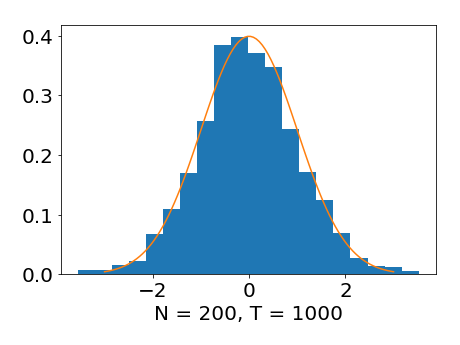}
		\end{subfigure}
		\caption{Histograms of estimated standardized loadings. ($N =50, 100, 200; T=250, 500, 1000; h=0.3$).  The normal density function is superimposed on the histograms (Cross-sectionally dependent errors)}
		\label{hist_loadings_cross}
	\end{figure}
	
	\begin{figure}[H]
		\centering
		\begin{subfigure}{.22\textwidth}
			\centering
			\includegraphics[width=1\linewidth]{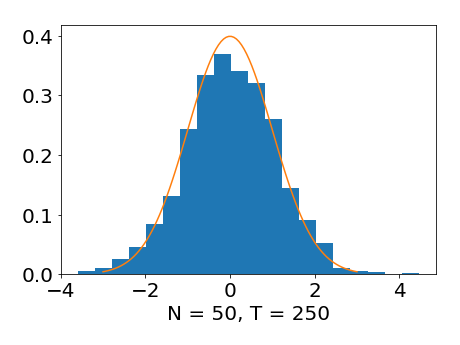}
		\end{subfigure}%
		\begin{subfigure}{.22\textwidth}
			\centering
			\includegraphics[width=1\linewidth]{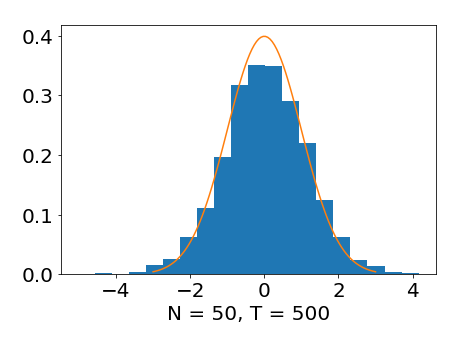}
		\end{subfigure}%
		\begin{subfigure}{.22\textwidth}
			\centering
			\includegraphics[width=1\linewidth]{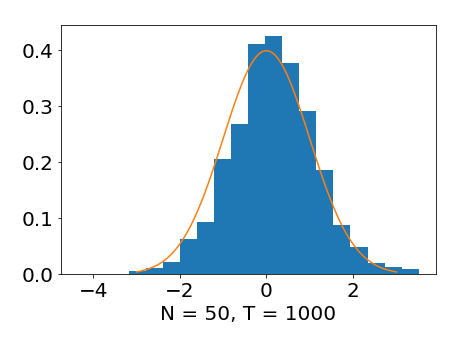}
		\end{subfigure}
		\begin{subfigure}{.22\textwidth}
			\centering
			\includegraphics[width=1\linewidth]{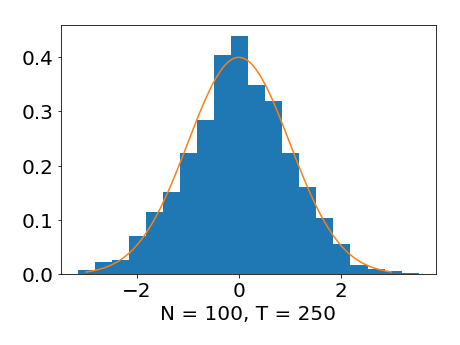}
		\end{subfigure}%
		\begin{subfigure}{.22\textwidth}
			\centering
			\includegraphics[width=1\linewidth]{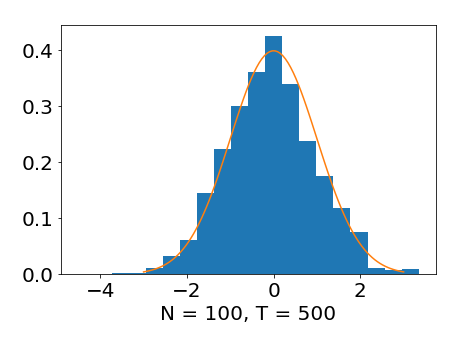}
		\end{subfigure}%
		\begin{subfigure}{.22\textwidth}
			\centering
			\includegraphics[width=1\linewidth]{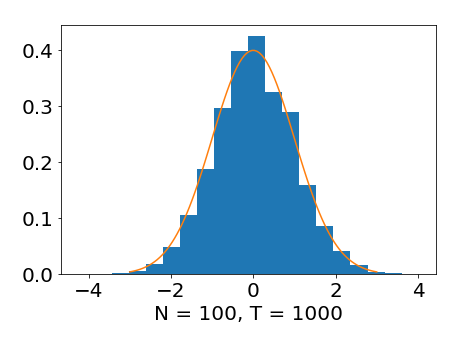}
		\end{subfigure}
		\begin{subfigure}{.22\textwidth}
			\centering
			\includegraphics[width=1\linewidth]{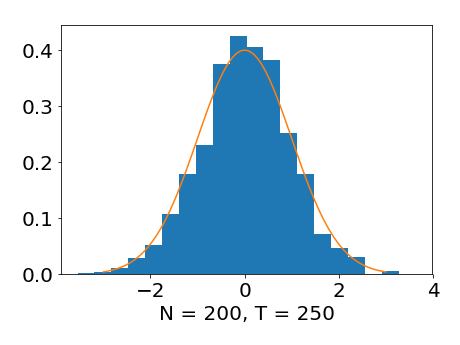}
		\end{subfigure}%
		\begin{subfigure}{.22\textwidth}
			\centering
			\includegraphics[width=1\linewidth]{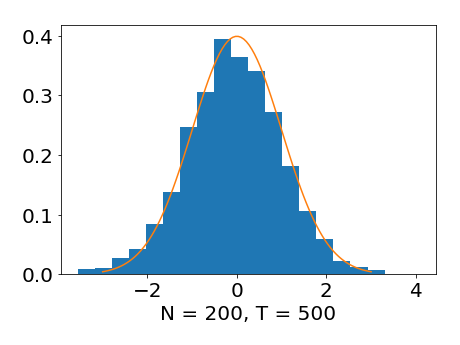}
		\end{subfigure}%
		\begin{subfigure}{.22\textwidth}
			\centering
			\includegraphics[width=1\linewidth]{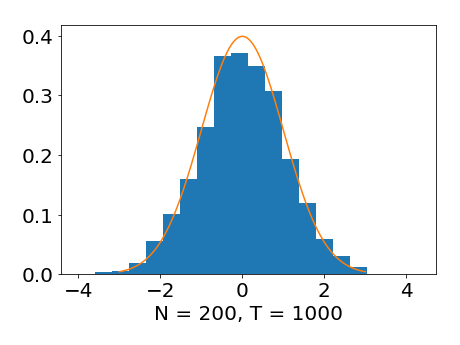}
		\end{subfigure}
		\caption{Histograms of estimated standardized common components. ($N =50, 100, 200; T=250, 500, 1000; h=0.3$).  The normal density function is superimposed on the histograms (Cross-sectionally dependent errors)}
		\label{hist_comp_cross}
	\end{figure}
	
	\begin{figure}[H]
		\centering
		\begin{subfigure}{.22\textwidth}
			\centering
			\includegraphics[width=1\linewidth]{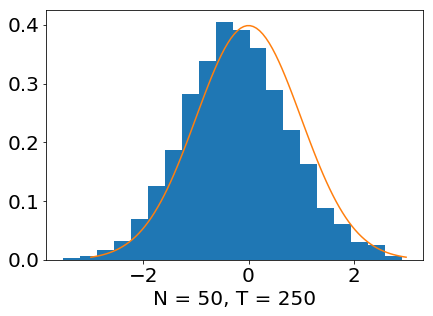}
		\end{subfigure}%
		\begin{subfigure}{.22\textwidth}
			\centering
			\includegraphics[width=1\linewidth]{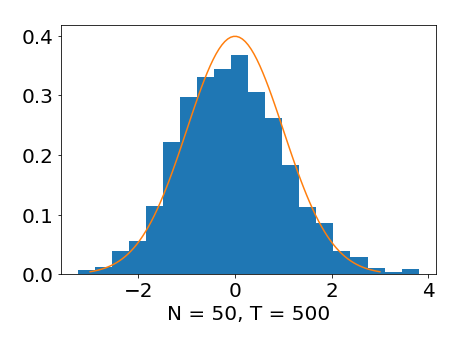}
		\end{subfigure}%
		\begin{subfigure}{.22\textwidth}
			\centering
			\includegraphics[width=1\linewidth]{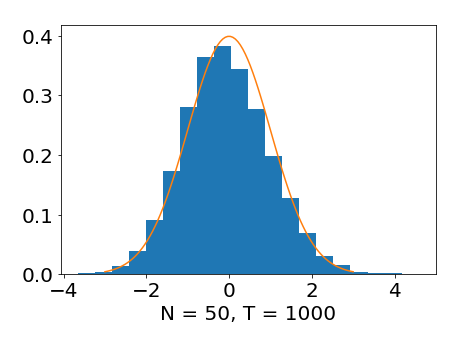}
		\end{subfigure}
		\begin{subfigure}{.22\textwidth}
			\centering
			\includegraphics[width=1\linewidth]{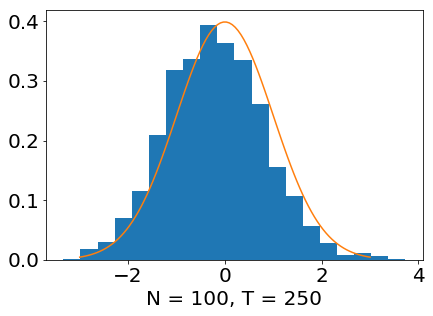}
		\end{subfigure}%
		\begin{subfigure}{.22\textwidth}
			\centering
			\includegraphics[width=1\linewidth]{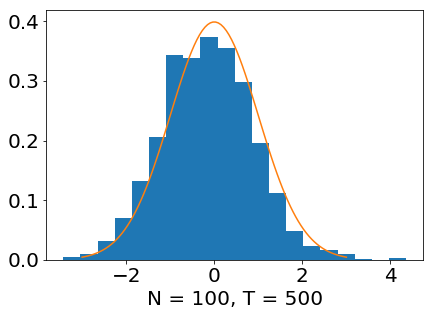}
		\end{subfigure}%
		\begin{subfigure}{.22\textwidth}
			\centering
			\includegraphics[width=1\linewidth]{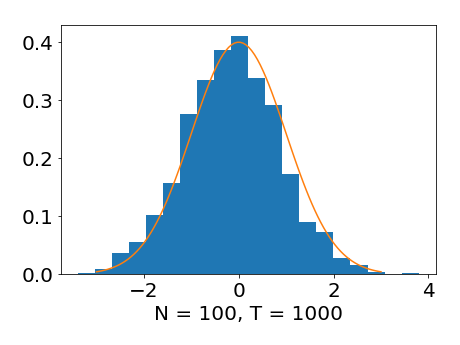}
		\end{subfigure}
		\begin{subfigure}{.22\textwidth}
			\centering
			\includegraphics[width=1\linewidth]{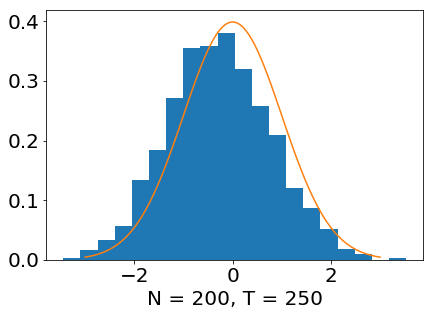}
		\end{subfigure}%
		\begin{subfigure}{.22\textwidth}
			\centering
			\includegraphics[width=1\linewidth]{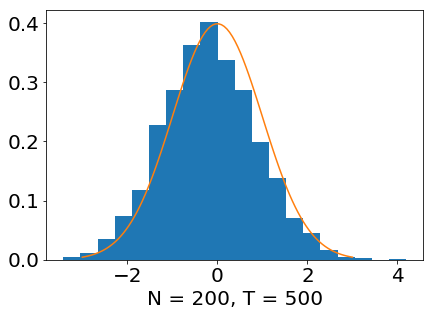}
		\end{subfigure}%
		\begin{subfigure}{.22\textwidth}
			\centering
			\includegraphics[width=1\linewidth]{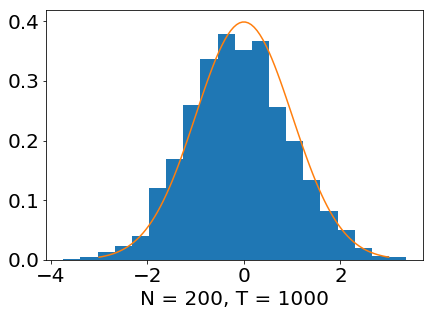}
		\end{subfigure}
		\caption{Histograms of estimated standardized and bias-corrected generalized correlation test statistic. ($N =50, 100, 200; T=250, 500, 1000; h=0.3$). The normal density function is superimposed on the histograms (Cross-sectionally dependent errors)}
		\label{hist_gen_cross}
	\end{figure}

	\newpage

	\section{Proof of Asymptotic Results}\label{sec:proofs}
	
	\subsection{Proofs of Asymptotic Distribution Results}
	
	Define $X^s_t = \Lambda(S_t) F^s_t + e^s_t$ and $\bar{X}_t^s = \Lambda(s) F^s_t + e^s_t$, we have $\Delta X^s_t = X^s_t - \bar{X}_t^s = (\Lambda(S_t) - \Lambda(s))F^s_t$. Similarly, $\bar{X}_t = \Lambda(s) F_t + e_t$ and $\Delta X_t = X_t - \bar{X}_t = (\Lambda(S_t) - \Lambda(s))F_t$.
	
	\begin{lemma} \label{lemma1}
		Under Assumption \ref{Ass:Ident}-\ref{ass_loading}, $h \rightarrow 0$, $Th \rightarrow \infty$, if there exists $T_0$ with $T_0/T \rightarrow 0$ such that $S_{T_0} \sim \pi$, then 
		\begin{enumerate}
			\item $\frac{T(s)}{T} \xrightarrow{p} \pi(s)$  and $\frac{T}{T(s)}  \xrightarrow{p} \frac{1}{\pi(s)}$ \label{ts}
			\item $\frac{1}{N} \sum_{i = 1}^N ( \Delta X_{it}^s )^2  = O_p(h)$ and $\frac{1}{T(s)} \sum_{t = 1}^T (\Delta X_{it}^s)^2 = O_p(h^2) $ \label{deltaX}
		\end{enumerate}
	\end{lemma}
	
	\begin{proof}[Proof of Lemma \ref{lemma1}.\ref{ts}]
		If $S_{T_0} \sim \pi$, then $S_t \sim \pi $ for all $t \geq T_0$. 
		Since $\pi(s) > 0$,
		\begin{eqnarray*}
			\frac{T(s)}{T}  &=& \frac{1}{T} \sum_{t = 1}^{T_0 - 1} K_s(S_t) +  \frac{1}{T}\sum_{t = T_0}^T K_s(S_t) \\
			&=&  \frac{T-T_0}{T} \Bigg[  \pi(s) + O(h^2) + O_p\left(\frac{1}{\sqrt{Th}}\right)   \Bigg] + o_p(1) =\pi(s) + o_p(1)
		\end{eqnarray*}
		following $T_0/T \rightarrow 0 $. Moreover, since  $\pi(s) > 0$, we have 
		$$\frac{T}{T(s)}  = O_p(1).$$
	\end{proof}
	
	\begin{proof}[Proof of Lemma \ref{lemma1}.\ref{deltaX}]
		
		\begin{eqnarray*}
			(\Delta X_{it}^s)^2  &=& ((\Lambda_i(S_t) - \Lambda_i(s))^\T  F^s_t)^2 \\
			&\leq& K_s(S_t) \norm{ \Lambda_i(S_t) - \Lambda_i(s)}^2  \norm{F_t }^2 \\
			&\stackrel{S_t = s + uh}{=}& \frac{1}{h}K(u) \norm{ \Lambda_i(S_t) - \Lambda_i(s)}^2  \norm{F_t }^2\\
			&\leq& \frac{1}{h}K(u) (C^2u^2h^2)  \norm{F_t }^2 \\
			&=& O_p(h)
		\end{eqnarray*}
		by $K_s(S_t) = \frac{1}{h} K(\frac{S_t - s}{h}) = \frac{1}{h}K(u)$, when $S_t = s + uh$, and by Assumption \ref{ass_loading}, $\norm{\Lambda_i(s + \Delta s) - \Lambda_i(s)} \leq C \Delta s$, $\forall s, \Delta s$ and $i$ and $\int u^4 k(u) du$ exists so $\lim_{u \rightarrow \infty} u^2 k(u)$ is bounded. Therefore,
		\begin{eqnarray*}
			\frac{1}{N} \sum_{i = 1}^N (\Delta X_{it}^s )^2   = O_p(h) \label{deltaX_N}.
		\end{eqnarray*} 
		
		Furthermore,
		\begin{eqnarray*}
			\frac{1}{T(s)} \sum_{t = 1}^T  (\Delta X_{it}^s)^2 &=& \frac{1}{T(s)} \sum_{t = 1}^T ((\Lambda_i(S_t) - \Lambda_i(s))^\T  F^s_t)^2 \\ 
			&=& \frac{1}{T(s)} \sum_{t = 1}^T K_s(S_t) ((\Lambda_i(S_t) - \Lambda_i(s))^\T  F_t)^2 \\ 
			&\leq& \frac{1}{T(s)} \sum_{t = 1}^T K_s(S_t) \norm{(\Lambda_i(S_t) - \Lambda_i(s)}^2 \norm{F_t}^2  \\ 
			&\leq& \left( \frac{1}{T(s)} \sum_{t = 1}^T K_s(S_t) \norm{(\Lambda_i(S_t) - \Lambda_i(s)}^4 \right)^{1/2} \left(\frac{1}{T(s)} \sum_{t = 1}^T K_s(S_t) \norm{F_t}^4 \right)^{1/2}.
		\end{eqnarray*}
		By Assumption \ref{ass_factor}, $\max_t\, \+E[\norm{F_t}^4 | S_t = s]  \leq \bar{F} < \infty$, so $(\frac{1}{T(s)} \sum_{t = 1}^T K_s(S_t) \norm{F_t}^4)^{1/2} = \+E\left[ \norm{F_t}^4 | S_t = s \right] + O_p(h^2) + O_P(\frac{1}{\sqrt{Th}}) = O_p(1)$. Also,
		\begin{eqnarray*}
			\frac{1}{T} \sum_{t = 1}^T K_s(S_t) \norm{\Lambda_i(S_t) - \Lambda_i(s)}^4 &\leq& \frac{C}{T} \sum_{t = 1}^T K_s(S_t) (S_t - s)^4 \\
			&=& C \int K(u) u^4h^4 \pi(s+uh)du = O_p(h^4)
		\end{eqnarray*}
		when $\int u^4 k(u) du$ exists. 
		Therefore, 
		
		\begin{eqnarray*}
			\frac{1}{T(s)} \sum_{t = 1}^T (\Delta X_{it}^s)^2 = O_p(h^2).
		\end{eqnarray*} 
	\end{proof}

	\begin{lemma}\label{lemma:simplify-assumption6}
		Let $R_K = \int K^2(u) du$. If there exists $T_0$ with $T_0/T \rightarrow 0$ such that $S_{T_0} \sim \pi$, then for all $N$, $T$, and $h$ that satisfies $h \rightarrow 0$, $Nh \rightarrow \infty$, we have 
		$Th \rightarrow \infty$, $Nh^2 \rightarrow 0$ and $Th^3 \rightarrow 0$,
		\begin{enumerate}
			\item  $\+E \norm{\sqrt{\frac{Th}{N}} \frac{1}{T(s)}  \sum_{u = 1}^T \sum_{i = 1}^N K_s(S_u) F_u(e_{iu}e_{it}- \+E[e_{iu}e_{it}])}^2  $ \\ $= \frac{1}{N T} \sum_{l = 1}^r \sum_{l^\prime = 1}^r  \sum_{i = 1}^N \sum_{j = 1}^{N} \left(\sum_{u = 1}^{T} \frac{R_K}{\pi(s)} \gamma^s_{Fe, t, l, l^\prime}(i,j,u,u)  +  h \sum_{p \neq u} \gamma^s_{Fe, t, l, l^\prime}(i,j,u,p) \right) + o_p(1)  $.
			\item  $\+E \norm{\sqrt{\frac{Th}{N}} \frac{1}{T(s)} \sum_{t = 1}^T \sum_{i = 1}^N K_s(S_t) F_t  \Lambda_i(s)^\T   e_{it}}^2 $\\ $ = \frac{1}{N T} \sum_{l = 1}^{r^2} \sum_{l^\prime = 1}^{r^2}  \sum_{i = 1}^N \sum_{j = 1}^{N} \left(\sum_{u = 1}^{T} \frac{R_K}{\pi(s)} \gamma^s_{F\Lambda e, l, l^\prime}(i,j,u,u)  +  h \sum_{p \neq u} \gamma^s_{F\Lambda e, l, l^\prime}(i,j,u,p) \right) + o_p(1) $.
			\item  $ \+E \norm{\sqrt{\frac{Th}{N}} \frac{1}{T(s)} \sum_{t = 1}^T \sum_{l = 1}^N K_s(S_t) \Lambda_l(s)(e_{lt}e_{it} - \+E[e_{lt}e_{it}])}^2 $ \\ $= \frac{1}{N T} \sum_{l = 1}^r \sum_{l^\prime = 1}^r  \sum_{i = 1}^N \sum_{j = 1}^{N} \left(\sum_{u = 1}^{T} \frac{R_K}{\pi(s)} \gamma^s_{\Lambda e,m, l, l^\prime}(i,j,u,u)  +  h \sum_{p \neq u} \gamma^s_{\Lambda e, m, l, l^\prime}(i,j,u,p) \right) + o_p(1)$.
			\item $ \+E \norm{\frac{\sqrt{Th}}{\sqrt{N} T(s_l)}\sum_{i=1}^N (F^{s_l})^\T  e_i^{s_l}\lambda^\T_{l'i} }^2 $ \\
			$= \frac{1}{N T} \sum_{m = 1}^{r^2} \sum_{m^\prime = 1}^{r^2}  \sum_{i = 1}^N \sum_{j = 1}^{N} \left(\sum_{u = 1}^{T} \frac{R_K}{\pi(s_l)} \gamma^{s_l, s_{l^\prime}}_{F\Lambda e, m, m^\prime}(i,j,u,u)  +  h \sum_{p \neq u} \gamma^{s_l, s_{l^\prime}}_{F\Lambda e, m, m^\prime}(i,j,u,p) \right) + o_p(1)$
			\item $E\norm{\frac{\sqrt{Th}}{NT(s_l)} \sum_{i=1}^N \sum_{k=1}^N \lambda_{li} \lambda^\T_{l'i}  \sum_{t=1}^T [e^{s_l}_{it}e^{s_l}_{kt} - \+E(e^{s_l}_{it}e^{s_l}_{kt})]  }^2 $ \\
			$=\frac{1}{N T} \sum_{m = 1}^{r^2} \sum_{m^\prime = 1}^{r^2}  \sum_{i = 1}^N    \sum_{i^\prime = 1}^N \sum_{j = 1}^{N}  \sum_{j^\prime = 1}^{N} \Big(\sum_{u = 1}^{T} \frac{R_K}{\pi(s_l)} \gamma^{s_l, s_{l^\prime}}_{\Lambda e, m, m^\prime}(i,i^\prime,j,j^\prime,u,u) $ \\ $ +  h \sum_{p \neq u} \gamma^{s_l, s_{l^\prime}}_{\Lambda e, m, m^\prime}(i,i^\prime,j,j^\prime,u,p) \Big)$
			\item $E\norm{\frac{Th}{NT^2(s_l)} \sum_{i=1}^N \sum_{k=1}^N (F^{s_l})^\T  e_k^{s_l} \lambda^\T_{l'i}  \sum_{t=1}^T [e^{s_l}_{it}e^{s_l}_{kt} - \+E(e^{s_l}_{it}e^{s_l}_{kt})]  }^2$ \\
			$= \frac{1}{N^2 T^2}  \sum_{m} \sum_{m^\prime} \sum_{i = 1}^{N}  \sum_{i^\prime = 1}^{N} \sum_{j = 1}^{N}  \sum_{j^\prime = 1}^{N}  \sum_{u = 1}^{T}  \sum_{u^\prime = 1}^{T} \sum_{p = 1}^{T}  \sum_{p^\prime = 1}^{T}   $ \\ $ c(u,u^\prime,p,p^\prime)  \cdot \gamma^{s_l}_{F \Lambda e, m, m^\prime} (i,i^\prime,j,j^\prime,u,u^\prime,p, p^\prime)  + o_p(1)$
		\end{enumerate}
		$\gamma^s_{Fe, t, l, l^\prime}(i,j,u,p) $, $\gamma^s_{F \Lambda e, t, l, l^\prime}(i,j,u,p) $ and $\gamma^s_{\Lambda e, t, l, l^\prime}(i,j,u,p) $ are defined in Assumption \ref{ass_mom}, and $\gamma^{s_l, s_{l^\prime}}_{F\Lambda e, m, m^\prime}(i,j,u,p) $, $\gamma^{s_l, s_{l^\prime}}_{\Lambda e, m, m^\prime}(i,i^\prime,j,j^\prime,u,p) $ $\gamma^{s_l}_{F \Lambda e, m, m^\prime} (i,i^\prime,j,j^\prime,u,u^\prime,p, p^\prime) $ and $c(u,u^\prime,p,p^\prime) $ are defined in Assumption \ref{doublesum}.
	\end{lemma}
	
	\begin{proof}[Proof of Lemma \ref{lemma:simplify-assumption6}]
		Without loss of generality, we can assume $T_0 = 1$ because the difference between $T_0 = 1$ and some other value in the value of Lemma \ref{lemma:simplify-assumption6}.1-6 is $o_p(1)$ following $T_0/T \rightarrow 0$.
		If we can show 
		\begin{eqnarray*}
			\+E \norm{\sqrt{\frac{Th}{N} } \frac{1}{T(s)}  \sum_{u = 1}^{T} \sum_{i = 1}^{N} K_s(S_u)  \*z_{iu} }^2  & =& \frac{1}{NT} \sum_{l} \sum_{l^\prime} \sum_{i = 1}^{N} \sum_{j = 1}^{N} \Bigg(  \sum_{u=1}^{T} \frac{R_K}{\pi(s)} \+E[ \tvec(\*z_{iu})_l \tvec(\*z_{ju})_{l^\prime}  |S_u = s]  \\ 
			&& + h \sum_{p \neq u}  \+E[ \tvec(\*z_{iu})_l \tvec(\*z_{jp})_{l^\prime}  |S_u = s, S_p = s]  \Bigg) + o_p(1),
		\end{eqnarray*}
		then we can plug $ F_u(e_{iu}e_{it}- \+E[e_{iu}e_{it}])$, $  F_t  \Lambda_i(s)^\T   e_{it}  $, $\Lambda_l(s)(e_{lt}e_{it} - \+E[e_{lt}e_{it}])$ and $ F_t  \Lambda_i(s_{l^\prime})^\T   e_{it}  $ into $\*z_{iu}$ and Lemmas \ref{lemma:simplify-assumption6}.1, \ref{lemma:simplify-assumption6}.2, \ref{lemma:simplify-assumption6}.3 and \ref{lemma:simplify-assumption6}.4 holds. 
		
		Note that 
		\begin{eqnarray*}
			\+E \norm{\sqrt{\frac{Th}{N} }  \frac{1}{T(s)}  \sum_{u = 1}^{T} \sum_{i = 1}^{N} K_s(S_u)  \*z_{iu} }^2  & =& \frac{Th}{NT(s)^2} \sum_{l} \sum_{l^\prime} \sum_{i = 1}^{N} \sum_{j = 1}^{N} \Bigg(  \sum_{u=1}^{T}  \+E[ K_s(S_u)^2 \tvec(\*z_{iu})_l \tvec(\*z_{ju})_{l^\prime} ] \\
			&& + \sum_{p \neq u}  \+E[ K_s(S_u) K_s(S_p)   \tvec(\*z_{iu})_l \tvec(\*z_{jp})_{l^\prime} ]  \Bigg), 
		\end{eqnarray*}
		
		
		Suppose  $\tvec(\*z_{iu})_l \tvec(\*z_{jp})_{l^\prime} = g_{i,j,l,l^\prime}(S_u, S_p) + \varepsilon_{i,j,l,l^\prime,u,p} $ with $\+E[\varepsilon_{i,j,l,l^\prime,u,p} | S_u, S_p] = 0$.
		
		Let us first consider $\+E[ K_s(S_u)^2 \tvec(\*z_{iu})_l \tvec(\*z_{ju})_{l^\prime} ] $. We have 
		\begin{align*}
		\+E[ K_s(S_u)^2 \tvec(\*z_{iu})_l \tvec(\*z_{ju})_{l^\prime} ]  =& \frac{1}{h^2} \int k \Big( \frac{s^\prime - s}{h} \Big)^2 \Big( g_{i,j,l,l^\prime}(s^\prime, s^\prime) + \varepsilon_{i,j,l,l^\prime,u,u}  \Big) \pi(s^\prime) ds^\prime \\
		=& \frac{1}{h^2} \int k \Big( \frac{s^\prime - s}{h} \Big)^2  g_{i,j,l,l^\prime}(s^\prime, s^\prime)  \pi(s^\prime) ds^\prime \\
		=& \frac{1}{h} \int k(x)^2   g_{i,j,l,l^\prime}(s+xh, s+xh)  \pi(s+xh) d x \\
		=& \frac{1}{h} \int k(x)^2   g_{i,j,l,l^\prime}(s, s)  \pi(s) d x + o \Big(\frac{1}{h}\Big) \\
		=& \frac{1}{h}  g_{i,j,l,l^\prime}(s, s)  R(k)   \pi(s) + o \Big(\frac{1}{h}\Big)  \\
		=& \frac{ R(k)   \pi(s) }{h} \+E[ \tvec(\*z_{iu})_l \tvec(\*z_{ju})_{l^\prime}  |S_u = s]  + o \Big(\frac{1}{h}\Big) 
		\end{align*}
		Next let us consider $ \+E[ K_s(S_u) K_s(S_p)   \tvec(\*z_{iu})_l \tvec(\*z_{jp})_{l^\prime} ]  $. We have 
		\begin{align*}
		\+E[ K_s(S_u) K_s(S_p)   \tvec(\*z_{iu})_l \tvec(\*z_{jp})_{l^\prime} ]  =&  \frac{1}{h^2} \int \int k \Big( \frac{s^\prime - s}{h} \Big) k \Big( \frac{s^{\prime\prime} - s}{h} \Big) \Big( g_{i,j,l,l^\prime}(s^\prime, s^{\prime\prime}) + \varepsilon_{i,j,l,l^\prime,u,p}  \Big) \pi(s^\prime, s^{\prime\prime}) ds^\prime  d s^{\prime\prime}\\
		=& \frac{1}{h^2} \int \int k \Big( \frac{s^\prime - s}{h} \Big) k \Big( \frac{s^{\prime\prime} - s}{h} \Big) g_{i,j,l,l^\prime}(s^\prime, s^{\prime\prime}) \pi(s^\prime, s^{\prime\prime}) ds^\prime  d s^{\prime\prime}\\
		=& \frac{1}{h^2} \int \int k \Big( \frac{s^\prime - s}{h} \Big) k \Big( \frac{s^{\prime\prime} - s}{h} \Big) g_{i,j,l,l^\prime}(s^\prime, s^{\prime\prime}) \pi(s^\prime)\pi( s^{\prime\prime}) ds^\prime  d s^{\prime\prime}\\
		=&\int \int k(x) k(y)  g_{i,j,l,l^\prime}(s+xh, s+yh)   \pi(s+xh) \pi(s+yh) d x dy  \\
		=& \int \int k(x)  k(y)    g_{i,j,l,l^\prime}(s, s)  \pi(s)^2  d x dy + o(1) \\
		=&  g_{i,j,l,l^\prime}(s, s)   \pi(s)^2 + o(1) \\
		=&  \pi(s)^2  \+E[ \tvec(\*z_{iu})_l \tvec(\*z_{jp})_{l^\prime}  |S_u = s, S_p = s]  + o(1) 
		\end{align*}
		From Lemma \ref{lemma1}, $\frac{T}{T(s)} \xrightarrow{p} \frac{1}{\pi(s)}$.  Then we have
		\begin{align*}
		\+E \norm{\sqrt{\frac{Th}{N} }  \frac{1}{T(s)}  \sum_{u = 1}^{T} \sum_{i - 1}^{N} K_s(S_u)  \*z_{iu} }^2  =& \frac{h}{N T \pi(s)^2} \sum_{l} \sum_{l^\prime} \sum_{i = 1}^{N} \sum_{j = 1}^{N} \Bigg(  \sum_{u=1}^{T}  \frac{ R(k)   \pi(s) }{h} \+E[ \tvec(\*z_{iu})_l \tvec(\*z_{ju})_{l^\prime}  |S_u = s] \\
		& + \sum_{p \neq u}  \pi(s)^2  \+E[ \tvec(\*z_{iu})_l \tvec(\*z_{jp})_{l^\prime}  |S_u = s, S_p = s]    \Bigg) + o_p(1) \\
		=& \frac{h}{N T} \sum_{l} \sum_{l^\prime} \sum_{i = 1}^{N} \sum_{j = 1}^{N} \Bigg(  \sum_{u=1}^{T}  \frac{ R(k)  }{ \pi(s)} \+E[ \tvec(\*z_{iu})_l \tvec(\*z_{ju})_{l^\prime}  |S_u = s] \\
		& + h \sum_{p \neq u}  \+E[ \tvec(\*z_{iu})_l \tvec(\*z_{jp})_{l^\prime}  |S_u = s, S_p = s]    \Bigg) + o_p(1)
		\end{align*}
		Hence, if we can plug $ F_u(e_{iu}e_{it}- \+E[e_{iu}e_{it}])$, $  F_t  \Lambda_i(s)^\T   e_{it}  $, $\Lambda_l(s)(e_{lt}e_{it} - \+E[e_{lt}e_{it}])$ and $ F_t  \Lambda_i(s_{l^\prime})^\T   e_{it}  $ into $\*z_{iu}$ and Lemmas \ref{lemma:simplify-assumption6}.1, \ref{lemma:simplify-assumption6}.2, \ref{lemma:simplify-assumption6}.3 and \ref{lemma:simplify-assumption6}.4 holds. \ref{lemma:simplify-assumption6}.5 can be shown similarly. 
		
		Let $\*z_{ikut} = F_u    \lambda^\T_{l'i} e_{ku}  [e_{it}e_{kt} - \+E(e_{it}e_{kt})] $. Then for Lemma \ref{lemma:simplify-assumption6}.6, we have 
		\begin{align*}
		&E\norm{\frac{Th}{NT^2(s_l)} \sum_{i=1}^N \sum_{k=1}^N (F^{s_l})^\T  e_k^{s_l} \lambda^\T_{l'i}  \sum_{t=1}^T [e^{s_l}_{it}e^{s_l}_{kt} - \+E(e^{s_l}_{it}e^{s_l}_{kt})]  }^2 \\
		=& \frac{T^2 h^2}{N^2 T^4(s_l)} \sum_{m} \sum_{m^\prime} \sum_{i = 1}^{N}  \sum_{i^\prime = 1}^{N} \sum_{j = 1}^{N}  \sum_{j^\prime = 1}^{N}  \sum_{u = 1}^{T}  \sum_{u^\prime = 1}^{T} \sum_{p = 1}^{T}  \sum_{p^\prime = 1}^{T}  \+E[ K_s(S_u) K_s(S_{u^\prime})   K_s(S_p) K_s(S_{p^\prime})  \tvec(\*z_{ijup})_m \tvec(\*z_{i^\prime j^\prime u^\prime p^\prime})_{m^\prime} ]
		\end{align*}
		
		Suppose  $ \tvec(\*z_{ijup})_m \tvec(\*z_{i^\prime j^\prime u^\prime p^\prime})_{m^\prime}  = g_{i,i^\prime, j, j^\prime, m ,m^\prime}(S_u, S_{u^\prime}, S_p, S_{p^\prime}) + \varepsilon_{i,i^\prime, j, j^\prime, m ,m^\prime,u,u^\prime,p,p^\prime} $  with  \\ $\+E[\varepsilon_{i,i^\prime, j, j^\prime, m ,m^\prime,u,u^\prime,p,p^\prime} | S_u, S_{u^\prime}, S_p, S_{p^\prime}] = 0$. With similar algebra manipulation,  if $u = u^\prime = p = p^\prime$, we have 
		\begin{align*}
		& \+E[ K_s(S_u) K_s(S_{u^\prime})   K_s(S_p) K_s(S_{p^\prime})  \tvec(\*z_{ijup})_m \tvec(\*z_{i^\prime j^\prime u^\prime p^\prime})_{m^\prime} ] \\
		=& \frac{1}{h^3}  \int k(x)^4 dx  \cdot \pi(s)   g_{i,i^\prime, j, j^\prime, m ,m^\prime}(s,s,s,s)+ o \Big(\frac{1}{h^3}\Big); 
		\end{align*}
		if three variables in $u, u^\prime , p, p^\prime $ are the same, we have
		\begin{align*}
		& \+E[ K_s(S_u) K_s(S_{u^\prime})   K_s(S_p) K_s(S_{p^\prime})  \tvec(\*z_{ijup})_m \tvec(\*z_{i^\prime j^\prime u^\prime p^\prime})_{m^\prime} ] \\
		=& \frac{1}{h^2}  \int k(x)^3 dx \cdot  \pi(s)^2   g_{i,i^\prime, j, j^\prime, m ,m^\prime}(s,s,s,s)+ o \Big(\frac{1}{h^2}\Big); 
		\end{align*}
		if two variables in $u, u^\prime , p, p^\prime $ are the same and other two are same as well, we have
		\begin{align*}
		& \+E[ K_s(S_u) K_s(S_{u^\prime})   K_s(S_p) K_s(S_{p^\prime})  \tvec(\*z_{ijup})_m \tvec(\*z_{i^\prime j^\prime u^\prime p^\prime})_{m^\prime} ] \\
		=& \frac{1}{h^2}  R(k)^2 \cdot \pi(s)^2   g_{i,i^\prime, j, j^\prime, m ,m^\prime}(s,s,s,s)+ o \Big(\frac{1}{h^2}\Big); 
		\end{align*}
		if two variables in $u, u^\prime , p, p^\prime $ are the same and other two are different, we have
		\begin{align*}
		& \+E[ K_s(S_u) K_s(S_{u^\prime})   K_s(S_p) K_s(S_{p^\prime})  \tvec(\*z_{ijup})_m \tvec(\*z_{i^\prime j^\prime u^\prime p^\prime})_{m^\prime} ] \\
		=& \frac{1}{h}  R(k) \cdot \pi(s)^3   g_{i,i^\prime, j, j^\prime, m ,m^\prime}(s,s,s,s)+ o \Big(\frac{1}{h}\Big); 
		\end{align*}
		if all variables in $u, u^\prime , p, p^\prime $ are different, we have
		\begin{align*}
		& \+E[ K_s(S_u) K_s(S_{u^\prime})   K_s(S_p) K_s(S_{p^\prime})  \tvec(\*z_{ijup})_m \tvec(\*z_{i^\prime j^\prime u^\prime p^\prime})_{m^\prime} ] \\
		=&  \pi(s)^4   g_{i,i^\prime, j, j^\prime, m ,m^\prime}(s,s,s,s)+ o(1).
		\end{align*}
		From Lemma \ref{lemma1}, $\frac{T}{T(s)} \xrightarrow{p} \frac{1}{\pi(s)}$.  Then we have
		\begin{align*}
		&E\norm{\frac{Th}{NT^2(s_l)} \sum_{i=1}^N \sum_{k=1}^N (F^{s_l})^\T  e_k^{s_l} \lambda^\T_{l'i}  \sum_{t=1}^T [e^{s_l}_{it}e^{s_l}_{kt} - \+E(e^{s_l}_{it}e^{s_l}_{kt})]  }^2 \\
		=& \frac{h^2}{N^2 T^2 \pi(s_l)^4 } \sum_{m} \sum_{m^\prime} \sum_{i = 1}^{N}  \sum_{i^\prime = 1}^{N} \sum_{j = 1}^{N}  \sum_{j^\prime = 1}^{N}  \sum_{u = 1}^{T}  \sum_{u^\prime = 1}^{T} \sum_{p = 1}^{T}  \sum_{p^\prime = 1}^{T} \\
		& \quad  \+E[ K_{s_l}(S_u) K_{s_l}(S_{u^\prime})   K_{s_l}(S_p) K_{s_l}(S_{p^\prime})  \tvec(\*z_{ijup})_m \tvec(\*z_{i^\prime j^\prime u^\prime p^\prime})_{m^\prime} ]  + o_p(1) \\
		=&  \frac{1}{N^2 T^2}  \sum_{m} \sum_{m^\prime} \sum_{i = 1}^{N}  \sum_{i^\prime = 1}^{N} \sum_{j = 1}^{N}  \sum_{j^\prime = 1}^{N}  \sum_{u = 1}^{T}  \sum_{u^\prime = 1}^{T} \sum_{p = 1}^{T}  \sum_{p^\prime = 1}^{T}    \\
		& \quad c(u,u^\prime,p,p^\prime) \cdot  \+E[ \tvec(\*z_{ijup})_m \tvec(\*z_{i^\prime j^\prime u^\prime p^\prime})_{m^\prime}  |S_u = s_l, S_{u^\prime} = s_l, S_p = s_l, S_{p^\prime} = s_l]   + o_p(1),
		\end{align*}
		where 
		\begin{align*}
		c(u,u^\prime,p,p^\prime) = \begin{cases}
		\frac{\int k(x)^4 dx }{ h \pi(s_l)^3 } & u = u^\prime = p = p^\prime \\
		\frac{\int k(x)^3 dx }{ \pi(s_l)^2 } & \text{three in $\{u, u^\prime, p, p^\prime\}$ are the same}  \\
		\frac{R(k)^2}{ \pi(s_l)^2 } & \text{two in $\{u, u^\prime, p, p^\prime\}$ are the same, the other two are the same}  \\
		\frac{R(k) h}{ \pi(s_l) } & \text{two in $\{u, u^\prime, p, p^\prime\}$ are the same, the other two are  distinct}  \\
		h^2 & \text{all in $\{u, u^\prime, p, p^\prime\}$ are distinct}  \\
		\end{cases}
		\end{align*}

	\end{proof}

	\begin{proof}[Proof of Theorem 1]
		Since $\hat{F}^s$ are $\sqrt{T(s)}$ times eigenvectors corresponding to the $r$ largest eigenvalues of matrix $\frac{1}{NT} (X^s)^\T X^s$. $V^s_r$ is the diagonal matrix with diagonal values equal to $r$ largest eigenvalues in decreasing order of matrix $\frac{1}{NT(s)} (X^s)^\T X^s$ and all eigenvalues are bounded from 0. We plug $X^s = \Lambda(s) (F^s)^\T  + e^s + \Delta X^s = \bar{X}^s + \Delta X^s$ into
		\begin{eqnarray*}
			\left(\frac{1}{NT(s)} (X^s)^\T X^s \right) \hat{F}^s = \hat{F}^s V^s_r
		\end{eqnarray*}
		and get
		\begin{eqnarray*}
			\frac{1}{NT(s)} [ F^s \Lambda(s)^\T  \Lambda(s) (F^s)^\T  \hat{F}^s + F^s \Lambda(s)^\T  e^s \hat{F}^s + (e^s)^\T  \Lambda(s) (F^s)^\T  \hat{F}^s  && \\
			+ (e^s)^\T  e^s \hat{F}^s + (\Delta X^s)^\T \bar{X}^s \hat{F}^s + (\bar{X}^s)^\T  \Delta X^s\hat{F}^s  + (\Delta X^s)^\T \Delta X^s \hat{F}^s]   &=& \hat{F}^s V^s_r.
		\end{eqnarray*}
		Define $H^s = \frac{\Lambda(s)^\T  \Lambda(s)}{N} \frac{(F^s)^\T  \hat{F}^s}{T(s)}  (V^s_r)^{-1}$, we have 
		\begin{eqnarray*}
			&&V^s_r (\hat{F}^s_t - (H^s)^\T  F^s_t) \\
			&=& \frac{1}{NT(s)}\left[(\hat{F}^s)^\T  (e^s)^\T  \Lambda(s) F^s_t + (\hat{F}^s)^\T  F^s \Lambda(s)^\T  e^s_t + (\hat{F}^s)^\T  (e^s)^\T  e^s_t \right. \\
			&&+ \left. (\hat{F}^s)^\T  (\Delta X^s)^\T \bar{X}_t^s  + (\hat{F}^s)^\T  (\bar{X}^s)^\T  \Delta X^s_t + (\hat{F}^s)^\T (\Delta X^s)^\T \Delta X^s_t \right] \\
			&=& \frac{1}{T(s)} \left[ \sum_{u = 1}^T\hat{F}^s_u \frac{(e^s_u)^\T \Lambda(s) F^s_t}{N} \right] + \frac{1}{T(s)} \left[\sum_{u = 1}^T\hat{F}^s_u \frac{(F^s_u)^\T  \Lambda(s)^\T  e^s_t}{N} \right]  \\
			&&+ \frac{1}{T(s)} \left[ \sum_{u = 1}^T\hat{F}^s_u \frac{(e^s_u)^\T  e^s_t - \+E[(e^s_u)^\T  e^s_t]}{N}  \right]  + \frac{1}{T(s)} \left[ \sum_{u = 1}^T\hat{F}^s_u \frac{\+E[(e^s_u)^\T  e^s_t]}{N} \right] \\
			&&+ \frac{1}{T(s)} \left[ \sum_{u = 1}^T\hat{F}^s_u \frac{(\Delta X^s_u)^\T \bar{X}_t^s}{N} \right] + \frac{1}{T(s)} \left[ \sum_{u = 1}^T\hat{F}^s_u \frac{( \bar{X}^s_u)^\T \Delta X^s_t}{N} \right] + \frac{1}{T(s)} \left[ \sum_{u = 1}^T\hat{F}^s_u \frac{(\Delta X^s_u)^\T \Delta X^s_t}{N} \right].
		\end{eqnarray*}
		
		We want to prove  $\frac{1}{T}\sum_{t = 1}^T \norm{V^s_r (\hat{F}^s_t - (H^s)^\T  F^s_t)}^2$ converges to 0 in probability at rate $\delta^2_{NT, h}$. It is equivalent to prove $\frac{1}{T(s)}\sum_{t = 1}^T \norm{V^s_r (\hat{F}^s_t - (H^s)^\T  F^s_t)}^2$ converges to 0 in probability at rate $\delta^2_{NT, h}$ by lemma \ref{lemma1}.\ref{ts}. First, let 
		\begin{eqnarray*}
			\gamma_N^s(u,t) &=&  \frac{K_s^{1/2}(S_u) K_s^{1/2}(S_t)\+E[e^\T_u e_t]}{N} = K_s^{1/2}(S_u) K_s^{1/2}(S_t) \gamma_N(u,t) \\
			\zeta^s_{ut} &=& \frac{(e^s_u)^\T  e^s_t - K_s^{1/2}(S_u) K_s^{1/2}(S_t)\+E[e^\T_u e_t]}{N} = K_s^{1/2}(S_u) K_s^{1/2}(S_t) \zeta_{ut}\\
			\eta^s_{ut} &=& \frac{(F^s_u)^\T  \Lambda(s)^\T  e^s_t}{N} = K_s^{1/2}(S_u) K_s^{1/2}(S_t) \eta_{ut} \\
			\epsilon^s_{ut} &=& \frac{(e^s_u)^\T \Lambda(s) F^s_t}{N} = K_s^{1/2}(S_u) K_s^{1/2}(S_t) \epsilon_{ut}
		\end{eqnarray*}
		and 
		\begin{eqnarray*}
			a_t &=& \frac{1}{T(s)^2} \norm{\sum_{u = 1}^T\hat{F}^s_u \frac{K_s^{1/2}(S_u) K_s^{1/2}(S_t)\+E[e^\T_u e_t]}{N}}^2 = \frac{1}{T(s)^2} \norm{\sum_{u = 1}^T\hat{F}^s_u \gamma_N^s(u,t)}^2\\
			b_t &=& \frac{1}{T(s)^2} \norm{\sum_{u = 1}^T\hat{F}^s_u \frac{(e^s_u)^\T  e^s_t - K_s^{1/2}(S_u) K_s^{1/2}(S_t)\+E[e^\T_u e_t]}{N}}^2 = \frac{1}{T(s)^2} \norm{\sum_{u = 1}^T\hat{F}^s_u \zeta^s_{ut}}^2 \\ 
			c_t &=& \frac{1}{T(s)^2} \norm{\sum_{u = 1}^T\hat{F}^s_u \frac{(F^s_u)^\T  \Lambda(s)^\T  e^s_t}{N}}^2 = \frac{1}{T(s)^2} \norm{\sum_{u = 1}^T\hat{F}^s_u \eta^s_{ut}}^2\\
			d_t &=& \frac{1}{T(s)^2} \norm{\sum_{u = 1}^T\hat{F}^s_u \frac{(e^s_u)^\T \Lambda(s) F^s_t}{N}}^2 = \frac{1}{T(s)^2} \norm{\sum_{u = 1}^T\hat{F}^s_u \epsilon^s_{ut}}^2\\
			f_t &=& \frac{1}{T(s)^2} \norm{\sum_{u = 1}^T\hat{F}^s_u \frac{(\Delta X^s_u)^\T \bar{X}_t^s}{N}}^2 \\
			g_t &=& \frac{1}{T(s)^2} \norm{\hat{F}^s_u \frac{( \bar{X}^s_u)^\T \Delta X^s_t}{N}}^2 \\
			h_t &=& \frac{1}{T(s)^2} \norm{\sum_{u = 1}^T\hat{F}^s_u \frac{(\Delta X^s_u)^\T \Delta X^s_t}{N}}^2,
		\end{eqnarray*}
		then
		\begin{eqnarray*}
			\norm{V^s_r (\hat{F}^s_t - (H^s)^\T  F^s_t)}^2 &\leq& 4 \{a_t + b_t + c_t + d_t + f_t + g_t + h_t\} \\
			\frac{1}{T(s)}\sum_{t = 1}^T \norm{V^s_r (\hat{F}^s_t - (H^s)^\T  F^s_t)}^2 &\leq& \frac{4}{T(s)}\sum_{t = 1}^T (a_t + b_t + c_t + d_t  + f_t + g_t + h_t)
		\end{eqnarray*}
		
		Next is to show $\frac{1}{T(s)}\sum_{t = 1}^T a_t, \frac{1}{T(s)}\sum_{t = 1}^T b_t, \frac{1}{T(s)}\sum_{t = 1}^T c_t, \frac{1}{T(s)}\sum_{t = 1}^T d_t, \frac{1}{T(s)}\sum_{t = 1}^T f_t,  \frac{1}{T(s)}\sum_{t = 1}^T g_t, \\ \frac{1}{T(s)}\sum_{t = 1}^T h_t$ converge to 0 at rate at least $\delta^2_{NT, h}$. \\ 
		Since $(\hat{F}^s)^\T \hat{F}^s/T(s) = I_r$, $tr((\hat{F}^s)^\T \hat{F}^s/T(s)) = \frac{1}{T(s)}\sum_{t = 1}^T \norm{\hat{F}^s_t}^2 = r = O(1)$.
		\begin{eqnarray*}
			\frac{1}{T(s)}\sum_{t = 1}^T a_t &=& \frac{1}{T} \sum_{t = 1}^T \frac{1}{T(s)^2} \norm{\sum_{u = 1}^T\hat{F}^s_u \gamma_N^s(u,t)}^2 \\
			&\leq& \frac{1}{T(s)} \left( \frac{1}{T(s)} \sum_{u = 1}^T \norm{\hat{F}^s_u}^2 \right) \left( \frac{1}{T(s)}\sum_{t = 1}^T \sum_{u = 1}^T  \gamma_N^s(u,t)^2 \right) \\
			&=& O_p \left(\frac{1}{T} \right) O\left(1\right) O_p\left( \frac{1}{h} \right)\\
			&=& O_p \left(\frac{1}{Th} \right)
		\end{eqnarray*}
		where we use Lemma \ref{lemma1}.1 and $\frac{1}{T(s)}\sum_{t = 1}^T \sum_{u = 1}^T \gamma^s_N(u,t)^2 = O_p(\frac{1}{h})$. The proof is as follows.
		
		Let $\rho(u,t) = \gamma_N(u,t)/[\gamma_N(u,u) \gamma_N(t,t)]^{1/2}$, then $\rho(u,t) \leq 1$. Also, $\gamma_N(t,t) \leq M$ from Assumption \ref{ass_err}.2. Therefore, $\gamma_N(u,t)^2 \leq M |\gamma_N(u,t)|$. Moreover, from Assumption \ref{ass_err}.2, $\sum_{t = 1}^T  |\gamma_N(u,t)| \leq M$. Together with $K_s(S_t) = \frac{1}{h}K(\frac{S_t - s}{h}) = O_p(\frac{1}{h})$, we have
		
		\begin{eqnarray} \label{gamma}
		\frac{1}{T(s)}\sum_{t = 1}^T \sum_{u = 1}^T \gamma^s_N(u,t)^2 &=& \frac{1}{T(s)}\sum_{t = 1}^T \sum_{u = 1}^T \frac{K_s(S_u) K_s(S_t)\+E[e^\T_u e_t]^2}{N^2} \nonumber \\
		&\leq&\max_t K_s(S_t)  \frac{1}{T(s)}  \sum_{u = 1}^T  K_s(S_u)  \left( \max_{k} \sum_{t = 1}^T \frac{\+E[e^\T_u e_t]^2}{N^2} \right) \nonumber\\
		&\leq& \max_t K_s(S_t)  \frac{1}{T(s)}  \sum_{u = 1}^T  K_s(S_u) \left( \max_{k} \sum_{t = 1}^T  \gamma_N^2(u,t) \right) \nonumber \\
		&\leq& \max_t K_s(S_t)  \frac{1}{T(s)}  \sum_{u = 1}^T  K_s(S_u) \left( M \max_{k} \sum_{t = 1}^T  |\gamma_N(u,t)| \right) \nonumber\\
		&=& O_p \left( \frac{1}{h} \right) O_p(1) O_p(1) = O_p \left( \frac{1}{h} \right).
		\end{eqnarray} 
		
		\begin{eqnarray*}
			\frac{1}{T(s)}\sum_{t = 1}^T b_t &=& \frac{1}{T(s)}\sum_{t = 1}^T \frac{1}{T(s)^2} \norm{\sum_{u = 1}^T\hat{F}^s_u \zeta^s_{ut}}^2 \\
			&\leq& \frac{1}{T(s)} \left( \frac{1}{T(s)}\sum_{u = 1}^T \norm{\hat{F}^s_u}^2 \right) \left( \frac{1}{T(s)^2} \sum_{u = 1}^T \sum_{l = 1}^T K_s(S_u) K_s(S_l)  \left[ \sum_{t = 1}^T K_s(S_t) \zeta_{ut} \zeta_{lt} \right]^2\right)^{1/2} \\
			&=& \frac{1}{T(s)} O_p(1) O_p \left(\frac{T}{N}\right) \\
			&=& O_p\left(\frac{1}{N} \right),
		\end{eqnarray*}
		where we use $\frac{1}{T(s)^2} \sum_{u = 1}^T \sum_{l = 1}^T K_s(S_u) K_s(S_l)  \left[ \sum_{t = 1}^T K_s(S_t) \zeta_{ut} \zeta_{lt} \right]^2 = O_p\left(\frac{T^2}{N^2}\right)$. 
		
		Note that $\left[ \sum_{t = 1}^T K_s(S_t) \zeta_{ut} \zeta_{lt} \right]^2 = \sum_{t = 1}^T \sum_{m = 1}^T K_s(S_t) K_s(S_m) \zeta_{ut} \zeta_{lt} \zeta_{km} \zeta_{lm}$ and $\+E[\zeta_{ut} \zeta_{lt} \zeta_{km} \zeta_{lm}] \leq \max_{u,t} \+E[\zeta_{ut}]^4 = O_p\left(\frac{1}{N^2}\right)$ by Assumption \ref{ass_err}.5. Thus, $\frac{1}{T(s)^2} \sum_{u = 1}^T \sum_{l = 1}^T K_s(S_u) K_s(S_l)  \left[ \sum_{t = 1}^T K_s(S_t) \zeta_{ut} \zeta_{lt} \right]^2 =  T^2 \frac{T(s)^2}{T^2} \left( \frac{1}{T(s)} \sum_{t = 1}^T K_s(S_t) \right)^4 O_p\left(\frac{1}{N^2}\right) = O_p \left(\frac{T^2}{N^2}\right)$ 
		
		\begin{eqnarray*}
			\frac{1}{T(s)}\sum_{t = 1}^T c_t &=& \frac{1}{T(s)}\sum_{t = 1}^T \frac{1}{T(s)^2} \norm{\sum_{u = 1}^T\hat{F}^s_u \eta^s_{ut}}^2 \\
			&\leq& \frac{1}{T(s)}\sum_{t = 1}^T \left( \frac{1}{N^2} K_s(S_t) \norm{\Lambda(s)^\T e_t}^2 \right) \left( \frac{1}{T(s)}\sum_{u = 1}^T \norm{\hat{F}^s_u}^2 \right) \left(\frac{1}{T(s)}\sum_{u = 1}^T K_s(S_u) \norm{F_u}^2 \right) \\
			&=& O_p\left(\frac{1}{N} \right) O_p(1) O_p(1) \\
			&=& O_p\left( \frac{1}{N} \right),
		\end{eqnarray*}
		where we use $\frac{1}{T(s)}\sum_{u = 1}^T K_s(S_u) \norm{F_u}^2 = O_p(1)$ by Assumption \ref{ass_factor} and 
		
		\begin{eqnarray*}
			\frac{1}{T(s)}\sum_{t = 1}^T \frac{1}{N^2} K_s(S_t) \norm{\Lambda(s)^\T e_t}^2 &=& \frac{1}{N} \left( \frac{1}{T(s)}\sum_{t = 1}^T  K_s(S_t) \norm{\Lambda(s)^\T e_t/\sqrt{N}}^2 \right) \\ 
			&=& \frac{1}{N}\left( \+E\left[\norm{\Lambda(s)^\T e_t/\sqrt{N}}^2|S_t = s \right] + O_p(h^2) + O_p\left(\frac{1}{\sqrt{Th}}\right) \right) \\
			&=& O_p\left(\frac{1}{N}\right)
		\end{eqnarray*} 
		by the independence of $e$ and $S$, $\+E\left[\norm{\Lambda(s)^\T e_t/\sqrt{N}}^2|S_t = s \right] = \+E\left[\norm{\Lambda(s)^\T e_t/\sqrt{N}}^2 \right]$, Lemma 1 (ii) in \cite{bai2002determining} and Assumption \ref{ass_loading} and \ref{ass_mom}.3.
		
		$\frac{1}{T(s)}\sum_{t = 1}^T d_t = O_p\left(\frac{1}{N}\right)$. The proof of is similar to $\frac{1}{T(s)}\sum_{t = 1}^T c_t$.
		
		\begin{eqnarray*}
			\frac{1}{T(s)}\sum_{t = 1}^T f_t &=& \frac{1}{T(s)}\sum_{t = 1}^T  \frac{1}{T(s)^2} \norm{\sum_{u = 1}^T\hat{F}^s_u \frac{(\Delta X^s_u)^\T \bar{X}_t^s}{N}}^2 \\
			&\leq& \left( \frac{1}{T(s)}\sum_{u = 1}^T \norm{\hat{F}^s_u}^2 \right) \left(\frac{1}{T(s)^2} \sum_{t = 1}^T  \sum_{u = 1}^T \left(\frac{(\Delta X^s_u)^\T \bar{X}_t^s}{N}\right)^2 \right) \\
			&\leq&\left(\frac{1}{T(s)}\sum_{u = 1}^T \norm{\hat{F}^s_u}^2 \right) \left( \frac{1}{T(s)^2} \sum_{t = 1}^T  \sum_{u = 1}^T \left(\sum_{i = 1}^N \frac{\Delta X_{iu}^s\bar{X}_{it}^s}{N}\right)^2 \right) \\
			&\leq& \left(\frac{1}{T(s)}\sum_{u = 1}^T \norm{\hat{F}^s_u}^2 \right) \left( \frac{1}{N^2T(s)^2} \sum_{t = 1}^T  \sum_{u = 1}^T \left(\sum_{i = 1}^N (\Delta X_{iu}^s)^2\right) \left(\sum_{i = 1}^N (\bar{X}_{it}^s)^2\right) \right) \\
			&\leq& \left(\frac{1}{T}\sum_{u = 1}^T \norm{\hat{F}^s_u}^2\right) \left(\frac{1}{NT(s)} \sum_{u = 1}^T  \sum_{i = 1}^N (\Delta X_{iu}^s)^2\right) \left(\frac{1}{NT(s)}  \sum_{t = 1}^T \sum_{i = 1}^N (\bar{X}_{it}^s)^2\right) \\
			&=& O_p(1) O_p(h^2) O_p(1)\\
			&=& O_p(h^2).  \label{fj}
		\end{eqnarray*}
		by Assumption \ref{ass_factor}, Lemma \ref{lemma1}.\ref{deltaX} and $\frac{1}{T(s)} \sum_{t = 1}^T (\bar{X}_{it}^s)^2 = \frac{1}{T(s)} \sum_{t = 1}^T K_s(S_t)(\Lambda_i(s)F_t + e_t)^2 = O_p(1)$.
		
		$\frac{1}{T(s)}\sum_{t = 1}^T g_t = O_p(h^2)$. The proof of $\frac{1}{T(s)}\sum_{t = 1}^T g_t$ is similar to $\frac{1}{T(s)}\sum_{t = 1}^T f_t$. 
		
		\begin{eqnarray*}
			\frac{1}{T(s)}\sum_{t = 1}^T h_t &=& \frac{1}{T(s)}\sum_{t = 1}^T  \frac{1}{T(s)^2} \norm{\sum_{u = 1}^T\hat{F}^s_u \frac{(\Delta X^s_u)^\T \Delta X^s_t}{N}}^2 \\
			&\leq& \left(\frac{1}{T(s)}\sum_{u = 1}^T \norm{\hat{F}^s_u}^2\right) \left(\frac{1}{T(s)^2} \sum_{t = 1}^T  \sum_{u = 1}^T \norm{\frac{(\Delta X^s_u)^\T  \Delta X^s_t}{N}}^2 \right) \\
			&\leq&  \left(\frac{1}{T(s)}\sum_{u = 1}^T \norm{\hat{F}^s_u}^2\right) \left(\frac{1}{T(s)^2} \sum_{t = 1}^T  \sum_{u = 1}^T  \frac{1}{N^2} \norm{\Delta X^s_t}^2 \norm{\Delta X^s_u}^2 \right) \\
			&=& O_p(h^4) .
		\end{eqnarray*}
		
		Therefore, we have 
		\begin{eqnarray*}
			\frac{1}{T(s)}\sum_{t = 1}^T \norm{V^s_r (\hat{F}^s_t - (H^s)^\T  F^s_t)}^2 = O_p\left(\max\left(\frac{1}{N}, \frac{1}{Th}, h^2\right)\right).
		\end{eqnarray*}
		Denote $\delta_{NT, h} = \min(\sqrt{N}, \sqrt{Th})$. When $\delta_{NT, h}h \rightarrow 0$ and by lemma \ref{lemma1}.1, we have $$\frac{1}{T}\sum_{t = 1}^T \norm{V^s_r (\hat{F}^s_t - (H^s)^\T  F^s_t)}^2 = O_p(\delta_{NT, h}^{-2}).$$
		
		Similarly, we can decompose 
		\begin{eqnarray*}
			\left(\frac{1}{NT(s)} X^s (X^s)^\T  \right) \bar \Lambda(s) = \bar \Lambda(s) V^s_r,
		\end{eqnarray*}
		where $\bar \Lambda(s)$ is eigenvectors corresponding to top $r$ eigenvalues of $\frac{1}{NT(s)} X^s (X^s)^\T $ and has $\bar \Lambda(s)^\T  \Lambda(s)/N = I_r$, and get
		\begin{eqnarray*}
			\frac{1}{NT(s)} [ \Lambda(s) (F^s)^\T  F^s \Lambda(s)^\T  \bar\Lambda(s) +  e^s F^s \Lambda(s)^\T  \bar\Lambda(s) +  \Lambda(s) (F^s)^\T  (e^s)^\T  \bar\Lambda(s)  && \\
			+  e^s (e^s)^\T  \bar\Lambda(s) + \bar{X}^s(\Delta X^s)^\T  \bar\Lambda(s) +  \Delta X^s (\bar{X}^s)^\T  \bar\Lambda(s)  + \Delta X^s(\Delta X^s)^\T  \bar\Lambda(s)]   &=& \bar\Lambda(s) V^s_r.
		\end{eqnarray*}
		Define $\bar H^s = \frac{(F^s)^\T F^s}{T(s)} \frac{\Lambda(s)^\T  \bar{\Lambda}(s)}{N} (V^{s}_r)^{-1}$ and decompose $V^s_r (\bar \Lambda(s) - \bar H^s \Lambda(s))$. By similar approach, we can show 
		\begin{eqnarray*}
			\frac{1}{N}\sum_{i = 1}^N \norm{\bar \Lambda_i(s) - (\bar H^s)^\T  \Lambda_i(s) }^2 = O_p(\delta_{NT, h}^{-2}),
		\end{eqnarray*}
		From \cite{Bai2008} page 10, we have $\hat{\Lambda}(s) = \Lambda(s) (V^{s}_r)^{1/2}$. From Lemma \ref{lemma4}.1 (shown later), we have 
		\[(H^s)^\T  = (V^s)^{-\frac{1}{2}} (\Upsilon^s)^\T  (\Sigma_{\Lambda(s)})^{\frac{1}{2}} + o_p(1/\delta_{NT,h}),\]
		where $\Upsilon^s$ eigenvectors corresponding to top eigenvalues of $  \Sigma_{\Lambda(s)}^{1/2}\Sigma_{F|s} \Sigma_{\Lambda(s)}^{1/2}$ and $(\Upsilon^s)^\T \Upsilon^s = I$. Similarly, 
		\[(\bar H^s)^\T  = (V^s)^{-\frac{1}{2}} (\bar \Upsilon^s)^\T  (\Sigma_{F|s})^{\frac{1}{2}} + o_p(1/\delta_{NT,h}),\]
		where $\bar \Upsilon^s$ eigenvectors corresponding to top eigenvalues of $ \Sigma_{F|s}^{1/2} \Sigma_{\Lambda(s)}\Sigma_{F|s}^{1/2}$ and $(\bar \Upsilon^s)^\T \bar \Upsilon^s = I$. From the definition of $\Upsilon^s$ and $\bar \Upsilon^s$, we have $\bar \Upsilon^s = \Sigma_{F|s}^{1/2} \Sigma_{\Lambda(s)}^{1/2} \Upsilon^s (V^s)^{-1/2} $. Thus,
		\begin{eqnarray*}
			\hat \Lambda_i(s) - (H^s)^{-1} \Lambda_i(s)&=& \left( (V^{s})^{1/2} \bar \Lambda_i(s) - (V^{s})^{1/2} (\Upsilon^s)^\T  \Sigma_{\Lambda(s)}^{-1/2}  \Lambda_i(s) \right) (1+ o_p(1/\delta_{NT,h})) \\
			&=& \left(\bar \Lambda_i(s) - (\bar H^s)^\T  \Lambda_i(s) \right) (1+ o_p(1/\delta_{NT,h})).
		\end{eqnarray*}
		Thus,
		\begin{eqnarray*}
			\frac{1}{N} \sum_{i=1}^N \norm{\hat{\Lambda}_i(s) - (H^s)^{-1} \Lambda_i(s)}^2 = O_p(\delta_{NT, h}^{-2}).
		\end{eqnarray*}
		
	\end{proof}

	\begin{lemma}\label{lemma2} 
		Under Assumption \ref{Ass:Ident}-\ref{ass_eigen},$Th \rightarrow \infty$, $\delta_{NT, h}h \rightarrow 0$, 
		\begin{enumerate}
			\item $\frac{1}{T(s)} (\hat{F}^s)^\T  (\frac{1}{NT(s)} (X^s)^\T X^s ) (\hat{F}^s)\xrightarrow{p} V^s$, where $V^s$ is the diagonal matrix consisting of the eigenvalues of $\Sigma_{\Lambda(s)}\Sigma_{F|s}$. \label{lemma21}
			\item $\frac{(\hat{F}^s)^\T  F^s}{T(s)} \xrightarrow{p} Q^s$, where $Q^s =(V^s)^{1/2} (\Upsilon^s)^\T  \Sigma_{\Lambda(s)} ^{-1/2}$ are eigenvalues of $  \Sigma_{\Lambda(s)}^{1/2}\Sigma_{F|s} \Sigma_{\Lambda(s)}^{1/2}$, and $\Upsilon^s$ is the corresponding eigenvector matrix such that $(\Upsilon^s)^\T \Upsilon^s = I$. \label{lemma22}
		\end{enumerate}
	\end{lemma}

	\begin{proof}[Proof of Lemma \ref{lemma2}]
		Left multiply $\left(\frac{1}{NT(s)} (X^s)^\T X^s \right) \hat{F}^s = \hat{F}^s V^s_r$ on both sides by $\frac{1}{T} \left(\frac{\Lambda(s)^\T \Lambda(s)}{N}\right)^{1/2} \left( F^s\right)'$
		\begin{eqnarray*}
			\left(\left(\frac{\Lambda(s)^\T \Lambda(s)}{N}\right)^{1/2} \frac{1}{T(s)} (F^s)^\T  \left(\frac{1}{NT(s)} (X^s)^\T X^s \right) \right) \hat{F}^s = \left(\frac{\Lambda(s)^\T \Lambda(s)}{N}\right)^{1/2} \frac{(F^s)^\T  \hat{F}^s}{T(s)} V^s_r
		\end{eqnarray*}
		Expanding $(X^s)^\T X^s$ with $X^s = \Lambda(s) (F^s)^\T  + e^s + \Delta X^s$
		
		\begin{eqnarray*}
			\left(\frac{\Lambda(s)^\T \Lambda(s)}{N}\right)^{1/2} \frac{(F^s)^\T F^s}{T(s)} \frac{\Lambda(s)^\T \Lambda(s)}{N}  \frac{(F^s)^\T  \hat{F}^s}{T(s)} + d_{NT}^s = \left(\frac{\Lambda(s)^\T \Lambda(s)}{N}\right)^{1/2} \frac{(F^s)^\T  \hat{F}^s}{T(s)} V^s_r,
		\end{eqnarray*}
		where 
		\begin{eqnarray*}
			d_{NT}^s &=& \left(\frac{\Lambda(s)^\T \Lambda(s)}{N}\right)^{1/2} \left[\frac{(F^s)^\T F^s}{T(s)} \frac{\Lambda(s)^\T e^s \hat{F}^s}{NT(s)}  + \frac{1}{NT(s)} \frac{(F^s)^\T  (e^s)^\T  \Lambda(s)(F^s)^\T \hat{F}^s}{T(s)}  \right. \\
			&&+   \frac{1}{NT(s)} \frac{(F^s)^\T  (e^s)^\T  e^s  \hat{F}^s}{T(s)}  + \frac{1}{NT(s)^2} (F^s)^\T  (\Delta X^s)^\T  \bar{X}^s \hat{F}^s\\ 
			&&+ \left. \frac{1}{NT(s)^2} (F^s)^\T  (\bar{X}^s)^\T  \Delta X^s \hat{F}^s+  \frac{1}{NT(s)^2} (F^s)^\T (\Delta X^s)^\T  \Delta X^s \hat{F}^s \right].
		\end{eqnarray*}
		
		By Assumption \ref{ass_loading}.1, $\frac{\Lambda(s)^\T \Lambda(s)}{N} = \Sigma_{\Lambda(s)}  + O_p\left(\frac{1}{\sqrt{N}}\right)$. 
		
		By Assumption \ref{ass_factor}, $\frac{(F^s)^\T F^s}{T(s)} =  \Sigma_{F|s} + O_p(h^2) + O_p\left( \frac{1}{\sqrt{Th}}\right)$.
		
		The first term in $d_{NT}^s$, $\frac{(F^s)^\T F^s}{T(s)} \frac{\Lambda(s)^\T e^s \hat{F}^s}{NT(s)} $, has 
		\begin{eqnarray*}
			\norm{\frac{\Lambda(s)^\T e^s \hat{F}^s}{NT(s)}}^2 &\leq& \frac{1}{N} \left( \frac{1}{NT(s)} \norm{\Lambda(s)^\T e^s}^2 \right) \left( \frac{1}{T(s)} \norm{\hat{F}^s}^2 \right) \\
			&=& \frac{1}{N} \frac{T}{T(s)} \left( \frac{1}{T} \sum_{t = 1}^T \norm{\frac{1}{\sqrt{N}} \sum_{i = 1}^N\Lambda_i(s)^\T e^s_{it}}^2 \right) \left( \frac{1}{T(s)}\sum_{t = 1}^T   \norm{\hat{F}^s_t}^2\right) \\
			&=& \frac{1}{N} O_p(1) O_p(1) O_p(1) = O_p\left(\frac{1}{N}\right)
		\end{eqnarray*}
		
		by Assumption \ref{ass_mom}.3. Therefore, $\frac{(F^s)^\T F^s}{T(s)} \frac{\Lambda(s)^\T e^s \hat{F}^s}{NT(s)}  =  O_p\left(\frac{1}{\sqrt{N}}\right)$.
		
		The second term in $d_{NT}^s$, $\frac{1}{NT(s)} \frac{(F^s)^\T  (e^s)^\T  \Lambda(s)(F^s)^\T \hat{F}^s}{T(s)}$, has
		\begin{eqnarray*}
			\frac{1}{NT(s)} (F^s)^\T  (e^s)^\T  \Lambda(s) &=& \frac{1}{NT(s)} \sum_{i = 1}^N \sum_{t = 1}^T K_s(S_t) F^s_t \Lambda_i(s)^\T  e_{it} = O_p\left(\frac{1}{\sqrt{NTh}}\right) 
		\end{eqnarray*}
		by Assumption \ref{ass_mom}.2 and Lemma \ref{lemma:simplify-assumption6} and 
		\begin{eqnarray*}
			\frac{(F^s)^\T \hat{F}^s}{T(s)} &=& \frac{(F^s)^\T (F^s H^s + \hat{F}^s - F^s H^s)}{T(s)} = \frac{(F^s)^\T F^s H^s}{T(s)} + \frac{(F^s)^\T (\hat{F}^s - F^s H)}{T(s)} = O_p(1)
		\end{eqnarray*}
		by Assumption \ref{ass_factor}, $H^s = O_p(1)$ and \\  $\norm{\frac{(F^s)^\T (\hat{F}^s - F^s H)}{T(s)}}^2 = \left(\frac{1}{T(s)} \sum_{t = 1}^T \norm{F^s_t}^2\right) \left(\frac{1}{T(s)} \sum_{t = 1}^T \norm{\hat{F}^s_t -(H^s)^\T  F^s_t}^2 \right) = O_p\left(\frac{1}{\delta_{NT, h}^2}\right) $.
		
		Therefore, $\frac{1}{NT(s)} \frac{(F^s)^\T  (e^s)^\T  \Lambda(s)(F^s)^\T \hat{F}^s}{T(s)} =  O_p\left(\frac{1}{\sqrt{N}}\right)$.

		The third term in $d_{NT(s)}^s$, $\frac{1}{NT(s)} \frac{(F^s)^\T  (e^s)^\T  e^s  \hat{F}^s}{T(s)}$, has
		\begin{eqnarray*}
			\frac{1}{N^2 T(s)^4} \norm{(F^s)^\T  (e^s)^\T  e^s  \hat{F}^s}^2 &\leq& \left(\frac{1}{T(s)} \norm{F^s}^2 \right) \left( \frac{1}{N^2T(s)^3} \norm{(F^s)^\T  (e^s)^\T  e^s}^2 \right),
		\end{eqnarray*}
		where
		\begin{eqnarray*}
			&& \frac{1}{N^2T(s)^3} \norm{(F^s)^\T  (e^s)^\T  e^s}^2 \\ 
			&=& \frac{1}{N^2T(s)^3} tr ((F^s)^\T  (e^s)^\T  e^s (e^s)^\T  e^s F^s) \\
			&=& \frac{1}{N^2T(s)^3} \sum_{t = 1}^T  \sum_{l = 1}^T  \sum_{u = 1}^T  \sum_{i = 1}^N  \sum_{m = 1}^N K_s(S_t) K_s(S_l) K_s(S_u) tr(F_t e_{it} e_{il} e_{ml} e_{mu} F_u^\T ) \\
			&=& \frac{1}{N^2T(s)^3} \sum_{t = 1}^T  \sum_{l = 1}^T  \sum_{u = 1}^T  K_s(S_t) K_s(S_l) K_s(S_u) tr(F_t e_t^\T  e_l e_l^\T  e_u F_u^\T ) \\
			&=& \frac{1}{N^2T(s)^3} \sum_{t = 1}^T  \sum_{l = 1}^T  \sum_{u = 1}^T  K_s(S_t) K_s(S_l) K_s(S_u)  e_t^\T  e_l e_l^\T  e_u F_u^\T  F_t.
		\end{eqnarray*}
		Since
		\begin{eqnarray*}
			\frac{1}{N^2} \+E\left[ e_t^\T  e_l e_l^\T  e_u F_u^\T  F_t \right]
			&=& \gamma_N(t,l) \gamma_N(l,u) \+E(F_u^\T  F_t) + \gamma_N(t,l) \+E[\zeta_{lu} F_u^\T  F_t]  + \gamma_N(l,u) \+E[\zeta_{tl} F_u^\T  F_t] + \+E[\zeta_{tl} \zeta_{lu} F_u^\T  F_t], 
		\end{eqnarray*}
		we have
		\begin{eqnarray*}
			&& \frac{1}{T(s)^3} \sum_{t = 1}^T  \sum_{l = 1}^T  \sum_{u = 1}^T K_s(S_t) K_s(S_l) K_s(S_u) \gamma_N(t,l) \gamma_N(l,u) \+E[F_u^\T  F_t]  \\
			&\leq& \frac{1}{T(s)^2 h^2} \left( \max_{t} \+E\norm{F_t}^2 \right) \left( \max_t h K_s(S_t) \right) \left( \max_k h K_s(S_u) \right) \left(\frac{1}{T(s)} \sum_{l = 1}^T K_s(S_l)\right) \max_l \left( \sum_{l = 1}^T \gamma_N(t,l)\right)^2 \\
			&=& O_p\left( \frac{1}{T^2h^2} \right)
		\end{eqnarray*}
		and
		\begin{eqnarray*}
			&&  \frac{1}{T(s)^3} \sum_{t = 1}^T  \sum_{l = 1}^T  \sum_{u = 1}^T K_s(S_t) K_s(S_l) K_s(S_u) \gamma_N(t,l) \+E[\zeta_{lu} F_u^\T  F_t] \\
			&\leq& \frac{1}{T(s)h\sqrt{N}}  \left(\max_{t,l} \left(\sum_{l = 1}^T |\gamma_N(t,l)| \right) \left(\max_l h K_s(S_l) \right)\left( \max_{l,u} \+E\left[ \sqrt{N} \zeta_{lu} \right]^2 \right)^{1/2} \left(\max_t \+E\left[\norm{F_t}^4 |S_t = s \right] \right)^{1/2} \right) \\
			&=& O_p\left(\frac{1}{Th\sqrt{N}}\right)
		\end{eqnarray*}
		by Assumption \ref{ass_err}.2, \ref{ass_err}.5, and \ref{ass_factor}.
		
		Term $\gamma_N(l,u) \+E[\zeta_{jl} F_u^\T  F_t]$ is similar to $\gamma_N(t,l) \+E[\zeta_{lu} F_u^\T  F_t]$.
		
		\begin{eqnarray*}
			&& \frac{1}{T(s)^3} \sum_{t = 1}^T  \sum_{l = 1}^T  \sum_{u = 1}^T K_s(S_t) K_s(S_l) K_s(S_u) \+E[\zeta_{jl} \zeta_{lu} F_u^\T  F_t] \\ 
			&\leq& \frac{1}{N} \left(\max_t \+E\left[\norm{F_t}^4 |S_t = s \right] \right)^{1/2}  \left( \max_{j,k}\left( \+E\left[ \sqrt{N} \zeta_{jk}\right]^4 \right)^{1/2} \right) \left( \frac{1}{T(s)} \sum_{l = 1}^T K_s(S_l)\right)^3 \\
			&\leq& O_p\left(\frac{1}{N}\right)
		\end{eqnarray*}
		by Assumption \ref{ass_factor} and \ref{ass_err}.2.

		Therefore, $\frac{1}{NT(s)} \frac{(F^s)^\T  (e^s)^\T  e^s  \hat{F}^s}{T(s)} = O_p(\frac{1}{Th}) + O_p(\frac{1}{\sqrt{N}}).$

		The fourth term in $d_{NT}^s$, $\frac{1}{NT(s)^2} (F^s)^\T  (\Delta X^s)^\T  \bar{X}^s \hat{F}^s$, has
		\begin{eqnarray*}
			&& \norm{\frac{1}{NT(s)^2} (F^s)^\T  (\Delta X^s)^\T  \bar{X}^s \hat{F}^s}^2 \\ 
			&\leq& \left( \frac{1}{T(s)}\norm{F^s}^2\right) \left( \frac{1}{N^2T(s)^2}\norm{(\Delta X^s)^\T  \bar{X}^s}^2\right) \left( \frac{1}{T(s)} \norm{ \hat{F}^s}^2\right) \\
			&=&  \left(\frac{1}{T(s)} \sum_{t = 1}^T \norm{F^s_t}^2\right) \left( \frac{1}{T(s)^2} \sum_{t = 1}^T  \sum_{u = 1}^T \norm{\frac{(\Delta X^s_u)^\T \bar{X}_t^s}{N}}^2 \right)  \left(\frac{1}{T(s)}  \sum_{t = 1}^T \norm{ \hat{F}^s_t}^2 \right) \\
			&=& O_p(1)  O_p(h^2) O_p(1) = O_p(h^2) 
		\end{eqnarray*}
		by Assumption \ref{ass_factor}, Lemma \ref{lemma1}.\ref{deltaX} and Cauchy-Schwarz inequality. Thus, $\frac{1}{NT(s)^2} (F^s)^\T  (\Delta X^s)^\T  \bar{X}^s \hat{F}^s = O_p(h).$
		
		The fifth term in $d_{NT}^s$ is similar to the fourth term.
		
		The sixth term in $d_{NT}^s$, $\frac{1}{NT(s)^2} (F^s)^\T (\Delta X^s)^\T  \Delta X^s \hat{F}^s$, has
		\begin{eqnarray*}
			&& \norm{\frac{1}{NT(s)^2} (F^s)^\T (\Delta X^s)^\T  \Delta X^s \hat{F}^s}^2 \\ 
			&\leq& \left(\frac{1}{T(s)}\norm{F^s}^2\right) \left( \frac{1}{N^2T(s)^2}\norm{(\Delta X^s)^\T  \Delta X^s}^2\right) \left(\frac{1}{T(s)} \norm{ \hat{F}^s}^2\right) \\
			&=& \left( \frac{1}{T(s)} \sum_{t = 1}^T \norm{F^s_t}^2 \right) \left( \frac{1}{T(s)^2} \sum_{t = 1}^T  \sum_{u = 1}^T \norm{\frac{(\Delta X^s_u)^\T \Delta X^s_t}{N}}^2 \right)  \left( \frac{1}{T(s)}  \sum_{t = 1}^T \norm{ \hat{F}^s_t}^2 \right) \\
			&=& O_p(1) O_p(h^4) O_p(1) = O_p(h^4).
		\end{eqnarray*}
		
		Therefore, $\frac{1}{T^2N} (F^s)^\T (\Delta X^s)^\T  \Delta X^s \hat{F}^s = O_p(h^2)$.
		
		With the convergence rate of the first term to the seventh term in $d_{NT}^s$ and $
		h \rightarrow 0$, 
		$$d_{NT}^s = O_p \left( \frac{1}{\delta_{NT, h}}\right).$$

		Let 
		$B_{NT}^s = \left(\frac{\Lambda(s)^\T \Lambda(s)}{N}\right)^{1/2} \left(\frac{(F^s)^\T F^s}{T(s)}\right) \left(\frac{\Lambda(s)^\T \Lambda(s)}{N}\right)^{1/2}$ and 
		$R_{NT}^s = \left(\frac{\Lambda(s)^\T \Lambda(s)}{N}\right)^{1/2} \left( \frac{(F^s)^\T  \hat{F}^s}{T(s)} \right)$, (we need $\Sigma_{\Lambda(s)}\Sigma_{F|s}$ to be positive definite and eigenvalues are distinct so that $R_{NT}^s$ is invertible and its eigenvalues are bounded away from 0.) we have
		\begin{gather*}
		[B_{NT}^s + d_{NT}^s (R_{NT}^s)^{-1}]R_{NT}^s = R_{NT}^s V_r^s.
		\end{gather*}
		Let $\Upsilon^s_{NT} = R_{NT}^s ((V_r^s)^*)^{-1/2}$, so that each column of $\Upsilon^s_{NT}$ has unit length, where $(V_r^s)^*$ is a diagonal matrix consisting of the diagonal element of $(R_{NT}^s)^\T R_{NT}^s = \frac{(\hat{F}^s)^\T F^s}{T(s)} \frac{\Lambda(s)^\T \Lambda(s)}{N}  \frac{(F^s)^\T  \hat{F}^s}{T(s)} = V_r^s$. We have
		\begin{gather*}
		[B_{NT}^s + d_{NT}^s (R_{NT}^s)^{-1}]\Upsilon^s_{NT} = \Upsilon^s_{NT} V_r^s
		\end{gather*}
		Note that $B_{NT}^s + d_{NT}^s (R_{NT}^s)^{-1} \xrightarrow{p} B =  \Sigma_{\Lambda(s)}^{1/2}\Sigma_{F|s}\Sigma_{\Lambda(s)}^{1/2}$ because $d_{NT}^s = O_p\left(\frac{1}{\delta_{NT, h}}\right) = o_p(1)$. By Assumption \ref{ass_eigen}, the eigenvalues of $B$ are distinct. By the continuity of eigenvalues, $B_{NT}^s$ have distinct eigenvalues for large $N$, $T$ and small $h$.

		By the perturbation theory for eigenvalues of Hermitian matrices (e.g., \cite{stewart1990matrix}), we have 
		
		\begin{eqnarray*}
			V_r^s = V^s + O_p\left(\frac{1}{\delta_{NT, h}}\right).
		\end{eqnarray*}
		Thus, the eigenvector matrix, $\Upsilon^s_{NT}$, is uniquely determined. By the eigenvector perturbation theory (\cite{franklin2012matrix}),  there exist a unique eigenvector matrix $\Upsilon^s$ such that $\norm{\Upsilon^s_{NT} - \Upsilon^s}_F = o_p(1)$. We have
		\begin{eqnarray*}
			\frac{(F^s)^\T \hat{F}^s}{T} &=& \left(\frac{\Lambda(s)^\T \Lambda(s)}{N}\right)^{-1/2} \Upsilon^s_{NT} ((V_r^s)^*)^{1/2} \\
			&=& \left(\frac{\Lambda(s)^\T \Lambda(s)}{N}\right)^{-1/2} \Upsilon^s (V^s)^{1/2} + O_p\left(\frac{1}{\delta_{NT, h}}\right).
		\end{eqnarray*}

	\end{proof}

	\begin{lemma}\label{lemma3}
		Under Assumption \ref{Ass:Ident}-\ref{ass_mom},$Th \rightarrow \infty$, $\delta_{NT,h}h \rightarrow 0$, $\sqrt{Nh}/(Th) \rightarrow 0$,
		\begin{enumerate}
			\item $\sqrt{Nh} \left( \frac{1}{T(s)} \sum_{u=1}^T \hat{F}^s_u \gamma_N^s(u,t) + \frac{1}{T(s)} \sum_{u=1}^T \hat{F}^s_u \zeta^s_{ut} \right) = o_p(1)$
			\item $\sqrt{Nh} \left( \frac{1}{T(s)} \sum_{u=1}^T \hat{F}^s_u \epsilon^s_{ut} \right) = o_p(1)$ 
			\item $\sqrt{N} \left( \frac{1}{T(s)} \sum_{u=1}^T  \hat{F}^s_u \frac{(\Delta X^s_u)^\T \bar{X}_t}{N}  \right) = o_p(1)$
		\end{enumerate}
	\end{lemma}
	
	\begin{proof}[Proof of Lemma \ref{lemma3}.1]
		\begin{eqnarray*}
			&& \frac{1}{T(s)} \sum_{u=1}^T \hat{F}^s_u \gamma_N^s(u,t) + \frac{1}{T(s)} \sum_{u=1}^T \hat{F}^s_u \zeta^s_{ut}  \\
			&=&  \frac{1}{T(s)}\sum_{u = 1}^T (\hat{F}^s_u - (H^s)^\T F^s_u) (e^s_u)^\T e^s_t/N + \frac{1}{T(s)} (H^s)^\T  \sum_{u = 1}^T F^s_u (e^s_u)^\T e^s_t/N
		\end{eqnarray*}
		
		The norm of the first term, $\frac{1}{T(s)}\sum_{u = 1}^T (\hat{F}^s_u - (H^s)^\T F^s_u) (e^s_u)^\T e^s_t/N$, has
		\begin{eqnarray*}
			&& \norm{\frac{1}{T(s)}\sum_{u = 1}^T (\hat{F}^s_u - (H^s)^\T F^s_u) (e^s_u)^\T e^s_t/N} \\
			&\leq& \left[\frac{1}{T(s)}\sum_{u = 1}^T \norm{\hat{F}^s_u - (H^s)^\T F^s_u}^2 \right]^{1/2} \left[\frac{1}{T(s)} \sum_{u = 1}^T [(e^s_u)^\T e^s_t/N]^2 \right]^{1/2}
		\end{eqnarray*}
		
		and term $\frac{1}{T(s)} \sum_{u = 1}^T [(e^s_u)^\T e^s_t/N]^2$ has
		\begin{eqnarray*}
			&&\frac{1}{T(s)} \sum_{u = 1}^T K_s(S_u) K_s(S_t) \+E [e^\T_ue_t/N]^2 \\ 
			&=& \frac{1}{T(s)} \sum_{u = 1}^T K_s(S_u) K_s(S_t) [\+E e^\T_ue_t/N]^2 + \frac{1}{T(s)} \sum_{u = 1}^T K_s(S_u) K_s(S_t) \+E(\zeta_{ut}^2) \\ &=& \frac{1}{T(s)} K_s(S_t) \left(\max_k K_s(S_u)\right) \sum_{u = 1}^T \gamma_N(u,t)^2 + \frac{1}{N} K_s(S_t) \max_k \+E(N \zeta_{ut}^2) \frac{1}{T(s)} \sum_{u = 1}^T K_s(S_u) \\
			&=& O_p\left(\frac{1}{Th^2}\right) + O_p\left(\frac{1}{Nh}\right)
		\end{eqnarray*}
		by Lemma \ref{lemma1}.1 and Assumption \ref{ass_err}.2.
		
		Therefore, 
		\begin{eqnarray*}
			\norm{\frac{1}{T(s)}\sum_{u = 1}^T (\hat{F}^s_u - (H^s)^\T F^s_u) (e^s_u)^\T e^s_t/N} = O_p\left( \frac{1}{\delta_{NT,h}}\right) O_p\left(\frac{1}{\sqrt{T}h} + \frac{1}{\sqrt{Nh}}\right).
		\end{eqnarray*}

		The second term, $\frac{1}{T(s)} \sum_{u = 1}^T F^s_u (e^s_u)^\T e^s_t/N$, has
		\begin{eqnarray*}
			\frac{1}{T(s)} \sum_{u = 1}^T F^s_u (e^s_u)^\T e^s_t/N &=& K_s(S_t)^{1/2} \frac{1}{NT(s)} \sum_{u = 1}^T K_s(S_u) \+E(F_u e_u^\T e_t) \\  
			&&+ K_s(S_t)^{1/2} \frac{1}{NT(s)} \sum_{u = 1}^T K_s(S_u) [F_u e_u^\T e_t - \+E(F_u e_u^\T e_t)]
		\end{eqnarray*}
		where 
		\begin{eqnarray*}
			\frac{1}{NT(s)} \norm{\sum_{u = 1}^T K_s(S_u) \+E(F_u e_u^\T e_t)} &\leq& \frac{1}{T(s)}  \sum_{u = 1}^T K_s(S_u) \norm{ \+E(F_u e_u^\T e_t/N)} \\ 
			&\leq& \frac{1}{T(s)} \left( \max_kK_s(S_u) \right) \sum_{u = 1}^T \norm{ \+E(F_u e_u^\T e_t/N)}  
			= O_p \left(\frac{1}{Th} \right)
		\end{eqnarray*}
		by Assumption \ref{ass_err}.6, and
		\begin{eqnarray*}
			K_s(S_t)^{1/2} \frac{1}{NT(s)} \sum_{u = 1}^T K_s(S_u) [F_u e_u^\T e_t - \+E(F_u e_u^\T e_t)] = O_p\left(\frac{1}{\sqrt{h}}\right)  O_p \left(\frac{1}{\sqrt{NTh}}\right) = O_p\left(\frac{1}{\sqrt{NT}h}\right)
		\end{eqnarray*}
		by Assumption \ref{ass_mom}.1 and Lemma \ref{lemma:simplify-assumption6}. Therefore,
		\begin{eqnarray*}
			\frac{1}{T(s)} \sum_{u = 1}^T F^s_u (e^s_u)^\T e^s_t/N = O_p\left(\frac{1}{Th}\right) + O_p\left(\frac{1}{\sqrt{NT}h}\right)
		\end{eqnarray*}
		and
		\begin{eqnarray*}
			&&\frac{1}{T(s)} \sum_{u=1}^T \hat{F}^s_u \gamma_N^s(u,t) + \frac{1}{T(s)} \sum_{u=1}^T \hat{F}^s_u \zeta^s_{ut} \\
			&=& \left(O_p \left(\frac{1}{\delta_{NT,h}} \right) O_p\left(\frac{1}{\sqrt{T}h} + \frac{1}{\sqrt{Nh}}\right)\right) +  \left(O_p \left(\frac{1}{Th}\right) + O_p\left( \frac{1}{\sqrt{NT}h}\right) \right)
		\end{eqnarray*}

		When $\sqrt{Nh}/(Th) \rightarrow 0$, $\frac{1}{T} \sum_{u=1}^T \hat{F}^s_u \gamma_N^s(u,t) + \frac{1}{T} \sum_{u=1}^T \hat{F}^s_u \zeta^s_{ut} = O_p\left(\frac{1}{\sqrt{Nh} \delta_{NT,h}}\right)$, so 
		
		$\sqrt{Nh} \left(\frac{1}{T} \sum_{u=1}^T \hat{F}^s_u \gamma_N^s(u,t) + \frac{1}{T} \sum_{u=1}^T \hat{F}^s_u \zeta^s_{ut} \right) = o_p(1)$
		
	\end{proof}

	\begin{proof}[Proof of Lemma \ref{lemma3}.2]
		\begin{eqnarray*}
			\frac{1}{T(s)} \sum_{u=1}^T \hat{F}^s_u \epsilon^s_{ut} = \frac{1}{T(s)}\sum_{u = 1}^T (\hat{F}^s_u - (H^s)^\T  F^s_u) \epsilon^s_{ut} + (H^s)^\T  \frac{1}{T(s)}\sum_{u = 1}^T  F^s_u \epsilon^s_{ut}.
		\end{eqnarray*}
		
		The first term has
		
		\begin{eqnarray*}
			&& \norm{\frac{1}{T(s)}\sum_{u = 1}^T (\hat{F}^s_u - (H^s)^\T  F^s_u) \epsilon^s_{ut}} \\ 
			&=& \norm{\frac{1}{NT(s)}\left(\sum_{u = 1}^T (\hat{F}^s_u - (H^s)^\T  F^s_u) (e^s_u)^\T \Lambda(s)\right) F^s_t} \\
			&\leq&  \left( \frac{1}{T(s)} \sum_{u = 1}^T \norm{\hat{F}^s_u - (H^s)^\T  F^s_u}^2  \right)^{1/2} \left( \frac{1}{T} \sum_{u = 1}^T \norm{(e^s_u)^\T \Lambda(s) F^s_t/N}^2\right)^{1/2}  \\
			&=& O_p\left(\frac{1}{\delta_{NT,h}} \right) O\left(\frac{1}{\sqrt{Nh}}\right)
		\end{eqnarray*}
		
		by
		\begin{eqnarray*}
			&& \frac{1}{T(s)} \sum_{u = 1}^T K_s(S_u) K_s(S_t) \norm{e_u^\T \Lambda(s) F_t/N}^2 \\ 
			&\leq& \frac{1}{N}  K_s(S_t) \left(\max_{s,k} \norm{N^{-1/2} e_u^\T \Lambda(s)}^2\right) \left( \norm{F_t}^2\right) \left( \frac{1}{T(s)} \sum_{u = 1}^T K_s(S_u)\right) \\
			&=&O\left(\frac{1}{N}\right) O\left( \frac{1}{h}\right) O_p(1) O_p(1) O_p(1) = O_p\left( \frac{1}{Nh}\right)
		\end{eqnarray*}
		by Assumption \ref{ass_mom}.3 and \ref{ass_factor}. Also,
		\begin{eqnarray*}
			\norm{\frac{1}{T(s)}\sum_{u = 1}^T  F^s_u \epsilon^s_{ut}}^2  &=& \norm{\frac{1}{T(s)} \sum_{u = 1}^T F^s_u (e^s_u)^\T \Lambda(s) F^s_t /N }^2  \\ 
			&\leq& \frac{1}{NTh} \left( \norm{ \frac{\sqrt{Th}}{\sqrt{N}T(s)} \sum_{u = 1}^T K_s(S_u) F_u e_u^\T \Lambda(s) }^2 \right) \left( K_s(S_t) \norm{F_t}^2\right) \\
			&=& O\left( \frac{1}{NTh}\right) O_p(1) O_p\left(\frac{1}{h}\right) = O_p\left(\frac{1}{NTh^2}\right)
		\end{eqnarray*}
		by Assumption \ref{ass_factor} and \ref{ass_mom}.2 and Lemma \ref{lemma:simplify-assumption6}. Therefore,
		\begin{eqnarray*}
			\frac{1}{T} \sum_{u=1}^T \hat{F}^s_u \epsilon^s_{ut} = O\left(\frac{1}{\delta_{NT,h}\sqrt{Nh}}\right) + O_p\left(\frac{1}{\sqrt{NT}h}\right)
		\end{eqnarray*}
		and then $\sqrt{Nh} \frac{1}{T} \sum_{u=1}^T \hat{F}^s_u \epsilon^s_{ut} = o_p(1)$.
	\end{proof}
	
	\begin{proof}[Proof of Lemma \ref{lemma3}.3]
		\begin{eqnarray*}
			&& \frac{1}{T(s)} \sum_{u=1}^T  \hat{F}^s_u \frac{(\Delta X^s_u)^\T \bar{X}_t}{N} \\
			&=&  \frac{1}{T(s)}\sum_{u = 1}^T (\hat{F}^s_u - (H^s)^\T F^s_u)  \frac{(\Delta X^s_u)^\T \bar{X}_t}{N}+ \frac{1}{T(s)} (H^s)^\T  \sum_{u = 1}^T F^s_u  \frac{(\Delta X^s_u)^\T \bar{X}_t}{N}
		\end{eqnarray*}
		
		The first term $\frac{1}{T(s)}\sum_{u = 1}^T (\hat{F}^s_u - (H^s)^\T F^s_u)  \frac{(\Delta X^s_u)^\T \bar{X}_t}{N}$ has
		
		\begin{eqnarray*}
			&& \norm{\frac{1}{T(s)}\sum_{u = 1}^T (\hat{F}^s_u - (H^s)^\T F^s_u)  \frac{(\Delta X^s_u)^\T \bar{X}_t}{N}}^2  \\
			&\leq& \left( \frac{1}{T(s)}\sum_{u = 1}^T \norm{\hat{F}^s_u - (H^s)^\T F^s_u}^2 \right) \left( \frac{1}{T(s)}\sum_{u = 1}^T \norm{\frac{(\Delta X^s_u)^\T \bar{X}_t}{N}}^2 \right) \\ 
			&=& O_p(\delta_{NT,h}^{-2}) O_p(h^2)
		\end{eqnarray*}
		by Lemma \ref{lemma1}.2.

		The second term $\frac{1}{T(s)} (H^s)^\T  \sum_{u = 1}^T F^s_u  \frac{(\Delta X^s_u)^\T \bar{X}_t}{N}$ has
		\begin{eqnarray*}
			\norm{\frac{1}{T(s)} \sum_{u = 1}^T F^s_u  \frac{(\Delta X^s_u)^\T \bar{X}_t}{N}}^2 
			&\leq&  \left( \frac{1}{T(s)} \sum_{u = 1}^T K_s(S_u) \norm{F_u}^2 \right) \left( \frac{1}{N^2T(s)} \sum_{u = 1}^T  \norm{(\Delta X^s_u)^\T \bar{X}_t}^2 \right) \\
			&=& O_p(1) O_p(h^2) =  O_p(h^2)
		\end{eqnarray*}
		by assumption \ref{ass_factor}. Therefore, $\frac{1}{T(s)} \sum_{u=1}^T  \hat{F}^s_u \frac{(\Delta X^s_u)^\T \bar{X}_t}{N} = O_p(h)$.
		
		Thus, $\sqrt{N} \frac{1}{T(s)} \sum_{u=1}^T  \hat{F}^s_u \frac{(\Delta X^s_u)^\T \bar{X}_t}{N} = \sqrt{N} O_p(h) = O_p \left( \sqrt{Nh^2} \right) = o_p(1)$, when $Nh^2 \rightarrow 0$.
		
	\end{proof}

	\begin{proof}[Proof of Theorem 2]
		Since $\left(\frac{1}{NT(s)} (X^s)^\T X^s \right) \hat{F}^s = \hat{F}^s V^s_r$, $V^s_r$ is full rank by Lemma \ref{lemma2}, and $H^s = \frac{\Lambda(s)^\T  \Lambda(s)}{N} \frac{(F^s)^\T  \hat{F}^s}{T(s)}  (V^s_r)^{-1}$,  we have
		\begin{eqnarray*}
			&& \hat{F}^s_t - (H^s)^\T  F^s_t \\
			&=& (V^s_r)^{-1} \left(  \frac{1}{NT(s)}(\hat{F}^s)^\T  (e^s)^\T  \Lambda(s) F^s_t +  \frac{1}{NT(s)}(\hat{F}^s)^\T  F^s \Lambda(s)^\T  e^s_t +  \frac{1}{NT(s)} (\hat{F}^s)^\T  (e^s)^\T  e^s_t \right. \\
			&&+ \left. \frac{1}{NT(s)} (\hat{F}^s)^\T  (\Delta X^s)^\T \bar{X}_t^s + \frac{1}{NT(s)} (\hat{F}^s)^\T  (\bar{X}^s)^\T  \Delta X^s_t + \frac{1}{NT(s)} (\hat{F}^s)^\T (\Delta X^s)^\T \Delta X^s_t  \right) \\
			&=& (V^s_r)^{-1} \left( \frac{1}{T(s)} \sum_{u=1}^T \hat{F}^s_u \gamma_N^s(u,t) + \frac{1}{T(s)} \sum_{u=1}^T \hat{F}^s_u \zeta^s_{ut} + \frac{1}{T(s)} \sum_{u=1}^T \hat{F}^s_u \eta^s_{ut} +  \frac{1}{T(s)} \sum_{u=1}^T \hat{F}^s_u \epsilon^s_{ut}  \right. \\ 
			&&+ \frac{1}{T(s)} \sum_{u=1}^T \hat{F}^s_u \frac{(\Delta X^s_u)^\T \bar{X}_t^s}{N} + \left. \frac{1}{T(s)} \sum_{u=1}^T \hat{F}^s_u \frac{( \bar{X}^s_u)^\T \Delta X^s_t}{N}  +  \frac{1}{T(s)} \sum_{u=1}^T \hat{F}^s_u \frac{(\Delta X^s_u)^\T \Delta X^s_t}{N} \right) \\
			&=& (V^s_r)^{-1} \left( A_1 + A_2 + A_3 + A_4 + A_5 + A_6 + A_7 \right)
		\end{eqnarray*}
		
		From Lemma \ref{lemma3}, $\sqrt{Nh} (A_1 + A_2) = o_p(1)$, $\sqrt{Nh} A_4 = o_p(1)$, $\sqrt{Nh} A_5 = o_p(1)$.
		
		Since $(V^s_r)^{-1} = O_p(1)$,
		
		\begin{eqnarray*}
			\sqrt{N} A_3 = K_s^{1/2}(S_t) \left( \frac{1}{T(s)}\sum_{u = 1}^T K_s(S_u) \hat{F}_u (F_u)^\T  \right) \left( \frac{1}{\sqrt{N}}\sum_{i = 1}^{N} \Lambda_i(s) e_{it} \right)
		\end{eqnarray*}
		
		From Lemma \ref{lemma2}.2, we have $\frac{1}{T(s)}\sum_{u = 1}^T (K_s(S_u) \hat{F}_u (F_u)^\T ) \stackrel{P}{\rightarrow} Q^s$. From Assumption \ref{ass_mom}.3, we have $\frac{1}{\sqrt{N}}\sum_{i = 1}^{N} \Lambda_i(s) e_{it} \xrightarrow{d} N(0, \Gamma_t^s)$. Therefore, from Slusky theorem, 
		\begin{eqnarray*}
			\frac{\sqrt{N} A_3}{K_s^{1/2}(S_t)} \xrightarrow{d} N(0, Q^s \Gamma_t^s (Q^s)^\T )
		\end{eqnarray*}
		
		For term $A_6$,
		\begin{eqnarray*}
			\frac{\sqrt{N} A_6}{K_s^{1/2}(S_t)} = \frac{\sqrt{N}}{T(s)} \sum_{u=1}^T \hat{F}^s_u \frac{( \bar{X}^s_u)^\T \Delta X_t}{N} 
			=  \left(\frac{1}{T(s)}  \sum_{u=1}^T \hat{F}^s_u (F^s_u)^\T  \right) \frac{\Lambda(s)^\T  (\Lambda(S_t)-\Lambda(s))}{\sqrt{N}} F_t = o_p(1)
		\end{eqnarray*}
		by Assumption \ref{ass_factor}, \ref{ass_loading} and when $\frac{1}{\sqrt{N}} \sum_{i = 1}^N \norm{\Lambda_i(S_t)-\Lambda_i(s)}  = o_p(1)$.
		
		For term $A_7$,
		\begin{eqnarray*}
			\frac{\sqrt{N} A_7}{K_s^{1/2}(S_t)} \frac{1}{T(s)} \sum_{u=1}^T \hat{F}^s_u \frac{(\Delta X^s_u)^\T \Delta X^s_t}{N} = \frac{\sqrt{N}}{T(s)} \sum_{u=1}^T \hat{F}^s_u (F^s_u)^\T  \frac{ (\Lambda(S_u)-\Lambda(s))^\T  (\Lambda(S_t)-\Lambda(s))}{\sqrt{N}} F_t = o_p(1)
		\end{eqnarray*}
		by Assumption \ref{ass_factor}, \ref{ass_loading}, $\norm{\Lambda_i(S_u)-\Lambda_i(s)} \leq 2 \bar{\Lambda}$ and when $\frac{1}{\sqrt{N}} \sum_{i = 1}^N \norm{\Lambda_i(S_t)-\Lambda_i(s)}  = o_p(1)$. 
		
		From Lemma \ref{lemma2}, $(V^s_r)^{-1} \xrightarrow{d} (V^s)^{-1}$.  If $\sqrt{Nh}/(Th) \rightarrow 0$, $\sqrt{Nh} \rightarrow \infty$ and $Nh^2 \rightarrow 0$ and for the $j$ where $\frac{1}{\sqrt{N}} \sum_{i = 1}^N \norm{\Lambda_i(S_t)-\Lambda_i(s)}  = o_p(1)$. From Slusky theorem, 
		
		\begin{eqnarray*}
			\sqrt{N} \left( \frac{\hat{F}^s_t}{K^{1/2}_s(S_t)} - (H^s)^\T  F_t \right) \xrightarrow{d} N(0, (V^s)^{-1} Q^s \Gamma_t^s (Q^s)^\T  (V^s)^{-1}) 
		\end{eqnarray*}
	\end{proof}

	\begin{lemma}\label{lemma4}
		Under Assumption \ref{Ass:Ident}-\ref{ass_eigen} and $Th \rightarrow \infty$, $\delta_{NT,h}h \rightarrow 0$,
		\begin{enumerate}
			\item $H^s = (Q^s)^{-1} + O_p\left(\frac{1}{\delta_{NT,h}}\right)$
			\item $H^s (H^s)^\T  = \Sigma_{F|s}^{-1} + O_p\left(\frac{1}{\delta_{NT,h}}\right)$
			\item $\frac{1}{T(s)} (\hat{F}^s - F^s H^s)^\T \underline{e}_i^s = O_p \left( \max \left( \frac{1}{\delta_{NT,h}^2}, h \right) \right)$, where $\underline{e}_i^s$ is the $i$-th row in $e^s$.
			\item $\frac{1}{T(s)} (\hat{F}^s - F^s H^s)^\T F^s = O_p \left( \max \left( \frac{1}{\delta_{NT,h}^2}, h \right) \right)$
			\item $\frac{1}{T(s)} (\hat{F}^s (H^s)^{-1}- F^s)^\T \hat{F}^s = O_p \left( \max \left( \frac{1}{\delta_{NT,h}^2}, h \right) \right)$
		\end{enumerate}
	\end{lemma}

	\begin{proof}[Proof of Lemma \ref{lemma4}.1]
		From Lemma \ref{lemma2},
		\begin{eqnarray*}
			(H^s)^\T  &=& (V^s_r)^{-1} \frac{(\hat{F}^s)^\T  F^s }{T(s)}  \frac{\Lambda(s)^\T  \Lambda(s)}{N} \\
			&=& (V^s)^{-1} Q^s \Sigma_{\Lambda(s)} + O_p \left( \frac{1}{\delta_{NT,h}} \right) \\
			&=& (V^s)^{-1} (V^s)^{\frac{1}{2}} (\Upsilon^s)^\T  (\Sigma_{\Lambda(s)})^{-\frac{1}{2}} \Sigma_{\Lambda(s)} +O_p\left( \frac{1}{\delta_{NT,h}}\right)  \\
			&=& (V^s)^{-\frac{1}{2}} (\Upsilon^s)^\T  (\Sigma_{\Lambda(s)})^{\frac{1}{2}} + O_p \left( \frac{1}{\delta_{NT,h}} \right) \\
			&=& ((Q^s)^{-1})^\T  + O_p \left( \frac{1}{\delta_{NT,h}}\right)
		\end{eqnarray*}

	\end{proof}

	\begin{proof}[Proof of Lemma \ref{lemma4}.2]
		Since $H^s = (Q^s)^{-1} + O_p \left(\frac{1}{\delta_{NT,h}} \right)$ and $Q^s = (V^s)^{1/2} (\Upsilon^s)^\T  \Sigma_{\Lambda(s)} ^{-1/2}$,
		
		\begin{eqnarray*}
			H^s (H^s)^\T  &=& \Sigma_{\Lambda(s)} ^{1/2} ((\Upsilon^s)^\T )^{-1} (V^s)^{-1/2} (V^s)^{-1/2} (\Upsilon^s)^\T  \Sigma_{\Lambda(s)} ^{1/2} + O_p\left(\frac{1}{\delta_{NT,h}} \right)
		\end{eqnarray*}
		
		Also, $V^s$ are eigenvalues of $\Sigma_{\Lambda(s)}^{1/2}\Sigma_{F|s} \Sigma_{\Lambda(s)}^{1/2}$ and $\Upsilon^s$ is the corresponding eigenvector matrix such that $(\Upsilon^s)^\T \Upsilon^s = I$.
		\begin{eqnarray*}
			\Sigma_{\Lambda(s)}^{1/2}\Sigma_{F|s} \Sigma_{\Lambda(s)}^{1/2} \Upsilon^s &=& \Upsilon^s V^s \\
			(V^s)^{-1} &=& (\Upsilon^s)^{-1} \Sigma_{\Lambda(s)}^{-1/2} \Sigma_{F|s}^{-1} \Sigma_{\Lambda(s)}^{-1/2} ((\Upsilon^s)^\T )^{-1}
		\end{eqnarray*}
		
		Therefore,
		\begin{eqnarray*}
			H^s (H^s)^\T  &=& \Sigma_{\Lambda(s)} ^{1/2} ((\Upsilon^s)^\T )^{-1} (\Upsilon^s)^{-1} \Sigma_{\Lambda(s)}^{-1/2} \Sigma_{F|s}^{-1} \Sigma_{\Lambda(s)}^{-1/2} ((\Upsilon^s)^\T )^{-1} (\Upsilon^s)^{-1} \Sigma_{\Lambda(s)} ^{1/2} + O_p \left( \frac{1}{\delta_{NT,h}}\right) \\
			&=& \Sigma_{F|s}^{-1} + O_p \left( \frac{1}{\delta_{NT,h}}\right)
		\end{eqnarray*}
	\end{proof}

	\begin{proof}[Proof of Lemma \ref{lemma4}.3]
		\begin{eqnarray*}
			&&\frac{1}{T(s)} (\hat{F}^s - F^s H^s)^\T \underline{e}_i^s \\  
			&=& \frac{1}{T(s)} \sum_{t = 1}^T (\hat{F}^s_t - (H^s)^\T F^s_t ) e_{it}^s  \\
			&=& (V^s_r)^{-1} \left[ \frac{1}{T(s)^2}\sum_{t = 1}^T \sum_{u = 1}^T \hat{F}^s_u \gamma_N^s(u,t) e_{it}^s + \frac{1}{T(s)^2} \sum_{t = 1}^T \sum_{u = 1}^T \hat{F}^s_u \zeta^s_{ut} e_{it}^s  + \frac{1}{T(s)^2} \sum_{t = 1}^T \sum_{u = 1}^T \hat{F}^s_u \eta^s_{ut} e_{it}^s  \right. \\
			&&+ \frac{1}{T(s)^2}\sum_{t = 1}^T \sum_{u = 1}^T \hat{F}^s_u \epsilon^s_{ut} e_{it}^s + \frac{1}{T(s)^2} \sum_{t = 1}^T \sum_{u = 1}^T \hat{F}^s_u \frac{(\Delta X^s_u)^\T \bar{X}_t^s}{N} e_{it}^s + \frac{1}{T(s)^2} \sum_{t = 1}^T \sum_{u = 1}^T \hat{F}^s_u \frac{( \bar{X}^s_u)^\T \Delta X^s_t}{N} e_{it}^s \\ 
			&&+ \left. \frac{1}{T(s)^2} \sum_{t = 1}^T \sum_{u = 1}^T \hat{F}^s_u \frac{(\Delta X^s_u)^\T \Delta X^s_t}{N} e_{it}^s \right]\\
			&=& (V^s_r)^{-1} \left[ \text{\RNum{1} + \RNum{2} + \RNum{3} + \RNum{4} + \RNum{5} + \RNum{6} + \RNum{7}} \right]
		\end{eqnarray*}
		
		For term \RNum{1}$ = \frac{1}{T(s)^2}\sum_{t = 1}^T \sum_{u = 1}^T  \hat{F}^s_u \gamma_N^s(u,t) e_{it}^s$, we have 
		\begin{eqnarray*}
			\text{\RNum{1}} &=& \frac{1}{T(s)^2}\sum_{t = 1}^T \sum_{u = 1}^T \hat{F}^s_u \gamma_N^s(u,t) \underline{e}_i^s \\
			&=& \frac{1}{T(s)^2}\sum_{t = 1}^T \sum_{u = 1}^T  (\hat{F}^s_u - (H^s)^\T F^s_u) \gamma_N^s(u,t) e_{it}^s + \frac{1}{T(s)^2} (H^s)^\T  \sum_{t = 1}^T \sum_{u = 1}^T  F^s_u \gamma_N^s(u,t) e_{it}^s 
		\end{eqnarray*}
		
		The norm of first term in \RNum{1} has 
		\begin{eqnarray*}
			&& \norm{\frac{1}{T(s)^2}\sum_{t = 1}^T \sum_{u = 1}^T (\hat{F}^s_u - (H^s)^\T F^s_u) \gamma_N^s(u,t) e_{it}^s} \\
			&\leq& \frac{1}{\sqrt{T(s)}} \left(\frac{1}{T(s)} \sum_{u = 1}^T \norm{\hat{F}^s_u - (H^s)^\T F^s_u}^2\right)^{1/2} \left(\frac{1}{T} \sum_{u = 1}^T \left( \left(\sum_{t = 1}^T  \gamma_N^s(u,t)^2\right)  \left(\frac{1}{T(s)} \sum_{t = 1}^T (e_{it}^s)^2\right) \right)\right)^{1/2}  \\
			&=& O_p\left( \frac{1}{\sqrt{T}}\right) O_p \left( \frac{1}{\delta_{NT,h}}\right) O_p\left(\frac{1}{\sqrt{h}}\right) = O_p\left( \frac{1}{\sqrt{Th}\delta_{NT,h}}\right)
		\end{eqnarray*}
		by $\frac{1}{T(s)} \sum_{t = 1}^T \norm{\hat{F}^s_u - (H^s)^\T F^s_u}^2 = O_p\left(\frac{1}{\delta_{NT,h}^2}\right)$, $\frac{1}{T(s)} \sum_{t = 1}^T \sum_{u = 1}^T \gamma_N^s(u,t)^2 = O_p\left(\frac{1}{h}\right)$ from (\ref{gamma}) and Assumption \ref{ass_err}.1.
		
		The second term in \RNum{1} is (ignore $H^s$ since it is $O_p(1)$)
		
		\begin{eqnarray*}
			&& E\norm{\frac{1}{T(s)^2}\sum_{t = 1}^T \sum_{u = 1}^T  F^s_u \gamma_N^s(u,t) e_{it}^s} \\
			&\leq& \+E\left[ \frac{1}{T(s)^2}\sum_{t = 1}^T \sum_{u = 1}^T K_s(S_t) K_s(S_u) |\gamma_N(u,t)| \left( \+E\left[ \norm{F_u}^2 |S_t, S_u \right] \right)^{1/2}  \left( \+E\left[e_{it}^2 |S_t, S_u \right] \right)^{1/2} \right] \\
			&=& O_p \left( \frac{1}{Th}\right)
		\end{eqnarray*}
		by (\ref{gamma}) and Assumption \ref{ass_factor}.
		
		Therefore, \RNum{1} = $\frac{1}{T(s)^2}\sum_{t = 1}^T \sum_{u = 1}^T  \hat{F}^s_u \gamma_N^s(u,t) e_{it}^s = O_p\left( \frac{1}{\sqrt{Th}\delta_{NT,h}}\right)$.

		For term \RNum{2}$ = \frac{1}{T^2} \sum_{t = 1}^T \sum_{u = 1}^T \hat{F}^s_u \zeta^s_{ut} e_{it}^s$, we have
		\begin{eqnarray*}
			\text{\RNum{2}} &=& \frac{1}{T(s)^2} \sum_{t = 1}^T \sum_{u = 1}^T \hat{F}^s_u \zeta^s_{ut} e_{it}^s \\
			&=& \frac{1}{T(s)^2} \sum_{t = 1}^T \sum_{u = 1}^T (\hat{F}^s_u - (H^s)^\T F^s_u) \zeta^s_{ut} e_{it}^s + \frac{1}{T(s)^2} (H^s)^\T  \sum_{t = 1}^T \sum_{u = 1}^T F^s_u \zeta^s_{ut} e_{it}^s
		\end{eqnarray*}
		
		The norm of first term in \RNum{2} has
		\begin{eqnarray*}
			&&\norm{\frac{1}{T(s)^2} \sum_{t = 1}^T \sum_{u = 1}^T (\hat{F}^s_u - (H^s)^\T F^s_u) \zeta^s_{ut} e_{it}^s} \\
			&\leq& \left( \frac{1}{T(s)} \sum_{u = 1}^T \norm{\hat{F}^s_u - (H^s)^\T F^s_u}^2\right)^{1/2} \left(\frac{1}{T(s)} \sum_{u = 1}^T \left(\frac{1}{T(s)} \sum_{t = 1}^T \zeta^s_{ut} e_{it}^s\right)^2\right)^{1/2}  \\
			&=& \left(\frac{1}{T} \sum_{u = 1}^T \norm{\hat{F}^s_u - (H^s)^\T F^s_u}^2\right)^{1/2} \left(\frac{1}{T(s)} \sum_{u = 1}^T K_s(S_u) \left(\frac{1}{T(s)} \sum_{t = 1}^T K_s(S_t) \zeta_{ut} e_{it}\right)^2\right)^{1/2}
		\end{eqnarray*}
		Since 
		\begin{eqnarray*}
			\frac{1}{T(s)} \sum_{t = 1}^T K_s(S_t) \zeta_{ut} e_{it} = \frac{1}{\sqrt{N}} \frac{1}{T(s)}\sum_{t = 1}^T K_s(S_t) \left(\frac{1}{\sqrt{N}} \sum_{i = 1}^N e_{iu} e_{it} - \+E[e_{iu} e_{it}]\right) e_{it} = O_p\left( \frac{1}{\sqrt{N}} \right)
		\end{eqnarray*}
		by Assumption \ref{ass_err}.1 and \ref{ass_err}.5, the first term in \RNum{2}, $\frac{1}{T(s)^2} \sum_{t = 1}^T \sum_{u = 1}^T (\hat{F}^s_u - (H^s)^\T F^s_u) \zeta^s_{ut} e_{it}^s = O_p \left(\frac{1}{\sqrt{N}\delta_{NT,h}}\right)$. The second term in \RNum{2} (ignore $H^s$) has
		\begin{eqnarray*}
			\frac{1}{T(s)^2} \sum_{t = 1}^T \sum_{u = 1}^T F^s_u \zeta^s_{ut} e_{it}^s = \frac{1}{\sqrt{NTh}} \left( \frac{1}{T(s)} \sum_{t = 1}^T K_s(S_t) z^s_t e_{it} \right),
		\end{eqnarray*}
		where $z^s_t = \frac{\sqrt{Th}}{\sqrt{N}T(s)} \sum_{u = 1}^T \sum_{l = 1}^N K_s(S_u) F_u[e_{lu}e_{lt} - \+E[e_{lu}e_{lt}]]$. By Assumption \ref{ass_mom}.1 and Lemma \ref{lemma:simplify-assumption6}, $\max_{t} E\norm{z^s_t}^2 \leq M$. We have 
		\begin{eqnarray*}
			\frac{1}{\sqrt{NTh}} \left( \frac{1}{T(s)} \sum_{t = 1}^T K_s(S_t) z^s_t e_{it} \right) = O_p\left( \frac{1}{\sqrt{NTh}} \right)
		\end{eqnarray*}
		
		Thus \RNum{2}$ = \frac{1}{T(s)^2} \sum_{t = 1}^T \sum_{u = 1}^T \hat{F}^s_u \zeta^s_{ut} e_{it}^s = O_p\left( \frac{1}{\sqrt{N}\delta_{NT,h}}\right)$.
		
		For term \RNum{3}$ = \frac{1}{T(s)^2} \sum_{t = 1}^T \sum_{u = 1}^T \hat{F}^s_u \eta^s_{ut} e_{it}^s$, we have
		\begin{eqnarray*}
			\text{\RNum{3}} = \frac{1}{T(s)^2} \sum_{t = 1}^T \sum_{u = 1}^T (\hat{F}^s_u - (H^s)^\T F^s_u) \eta^s_{ut} e_{it}^s + \frac{1}{T(s)^2} (H^s)^\T  \sum_{t = 1}^T \sum_{u = 1}^T F^s_u \eta^s_{ut} e_{it}^s
		\end{eqnarray*}

		The norm of the first term in \RNum{3} has
		\begin{eqnarray*}
			&& \norm{\frac{1}{T(s)^2} \sum_{t = 1}^T \sum_{u = 1}^T (\hat{F}^s_u - (H^s)^\T F^s_u) \eta^s_{ut} e_{it}^s} \\
			&\leq& \left( \frac{1}{T(s)} \sum_{u = 1}^T \norm{\hat{F}^s_u - (H^s)^\T F^s_u}^2\right)^{1/2} \left(\frac{1}{T(s)} \sum_{u = 1}^T \left(\frac{1}{T(s)} \sum_{t = 1}^T \eta^s_{ut} e_{it}^s\right)^2\right)^{1/2} \\
			&=& \left(\frac{1}{T} \sum_{u = 1}^T \norm{\hat{F}^s_u - (H^s)^\T F^s_u}^2\right)^{1/2} \left(\frac{1}{T} \sum_{u = 1}^T K_s(S_u) \left(\frac{1}{T} \sum_{t = 1}^T K_s(S_t) \eta_{ut} e_{it} \right)^2\right)^{1/2},
		\end{eqnarray*}
		where
		\begin{eqnarray*}
			\frac{1}{T(s)} \sum_{t = 1}^T K_s(S_t) \eta_{ut} e_{it} &=& \frac{1}{\sqrt{N}} F_u^\T  \frac{1}{T(s)} \sum_{t = 1}^T K_s(S_t) \left( \frac{1}{\sqrt{N}} \sum_{l = 1}^N \Lambda_l(s) e_{lt} \right) e_{it} \\
			&=&  \frac{1}{\sqrt{N}} F_u^\T  \frac{1}{T(s)} \sum_{t = 1}^T K_s(S_t) O_p(1)  e_{it} \\
			&=& O_p\left( \frac{1}{\sqrt{N}} \right)
		\end{eqnarray*}
		by Assumption \ref{ass_mom}.3. The first term in  \RNum{3} has $\frac{1}{T^2} \sum_{t = 1}^T \sum_{u = 1}^T (\hat{F}^s_u - (H^s)^\T F^s_u) \eta^s_{ut} e_{it}^s = O_p\left( \frac{1}{\sqrt{N}\delta_{NT,h}}\right)$. The second term in \RNum{3} (ignore $H^s$) is 
		\begin{eqnarray*}
			\frac{1}{T(s)^2} \sum_{t = 1}^T \sum_{u = 1}^T F^s_u \eta^s_{ut} e_{it}^s &=& \frac{1}{T(s)^2} \sum_{t = 1}^T \sum_{u = 1}^T K_s(S_u) F_uF_u^\T  (K_s(S_t)\Lambda(s)e_t) e_{it} \\
			&=&  \left(\frac{1}{T(s)}\sum_{u = 1}^T K_s(S_u) F_uF_u^\T \right) \left( \frac{1}{NT(s)} \sum_{t = 1}^T \sum_{l = 1}^N K_s(S_t) \Lambda_l(s)e_{lt} e_{it} \right) \\
			&=& O_p(1) \left(O_p\left(\frac{1}{\sqrt{NTh}}\right) + O_p\left( \frac{1}{N}\right)\right)
		\end{eqnarray*}
		by Assumption \ref{ass_factor} and 
		\begin{eqnarray*}
			&& \frac{1}{NT(s)}\sum_{l = 1}^N \sum_{t = 1}^T  K_s(S_t) \Lambda_l(s)e_{lt} e_{it}  \\
			&=& \frac{1}{NT(s)}\sum_{l = 1}^N \sum_{t = 1}^T  K_s(S_t) \Lambda_l(s)(e_{lt} e_{it} - \+E(e_{lt} e_{it} + \+E(e_{lt} e_{it}) \\ 
			&=& \frac{1}{NT(s)}\sum_{l = 1}^N \sum_{t = 1}^T  K_s(S_t) \Lambda_l(s) (e_{lt} e_{it} - \+E(e_{lt} e_{it}) + \frac{1}{NT(s)}\sum_{l = 1}^N \sum_{t = 1}^T  K_s(S_t) \Lambda_l(s) \+E(e_{lt} e_{it})
		\end{eqnarray*}
		by Assumption \ref{ass_err}.3, $|\+E(e_{lt} e_{it})| = |\tau_{il, t}| \leq |\tau_{il}|$, and by Assumption \ref{ass_loading}, $\norm{\Lambda_l(s)} \leq \bar{\Lambda} \leq \infty$, we have 
		\begin{eqnarray*}
			\norm{\frac{1}{NT(s)}\sum_{l = 1}^N \sum_{t = 1}^T  K_s(S_t) \Lambda_l(s) \+E(e_{lt} e_{it})} &\leq& \frac{\bar{\Lambda} }{NT} \sum_{l = 1}^N |\tau_{il}| \left( \sum_{t = 1}^T K_s(S_t)\right) = O_p\left( \frac{1}{N}\right)
		\end{eqnarray*}
		and $\frac{1}{NT}\sum_{l = 1}^N \sum_{t = 1}^T  K_s(S_t) \Lambda_l(s) (e_{lt} e_{it} - \+E(e_{lt} e_{it}) = O_p\left(\frac{1}{\sqrt{NTh}}\right)$ by Assumption \ref{ass_mom}.5 and Lemma \ref{lemma:simplify-assumption6}.
		
		Therefore, \RNum{3} = $\frac{1}{T(s)^2} \sum_{t = 1}^T \sum_{u = 1}^T \hat{F}^s_u \eta^s_{ut} e_{it}^s =  O_p\left(\frac{1}{\sqrt{N}\delta_{NT,h}}\right)$

		Term \RNum{4}$ = \frac{1}{T(s)^2}\sum_{t = 1}^T \sum_{u = 1}^T \hat{F}^s_u \epsilon^s_{ut} e_{it}^s$ can be proved in a similar way as \RNum{4} and has $O_p\left(\frac{1}{\sqrt{N}\delta_{NT,h}}\right)$.
		
		The term \RNum{5} = $\frac{1}{T(s)^2} \sum_{t = 1}^T \sum_{u = 1}^T \hat{F}^s_u \frac{(\Delta X^s_u)^\T \bar{X}_t^s}{N} e_{it}^s$ has

		\begin{eqnarray*}
			\text{ \RNum{5}} &=&  \frac{1}{T(s)^2} \sum_{t = 1}^T \sum_{u = 1}^T \hat{F}^s_u \frac{(\Delta X^s_u)^\T \bar{X}_t^s}{N}e_{it}^s \\
			&=& \frac{1}{T(s)^2} \sum_{t = 1}^T \sum_{u = 1}^T (\hat{F}^s_u - (H^s)^\T F^s_u) \frac{(\Delta X^s_u)^\T \bar{X}_t^s}{N} e_{it}^s + (H^s)^\T  \frac{1}{T(s)^2} \sum_{t = 1}^T \sum_{u = 1}^T F^s_u \frac{(\Delta X^s_u)^\T \bar{X}_t^s}{N}e_{it}^s.
		\end{eqnarray*}

		Proven in a similar approach as the first term \RNum{5} and $\frac{1}{T}\sum_{t=1}^T f_t$ in the proof Theorem 1,  The norm of first term in \RNum{5} 
		\begin{eqnarray*}
			\frac{1}{T(s)^2} \sum_{t = 1}^T \sum_{u = 1}^T (\hat{F}^s_u - (H^s)^\T F^s_u) \frac{(\Delta X^s_u)^\T \bar{X}_t^s}{N} e_{it}^s &=& O_p\left(\frac{h}{\delta_{NT,h}}\right).
		\end{eqnarray*}

		The second term in \RNum{5} (ignore $H^s$) has 
		\begin{eqnarray*}
			&& \norm{\frac{1}{T(s)^2} \sum_{t = 1}^T \sum_{u = 1}^T F^s_u \frac{(\Delta X^s_u)^\T \bar{X}_t^s}{N} e_{it}^s }^2  \\
			&\leq& \left(\frac{1}{T(s)} \sum_{u = 1}^T \norm{F^s_u}^2 \right) \left(\frac{1}{N^2T(s)^2} \sum_{t = 1}^T \sum_{u = 1}^T \norm{(\Delta X^s_u)^\T \bar{X}_t^s}^2 \right) \left(\frac{1}{T(s)} \sum_{t = 1}^T (e_{it}^s)^2\right) \\
			&=& O_p(h^2)
		\end{eqnarray*}
		from Assumption \ref{ass_factor}, \ref{ass_loading}, \ref{ass_err}.1 and Lemma \ref{lemma1}.2. Therefore, \RNum{5} = $\frac{1}{T(s)^2} \sum_{t = 1}^T \sum_{u = 1}^T \hat{F}^s_u \frac{(\Delta X^s_u)^\T \bar{X}_t^s}{N} e_{it}^s = O_p(h).$
		
		Term \RNum{6} = $\frac{1}{T(s)^2} \sum_{t = 1}^T \sum_{u = 1}^T \hat{F}^s_u \frac{( \bar{X}^s_u)^\T \Delta X^s_t}{N}e_{it}^s = O_p(h)$, which can be shown in a similar way as \RNum{5}.
		
		The term \RNum{7} = $\frac{1}{T(s)^2} \sum_{t = 1}^T \sum_{u = 1}^T \hat{F}^s_u \frac{(\Delta X^s_u)^\T \Delta X^s_t}{N} e_{it}^s$ has  
		
		\begin{eqnarray*}
			&& \frac{1}{T(s)^2} \sum_{t = 1}^T \sum_{u = 1}^T \hat{F}^s_u \frac{(\Delta X^s_u)^\T \Delta X^s_t}{N} e_{it}^s \\
			&=&  \frac{1}{NT(s)^2} \sum_{t = 1}^T \sum_{u = 1}^T (\hat{F}^s_u - (H^s)^\T  F^s_u ) (\Delta X^s_u)^\T  \Delta X^s_t e_{it}^s + \frac{1}{NT(s)^2} (H^s)^\T  \sum_{t = 1}^T \sum_{u = 1}^T F^s_u (\Delta X^s_u)^\T  \Delta X^s_t  F^s_t e_{it}^s = O_p(h^2)
		\end{eqnarray*}
		which can be proven in a similar approach as the first term in \RNum{5} and by Lemma \ref{lemma1}.2. Therefore,
		\begin{eqnarray*}
			\frac{1}{T(s)} (\hat{F}^s - F^s H^s)^\T \underline{e}_i^s &=& O_p\left( \frac{1}{\sqrt{Th}\delta_{NT,h}}\right) + O_p \left( \frac{1}{\sqrt{N}\delta_{NT,h}}\right) + O_p(h) = O_p \left( \max \left( \frac{1}{\delta_{NT,h}^2}, h \right) \right)
		\end{eqnarray*}
	\end{proof}
	
	\begin{proof}[Proof of Lemma \ref{lemma4}.4]
		\begin{eqnarray*}
			&& \frac{1}{T(s)} (\hat{F}^s - F^s H^s)^\T F^s \\  
			&=& \frac{1}{T(s)} \sum_{t = 1}^T (\hat{F}^s_t - (H^s)^\T F^s_t ) (F^s_t)^\T   \\
			&=& (V^s_r)^{-1} \left[\frac{1}{T(s)^2}\sum_{t = 1}^T \sum_{u = 1}^T \hat{F}^s_u (F^s_t)^\T  \gamma_N^s(u,t)  + \frac{1}{T(s)^2} \sum_{t = 1}^T \sum_{u = 1}^T \hat{F}^s_u (F^s_t)^\T  \zeta^s_{ut} + \frac{1}{T(s)^2} \sum_{t = 1}^T \sum_{u = 1}^T \hat{F}^s_u (F^s_t)^\T  \eta^s_{ut}  \right. \\
			&&+ \frac{1}{T(s)^2}\sum_{t = 1}^T \sum_{u = 1}^T \hat{F}^s_u (F^s_t)^\T  \epsilon^s_{ut} + \frac{1}{T(s)^2} \sum_{t = 1}^T \sum_{u = 1}^T \hat{F}^s_u (F^s_t)^\T \frac{(\Delta X^s_u)^\T \bar{X}_t^s}{N} + \frac{1}{T(s)^2} \sum_{t = 1}^T \sum_{u = 1}^T \hat{F}^s_u (F^s_t)^\T  \frac{( \bar{X}^s_u)^\T \Delta X^s_t}{N} \\ 
			&&+ \left. \frac{1}{T(s)^2} \sum_{t = 1}^T \sum_{u = 1}^T \hat{F}^s_u (F^s_t)^\T  \frac{(\Delta X^s_u)^\T \Delta X^s_t}{N} \right] \\ 
			&=& (V^s_r)^{-1} \left[\text{\RNum{1} + \RNum{2} + \RNum{3} + \RNum{4} + \RNum{5} + \RNum{6} + \RNum{7}} \right].
		\end{eqnarray*}
		
		The term \RNum{1}$ = \frac{1}{T(s)^2}\sum_{t = 1}^T \sum_{u = 1}^T \hat{F}^s_u (F^s_t)^\T  \gamma_N^s(u,t)$,  we have 
		\begin{eqnarray*}
			\text{\RNum{1}} &=& \frac{1}{T(s)^2}\sum_{t = 1}^T \sum_{u = 1}^T \hat{F}^s_u  (F^s_t)^\T  \gamma_N^s(u,t) \\
			&=& \frac{1}{T(s)^2}\sum_{t = 1}^T \sum_{u = 1}^T  (\hat{F}^s_u - (H^s)^\T F^s_u)  (F^s_t)^\T  \gamma_N^s(u,t) + \frac{1}{T(s)^2} (H^s)^\T  \sum_{t = 1}^T \sum_{u = 1}^T  F^s_u  (F^s_t)^\T  \gamma_N^s(u,t)
		\end{eqnarray*}
		
		The norm of first term in \RNum{1} has 
		\begin{eqnarray*}
			&& \norm{\frac{1}{T(s)^2}\sum_{t = 1}^T \sum_{u = 1}^T (\hat{F}^s_u - (H^s)^\T F^s_u)(F^s_t)^\T  \gamma_N^s(u,t)} \\
			&\leq& \frac{1}{\sqrt{T(s)}} \left( \frac{1}{T(s)} \sum_{u = 1}^T \norm{\hat{F}^s_u - (H^s)^\T F^s_u}^2\right)^{1/2} \left[ \left(\frac{1}{T(s)} \sum_{u = 1}^T  \sum_{t = 1}^T  \gamma_N^s(u,t)^2\right) \left(\frac{1}{T(s)} \sum_{t = 1}^T\norm{F^s_t}^2\right) \right]^{1/2}  \\
			&=& O_p\left(\frac{1}{\sqrt{T}}\right) O_p \left( \frac{1}{\delta_{NT,h}}\right) O_p \left( \frac{1}{\sqrt{h}} \right) = O_p \left( \frac{1}{\sqrt{Th}\delta_{NT,h}} \right)
		\end{eqnarray*}
		by $\frac{1}{T(s)} \sum_{t = 1}^T \norm{\hat{F}^s_u - (H^s)^\T F^s_u}^2 = O_p\left(\frac{1}{\delta_{NT,h}^2}\right)$, $\frac{1}{T(s)} \sum_{t = 1}^T \sum_{u = 1}^T \gamma_N^s(u,t)^2 = O_p\left(\frac{1}{h}\right)$ from (\ref{gamma}) and Assumption 3 (b). 
		
		The second term in \RNum{1} is (ignore $H^s$ since it is $O_p(1)$)
		
		\begin{eqnarray*}
			&& \+E \norm{\frac{1}{T(s)^2}\sum_{t = 1}^T \sum_{u = 1}^T  F^s_u (F^s_t)^\T  \gamma_N^s(u,t)} \\
			&\leq& \+E\left[ \frac{1}{T(s)^2}\sum_{t = 1}^T \sum_{u = 1}^T K_s(S_t) K_s(S_u) |\gamma_N(u,t)| \left( \+E\left[ \norm{F_u}^2 |S_t, S_u \right] \right)^{1/2}  \left( \+E\left[\norm{F_t}^2 |S_t, S_u \right] \right)^{1/2} \right] \\
			&=& O_p\left(\frac{1}{Th}\right)
		\end{eqnarray*}
		by (\ref{gamma}) and Assumption \ref{ass_factor}.
		
		Therefore, \RNum{1} = $\frac{1}{T(s)^2}\sum_{t = 1}^T \sum_{u = 1}^T  \hat{F}^s_u  (F^s_t)^\T \gamma_N^s(u,t) = O_p\left(\frac{1}{\sqrt{Th}\delta_{NT,h}}\right)$.

		The term \RNum{2}$ = \frac{1}{T(s)^2} \sum_{t = 1}^T \sum_{u = 1}^T \hat{F}^s_u (F^s_t)^\T  \zeta^s_{ut}$ has
		\begin{eqnarray*}
			\text{\RNum{2}} &=& \frac{1}{T(s)^2} \sum_{t = 1}^T \sum_{u = 1}^T \hat{F}^s_u (F^s_t)^\T  \zeta^s_{ut} \\ 
			&=&\frac{1}{T(s)^2} \sum_{t = 1}^T \sum_{u = 1}^T (\hat{F}^s_u - (H^s)^\T ) F^s_u (F^s_t)^\T  \zeta^s_{ut} + \frac{1}{T(s)^2} (H^s)^\T  \sum_{t = 1}^T \sum_{u = 1}^T F^s_u (F^s_t)^\T  \zeta^s_{ut}.
		\end{eqnarray*}
		
		The norm of first term in \RNum{2} has
		\begin{eqnarray*}
			&&\norm{\frac{1}{T(s)^2} \sum_{t = 1}^T \sum_{u = 1}^T (\hat{F}^s_u - (H^s)^\T F^s_u) (F^s_t)^\T  \zeta^s_{ut}} \\
			&\leq& \left( \frac{1}{T(s)} \sum_{u = 1}^T \norm{\hat{F}^s_u - (H^s)^\T F^s_u}^2\right)^{1/2} \left(\frac{1}{T(s)} \sum_{u = 1}^T K_s(S_u) \norm{\frac{1}{T} \sum_{t = 1}^T K_s(S_t) F_t \zeta_{ut}  }^2\right)^{1/2}  \\
			&=& O_p\left(\frac{1}{\delta_{NT,h}}\right) O_p\left(\frac{1}{\sqrt{NTh}}\right) = O_p\left(\frac{1}{\sqrt{NTh}\delta_{NT,h}}\right)
		\end{eqnarray*}
		by Assumption \ref{ass_mom}.1 and Lemma \ref{lemma:simplify-assumption6}.

		Therefore, the first term in \RNum{2}, $\frac{1}{T(s)^2} \sum_{t = 1}^T \sum_{u = 1}^T (\hat{F}^s_u - (H^s)^\T F^s_u) (F^s_t)^\T  \zeta^s_{ut} = O_p\left( \frac{1}{\sqrt{NTh}\delta_{NT,h}}\right)$.
		
		The second term in \RNum{2} (ignore $H^s$) is 
		\begin{eqnarray*}
			\frac{1}{T(s)^2}\sum_{t = 1}^T \sum_{u = 1}^T F^s_u (F^s_t)^\T  \zeta^s_{ut} = \frac{1}{\sqrt{NTh}} \left( \frac{1}{T(s)} \sum_{t = 1}^T K_s(S_t) z^s_t F_t^\T  \right) = O_p \left( \frac{1}{\sqrt{NTh}} \right)
		\end{eqnarray*}
		where $z^s_t = \frac{\sqrt{Th}}{\sqrt{N}T(s)} \sum_{u = 1}^T \sum_{l = 1}^N K_s(S_u) F_u(e_{lu} e_{lt} - \+E[e_{lu}e_{lt}])$, by assumption \ref{ass_mom}.1 and Lemma \ref{lemma:simplify-assumption6}.

		Thus \RNum{2}$ = \frac{1}{T(s)^2} \sum_{t = 1}^T \sum_{u = 1}^T \hat{F}^s_u (F^s_t)^\T  \zeta^s_{ut} = O_p\left(\frac{1}{\sqrt{NTh}}\right)$.
		
		The term \RNum{3}$ = \frac{1}{T(s)^2} \sum_{t = 1}^T \sum_{u = 1}^T \hat{F}^s_u (F^s_t)^\T  \eta^s_{ut}$ has
		\begin{eqnarray*}
			\text{\RNum{3}} &=& \frac{1}{T(s)^2} \sum_{t = 1}^T \sum_{u = 1}^T \hat{F}^s_u (F^s_t)^\T  \zeta^s_{ut} \\
			&=& \frac{1}{T(s)^2} \sum_{t = 1}^T \sum_{u = 1}^T (\hat{F}^s_u - (H^s)^\T  F^s_u) (F^s_t)^\T  \eta^s_{ut} + \frac{1}{T(s)^2} (H^s)^\T  \sum_{t = 1}^T \sum_{u = 1}^T F^s_u (F^s_t)^\T  \eta^s_{ut}
		\end{eqnarray*}
		
		The norm of the first term in \RNum{3} has
		\begin{eqnarray*}
			&& \norm{\frac{1}{T(s)^2} \sum_{t = 1}^T \sum_{u = 1}^T (\hat{F}^s_u - (H^s)^\T F^s_u) (F^s_t)^\T  \eta^s_{ut}} \\
			&\leq& \left( \frac{1}{T(s)}\sum_{u = 1}^T \norm{ (\hat{F}^s_u - (H^s)^\T F^s_u)}^2\right)^{1/2} \left( \frac{1}{T(s)} \sum_{u = 1}^T K_s(S_u) \norm{ \frac{1}{T} \sum_{t = 1}^T K_s(S_t) F_t^\T  \eta_{ut}}^2\right)^{1/2} \\
			&=& O_p\left(\frac{1}{\delta_{NT,h}}\right) O_p\left(\frac{1}{\sqrt{NTh}}\right)
		\end{eqnarray*}
		by Assumption \ref{ass_mom}.2 and Lemma \ref{lemma:simplify-assumption6}.

		For the second term in \RNum{3},
		\begin{eqnarray*}
			\frac{1}{T(s)^2} \sum_{t = 1}^T \sum_{u = 1}^T F^s_u (F^s_t)^\T  \eta^s_{ut} &=& \frac{1}{T(s)^2} \sum_{t = 1}^T \sum_{u = 1}^T F^s_u (F^s_t)^\T  \frac{(F^s_u)^\T  \Lambda(s)^\T  e^s_t}{N} \\
			&=& \left( \frac{1}{T(s)} \sum_{u = 1}^T K_s(S_u) F_u F_u^\T  \right) \left( \frac{1}{NT(s)} \sum_{t = 1}^T \sum_{l = 1}^N K_s(S_t) F_t^\T  \Lambda_l(s) e_{lt} \right) \\
			&=& O_p(1) O_p\left( \frac{1}{\sqrt{NTh}} \right)
		\end{eqnarray*}
		by Assumption \ref{ass_mom}.2 and Lemma \ref{lemma:simplify-assumption6}. Therefore, \RNum{3} = $\frac{1}{T(s)^2} \sum_{t = 1}^T \sum_{u = 1}^T \hat{F}^s_u (F^s_t)^\T  \eta^s_{ut} = O_p\left( \frac{1}{\sqrt{NTh}}\right)$.
		
		$\text{\RNum{4}} = O_p\left(\frac{1}{\sqrt{NTh}}\right)$ can be proved in a similar way. 
		
		The term \text{\RNum{5}} = $\frac{1}{T(s)^2} \sum_{t = 1}^T \sum_{u = 1}^T \hat{F}^s_u (F^s_t)^\T \frac{(\Delta X^s_u)^\T \bar{X}_t^s}{N}$ has

		\begin{eqnarray*}
			\text{ \RNum{5}} &=&  \frac{1}{T(s)^2} \sum_{t = 1}^T \sum_{u = 1}^T \hat{F}^s_u \frac{(\Delta X^s_u)^\T \bar{X}_t^s}{N} (F^s_t)^\T  \\
			&=& \frac{1}{T(s)^2} \sum_{t = 1}^T \sum_{u = 1}^T (\hat{F}^s_u - (H^s)^\T F^s_u) \frac{(\Delta X^s_u)^\T \bar{X}_t^s}{N} (F^s_t)^\T  + (H^s)^\T  \frac{1}{T(s)^2} \sum_{t = 1}^T \sum_{u = 1}^T F^s_u \frac{(\Delta X^s_u)^\T \bar{X}_t^s}{N} (F^s_t)^\T 
		\end{eqnarray*}

		Similarly as the term \RNum{5} in the proof of Lemma \ref{lemma4}.3, the first term in \RNum{5} has
		\begin{eqnarray*}
			\frac{1}{T(s)^2} \sum_{t = 1}^T \sum_{u = 1}^T (\hat{F}^s_u - (H^s)^\T F^s_u) \frac{(\Delta X^s_u)^\T \bar{X}_t^s}{N} (F^s_t)^\T  &=& O_p\left(\frac{h}{\delta_{NT,h}}\right)
		\end{eqnarray*}
		
		and the second term in \RNum{5} (ignore $H^s$) has 
		\begin{eqnarray*}
			&& \norm{\frac{1}{T(s)^2} \sum_{t = 1}^T \sum_{u = 1}^T F^s_u \frac{(\Delta X^s_u)^\T \bar{X}_t^s}{N} (F^s_t)^\T  }^2  \\
			&\leq& \left(\frac{1}{T(s)} \sum_{u = 1}^T \norm{F^s_u}^2\right) \left(\frac{1}{N^2T(s)^2} \sum_{t = 1}^T \sum_{u = 1}^T \norm{(\Delta X^s_u)^\T \bar{X}_t^s}^2\right) \left( \frac{1}{T} \sum_{t = 1}^T \norm{F^s_t}^2\right) \\
			&=& O_p(h^2)
		\end{eqnarray*}
		by Assumption \ref{ass_factor} and the proof of $\frac{1}{T}\sum_{t=1}^T f_t$ in Theorem \ref{thm_consistency}.
		
		Therefore, \text{\RNum{5}} = $\frac{1}{T(s)^2} \sum_{t = 1}^T \sum_{u = 1}^T \hat{F}^s_u (F^s_t)^\T \frac{(\Delta X^s_u)^\T \bar{X}_t^s}{N} = O_p(h)$.
		
		The term \text{\RNum{6}} = $\frac{1}{T^2} \sum_{t = 1}^T \sum_{u = 1}^T \hat{F}^s_u (F^s_t)^\T  \frac{( \bar{X}^s_u)^\T \Delta X^s_t}{N} = O_p(h)$ similarly as \text{\RNum{5}}.
		
		The term \text{\RNum{7}} = $\frac{1}{T^2} \sum_{t = 1}^T \sum_{u = 1}^T \hat{F}^s_u (F^s_t)^\T  \frac{(\Delta X^s_u)^\T \Delta X^s_t}{N}$ has
		\begin{eqnarray*}
			\text{\RNum{7}} &=& \frac{1}{T^2} \sum_{t = 1}^T \sum_{u = 1}^T \hat{F}^s_u (F^s_t)^\T  \frac{(\Delta X^s_u)^\T \Delta X^s_t}{N} \\
			&=&  \frac{1}{NT^2} \sum_{t = 1}^T \sum_{u = 1}^T (\hat{F}^s_u - (H^s)^\T  F^s_u ) (\Delta X^s_u)^\T \Delta X^s_t  (F^s_t)^\T  + \frac{1}{NT^2} (H^s)^\T  \sum_{t = 1}^T \sum_{u = 1}^T F^s_u (\Delta X^s_u)^\T \Delta X^s_t (F^s_t)^\T  \\
			&=& O_p(h^2)
		\end{eqnarray*}
		similarly as the term \RNum{7} in the proof of lemma \ref{lemma4}.3.
		
		Therefore,
		\begin{eqnarray*}
			\frac{1}{T} (\hat{F}^s - F^s H^s)^\T F^s &=& (V^s_r)^{-1} [\text{\RNum{1} + \RNum{2} + \RNum{3} + \RNum{4}  + \RNum{5} + \RNum{6} + \RNum{7}}] \\
			&=&  O_p\left(\frac{1}{\sqrt{Th}\delta_{NT,h}}\right) + O_p\left(\frac{1}{\sqrt{NTh}}\right) + O_p(h) \\
			&=& O_p \left( \max \left( \frac{1}{\delta_{NT,h}^2}, h \right) \right)
		\end{eqnarray*}

	\end{proof}
	
	\begin{proof}[Proof of Lemma \ref{lemma4}.5]
		\begin{eqnarray*}
			&& \frac{1}{T(s)} (\hat{F}^s (H^s)^{-1}- F^s)^\T \hat{F}^s \\ 
			&=& ((H^s)^{-1})^\T  \frac{1}{T(s)} (\hat{F}^s- F^s H^s)^\T \hat{F}^s \\
			&=& ((H^s)^{-1})^\T  \frac{1}{T(s)} (\hat{F}^s- F^s H^s)^\T  (\hat{F}^s- F^s H^s) + ((H^s)^{-1})^\T   \frac{1}{T(s)} (\hat{F}^s- F^s H^s)^\T  F^s H^s \\
			&=& O_p \left( \max \left( \frac{1}{\delta_{NT,h}^2}, h \right) \right)
		\end{eqnarray*}
		from lemma \ref{lemma4}.4 when $\delta_{NT,h}^2 h \rightarrow 0$.
		
	\end{proof}
	
	\begin{proof}[Proof of Theorem 3]
		From $\hat{\Lambda}(s)  = \frac{X^s \hat{F}^s}{T(s)} $ and $X^s = \Lambda(s) (F^s)^\T  + e^s + \Delta X^s$, we have $\hat{\Lambda}(s) = \frac{1}{T(s)} \Lambda(s)  (F^s)^\T  \hat{F}^s + \frac{1}{T(s)} e^s \hat{F}^s + \frac{1}{T(s)} \Delta X^s  \hat{F}^s $, so $\hat{\Lambda}_i(s) = \frac{1}{T} (\hat{F}^s)^\T  F^s \Lambda_i(s) + \frac{1}{T} (\hat{F}^s)^\T  \underline{e}_i^s +  \frac{1}{T} (\hat{F}^s)^\T  \Delta X^s_i  $. Writing $F^s = F^s - \hat{F}^s (H^s)^{-1} + \hat{F}^s (H^s)^{-1}$ and $\hat{F}^s = \hat{F}^s - F^s H^s + F^s H^s$, we have
		\begin{eqnarray*}
			\hat{\Lambda}_i(s) &=& \frac{1}{T(s)} (\hat{F}^s)^\T  F^s \Lambda_i(s) + \frac{1}{T(s)} (\hat{F}^s)^\T  \underline{e}_i^s +  \frac{1}{T(s)} (\hat{F}^s)^\T  \Delta X^s_i \\  
			&=& \frac{1}{T(s)} (\hat{F}^s)^\T  (F^s - \hat{F}^s (H^s)^{-1} + \hat{F}^s (H^s)^{-1}) \Lambda_i(s) \\
			&& + \frac{1}{T(s)} (\hat{F}^s - F^s H^s + F^s H^s)^\T  \underline{e}_i^s + \frac{1}{T(s)} (\hat{F}^s - F^s H^s + F^s H^s)^\T  \Delta X^s_i  \\
			&=& (H^s)^{-1} \Lambda_i(s) + \frac{1}{T(s)} (\hat{F}^s)^\T  (F^s - \hat{F}^s (H^s)^{-1}) \Lambda_i(s) \\ 
			&&+ \frac{1}{T(s)} (\hat{F}^s - F^s H^s)^\T \underline{e}_i^s + \frac{1}{T(s)} (H^s)^\T  (F^s)^\T  \underline{e}_i^s +  \frac{1}{T(s)} (\hat{F}^s - F^s H^s)^\T \Delta X^s_i  + \frac{1}{T(s)} (H^s)^\T  (F^s)^\T  \Delta X^s_i 
		\end{eqnarray*}
		and
		
		\begin{eqnarray*}
			\hat{\Lambda}_i(s) - (H^s)^{-1} \Lambda_i(s)  &=& \frac{1}{T(s)} (H^s)^\T  (F^s)^\T  \underline{e}_i^s  + \frac{1}{T(s)} (\hat{F}^s)^\T  (F^s - \hat{F}^s (H^s)^{-1}) \Lambda_i(s) + \frac{1}{T(s)} (\hat{F}^s - F^s H^s)^\T \underline{e}_i^s \\
			&&+ \frac{1}{T(s)} (\hat{F}^s - F^s H^s)^\T \Delta X^s_i  + \frac{1}{T(s)} (H^s)^\T  (F^s)^\T  \Delta X^s_i \\
			&=& \frac{1}{T(s)} (H^s)^\T  (F^s)^\T  \underline{e}_i^s + O_p \left( \max \left( \frac{1}{\delta_{NT,h}^2}, h \right) \right) + O_p(h)
		\end{eqnarray*}
		by Lemma \ref{lemma4} and  $\norm{\frac{1}{T}(F^s)^\T  \Delta X^s_i }^2 = \norm{\frac{1}{T} \sum_{t = 1}^T F^s_t \Delta X^s_{it} }^2 \leq (\frac{1}{T} \sum_{t = 1}^T \norm{F^s_t}^2) (\frac{1}{T} \sum_{t = 1}^T \norm{\Delta X^s_{it}}^2) = O_p(h^2) $.

		If $\sqrt{Th}/N \rightarrow 0$, $Th \rightarrow \infty$, $Th^3 \rightarrow 0$, we have
		\begin{eqnarray*}
			\sqrt{Th} (\hat{\Lambda}_i(s) - (H^s)^{-1} \Lambda_i(s)) &=& \frac{\sqrt{Th}}{T(s)} (H^s)^\T  (F^s)^\T  \underline{e}_i^s + o_p(1) \\
			&=& (H^s)^\T  \frac{\sqrt{Th}}{T(s)} \sum_{t = 1}^T K_s(S_t) F_t e_{it} + o_p(1)
		\end{eqnarray*}
		
		From Lemma \ref{lemma3}.1, Assumption \ref{ass_mom}.4 and slusky theorem,
		
		\begin{eqnarray*}
			\sqrt{Th}(\hat{\Lambda}_i(s) - (H^s)^{-1} \Lambda_i(s)) \xrightarrow{d} N(0, ((Q^s)^\T )^{-1} \Phi^s_i (Q^s)^{-1})
			.\end{eqnarray*}

	\end{proof}

	\begin{proof}[Proof of Theorem 4]
		Since $C_{it,s} = F_t^\T \Lambda_i(s)$ and $\hat{C}_{it,s} = \hat{F}_t\hat{\Lambda}_i(s) = \left(\frac{\hat{F}^s_t}{K^{1/2}_s(S_t)}\right)'\hat{\Lambda}_i(s)$,
		\begin{eqnarray*}
			\hat{C}_{it,s} - C_{it,s} &=&  \left( \frac{\hat{F}^s_t}{K^{1/2}_s(S_t)}\right)^\T  \hat{\Lambda}_i(s) - F_t^\T \Lambda_i(s) \\
			&=& \Lambda_i(s)^\T  ((H^s)^{-1})^\T  \left( \frac{\hat{F}^s_t}{K^{1/2}_s(S_t)} - (H^s)^\T  F_t \right) + F_t^\T  H^s (\hat{\Lambda}_i(s) - (H^s)^{-1} \Lambda_i(s)) \\ 
			&&+ (\hat{\Lambda}_i(s) - (H^s)^{-1} \Lambda_i(s))^\T \left(\frac{\hat{F}^s_t}{K^{1/2}_s(S_t)} - (H^s)^\T  F_t \right)
		\end{eqnarray*}
		
		From the limiting distribution of estimated factors and the limiting distribution of estimated factor loadings, $\hat{\Lambda}_i(s) - (H^s)^{-1} \Lambda_i(s) = O_p\left( \frac{1}{\sqrt{Th}}\right)$ for all $i$ and $\frac{\hat{F}^s_t}{K^{1/2}_s(S_t)} - (H^s)^\T  F_t = O_p\left( \frac{1}{\sqrt{N}}\right)$ for the $t$ where $\frac{1}{\sqrt{N}} \sum_{i = 1}^N \norm{\Lambda_i(S_t)-\Lambda_i(s)}  = o_p(1)$. We have $$\left(\hat{\Lambda}_i(s) - (H^s)^{-1} \Lambda_i(s)\right)^\T \left(\frac{\hat{F}^s_t}{K^{1/2}_s(S_t)} - (H^s)^\T  F_t\right) = O_p\left(\frac{1}{\delta_{NT,h}^2}\right),$$
		$$ \delta_{NT,h} \left(\frac{\hat{F}^s_t}{K^{1/2}_s(S_t)} - (H^s)^\T  F_t\right) = \frac{\delta_{NT,h}}{\sqrt{N}} (V_r^s)^{-1} \frac{(\hat{F}^s)^\T  F^s}{T(s)}  \left( \frac{1}{\sqrt{N}}\sum_{i = 1}^{N} \Lambda_i(s) e_{it}\right) + O_p\left(\frac{1}{\delta_{NT,h}}\right),$$
		and 
		$$
		\delta_{NT,h} (\hat{\Lambda}_i(s) - (H^s)^{-1} \Lambda_i(s)) =\frac{\delta_{NT,h}}{\sqrt{Th}} (H^s)^\T  \frac{\sqrt{Th}}{T(s)} \sum_{t = 1}^T F^s_t e^s_{it} + O_p\left(\frac{1}{\delta_{NT,h}}\right).$$
		Combining these three equalities, we have 
		\begin{eqnarray*}
			\delta_{NT,h}  (\hat{C}_{it,s} - C_{it,s}) &=& \frac{\delta_{NT,h}}{\sqrt{N}}\Lambda_i(s)^\T  ((H^s)^{-1})^\T  (V_r^s)^{-1} \frac{(\hat{F}^s)^\T  F^s}{T(s)}  \left( \frac{1}{\sqrt{N}}\sum_{i = 1}^{N} \Lambda_i(s) e_{it}\right) \\
			&& + \frac{\delta_{NT,h}}{\sqrt{Th}} F_t^\T  H^s (H^s)^\T  \frac{\sqrt{Th}}{T(s)} \sum_{t = 1}^T F^s_t e^s_{it} + O_p\left(\frac{1}{\delta_{NT,h}}\right) 
		\end{eqnarray*}
		Given $H^s = \frac{\Lambda(s)^\T  \Lambda(s)}{N} \frac{(F^s)^\T  \hat{F}^s}{T(s)}  (V^s_r)^{-1}$, $((H^s)^{-1})^\T  (V_r^s)^{-1} \frac{(\hat{F}^s)^\T  F^s}{T(s)}  = \left(\frac{\Lambda(s)^\T \Lambda(s)}{N}\right)^{-1} = \Sigma_{\Lambda(s)}^{-1} + O_p(1/\sqrt{N})$. From Lemma \ref{lemma4}.2, we have $H^s (H^s)^\T  = \Sigma_{F|s}^{-1} + O_p\left(\frac{1}{\delta_{NT,h}}\right)$. Thus, 
		\begin{eqnarray*}
			\delta_{NT,h}  (\hat{C}_{it,s} - C_{it,s}) = \frac{\delta_{NT,h}}{\sqrt{N}}\Lambda_i(s)^\T  \Sigma_{\Lambda(s)}^{-1} \left( \frac{1}{\sqrt{N}}\sum_{i = 1}^{N} \Lambda_i(s) e_{it}\right) + \frac{\delta_{NT,h}}{\sqrt{Th}} F_t^\T   \Sigma_{F|s}^{-1} \frac{\sqrt{Th}}{T(s)} \sum_{t = 1}^T F^s_t e^s_{it} + O_p\left(\frac{1}{\delta_{NT,h}}\right) 
		\end{eqnarray*}
		Let $\xi_{NT,h} = \Lambda_i(s)^\T  \Sigma_{\Lambda(s)}^{-1} \left( \frac{1}{\sqrt{N}}\sum_{i = 1}^{N} \Lambda_i(s) e_{it}\right) $ and $\zeta_{NT,h} = F_t^\T   \Sigma_{F|s}^{-1} \frac{\sqrt{Th}}{T(s)} \sum_{t = 1}^T F^s_t e^s_{it}$. From Slusky's Theorem and Assumption \ref{ass_mom}, we have $\xi_{NT,h}\xrightarrow{d} \xi \stackrel{d}{=} N(0, V_{it,s})$, where $V_{it,s} = \Lambda_i(s)^\T  \Sigma_{\Lambda(s)}^{-1} \Gamma_{j,s} \Sigma_{\Lambda(s)}^{-1}  \Lambda_i(s)$; $\zeta_{NT,h} \xrightarrow{d} \zeta \stackrel{d}{=} N(0, W_{it,s})$, where $W_{it,s} = F_t^\T  \Sigma_{F|s}^{-1} \Phi_{i,s} \Sigma_{F|s}^{-1} F_t$. Since $\xi_{NT,h}$ is sum of cross sectional random variables and $\zeta_{NT,h}$ is the sum of serial random variables, $\xi_{NT,h}$ and $\zeta_{NT,h}$ are asymptotically independent. Also, $(\xi_{NT,h},\zeta_{NT,h})$ jointly converge to a bivariate normal distribution. Let $a_{NT,h} = \delta_{NT,h}/\sqrt{N}$ and $b_{NT,h} = \delta_{NT,h}/\sqrt{Th}$, we have 
		$$
		\delta_{NT,h}(\hat{C}_{it,s} - C_{it,s}) = a_{NT,h} \xi_{NT,h} + b_{NT,h}\zeta_{NT,h} + O_p\left(\frac{1}{\delta_{NT,h}}\right)
		$$
		$a_{NT,h}$ and $b_{NT,h}$ are bounded and nonrandom sequences. If they converge to some constants, $\delta_{NT,h}(\hat{C}_{it,s} - C_{it,s})$ is asymptotic normal from Slustky's Theorem. Otherwise, if $a_{NT,h}$ and $b_{NT,h}$ are not convergent sequences, we can use the almost sure representation theory as \cite{bai2003inferential}. Since $(\xi_{NT,h}, \zeta_{NT,h}) \xrightarrow{d} (\xi, \zeta)$, there exist random vectors $(\xi^\ast_{NT,h}, \zeta^\ast_{NT,h}) $ and $(\xi^\ast, \zeta^\ast)$ with the same distribution as $(\xi_{NT,h}, \zeta_{NT,h}) $ and $(\xi, \zeta)$, and $(\xi^\ast_{NT,h}, \zeta^\ast_{NT,h}) \xrightarrow{a.s.} (\xi, \zeta)$. We have
		$$
		a_{NT,h} \xi^\ast_{NT,h} + b_{NT,h}\zeta^\ast_{NT,h} = a_{NT,h} \xi^\ast + b_{NT,h} \zeta^\ast + a_{NT,h} (\xi^\ast_{NT,h} - \xi^\ast) + b_{NT,h} (\zeta^\ast_{NT,h}-\zeta^\ast)
		$$
		Because of the almost sure convergence, $a_{NT,h} (\xi^\ast_{NT,h} - \xi^\ast) = o_p(1)$ and $ b_{NT,h} (\zeta^\ast_{NT,h}-\zeta^\ast) = o_p(1)$. $\xi^\ast$ and $\zeta^\ast$ are independent normal random variables with variances $V_{it,s}$ and $W_{it,s}$.  We have $a_{NT,h} \xi^\ast + b_{NT,h} \zeta^\ast  \stackrel{d}{=} N(0, a_{NT,h}^2 V_{it,s} + b_{NT,h}^2 W_{it,s})$. Thus, 
		$$
		\frac{\delta_{NT,h}(\hat{C}_{it,s} - C_{it,s}) }{(a_{NT,h}^2 V_{it,s} + b_{NT,h}^2 W_{it,s})^{1/2}} \xrightarrow{d} N(0,1),
		$$
		which is equivalent to 
		$$
		\frac{\hat{C}_{it,s} - C_{it,s}}{\left( \frac{1}{N} V_{it,s} + \frac{1}{Th}W_{it,s}\right)^{1/2}} \xrightarrow{d} N(0,1).
		$$
	\end{proof}
	
	\subsection{Proof of Test for Constant Factor Loading }
	
	\begin{proof}[Proof of Lemma \ref{lemma_leqr}]
		Define $\rho = \tr \left\lbrace \left( \frac{1}{N}\Lambda^\T_1 \Lambda_1 \right)^{-1} \left( \frac{1}{N}\Lambda^\T_1 \Lambda_2 \right) \left( \frac{1}{N}\Lambda^\T_2 \Lambda_2 \right)^{-1} \left( \frac{1}{N} \Lambda^\T_2 \Lambda_1 \right) \right\rbrace$, let $\tilde{\Lambda}_1 =   \frac{1}{\sqrt{N}} \Lambda_1  \left( \frac{1}{N}\Lambda^\T_1 \Lambda_1 \right)^{-1/2}$ and $\tilde{\Lambda}_2 = \frac{1}{\sqrt{N}} \Lambda_2  \left( \frac{1}{N}\Lambda^\T_2 \Lambda_2 \right)^{-1/2}$, and therefore, $\tilde{\Lambda}_1^T \tilde{\Lambda}_1 = I_{k_1}$ and $\tilde{\Lambda}_2^T \tilde{\Lambda}_2 = I_{k_2}$, $I_{k_1} \in \mathbb{R}^{k_1 \times k_1}$ and $I_{k_2} \in \mathbb{R}^{k_2 \times k_2}$ are identity matrices.  
		
		As a result, $\rho = \tr \left\lbrace \left( \tilde{\Lambda}^\T_1 \tilde{\Lambda}_2 \right) \left(  \tilde{\Lambda}^\T_2 \tilde{\Lambda}_1 \right) \right\rbrace$. Without loss of generalization, we assume $k_1 \geq k_2$, $\tr \left(\tilde{\Lambda}_2  \tilde{\Lambda}^\T_2   \right) = \tr \left(\tilde{\Lambda}^\T_2  \tilde{\Lambda}_2   \right) = k_2$.\footnote{If $k_2 \geq k_1$, since $\rho = \tr \left\lbrace \left( \tilde{\Lambda}^\T_1 \tilde{\Lambda}_2 \right) \left( \tilde{\Lambda}^\T_2 \tilde{\Lambda}_1 \right) \right\rbrace = \tr \left\lbrace \left(\tilde{\Lambda}^\T_2 \tilde{\Lambda}_1 \right) \left( \tilde{\Lambda}^\T_1 \tilde{\Lambda}_2 \right)  \right\rbrace$, we study $\rho = \tr \left\lbrace \tilde{\Lambda}^\T_2 \left( \tilde{\Lambda}_1 \tilde{\Lambda}^\T_1  \right) \tilde{\Lambda}_2 \right\rbrace$. } Since $\tilde{\Lambda}^\T_2  \tilde{\Lambda}_2$ is a real symmetric matrix, $\tilde{\Lambda}^\T_2  \tilde{\Lambda}_2$ can be decomposed as
		$\tilde{\Lambda}^\T_2  \tilde{\Lambda}_2 = Q \Sigma Q^\T $
		where $Q\in \mathbb{R}^{N \times N}$ is an orthonormal matrix, and $\Sigma$ is a diagonal matrix whose entries are the eigenvalues of $\tilde{\Lambda}^\T_2  \tilde{\Lambda}_2$. We have $\tr (\Sigma) = \tr(\tilde{\Lambda}^\T_2  \tilde{\Lambda}_2) = k_2$
		
		Since $Q$ is an orthonormal basis, there exists $A = [\alpha_1, \alpha_2, \cdots, \alpha_{k_1}]$, such that $\tilde{\Lambda}_1 = Q A$. The columns of $A$ are orthogonal. The reason is that $A = Q^\T  \tilde{\Lambda}_1$, $\alpha_i^\T  \alpha_i = \tilde{\lambda}_{1i}^T Q Q^\T  \tilde{\lambda}_{1i}^T = 1$, where $\tilde{\lambda}_{1i}$ is the $i$-th column in $\tilde{\Lambda}_1$. In addition, $\alpha_i^\T  \alpha_t = \tilde{\lambda}_{1i}^T Q Q^\T  \tilde{\lambda}_{1j}^T = 0$. 
		
		\begin{eqnarray*}
			\rho &=& \tr \left\lbrace  \tilde{\Lambda}^\T_1 \left(\tilde{\Lambda}_2  \tilde{\Lambda}^\T_2\right) \tilde{\Lambda}_1  \right\rbrace \\ 
			&=& \tr(A^\T  Q^\T  Q \Sigma Q^\T  Q A) = \tr (A^\T  \Sigma A) \\
			&=& \sum_{j=1}^N \left( \sum_{i=1}^{k_1} \alpha^2_{ij} \right) \sigma_{jj} \\
			&\leq& \sum_{j=1}^N  \sigma_{jj} = \tr(\Sigma) = k_2 = k,
		\end{eqnarray*}
		where $\alpha_{ij}$ is the $j$-th element in $\alpha_{i}$, $\sigma_{jj}$ is the $j$-th diagonal element of $\Sigma$ and $k = \min(k_1, k_2) = k_2$. The last inequality is derived from $\sum_{i=1}^{k_1} \alpha^2_{ij}  \leq 1, \forall j$. The reason is that the columns of $A$ are orthogonal, there exists $\bar{A} \in \mathbb{R}^{N \times (N-k_1)}$, such that let $\tilde{A} = [A, \bar{A}]$, we have $\tilde{A}$ is an orthonormal matrix. We have $\tilde{A} \tilde{A}^\T  = I_N$. Therefore, $\forall j$, $1 = \sum_{i=1}^{k_1} \alpha^2_{ij} + \sum_{i=1}^{N - k_1} \bar{\alpha}^2_{ij} \geq  \sum_{i=1}^{k_1} \alpha^2_{ij}$, where $\bar{\alpha}_{ij}$ is the $j$-th element in $i$-th column in $\bar{A}$.

	\end{proof}
	
	\begin{proof}[Proof of Theorem \ref{thm_rho}]
		From the proof of Theorem 1 and Theorem 2, and denote $e^{s_l}_i \in \mathbb{R}^{N \times 1}$ as the $i$-th row in $e^{s_l}$, we have
		
		\begin{eqnarray}
		\nonumber \bar{\Lambda}_i(s_l) - (H^{s_l})^\T  \Lambda_i(s_l) &=& \frac{1}{N} \left[ \sum_{k=1}^N (V_r^{s_l})^{-1} \bar{\Lambda}_k(s_l) \frac{\+E\left[(e^{s_l}_k)^\T  e^{s_l}_i\right]}{T(s_l)} + \sum_{k=1}^N (V_r^{s_l})^{-1} \bar{\Lambda}_k(s_l) \frac{(e^{s_l}_k)^\T  e^{s_l}_i - \+E\left[(e^{s_l}_k)^\T  e^{s_l}_i\right]}{T(s_l)} \right. \\
		\nonumber &+& \left. \sum_{k=1}^N (V_r^{s_l})^{-1} \bar{\Lambda}_k(s_l) \frac{\Lambda^\T_k(s_l)(F^{s_l})^\T  e_i^{s_l}}{T(s_l)} + \sum_{k=1}^N (V_r^{s_l})^{-1} \bar{\Lambda}_k(s_l) \frac{(e^{s_l}_k)^\T  F^{s_l} \Lambda_i(s_l)}{T(s_l)} \right] + O_p(h) \\
		&=& \frac{1}{N}\sum_{k=1}^N \gamma^{s_l}_{N}(k,i) +  \frac{1}{N}\sum_{k=1}^N \zeta^{s_l}_{ki} +  \frac{1}{N}\sum_{k=1}^N \eta^{s_l}_{ki} +  \frac{1}{N}\sum_{k=1}^N \epsilon^{s_l}_{ki}  + O_p(h) \label{eqn-decom-lambda}
		\end{eqnarray}
		From the proof of Theorem \ref{ass_factor} and Theorem \ref{ass_loading}, and similar to \cite{pelger2017}, we have 
		\begin{eqnarray*}
			u_{li1} &=& \frac{1}{N}\sum_{k=1}^N \gamma^{s_l}_{N}(k,i) = O_p\left( \frac{1}{\sqrt{N \delta_{NT,h}}} \right) \\
			u_{li2} &=& \frac{1}{N}\sum_{k=1}^N \zeta^{s_l}_{ki} = O_p\left( \frac{1}{\sqrt{Th \delta_{NT,h}}} \right) \\
			u_{li3} &=& \frac{1}{N}\sum_{k=1}^N \eta^{s_l}_{ki} = O_p\left( \frac{1}{\sqrt{Th}} \right) \\
			u_{li4} &=& \frac{1}{N}\sum_{k=1}^N \epsilon^{s_l}_{ki} = O_p\left( \frac{1}{\sqrt{Th \delta_{NT,h}}} \right)
		\end{eqnarray*}
		
		Under $\sqrt{Th}/N \rightarrow 0$, $ \frac{1}{N}\sum_{k=1}^N \eta^{s_l}_{ki}$ is the dominate term in the asymptotic distribution of $\bar{\Lambda}_i(s_l)$. We have
		
		\begin{eqnarray*}
			\frac{1}{N}\bar{\Lambda}^\T_l\bar{\Lambda}_{l'} &=& \frac{1}{N}\sum_{i=1}^N  \left( H^{s_l\prime } \lambda_{li} + u_{li1} + u_{li2} + u_{li3} + u_{li4}  \right) \left( (H^{s_{l'}})^T \lambda_{l'i} + u_{l'i1} + u_{l'i2} + u_{l'i3} + u_{l'i4} \right)^\T  + O_p(h) \\
			&=& \left( \frac{1}{N}\sum_{i=1}^N  (H^{s_l})^\T  \lambda_{li} \lambda_{l'i}^\T  H^{s_{l'}} +  \sum_{a=1}^4 \frac{1}{N}\sum_{i=1}^N  (H^{s_l})^\T  \lambda_{li} u_{l'ia}^\T  + \sum_{a=1}^4 \frac{1}{N}\sum_{i=1}^N u_{lia}\lambda_{l'i}^\T  H^{s_{l'}} + \Upsilon_{ll'}  \right)   + O_p(h) \\
		\end{eqnarray*}
		where $\lambda_{li} = \Lambda_i(s_l) $ and $\Upsilon_{ll'}=\frac{1}{N} \sum_{i=1}^N \left( u_{li1} + u_{li2} + u_{li3} + u_{li4} \right) \left( u_{l'i1} + u_{l'i2} + u_{l'i3} + u_{l'i4} \right)^\T $. 
		
		Next is to analyze $ \frac{1}{N}\sum_{i=1}^N u_{l'ia}\lambda_{li}^\T $
		, $a=1,2,3,4$. Let $v_{li} = (H^{s_l})^\T  \frac{\sqrt{Th}}{T(s_l)} \left(\frac{1}{N} \sum_{k = 1}^N \lambda_{lk} \lambda^\T_{lk} \right) \left((F^{s_l})^\T  e_i^{s_l} \right)$.
		
		\begin{eqnarray*}
			\frac{1}{N}\sum_{i=1}^N  u_{li1} \lambda_{l'i}^\T  &=&  (V_r^{s_l})^{-1} \frac{1}{N^2}\sum_{i=1}^N \sum_{k=1}^N  \left(  \bar{\Lambda}_k(s_l) - (H^{s_l})^\T  \lambda_{lk} \right) \frac{\+E\left[(e^{s_l}_k)^\T  e^{s_l}_i\right]}{T(s_l)}   \lambda^\T_{l'i}    \\
			&+&   (V_r^{s_l})^{-1} (H^{s_l})^\T   \frac{1}{N^2}\sum_{i=1}^N \sum_{k=1}^N \lambda_{lk} \frac{\+E\left[(e^{s_l}_k)^\T  e^{s_l}_i\right]}{T(s_l)}   \lambda^\T_{l'i} \\
			&=& \left( (V_r^{s_l})^{-1} \frac{1}{N^2}\sum_{i=1}^N \sum_{k=1}^N  \left( \frac{1}{\sqrt{Th}} v_{lk} \right) \frac{\+E\left[(e^{s_l}_k)^\T  e^{s_l}_i\right]}{T(s_l)}  \lambda^\T_{l'i} \right. \\
			&+& \left. (V_r^{s_l})^{-1} (H^{s_l})^\T   \frac{1}{N^2}\sum_{i=1}^N \sum_{k=1}^N \lambda_{lk} \frac{\+E\left[(e^{s_l}_k)^\T  e^{s_l}_i\right]}{T(s_l)}   \lambda^\T_{l'i} \right) (1+o_p(1)) \\
			&=& \left( (V_r^{s_l})^{-1} (H^{s_l})^\T   \frac{1}{N^2}\sum_{i=1}^N \sum_{k=1}^N \lambda_{lk} \frac{\+E\left[(e^{s_l}_k)^\T  e^{s_l}_i\right]}{T(s_l)}   \lambda^\T_{l'i} \right) (1+o_p(1)) = O_p\left( \frac{1}{N} \right).
		\end{eqnarray*}
		The third equality follows from $\frac{1}{N^2}\sum_{i=1}^N \sum_{k=1}^N  \left( \frac{1}{\sqrt{Th}} v_{lk} \right) \frac{\+E\left[(e^{s_l}_k)^\T  e^{s_l}_i\right]}{T(s_l)}  \lambda^\T_{l'i} $\\ $=  O_p(1) \cdot \frac{1}{N}\sum_{i=1}^N \left\lbrace \frac{1}{N T(s_l)}  \sum_{k=1}^N  \left((F^{s_l})^\T  e_k^{s_l} \right) \frac{\+E\left[(e^{s_l}_k)^\T  e^{s_l}_i\right]}{T(s_l)} \right\rbrace \lambda^\T_{l'i}$ $=O_p\left( \frac{1}{N\sqrt{T(s_l)}} \right)$ by Assumption \ref{ass_loading}.1, Assumption \ref{ass_err}.3 ($\frac{\+E\left[(e^{s_l}_k)^\T  e^{s_l}_i\right]}{T(s_l)} \rightarrow \+E[e_{kt}e_{it}|S_t=s_l] = \tau_{ki,t}$) and Assumption \ref{ass_mom}.4. Also, \\ $ \frac{1}{N^2}\sum_{i=1}^N \sum_{k=1}^N \lambda_{lk} \frac{\+E\left[(e^{s_l}_k)^\T  e^{s_l}_i\right]}{T(s_l)}   \lambda^\T_{l'i} = \frac{1}{N^2}\sum_{i=1}^N\left\lbrace \sum_{k=1}^N \frac{\+E\left[(e^{s_l}_k)^\T  e^{s_l}_i\right]}{T(s_l)} \lambda_{lk} \right\rbrace \lambda^\T_{l'i}  = O_p\left( \frac{1}{N} \right)$ by Assumption \ref{ass_err}.3 and the fact that $\lambda_{lk}$ as a function of $s_l$ is independent of $e^{s_l}_k$ and $e^{s_l}_i$.

		\begin{eqnarray*}
			\frac{1}{N}\sum_{i=1}^N  u_{li2} \lambda_{l'i}^\T  &=&  (V_r^{s_l})^{-1} \frac{1}{N^2}\sum_{i=1}^N \sum_{k=1}^N  \left(  \bar{\Lambda}_k(s_l) - (H^{s_l})^\T  \lambda_{lk} \right) \frac{(e^{s_l}_k)^\T  e^{s_l}_i - \+E\left[(e^{s_l}_k)^\T  e^{s_l}_i\right]}{T(s_l)}  \lambda^\T_{l'i}   \\
			&+& (V_r^{s_l})^{-1} (H^{s_l})^\T  \frac{1}{N^2}\sum_{i=1}^N \sum_{k=1}^N \lambda_{lk}  \frac{(e^{s_l}_k)^\T  e^{s_l}_i - \+E\left[(e^{s_l}_k)^\T  e^{s_l}_i\right]}{T(s_l)}   \lambda^\T_{l'i} \\
			&=& \left( (V_r^{s_l})^{-1} \frac{1}{N^2}\sum_{i=1}^N \sum_{k=1}^N  \left( \frac{1}{\sqrt{Th}} v_{lk} \right) \frac{(e^{s_l}_k)^\T  e^{s_l}_i - \+E\left[(e^{s_l}_k)^\T  e^{s_l}_i\right]}{T(s_l)}  \lambda^\T_{l'i} \right.  \\
			&+& \left. (V_r^{s_l})^{-1} (H^{s_l})^\T  \frac{1}{N^2}\sum_{i=1}^N \sum_{k=1}^N \lambda_{lk}  \frac{(e^{s_l}_k)^\T  e^{s_l}_i - \+E\left[(e^{s_l}_k)^\T  e^{s_l}_i\right]}{T(s_l)}   \lambda^\T_{l'i} \right) (1+o_p(1)) \\
			&=& O_p\left( \frac{1}{N\sqrt{Th}} \right)
		\end{eqnarray*}
		by Assumption \ref{doublesum}.\ref{veeminueeelam} and \ref{doublesum}.\ref{lameeminuslam}.
		
		\begin{eqnarray*}
			\frac{1}{N}\sum_{i=1}^N  u_{li3} \lambda_{l'i}^\T  &=& (V_r^{s_l})^{-1} \frac{1}{N^2}\sum_{i=1}^N \sum_{k=1}^N  \left(  \bar{\Lambda}_k(s_l) - (H^{s_l})^\T  \lambda_{lk} \right) \frac{\lambda^\T_{lk} (F^{s_l})^\T  e_i^{s_l}}{T(s_l)} \lambda^\T_{l'i}   \\
			&+&  (V_r^{s_l})^{-1} (H^{s_l})^\T  \frac{1}{N^2}\sum_{i=1}^N \sum_{k=1}^N \lambda_{lk}  \frac{\lambda^\T_{lk} (F^{s_l})^\T  e_i^{s_l}}{T(s_l)} \lambda^\T_{l'i} \\
			&=& \left( (V_r^{s_l})^{-1} \frac{1}{\sqrt{Th}}   \left( \frac{1}{N} \sum_{k=1}^N v_{lk}\lambda^\T_{lk} \right)   \left( \frac{1}{N}\sum_{i=1}^N \frac{ (F^{s_l})^\T  e_i^{s_l}\lambda^\T_{l'i}}{T(s_l)} \right) \right. \\
			&+& \left. (V_r^{s_l})^{-1} (H^{s_l})^\T  \left(\frac{1}{N}\sum_{k=1}^N   \lambda_{lk}\lambda^\T_{lk}  \right) \left( \frac{1}{N}\sum_{i=1}^N \frac{(F^{s_l})^\T  e_i^{s_l}\lambda^\T_{l'i} }{T(s_l)}  \right) \right) (1+o_p(1))\\
			&=& \left( (V_r^{s_l})^{-1} (H^{s_l})^\T  \left(\frac{1}{N}\sum_{k=1}^N   \lambda_{lk}\lambda^\T_{lk}  \right) \left( \frac{1}{N}\sum_{i=1}^N \frac{(F^{s_l})^\T  e_i^{s_l}\lambda^\T_{l'i} }{T(s_l)}  \right)\right) (1+o_p(1)) \\
			&=& O_p\left( \frac{1}{\sqrt{NTh}} \right) 
		\end{eqnarray*}
		by Assumption \ref{doublesum}.\ref{vlamfelam} and Assumption \ref{ass_loading}.1 and $(V_r^{s_l})^{-1} \frac{1}{\sqrt{Th}}   \left( \frac{1}{N} \sum_{k=1}^N v_{lk}\lambda^\T_{lk} \right)   \left( \frac{1}{N}\sum_{i=1}^N \frac{ (F^{s_l})^\T  e_i^{s_l}\lambda^\T_{l'i}}{T(s_l)} \right) = O_p\left( \frac{1}{NTh} \right) $.
		
		\begin{eqnarray*}
			\frac{1}{N}\sum_{i=1}^N  u_{li4} \lambda_{l'i}^\T  &=& (V_r^{s_l})^{-1} \frac{1}{N^2}\sum_{i=1}^N \sum_{k=1}^N  \left(  \bar{\Lambda}_k(s_l) - (H^{s_l})^\T  \lambda_{lk} \right)\frac{(e_k^{s_l})^\T  F^{s_l} \lambda_{li} }{T(s_l)} \lambda^\T_{l'i}   \\
			&+& (V_r^{s_l})^{-1} (H^{s_l})^\T  \frac{1}{N^2}\sum_{i=1}^N \sum_{k=1}^N \lambda_{lk}  \frac{(e_k^{s_l})^\T  F^{s_l} \lambda_{li} }{T(s_l)}  \lambda^\T_{l'i} \\
			&=& \left( (V_r^{s_l})^{-1} (H^{s_l})^\T   \left(\frac{1}{N} \sum_{j = 1}^N \lambda_{lj} \Lambda^\T_{lj} \right) \left(  \frac{1}{N T(s_l)^2}  \sum_{k=1}^N  (F^{s_l})^\T  e_k^{s_l} (e_k^{s_l})^\T  F^{s_l}  \right)  \left( \frac{1}{N} \sum_{i=1}^N \lambda_{li}  \lambda^\T_{l'i} \right) \right.  \\
			&+& \left. (V_r^{s_l})^{-1} (H^{s_l})^\T   \left( \frac{1}{N} \sum_{k=1}^N  \frac{\lambda_{lk} (e_k^{s_l})^\T  F^{s_l}  }{T(s_l)} \right) \left( \frac{1}{N} \sum_{i=1}^N \lambda_{li}  \lambda^\T_{l'i} \right) \right) (1+o_p(1)) \\ 
			&=& O_p\left( \frac{1}{Th} \right), 
		\end{eqnarray*}
		where the first term is $O_p\left( \frac{1}{Th} \right)$ and the second term is $O_p\left( \frac{1}{\sqrt{NTh}} \right) $.
		
		Derivations for $\frac{1}{N}\sum_{i=1}^N   \lambda_{li}u_{l^{\prime}ia}^T$, $a=1,2,3,4$ are similar. 
		
		In $\Upsilon_{ll'}$, the leading term is 
		\begin{eqnarray*}
			\frac{1}{N}\sum_{i=1}^N  u_{li3}  u_{l'i3}^T  &=& \left(  (V_r^{s_l})^{-1} (H^{s_l})^\T  \left(\frac{1}{N}\sum_{i=1}^N   \lambda_{lk}\lambda^\T_{lk}  \right) \left( \frac{1}{N}\sum_{i=1}^N \frac{(F^{s_l})^\T  e_i^{s_l} }{T(s_l)} \left(\frac{ (F^{s_{l'}})^T e_i^{s_{l'}} }{T(s_{l'})}  \right)^T \right) \right. \\
			&& \times \left. \left(\frac{1}{N}\sum_{i=1}^N   \lambda_{l'k}\Lambda^\T_{l'k}  \right)H^{s_{l'}}(V_r^{s_{l'}})^{-1} \right) (1+o_p(1)) \\
			&=& O_p\left(\frac{1}{Th}\right)
		\end{eqnarray*}
		
		Let 
		\begin{align*}
		w_{u_l, \lambda_{l'}, 1} =&  (V_r^{s_l})^{-1} (H^{s_l})^\T  \left(\frac{1}{N}\sum_{k=1}^N   \lambda_{lk}\lambda^\T_{lk}  \right) \left( \frac{1}{N}\sum_{i=1}^N \frac{ (F^{s_l})^\T  e_i^{s_l}\lambda^\T_{l'i} }{T(s_l)}  \right)  H^{s_{l'}} , \\
		w_{u_l, \lambda_{l'}, 2} =&  (V_r^{s_l})^{-1} (H^{s_l})^\T  \left( \frac{1}{N} \sum_{k=1}^N  \frac{\lambda_{lk} (e_k^{s_l})^\T  F^{s_l}  }{T(s_l)} \right) \left( \frac{1}{N} \sum_{i=1}^N \lambda_{li}  \lambda^\T_{l'i} \right)  H^{s_{l'}}, \\
		w_{\lambda_l, u_{l'}, 1} =&  (H^{s_l})^\T  \left( \frac{1}{N} \sum_{k=1}^N  \frac{\lambda_{lk} (e_k^{s_{l'}})^\T  F^{s_{l'}}  }{T(s_{l'})} \right)  \left(\frac{1}{N}\sum_{k=1}^N   \lambda_{l'k}\Lambda^\T_{l'k}  \right) H^{s_{l'}} (V_r^{s_{l'}})^{-1},\\ 
		w_{\lambda_l, u_{l'}, 2} =&  (H^{s_l})^\T    \left(\frac{1}{N}\sum_{k=1}^N   \lambda_{lk}\Lambda^\T_{l'k}  \right) \left( \frac{1}{N}\sum_{i=1}^N \frac{ (F^{s_{l'}})^T e_i^{s_{l'}}\lambda^\T_{l'i} }{T(s_{l'})}  \right)  H^{s_{l'}} (V_r^{s_{l'}})^{-1},\\
		x_{u,v,p,w} =&(V_r^{s_p})^{-1} (H^{s_p})^T \left(\frac{1}{N}\sum_{i=1}^N   \lambda_{pk}\Lambda^\T_{uk}  \right) \left( \frac{1}{N}\sum_{i=1}^N \frac{ (F^{s_{u}})^T e_i^{s_{u}} }{T(s_{u})} \left(\frac{(F^{s_{v}})^T e_i^{s_{v}} }{T(s_{v})}  \right)^T \right) \left(\frac{1}{N}\sum_{i=1}^N   \lambda_{vk}\Lambda^\T_{wk}  \right) \\&H^{s_{w}}(V_r^{s_{w}})^{-1},\\
		z_{p,w} =& (V_r^{s_p})^{-1} (H^{s_p})^T \frac{1}{N^2}\sum_{i=1}^N \sum_{k=1}^N \lambda_{pk}  \frac{(e^{s_p}_k)^T e^{s_p}_i}{T(s_p)}   \Lambda^\T_{wi} H^{s_w}, \\
		x_{l,l'} =& x_{l, l', l, l'}  + x_{l,l,l, l'} + x_{l, l',l',l'},\\
		y_{l,l'} =& z_{l,l'}+z_{l',l}.
		\end{align*}
		
		Given $\sqrt{N}/(Th) \rightarrow 0$, $\sqrt{Th}/N \rightarrow 0$, $Nh^2 \rightarrow 0$ and $Th^3 \rightarrow 0$ (Theorem \ref{thm_loading} holds and $O_p(h)$ in Equation (\ref{eqn-decom-lambda}) does not dominate), the limiting distribution of $\frac{1}{N}\bar{\Lambda}^\T_l\bar{\Lambda}_{l'} $ ($l, l'=1,2$) is
		\begin{eqnarray*}
			\frac{1}{N}\bar{\Lambda}^\T_l\bar{\Lambda}_{l'} - \left(  \frac{1}{N}\sum_{i=1}^N (H^{s_l})^\T  \lambda_{li} \lambda_{l'i}^\T  H^{s_{l'}} + x_{l, l'} + y_{l, l'}  \right) = \left( w_{u_l, \lambda_{l'}, 1} + w_{u_l, \lambda_{l'}, 2} + w_{ \lambda_{l}, u_{l'}, 1} + w_{\lambda_{l}, u_{l'},  2} \right) (1+o_p(1))
		\end{eqnarray*}
		and we have
		\begin{eqnarray*}
			w_{u_l, \lambda_{l'}, 1} &\xrightarrow{p}& \left(V^{s_l} \right)^{-1}  ((Q^{s_l})^\T )^{-1} \Sigma_{\lambda_l,\lambda_l} \mu_{l,l'} (Q^{s_{l'}} )^{-1} \\
			w_{u_l, \lambda_{l'}, 2} &\xrightarrow{p}&  \left(V^{s_l} \right)^{-1} ((Q^{s_l})^\T )^{-1} \mu^\T_{l,l} \Sigma_{\lambda_l,\lambda_{l'}} (Q^{s_{l'}} )^{-1} \\
			w_{\lambda_{l}, u_{l'}, 1} &\xrightarrow{p}& ((Q^{s_l})^\T )^{-1} \mu^\T_{l^{\prime},l} \Sigma_{\lambda_{l'},\lambda_{l'}} (Q^{s_{l'}})^{-1} \left( V^{s_{l'}} \right)^{-1}  \\
			w_{\lambda_{l}, u_{l'}, 2} &\xrightarrow{p}& ((Q^{s_l})^\T )^{-1}  \Sigma_{\lambda_l,\lambda_{l'}} \mu_{l',l'} (Q^{s_{l'}})^{-1} \left( V^{s_{l'}} \right)^{-1},  
		\end{eqnarray*}
		where $\mu_{l,l'} = \frac{1}{N T(s_l)}(F^{s_l})^\T  e^{s_l}\Lambda^\T_{l'} = \frac{1}{N}\sum_{i=1}^N \frac{(F^{s_l})^\T  e_i^{s_l}\lambda^\T_{l'i} }{T(s_l)}$.
		
		For $w_{u_l, \lambda_{l'}, 1} \xrightarrow{p} \left(V^{s_l} \right)^{-1}  ((Q^{s_l})^\T )^{-1} \Sigma_{\lambda_l,\lambda_l} \mu_{l,l'} (Q^{s_{l'}} )^{-1}$, let $M_{l,l',1} =  \left(V^{s_l} \right)^{-1}  ((Q^{s_l})^\T )^{-1} \Sigma_{\lambda_l,\lambda_l} $ and $M_{l,l',2} = (Q^{s_{l'}} )^{-1}$, we have
		
		\begin{eqnarray*}
			\tvec \left( w_{u_l, \lambda_{l'}, 1} \right) \xrightarrow{p} \left( M_{l,l',2}^\T  \otimes M_{l,l',1}  \right) \tvec \left( \mu_{l,l'} \right)
		\end{eqnarray*}
		
		For $w_{u_l, \lambda_{l'}, 2} \xrightarrow{p} \left(V^{s_l} \right)^{-1} ((Q^{s_l})^\T )^{-1} \mu^\T_{l,l} \Sigma_{\lambda_l,\lambda_{l'}} (Q^{s_{l'}} )^{-1} $, let $M_{l,l',3} = \left(V^{s_l} \right)^{-1} ((Q^{s_l})^\T )^{-1} $ and $M_{l,l',4} = \Sigma_{\lambda_l,\lambda_{l'}} (Q^{s_{l'}} )^{-1} $, we have
		\begin{eqnarray*}
			\tvec \left( w_{u_l, \lambda_{l'}, 2} \right) \xrightarrow{p} \left( M_{l,l',3} \otimes M_{l,l',4}^\T   \right) \tvec \left( \mu_{l,l} \right)
		\end{eqnarray*}
		
		For $w_{\lambda_{l}, u_{l'}, 1} \xrightarrow{p}((Q^{s_l})^\T )^{-1}  \mu^\T_{l^{\prime},l} \Sigma_{\lambda_{l'},\lambda_{l'}} (Q^{s_{l'}})^{-1} \left( V^{s_{l'}} \right)^{-1}  $, let $M_{l,l',5} = ((Q^{s_l})^\T )^{-1}$ and \\ $M_{l,l',6} = \Sigma_{\lambda_{l'},\lambda_{l'}} (Q^{s_{l'}})^{-1} \left( V^{s_{l'}} \right)^{-1}  $,  we have
		\begin{eqnarray*}
			\tvec \left( w_{\lambda_{l}, u_{l'}, 1} \right) \xrightarrow{p} \left( M_{l,l',5} \otimes M_{l,l',6}^\T   \right) \tvec \left( \mu_{l',l} \right)
		\end{eqnarray*}
		
		For $w_{\lambda_{l}, u_{l'}, 2} \xrightarrow{p}  ((Q^{s_l})^\T )^{-1}  \Sigma_{\lambda_l,\lambda_{l'}} \mu_{l',l'} (Q^{s_{l'}})^{-1} \left( V^{s_{l'}} \right)^{-1}   $, let $M_{l,l',7} = ((Q^{s_l})^\T )^{-1} \Sigma_{\lambda_l,\lambda_{l'}} $ and $M_{l,l',8} = (Q^{s_{l'}})^{-1} \left( V^{s_{l'}} \right)^{-1}  $, we have
		\begin{eqnarray*}
			\tvec \left( w_{\lambda_{l}, u_{l'}, 2} \right) \xrightarrow{p} \left( M_{l,l',8}^\T \otimes M_{l,l',7}  \right) \tvec \left( \mu_{l',l'} \right) 
		\end{eqnarray*}
		
		As a result, 
		\begin{eqnarray*}
			&& \tvec \left( w_{u_l, \lambda_{l'}, 1} + w_{u_l, \lambda_{l'}, 2} + w_{ \lambda_{l}, u_{l'}, 1} + w_{\lambda_{l}, u_{l'},  2} \right) \\ 
			&\xrightarrow{p}& \left( M_{l,l',2}^\T  \otimes M_{l,l',1}  \right) \tvec \left( \mu_{l,l'} \right) + \left( M_{l,l',3} \otimes M_{l,l',4}^\T   \right) \tvec \left( \mu_{l,l} \right) \\
			&& + \left( M_{l,l',5} \otimes M_{l,l',6}^\T   \right) \tvec \left( \mu_{l',l} \right) + \left( M_{l,l',8}^\T \otimes M_{l,l',7}  \right) \tvec \left( \mu_{l',l'} \right) \\
			&\triangleq&  C_{l,l'}B
		\end{eqnarray*}
		where $B$ is defined as $B = \begin{bmatrix}
		\tvec\left( \mu_{1,1} \right) \\ \tvec\left( \mu_{1,2} \right) \\ \tvec\left( \mu_{2,1} \right) \\ \tvec\left( \mu_{2,2} \right) 
		\end{bmatrix}$ and
		
		\begin{eqnarray*}
			C_{1,1} &=& \begin{bmatrix}
				M_{1,1,2}^\T  \otimes M_{1,1,1} + M_{1,1,3} \otimes M_{1,1,4}^\T  + M_{1,1,5} \otimes M_{1,1,6}^\T  + M_{1,1,8}^\T  \otimes M_{1,1,7} & 0 & 0 & 0   
			\end{bmatrix} \\
			C_{1,2} &=& \begin{bmatrix}
				M_{1,2,3} \otimes M_{1,2,4}^\T  & M_{1,2,2}^\T  \otimes M_{1,2,1} & M_{1,2,5} \otimes M_{1,2,6}^\T  & M_{1,2,8}^\T  \otimes M_{1,2,7}
			\end{bmatrix} \\
			C_{2,1} &=& \begin{bmatrix}
				M_{2,1,8}^\T  \otimes M_{2,1,7} & M_{2,1,5} \otimes M_{2,1,6}^\T  & M_{2,1,2}^\T  \otimes M_{2,1,1} & M_{2,1,3} \otimes M_{2,1,4}^\T 
			\end{bmatrix} \\
			C_{2,2} &=& \begin{bmatrix}
				0 & 0 & 0& M_{2,2,2}^\T  \otimes M_{2,2,1} + M_{2,2,3} \otimes M_{2,2,4}^\T  + M_{2,2,5} \otimes M_{2,2,6}^\T  + M_{2,2,8}^\T  \otimes M_{2,2,7}.
			\end{bmatrix}
		\end{eqnarray*}
		
		From assumption in theorem 5 that $\sqrt{NTh}B = \xrightarrow{d} N(0, \Sigma_{B, B})$, and let $D = \begin{bmatrix}
		C_{1,1} \\ C_{1,2} \\ C_{2,1} \\ C_{2,2}
		\end{bmatrix}$,
		
		\begin{eqnarray*}
			\sqrt{NTh} \left( 
			\begin{bmatrix}
				\tvec\left(  \frac{1}{N} \bar{\Lambda}^\T_1\bar{\Lambda}_{1}\right)  \\ \tvec\left( \frac{1}{N}\bar{\Lambda}^\T_1\bar{\Lambda}_{2} \right)  \\ \tvec\left(  \frac{1}{N}\bar{\Lambda}^\T_2\bar{\Lambda}_{1} \right)  \\ \tvec\left(  \frac{1}{N}\bar{\Lambda}^\T_2\bar{\Lambda}_{2} \right) 
			\end{bmatrix} - 
			\begin{bmatrix}
				\tvec\left( \frac{1}{N} (H^{s_1})^\T  \Lambda^\T_1 \Lambda_1 H^{s_{1}} \right) \\ \tvec\left( \frac{1}{N} (H^{s_1})^\T  \Lambda^\T_1 \Lambda_2 H^{s_{2}} \right) \\ \tvec\left( \frac{1}{N} (H^{s_2})^\T  \Lambda^\T_2 \Lambda_1  H^{s_{1}} \right) \\ \tvec\left( \frac{1}{N} (H^{s_1})^\T  \Lambda^\T_2 \Lambda_2 H^{s_{2}} \right)
			\end{bmatrix} - \begin{bmatrix}
				\tvec\left(  x_{1,1} + y_{1,1}\right) \\ \tvec\left(  x_{1,2}+ y_{1,2} \right) \\ \tvec\left( x_{2,1} + y_{2,1}\right) \\ \tvec\left(   x_{2,2}+ y_{2,2} \right)
			\end{bmatrix} \right)  &\xrightarrow{d}& N\left( 0, D \Sigma_{B,B} D^\T \right)
		\end{eqnarray*}
		
		Let $G_1 = \frac{1}{N} (H^{s_1})^\T  \Lambda^\T_1 \Lambda_1 H^{s_{1}} $, $G_2 =  \frac{1}{N} (H^{s_1})^\T  \Lambda^\T_1 \Lambda_2 H^{s_{2}} $, $G_3 = \frac{1}{N} (H^{s_2})^\T  \Lambda^\T_2 \Lambda_1  H^{s_{1}} $, $G_4 =  \frac{1}{N} (H^{s_2})^\T  \Lambda^\T_2 \Lambda_2 H^{s_{2}} $ and define function $f$ to be
		\begin{eqnarray*}
			f \left( \begin{bmatrix}
				G_1 \\ G_2 \\ G_3 \\ G_4 
			\end{bmatrix} \right) = \tr\left( G_1^{-1} G_2 G_4^{-1} G_3 \right).
		\end{eqnarray*}
		$f(\cdot)$ has partial derivative
		\begin{eqnarray*}
			\frac{\partial f}{\partial G_1} &=& -(G_1^{-1} G_2 G_4^{-1} G_3 G_1^{-1} )^\T \\
			\frac{\partial f}{\partial G_2} &=& G_1^{-1} G_2 G_4^{-1} \\
			\frac{\partial f}{\partial G_3} &=& G_4^{-1} G_3 G_1^{-1} \\
			\frac{\partial f}{\partial G_4} &=& -(G_4^{-1} G_3 G_1^{-1} G_2 G_4^{-1} )^\T
		\end{eqnarray*}
		
		Let 
		$\xi = \begin{bmatrix}
		\tvec\left( \frac{\partial f}{\partial G_1} \right) \\ \tvec\left( \frac{\partial f}{\partial G_2} \right) \\ \tvec\left( \frac{\partial f}{\partial G_3} \right) \\ \tvec\left( \frac{\partial f}{\partial G_4} \right)
		\end{bmatrix}$, $\hat{\rho}$ is defined as 
		$\hat{\rho} = \tr \left\lbrace \left( \frac{1}{N}\bar{\Lambda}^\T_1\bar{\Lambda}_1 \right)^{-1} \left( \frac{1}{N}\bar{\Lambda}^\T_1\bar{\Lambda}_2 \right) \left( \frac{1}{N}\bar{\Lambda}^\T_2\bar{\Lambda}_2 \right)^{-1} \left( \frac{1}{N}\bar{\Lambda}^\T_2\bar{\Lambda}_1 \right) \right\rbrace $, $b$ is defined as $b = \begin{bmatrix}
		\tvec\left(  x_{1,1} + y_{1,1}\right) \\ \tvec\left(  x_{1,2}+ y_{1,2} \right) \\ \tvec\left( x_{2,1} + y_{2,1}\right) \\ \tvec\left(   x_{2,2}+ y_{2,2} \right)
		\end{bmatrix}$, and $\rho$ is defined as $\rho = \tr( G_1^{-1} G_2 G_4^{-1} G_3)$.\footnote{When $\Lambda_1, \Lambda_2 \in \mathbb{R}^{N \times r}$, and $\Lambda_1 = \Lambda_2 G$ for some matrix $G$, $\rho = \tr( G_1^{-1} G_2 G_4^{-1} G_3) = \tr((\Lambda^\T_1 \Lambda_1)^{-1} (\Lambda^\T_1 \Lambda_2) (\Lambda^\T_2 \Lambda_2)^{-1} (\Lambda^\T_2 \Lambda_1)) = \tr(I_r) = r$.}

		If $\sqrt{N}/(Th) \rightarrow 0$, $\sqrt{Th}/N \rightarrow 0$, $Nh^2 \rightarrow 0$ and $Th^3 \rightarrow 0$, we have 
		\begin{eqnarray*}
			\sqrt{NTh} (\hat{\rho} - r - \xi^\T  b) \xrightarrow{d} N(0, \xi^\T  D \Sigma_{B, B} D^\T \xi)
		\end{eqnarray*}

		The feasible test statistic is
		\begin{align*}
		\sqrt{NTh} \frac{(\hat{\rho} - r - \hat \xi^\T  \hat b)}{\sqrt{\hat \xi^\T  \hat D \hat \Sigma_{B, B} \hat D^\T \hat \xi} },
		\end{align*}
		where $\hat \xi$, $\hat b$, $\hat D$ and $\hat \Sigma_{B, B}$ are defined in Lemma \ref{consistentcovariance}. It is also shown in Lemma \ref{consistentcovariance} that $\hat{\xi} = \xi + O_p(1/\delta_{NT,h})$, $\hat b = b + O_p(\max(1/(T\delta_{NT,h}), 1/(N\delta_{NT,h})))$,  $\hat D = D + o_p(1)$ and $\hat \Sigma_{B, B} = \Sigma_{B, B} + o_p(1)$. We have 
		\[\hat \xi \hat b = \xi b + b O_p(1/\delta_{NT,h}) + \xi O_p(\max(1/(T \delta_{NT, h}), 1/(N \delta_{NT, h}))) + o_p(\max(1/(T \delta_{NT, h}), 1/(N \delta_{NT, h}))).\]
		
		Since $b = O_p(\max(1/N, 1/(Th)))$, when $\sqrt{NTh} O_p(\max(1/(T\delta_{NT,h}), 1/(N\delta_{NT,h}))) = o_p(1)$, equivalent to, $\sqrt{Th}/N \rightarrow 0$  and $\sqrt{N}/T \rightarrow 0$ (satisfied by assumptions about the rate conditions), the estimation errors multiplied  by $\sqrt{NTh}$ are $o_p(1)$. Thus, we have the feasible test statistic is asymptotically $N(0,1)$ distributed under $\mathcal{H}_0$. From Lemma \ref{lemma:rho}, the feasible test statistic diverges to $-\infty$ with probability 1 under $\mathcal{H}_1$.
		
	\end{proof}

	\begin{proof}[Proof of Lemma \ref{lemma:rho}]
		From the proof of Theorem \ref{thm_rho} and $NTh^3 \rightarrow 0$, it is shown that 
		$$
		\hat{\rho} = \tilde{\rho} + O_p\left( \frac{1}{\sqrt{NTh}} \right) + O_p\left( \frac{1}{Th} \right) + O_p\left( \frac{1}{N} \right),
		$$
		where $\tilde{\rho} = tr\left( \left( \frac{1}{N} (H^{s_1})^\T  \Lambda^\T_1 \Lambda_1 H^{s_{1}} \right)^{-1} \left(\frac{1}{N} (H^{s_1})^\T  \Lambda^\T_1 \Lambda_2 H^{s_{2}} \right) \left(  \frac{1}{N} (H^{s_2})^\T  \Lambda^\T_2 \Lambda_2 H^{s_{2}}\right)^{-1} \left( \frac{1}{N} (H^{s_2})^\T  \Lambda^\T_2 \Lambda_1  H^{s_{1}}\right) \right)$. 
		From the property of trace,
		$$\tilde{\rho} = tr\left( \left( \frac{1}{N}  \Lambda^\T_1 \Lambda_1  \right)^{-1} \left(\frac{1}{N}  \Lambda^\T_1 \Lambda_2  \right) \left(  \frac{1}{N} \Lambda^\T_2 \Lambda_2 \right)^{-1} \left( \frac{1}{N} \Lambda^\T_2 \Lambda_1 \right) \right).$$
		
		By Assumption (\ref{ass_lamjoint}), and delta method, similar as the proof of Theorem \ref{thm_rho}, 
		we have 
		$$
		\sqrt{N} \left(\tilde{\rho} - \bar{\rho} \right) \rightarrow N(0, \xi^\T  \Pi \xi),
		$$
		where $\xi = \begin{bmatrix}
		vec\left( -(G_1^{-1} G_2 G_4^{-1} G_3 G_1^{-1})^{\top} \right) \\ vec\left( G_1^{-1} G_2 G_4^{-1} \right) \\ vec\left( G_4^{-1} G_3 G_1^{-1} \right) \\ vec\left( -(G_4^{-1} G_3 G_1^{-1} G_2 G_4^{-1} )^{\top} \right) \end{bmatrix}$, $G_1 = \Sigma_{\Lambda_1, \Lambda_1}$, $G_2 = \Sigma_{\Lambda_1, \Lambda_2}$, $G_3 = \Sigma_{\Lambda_2, \Lambda_1}$, $G_4 = \Sigma_{\Lambda_2, \Lambda_2}$. Together with $$
		\hat{\rho} = \tilde{\rho} + O_p\left( \frac{1}{\sqrt{NTh}} \right) + O_p\left( \frac{1}{Th} \right) + O_p\left( \frac{1}{N} \right),
		$$ we have
		$$
		\sqrt{N} \left(\hat{\rho} - \bar{\rho} \right) \rightarrow N(0, \xi^\T  \Pi \xi)
		$$
	\end{proof}
	
	\subsection{Consistent Estimators for Terms in Theorem 2-4}
	Let the estimators for the projected and unprojected idiosyncratic components be $\hat e_{it}^s = X_{it}^s - \hat C_{it,s}^s = X_{it} - \hat \Lambda_i(s) \hat F_t^s$ and $\hat e_{it} = X_{it} - \hat C_{it,s} = X_{it} - \hat \Lambda_i(s) \hat F_t$. In Lemmas \ref{lemma:estimator-thm2-4}.1 and \ref{lemma:estimator-thm2-4}.3, we only consider the $t$s that satisfy $\frac{1}{\sqrt{N}} \sum_{i=1}^N \norm{\Lambda_i(S_t) - \Lambda_i(s)} = o_p(1)$. Thus, from (the proof of) Theorem \ref{thm_common}, $\hat e_{it}^s$ is a consistent estimator for $e_{it}^s$ for all $i$ and $t$ and $\hat e_{it}$ is a consistent estimator for $e_{it}$ for all $i$ and for $t$ where $t$ satisfies $\frac{1}{\sqrt{N}} \sum_{i=1}^N \norm{\Lambda_i(S_t) - \Lambda_i(s)} = o_p(1)$. 
	\begin{lemma}\label{lemma:estimator-thm2-4}
		As $N, Th \rightarrow \infty$, $h \rightarrow 0$, $Nh^2 \rightarrow 0$, $Th^3 \rightarrow 0$, under Assumption \ref{Ass:Ident}-\ref{ass_eigen},
		\begin{enumerate}
			\item \label{pi} Let $\Sigma_{e_t} = \+E[e_t e_t^T]$. Assume there are finitely many nonzeros in each row of $\Sigma_{e_t}$ and we know the set $\Omega_{e_t}$ of nonzero indices in $\Sigma_{e_t}$.  In Theorem \ref{thm_factor}, the consistent estimator for $\Pi_t = (V^s)^{-1}Q^s \Gamma_t^s (Q^s)^\T (V^s)^{-1}$, which is
			
			$$
			\Pi_t = plim\, (V_r^s)^{-1} \left(\frac{(\hat{F}^s)^\T F^s}{T(s)}  \right) \left( \frac{1}{N} \sum_{(i,j) \in \Omega_{e_t}} \Lambda_i(s)\Lambda^\T_j(s) \+E(e_{it}e_{jt}) \right) \left(\frac{(F^s)^\T \hat{F}^s }{T(s)} \right)  (V_r^s)^{-1},
			$$
			is 
			$$
			\hat{\Pi}_t = (V_r^s)^{-1} \left( \frac{1}{N} \sum_{(i,j) \in \Omega_{e_t}} \hat{\Lambda}_i(s)\hat{\Lambda}^\T_j(s) \hat{e}_{it} \hat{e}_{jt} \right) (V_r^s)^{-1}. 
			$$
			
			\item \label{theta} Let $\Sigma_{\und{e}_i} = \+E[\und{e}_i \und{e}_i^T]$. Assume there are finitely many nonzeros in each row of $\Sigma_{\und{e}_i}$ and we know the set $\Omega_{\und{e}_i}$ of nonzero indices in $\Sigma_{\und{e}_i}$.  In Theorem \ref{thm_loading}, the consistent estimator for $\Theta^s_i =  ((Q^s)^\T )^{-1} \Phi^s_i (Q^s)^{-1}$, which is 
			$$
			\Theta_i^s = plim\, ((Q^s)^\T )^{-1} \left( \frac{1}{T} \sum_{t = 1}^T \frac{R_K}{\pi(s)} \gamma_{FF}^s(t, t)   + \frac{h}{T} \sum_{(t,u) \in \Omega_{\und{e}_i}, t\neq u} \gamma_{FF}^s(t, u) \right) (Q^s)^{-1}
			$$
			is 
			$$
			\hat{\Theta}_i = \frac{Th}{T(s)^2}  \sum_{(t,u) \in \Omega_{\und{e}_i}}\hat{F}^s_t (\hat{F}^s_u)^\T  \hat{e}^s_{it} \hat{e}^s_{iu}.
			$$\footnote{We could also use HAC estimator for $\Phi_i^s$ on $\hat{F}^s_t \hat{e}^s_{it}$ similar as \cite{bai2003inferential} and proved similar to \cite{newey1994automatic}.}
			\item Assume there are finitely many nonzeros in each row of $\Sigma_{e_t}$ and $\Sigma_{\und{e}_i}$ and we know the sets $\Omega_{e_t}$ and $\Omega_{\und{e}_i}$ of nonzero indices in $\Sigma_{e_t}$ and $\Sigma_{\und{e}_i}$. In Theorem 4, the consistent estimator for $V_{it,s}$ is 
			$$
			\hat{V}_{it,s} = (\hat{\Lambda}_i(s))^\T  \left(\frac{1}{N} \sum_{i=1}^N (\hat{\Lambda}_i(s))^\T  \hat{\Lambda}_i(s)\right)^{-1} \left( \frac{1}{N} \sum_{l = 1}^N \hat{\Lambda}_i(s)\hat{\Lambda}^\T_i(s) \hat{e}^2_{it} \right)  \left(\frac{1}{N} \sum_{i=1}^N (\hat{\Lambda}_i(s))^\T  \hat{\Lambda}_i(s)\right)^{-1}  \hat{\Lambda}_i(s)
			$$
			and the consistent estimator for $W_{it,s}$ is 
			$$
			\hat{W}_{it,s} = (\hat{F}_t)^\T \hat{\Theta}_i (\hat{F}_t)^\T,
			$$
			where $\hat{\Theta}_i$ is defined in Lemma \ref{lemma:estimator-thm2-4}.\ref{theta}.

		\end{enumerate}
	\end{lemma}
	
	\begin{proof}[Proof of Lemma \ref{lemma:estimator-thm2-4}.1]
		For all $u$, $\hat{F}^s_u$ is the consistent estimator of $(H^s)^\T  F^s_u$ (without dividing $K^{1/2}_s(S_t)$) from the proof of Theorem \ref{thm_factor}; for all $i$,  $\hat{\Lambda}_i(s)$ is the consistent estimator for $(H^s)^{-1} \Lambda_i(s)$ from Theorem \ref{thm_loading}; $V_r^s$ is the consistent estimator for  $V^s$ from Lemma \ref{lemma3}.1; $\hat{e}_{it}$ is the consistent estimator for $e_{it}$ by $\hat{e}_{it} = X_{it} - \hat{C}_{it,s}$ and Theorem \ref{thm_common}. Moreover, the asymptotic convergence rate is $\delta_{NT,h}$. Similar as  Theorem 6 in \cite{bai2003inferential}, we can show
		\begin{enumerate}
			\item $\frac{1}{N} \sum_{(i,j) \in \Omega_{e_t}} \hat{\Lambda}_i(s)\hat{\Lambda}^\T_j(s) \hat{e}_{it} \hat{e}_{jt} - \frac{1}{N} \sum_{(i,j) \in \Omega_{e_t}} \hat{\Lambda}_i(s)\hat{\Lambda}^\T_j(s) e_{it} e_{jt} = O_p(1/\delta_{NT,h})$
			\item $\frac{1}{N} \sum_{(i,j) \in \Omega_{e_t}} \hat{\Lambda}_i(s)\hat{\Lambda}^\T_j(s) e_{it} e_{jt} - (H^s)^{-1} \frac{1}{N} \sum_{(i,j) \in \Omega_{e_t}} \hat{\Lambda}_i(s)\hat{\Lambda}^\T_j(s) e_{it} e_{jt} ((H^s)^{-1})^\T  = O_p(1/\delta_{NT,h})$
			\item  $\frac{1}{N} \sum_{(i,j) \in \Omega_{e_t}} \hat{\Lambda}_i(s)\hat{\Lambda}^\T_j(s) e_{it} e_{jt} -  \frac{1}{N} \sum_{(i,j) \in \Omega_{e_t}} \hat{\Lambda}_i(s)\hat{\Lambda}^\T_j(s) \+E[e_{it} e_{jt}] = O_p(1/\delta_{NT,h})$,
		\end{enumerate}
		where the last one is a special case of Theorem 2 in \cite{hansen2007asymptotic}. Together with $H^s = (Q^s)^{-1} + O_p(1/\delta_{NT,h})$ shown in Lemma \ref{lemma4}.1. From continuous mapping theorem, we have $\hat{\Pi}_t = \Pi_t + O_p(1/\delta_{NT,h})$. 
	\end{proof}
	
	\begin{proof}[Proof of Lemma \ref{lemma:estimator-thm2-4}.2]
		From the asymptotic distribution of kernel estimator (Section 3.2 in  \cite{hansen2007asymptotic}), we have 
		\begin{eqnarray*}
			&& \frac{T}{T(s)^2} \sum_{(t,u) \in \Omega_{\und{e}_i}}  K_s(S_t) K_s(S_u) F_t F_u^\T e_{it}e_{iu}\\
			&=& \left(\frac{1}{T} \sum_{t = 1}^T \frac{R_K}{\pi(s)} \gamma_{FF}^s(t, t)   + \frac{h}{T} \sum_{(t,u) \in \Omega_{\und{e}_i}, t\neq u} \gamma_{FF}^s(t, u) \right) + O(h^2) + O_p(1/\sqrt{Th}).
		\end{eqnarray*}
		Furthermore, similar as Lemma \ref{lemma:estimator-thm2-4}.\ref{pi}, we have 
		\begin{enumerate}
			\item $\frac{T}{T(s)^2} \sum_{(t,u) \in \Omega_{\und{e}_i}} \hat{F}^s_t (\hat{F}^s_u)^\T  \hat{e}^s_{it} \hat{e}^s_{iu} - \frac{T}{T(s)^2} \sum_{(t,u) \in \Omega_{\und{e}_i}} \hat{F}^s_t (\hat{F}^s_u)^\T  e^s_{it} e^s_{iu}  = O_p(1/\delta_{NT,h})$
			\item  $\frac{T}{T(s)^2} \sum_{(t,u) \in \Omega_{\und{e}_i}} \hat{F}^s_t (\hat{F}^s_u)^\T  e^s_{it} e^s_{iu} - \frac{T}{T(s)^2} \sum_{(t,u) \in \Omega_{\und{e}_i}} (H^s)^\T  F^s_t (F^s_u)^\T H^s e^s_{it} e^s_{iu} = O_p(1/\delta_{NT,h})$
		\end{enumerate}
		Recall $1/\delta_{NT,h} = \min\left(1/\sqrt{N}, 1/\sqrt{Th}, h \right)$, together with $H^s = (Q^s)^{-1} + O_p(1/\delta_{NT,h})$ and continuous mapping theorem, we have $\hat{\Theta}_i^s = \Theta_i^s + O_p(1/\delta_{NT,h})$.

	\end{proof}
	
	\begin{proof}[Proof of Lemma \ref{lemma:estimator-thm2-4}.3]
		$\hat{V}_{it,s}$ is a consistent estimator for $V_{it,s}$ followed from $\hat{\Lambda}_i(s)$ and $\hat{e}_{it}$ are the consistent estimator for $(H^s)^{-1} \Lambda_i(s)$ and $e_{it}$, and the proof of Lemma \ref{lemma:estimator-thm2-4}.\ref{pi}. $\hat{W}_{it,s}$ is a consistent estimator for $W_{it,s}$ followed from $\hat{F}^s_t$ and $\hat{e}^s_{it}$ are consistent estimators for $(H^s)^\T  F^s_t$ and $e_{it}$ the proof of Lemma \ref{lemma:estimator-thm2-4}.\ref{theta}.
	\end{proof}
	
	\subsection{Consistent Estimators for Terms in Theorem 5}
	Let the estimators for the projected idiosyncratic components be $\bar e^{s_l}_t = X^{s_l}_t - \bar \Lambda(s_l)\bar F^{s_l}_t$ and $\bar{\und{e}}^{s_l}_i = X^{s_l}_i - \bar F^{s_l} \bar \lambda_{li}$. From the proof of Theorem \ref{thm_common}, $\bar e^{s_l}_{it}$ is a consistent estimator for $e^{s_l}_{it}$ for all $i$ and $t$. 
	\begin{lemma}\label{consistentcovariance}
		As $N, Th \rightarrow \infty$, $h \rightarrow 0$, $\sqrt{N}/(Th) \rightarrow 0$, $\sqrt{Th}/N \rightarrow 0$, $Nh^2 \rightarrow 0$ and $Th^3 \rightarrow 0$, under Assumption \ref{Ass:Ident}-\ref{ass_mom}
		\begin{enumerate}
			\item Let $\hat G_1= \frac{1}{N}\bar \Lambda^{\top}_1 \bar \Lambda_1, \hat G_2= \frac{1}{N} \bar \Lambda_1^{\top} \bar \Lambda_2, \hat G_3= \frac{1}{N} \bar \Lambda_2^{\top} \bar \Lambda_1, \hat G_4= \frac{1}{N}\bar \Lambda_2^{\top} \bar \Lambda_2$ and plug $\hat{G_1}, \hat{G_2}, \hat{G_3}, \hat{G_4}$ in $\xi$ to get $\hat{\xi}$. We have $\hat G_i = G_i + O_p(1/\delta_{NT,h})$ for $i = 1, 2, 3, 4$ and therefore $\hat \xi = \xi + O_p(1/\delta_{NT,h})$.
			\item \label{xuvpw} Let $\Sigma_{e_T} = \+E[e^\T  e/N]$ and $\Sigma_{e_N} = \+E[e e^\T /T]$. Assume there are finitely many nonzeros in each row of $\Sigma_{e_T}$ and $\Sigma_{e_N}$  and we know the sets $\Omega_{e_T}$ and $\Omega_{e_N}$ of nonzero indices in $\Sigma_{e_T}$ and $\Sigma_{e_N}$. Let  
			\begin{align*}
			\hat x_{u,v,p,w} =& (\bar{V}_r^{s_p})^{-1}  \left(\frac{1}{N}\sum_{i=1}^N   \bar \lambda_{pi} \bar \Lambda^\T_{ui}  \right) \left( \frac{1}{N T(s_u) T(s_{v})} \sum_{(t_1, t_2) \in  \Omega_{e_T}} \bar{F}^{s_u}_{t_1} (\bar{F}^{s_{v}}_{t_2})^\T  (\bar{e}_{t_1}^{s_u})^\T  \bar{e}_{t_2}^{s_{v}} \right) \\
			& \left(\frac{1}{N}\sum_{i=1}^N   \bar \lambda_{vi} \bar \Lambda^\T_{wi}  \right)(\bar{V}_r^{s_{w}})^{-1}.
			\end{align*}
			and 
			$$\hat z_{p,w} = (\bar{V}_r^{s_p})^{-1} \frac{1}{N^2 T(s_p)} \sum_{(i,j) \in \Omega_{e_N}} \bar \lambda_{pi}  (\bar{\und{e}}^{s_p}_i)^\T \bar{\und{e}}^{s_p}_j  \bar \Lambda^\T_{wj} $$ 
			We have $\hat x_{u,v,p,w} =  x_{u,v,p,w} + O_p(1/(T\delta_{NT,h}))$ and $\hat z_{p,w} =  z_{p,w} + O_p(1/(N\delta_{NT,h}))$
			\item Let $\Sigma_{e} = \+E[\tvec(e) \tvec(e)^\T ]$. Assume there are finitely many nonzeros in each row of $\Sigma_{e}$ and we know the set $\Omega_{e}$ of nonzero indices in $\Sigma_{e}$, i.e., $\Omega_{e} = \{((i,t), (j,u))|\+E[e_{it}e_{ju}] \neq 0 \}$.  The consistent estimator for $\Sigma_{B,B} = \begin{bmatrix}
			\tvec(\mu_{1,1})\\\tvec(\mu_{1,2})\\\tvec(\mu_{2,1})\\\tvec(\mu_{2,2})
			\end{bmatrix}$, where $\mu_{l,l'} = \frac{1}{N T(s_l)}\sum_{i=1}^N \sum_{j=1}^{T} K_{s_l}(S_t) F_t e_{it}\lambda^\T_{l'i} $,  is 
			$$
			(NTh) \hat cov(\mu_{u,v,k,m}, \mu_{p,q,k',m'}) = \frac{Th}{N T(s_u) T(s_p)} \sum_{((i,t), (i',t')) \in \Omega_{e}} \bar F_{tk}^{s_u} \bar F_{t'k'}^{s_p} \bar \lambda_{u,im} \bar \lambda_{v,i'm'}  \bar e_{it}^{s_u} \bar e_{i't'}^{s_p}, 
			$$
			where $\mu_{u,v,k,m}$ is the $(k,m)$-th entry in $\mu_{u,v}$ and $\bar{\lambda}_{u,im}$ is the $(i,m)$-th entry in $\bar{\lambda}$.
			\item The consistent estimator for
			$D = \begin{bmatrix}
			C_{1,1} \\ C_{1,2} \\ C_{2,1} \\ C_{2,2}
			\end{bmatrix}$, where
			\begin{eqnarray*}
				C_{1,1} &=& \begin{bmatrix}
					M_{1,1,2}^\T  \otimes M_{1,1,1} + M_{1,1,3} \otimes M_{1,1,4}^\T  + M_{1,1,5} \otimes M_{1,1,6}^\T  + M_{1,1,8}^\T  \otimes M_{1,1,7} & 0 & 0 & 0   
				\end{bmatrix} \\
				C_{1,2} &=& \begin{bmatrix}
					M_{1,2,3} \otimes M_{1,2,4}^\T  & M_{1,2,2}^\T  \otimes M_{1,2,1} & M_{1,2,5} \otimes M_{1,2,6}^\T  & M_{1,2,8}^\T  \otimes M_{1,2,7}
				\end{bmatrix} \\
				C_{2,1} &=& \begin{bmatrix}
					M_{2,1,8}^\T  \otimes M_{2,1,7} & M_{2,1,5} \otimes M_{2,1,6}^\T  & M_{2,1,2}^\T  \otimes M_{2,1,1} & M_{2,1,3} \otimes M_{2,1,4}^\T 
				\end{bmatrix} \\
				C_{2,2} &=& \begin{bmatrix}
					0 & 0 & 0& M_{2,2,2}^\T  \otimes M_{2,2,1} + M_{2,2,3} \otimes M_{2,2,4}^\T  + M_{2,2,5} \otimes M_{2,2,6}^\T  + M_{2,2,8}^\T  \otimes M_{2,2,7}
				\end{bmatrix}
			\end{eqnarray*}
			is $\hat D$ that plugs $\hat{M}_{l,l',1} =  \left(\bar{V}_r^{s_l} \right)^{-1}  \frac{1}{N} \sum_{i=1}^N \bar{\lambda}_{li}\bar{\lambda}_{li}^\T $, $\hat{M}_{l,l',2} = I_r$, $\hat{M}_{l,l',3} =\left(\bar{V}_r^{s_l} \right)^{-1} $, $\hat{M}_{l,l',4} = \frac{1}{N} \sum_{i=1}^N \bar{\lambda}_{li}\bar{\lambda}_{l'i}^\T $, $\hat{M}_{l,l',5} = I_r$, $\hat{M}_{l,l',6} = \frac{1}{N} \sum_{i=1}^N \bar{\lambda}_{l'i}\bar{\lambda}_{l'i}^\T $, $\hat{M}_{l,l',7}= \frac{1}{N} \sum_{i=1}^N \bar{\lambda}_{li}\bar{\lambda}_{l'i}^\T $, $\hat{M}_{l,l',8} = \left(V_r^{s_{l'}} \right)^{-1}  $ in $M_{l,l',j}$ for $j = 1, \cdots, 8$.
		\end{enumerate}
	\end{lemma}
	
	\begin{proof}[Proof of Lemma \ref{consistentcovariance}.1]
		From Theorem \ref{thm_loading}, $\bar{\lambda}_{1i} = (H^{s_1})^\T  \lambda_{1i} + O_p(1/\delta_{NT,h})$ and $\bar{\lambda}_{2i} = (H^{s_2})^\T  \lambda_{2i} + O_p(1/\delta_{NT,h})$. Thus, for $v, w = 1,2$
		\begin{eqnarray*}
			\frac{1}{N}\bar{\Lambda}_v \bar{\Lambda}_{w} &=& \frac{1}{N}\sum_{i = 1}^N ((H^{s_v})^T \lambda_{vi} + O_p(1/\delta_{NT,h})) ((H^{s_w})^\T \lambda_{wi} + O_p(1/\delta_{NT,h})) \\ &=& \frac{1}{N}(H^{s_v})^T\Lambda^\T_v \Lambda_w H^{s_w} + O_p(1/\delta_{NT,h})),
		\end{eqnarray*}
		so $\hat{G} = G_i + O_p(1/\delta_{NT,h})$ for $i = 1, 2, 3, 4$. Note that $\xi = \begin{bmatrix}
		\tvec\left( -(G_1^{-1} G_2 G_4^{-1} G_3 G_1^{-1})^{\top} \right) \\ \tvec\left( G_1^{-1} G_2 G_4^{-1} \right) \\ \tvec\left( G_4^{-1} G_3 G_1^{-1} \right) \\ \tvec\left( -(G_4^{-1} G_3 G_1^{-1} G_2 G_4^{-1} )^{\top} \right) \end{bmatrix}$, given $\hat{G}_i = G_i + O_p(1/\delta_{NT,h})$ and the dimension of $G_i$ is fixed as $N$ and $T$ grow, we have $\hat G_1^{-1} \hat G_2 \hat G_4^{-1} \hat G_3 \hat G_1^{-1} = G_1^{-1} G_2 G_4^{-1} G_3 G_1^{-1} + O_p(1/\delta_{NT,h})$, $\hat G_1^{-1} \hat G_2 \hat G_4^{-1} = G_1^{-1} G_2 G_4^{-1} + O_p(1/\delta_{NT,h})$, $\hat G_4^{-1} \hat G_3 \hat G_1^{-1} = G_4^{-1} G_3 G_1^{-1} + O_p(1/\delta_{NT,h})$ and $\hat G_4^{-1} \hat G_3 \hat G_1^{-1} \hat G_2 \hat G_4^{-1} = G_4^{-1} G_3 G_1^{-1} G_2 G_4^{-1}+ O_p(1/\delta_{NT,h})$. Thus, $\hat \xi = \xi + O_p(1/\delta_{NT,h})$.
	\end{proof}
	
	\begin{proof}[Proof of Lemma \ref{consistentcovariance}.2]
		We first show $\hat x_{u,v,p,w} =  x_{u,v,p,w} + O_p(1/\delta_{NT,h})$. From Theorem \ref{thm_factor}, $\bar{F}^{s_u}_{t} = F^{s_u}_{t} + O_p(1/\delta_{NT,h})$ and $\bar{e}^{s_u}_{it} = e^{s_u}_{it} + O_p(1/\delta_{NT,h})$. Thus, $(\bar{e}_{t_1}^{s_u})^\T  \bar{e}_{t_2}^{s_{v}}/N = (e^{s_u}_{t_1})^\T  e_{t_2}^{s_{v}}/N + O_p(1/\delta_{NT,h})$ and   $ \bar{F}^{s_u}_{t_1} (\bar{F}^{s_{v}}_{t_2})^\T  (\bar{e}_{t_1}^{s_u})^\T  \bar{e}_{t_2}^{s_{v}} =  F^{s_u}_{t_1} (F^{s_{v}}_{t_2})^\T  (e_{t_1}^{s_u})^\T  e_{t_2}^{s_{v}} + O_p(1/\delta_{NT,h})$ for all $(t_1, t_2) \in \Omega_{e_T}$. Note that $|\Omega_{e_T}| = O(T)$ and $T(s_u)/T = O_p(1)$, we have
		\begin{eqnarray*}
			&& \frac{1}{N T(s_u) T(s_{v})} \sum_{(t_1, t_2) \in  \Omega_{e_T}} \bar{F}^{s_u}_{t_1} (\bar{F}^{s_{v}}_{t_2})^\T  (\bar{e}_{t_1}^{s_u})^\T  \bar{e}_{t_2}^{s_{v}}  \\
			&=& \frac{1}{N T(s_u) T(s_{v})} \sum_{(t_1, t_2) \in  \Omega_{e_T}} \left( F^{s_u}_{t_1} (F^{s_{v}}_{t_2})^\T  (e_{t_1}^{s_u})^\T  e_{t_2}^{s_{v}} + O_p(1/\delta_{NT,h}) \right) \\
			&=&\left( \frac{1}{N T(s_u) T(s_{v})} \sum_{(t_1, t_2) \in  \Omega_{e_T}} F^{s_u}_{t_1} (F^{s_{v}}_{t_2})^\T  (e_{t_1}^{s_u})^\T  e_{t_2}^{s_{v}}\right)  + O_p(1/T\delta_{NT,h})
		\end{eqnarray*}
		From Lemma \ref{consistentcovariance}.1, $\frac{1}{N}\sum_{i=1}^N   \bar \lambda_{pi} \bar \Lambda^\T_{ui} = \frac{1}{N}\sum_{i=1}^N \lambda_{pi} \Lambda^\T_{ui} + O_p(1/\delta_{NT,h}) $. From Lemma \ref{lemma3}.1, $\bar{V}_r^{s_p} = V_r^{s_p}  + O_p(1/\delta_{NT,h})$. Recall $ \frac{1}{N T(s_u) T(s_{v})} \sum_{(t_1, t_2) \in  \Omega_{e_T}} F^{s_u}_{t_1} (F^{s_{v}}_{t_2})^\T  (e_{t_1}^{s_u})^\T  e_{t_2}^{s_{v}} = O_p(1/(Th))$, we have $\hat x_{u,v,p,w} =  x_{u,v,p,w} + O_p(1/(T\delta_{NT,h}))$. 
		
		Next we show $\hat z_{p,w} = z_{p,w} + O_p(1/\delta_{NT,h})$. Note that $\frac{1}{T(p)} (\bar{\und{e}}^{s_p}_i)^\T \bar{\und{e}}^{s_p}_j = \frac{1}{T(p)} (\und{e}^{s_p}_i)^\T  \und{e}^{s_p}_j + O_p(1/\delta_{NT,h})$. Since  $\bar \lambda_{pi} = \lambda_{pi}+ O_p(1/\delta_{NT,h})$, and $\bar{V}_r^{s_p} = V_r^{s_p}  + O_p(1/\delta_{NT,h})$, we have 
		\begin{eqnarray*}
			&& \frac{1}{N^2 T(s_p)} \sum_{(i,j) \in \Omega_{e_N}} \bar \lambda_{pi}  (\bar{\und{e}}^{s_p}_i)^\T \bar{\und{e}}^{s_p}_j  \bar \Lambda^\T_{wj} \\
			&=& \frac{1}{N^2 T(s_p)} \sum_{(i,j) \in \Omega_{e_N}} \left( \lambda_{pi}  (\und{e}^{s_p}_i)^\T  \und{e}^{s_p}_j \Lambda^\T_{wj}+ O_p(1/\delta_{NT,h}) \right) \\
			&=& \left( \frac{1}{N^2 T(s_p)} \sum_{(i,j) \in \Omega_{e_N}}  \lambda_{pi}  (\und{e}^{s_p}_i)^\T  \und{e}^{s_p}_j \Lambda^\T_{wj}\right) + O_p(1/(N\delta_{NT,h})). 
		\end{eqnarray*}
		Together with $\bar V^{s_p}_r = V^{s_p}_r+ O_p(1/\delta_{NT,h})$, we have $\hat z_{p,w} = z_{p,w} = O_p(1/(N\delta_{NT,h}))$.
	\end{proof}
	
	\begin{proof}[Proof of Lemma \ref{consistentcovariance}.3]
		The argument is similar as Lemma \ref{lemma:estimator-thm2-4}.2. From the asymptotic distribution of kernel estimator (Section 3.2 in \citeauthor{hansen2007asymptotic}), we have
		\begin{eqnarray*}
			\frac{Th}{N T(s_u) T(s_p)} \sum_{((i,t), (i',t')) \in \Omega_{e}} F_{tk}^{s_u} F_{t'k'}^{s_p} \lambda_{u,im} \lambda_{v,i'm'}  e_{it}^{s_u} e_{i't'}^{s_p} = (NTh) cov(\mu_{u,v,k,m}, \mu_{p,q,k',m'}) + O(h^2) + O_p(1/\sqrt{Th})
		\end{eqnarray*}
		
		
		\begin{enumerate}
			\item $\frac{Th}{N T(s_u) T(s_p)} \sum_{((i,t), (i',t')) \in \Omega_{e}} \bar F_{tk}^{s_u} \bar F_{t'k'}^{s_p} \bar \lambda_{u,im} \bar \lambda_{v,i'm'}  \bar e_{it}^{s_u} \bar e_{i't'}^{s_p}$\\ $ -  \frac{Th}{N T(s_u) T(s_p)} \sum_{((i,t), (i',t')) \in \Omega_{e}} \bar F_{tk}^{s_u} \bar F_{t'k'}^{s_p} \bar \lambda_{u,im} \bar \lambda_{v,i'm'} e_{it}^{s_u} e_{i't'}^{s_p} = o_p(1)$
			\item $\frac{Th}{N T(s_u) T(s_p)} \sum_{((i,t), (i',t')) \in \Omega_{e}} \bar F_{tk}^{s_u} \bar F_{t'k'}^{s_p} \bar \lambda_{u,im} \bar \lambda_{v,i'm'}   e_{it}^{s_u}  e_{i't'}^{s_p}$\\ $ -  \frac{Th}{N T(s_u) T(s_p)} \sum_{((i,t), (i',t')) \in \Omega_{e}} \bar F_{tk}^{s_u} \bar F_{t'k'}^{s_p}  \lambda_{u,im}  \lambda_{v,i'm'} e_{it}^{s_u} e_{i't'}^{s_p} = o_p(1)$
			\item $\frac{Th}{N T(s_u) T(s_p)} \sum_{((i,t), (i',t')) \in \Omega_{e}} \bar F_{tk}^{s_u} \bar F_{t'k'}^{s_p} \bar \lambda_{u,im} \bar \lambda_{v,i'm'}   e_{it}^{s_u}  e_{i't'}^{s_p}$\\ $ -  \frac{Th}{N T(s_u) T(s_p)} \sum_{((i,t), (i',t')) \in \Omega_{e}}  F_{tk}^{s_u} F_{t'k'}^{s_p}  \lambda_{u,im}  \lambda_{v,i'm'} e_{it}^{s_u} e_{i't'}^{s_p} = o_p(1)$
		\end{enumerate}
		Recall $1/\delta_{NT,h} = \min\left(1/\sqrt{N}, 1/\sqrt{Th} \right)$, together with $H^s = (Q^s)^{-1} + o_p(1)$, we have \\ $(NTh) \hat cov(\mu_{u,v,k,m}, \mu_{p,q,k',m'})= (NTh) cov(\mu_{u,v,k,m}, \mu_{p,q,k',m'})+ o_p(1)$.
	\end{proof}
	
	\begin{proof}[Proof of Lemma \ref{consistentcovariance}.4]
		Recall $M_{l,l',1} =  \left(V_r^{s_l} \right)^{-1}  ((Q^{s_l})^\T )^{-1} \Sigma_{\lambda_l,\lambda_l} $, $M_{l,l',2} = (Q^{s_{l'}} )^{-1}$,\\ $M_{l,l',3} = \left(V_r^{s_l} \right)^{-1} ((Q^{s_l})^\T )^{-1} $, $M_{l,l',4} = \Sigma_{\lambda_l,\lambda_{l'}} (Q^{s_{l'}} )^{-1} $, $M_{l,l',5} = ((Q^{s_l})^\T )^{-1}$, \\ $M_{l,l',6} = \Sigma_{\lambda_{l'},\lambda_{l'}} (Q^{s_{l'}})^{-1} \left( V^{s_{l'}} \right)^{-1}  $, $M_{l,l',7} = ((Q^{s_l})^\T )^{-1} \Sigma_{\lambda_l,\lambda_{l'}} $, and $M_{l,l',8} = (Q^{s_{l'}})^{-1} \left( V_r^{s_{l'}} \right)^{-1}  $. 
		Note that $H^{s_l} = (Q^{s_l})^{-1} + O_p(1/\delta_{NT,h})$ shown in Lemma \ref{lemma4}.1 and $\bar \lambda_{l'i} =\lambda_{l'i} + O_p(1/\delta_{NT,h})$ from Theorem \ref{thm_loading}. The consistency of $\hat D$ holds following the same argument as Lemma \ref{consistentcovariance}.1.
	\end{proof}

	\subsection{Proof of Noisy State Process Model}
	As $\varepsilon_{it}$ is a vector while $e_{it}$ is a scalar, we state Assumption \ref{assumption_noisy}.1 for completeness: 
	\begin{customassumption}{8.1}\label{ass_varepsilon}
		There exists a positive constant $M < \infty$ such that for all $N$ and $T$:
		\begin{enumerate}
			\item \sloppy $\+E[ \varepsilon_{it}] = 0$, $\+E[ \norm{\varepsilon_{it}}^8] \leq M$ and $\varepsilon_{it}$ is independent of $S_t$, $F_t$ and $e_{it}$ for all $i$ and $t$.
			\item Weak time-series dependence: \sloppy $\+E\left[\norm{\mathcal{E}^\top_t \mathcal{E}_u/N}  \right] = \+E\left[\norm{\frac{1}{N} \sum_{i=1}^N \varepsilon^\top_{it} \varepsilon_{iu}} \right] = \widetilde \gamma_N(t,u)$. $
			\widetilde \gamma_N(t,t) \leq M$ for all $t$, $\widetilde \gamma_N(t,u) \leq M$ for all $t$ and $u$, and $\sum_{u=1}^T \norm{\widetilde \gamma_N(t,u)} \leq M$ for all $t$. 
			\item Weak cross-sectional dependence: \sloppy $\+E\left[\norm{\varepsilon_{it} \varepsilon^\top_{lt}} \right] = \widetilde \tau_{il, t}$, with $\widetilde \tau_{il, t} \leq \widetilde \tau_{il}$ for some $\widetilde \tau_{il}$ and $ \sum_{l = 1}^N \widetilde \tau_{il} \leq M$ for all $i$. 
			\item Weak total dependence: $\+E\left[\norm{\varepsilon_{it} \varepsilon^\top_{lu}}  \right] = \tau_{il, tu}$ and $\frac{1}{NT} \sum_{i = 1}^N \sum_{l = 1}^N \sum_{t = 1}^T \sum_{u = 1}^T \widetilde  \tau_{il,tu} \leq M$. 
			\item Bounded cross-sectional fourth moment correlation: \\ For every (t,u), $\+E \norm{N^{-1/2} \sum_{i=1}^N (\varepsilon_{it} \varepsilon^\top_{iu} - \+E[\varepsilon_{it} \varepsilon^\top_{iu}])}^4 \leq M$. 
		\end{enumerate}
	\end{customassumption}
	\begin{proof}[Proof of Corollary \ref{cor_noisy_ass_hold}]
		\begin{enumerate}
			\item 
			\begin{enumerate}
				\item We first show that under Assumption \ref{assumption_noisy}.1, Assumption \ref{ass_err} holds with $e_{it}$ replaced by $\varepsilon^\top_{it}F_t$. 
				
				Note that $\norm{\cdot}$ is the matrix Frobenius norm/vector 2-norm, from Cauchy-Schwartz inequality, we have $\norm{xy} \leq \norm{x}\norm{y}$ for any matrices/vectors $x$ and $y$.
				
				Assumption \ref{ass_err}.1 holds with $e_{it}$ replaced by $\varepsilon_{it}F_t$ because $\+E[\varepsilon_{it}F_t] = 0$ and 
				\[\+E[(F_t^\top \varepsilon_{it})^8] \leq \+E[\norm{F_t}^8 \norm{\varepsilon_{it}}^8] = \+E[\norm{F_t}^8 ] \+E[\norm{\varepsilon_{it}}^8]. \]
				
				Assumption \ref{ass_err}.2-4 holds with $e_{it}$ replaced by $\varepsilon_{it}F_t$ because
				\begin{eqnarray*}
					\+E[F_t^\top \varepsilon_{it} \varepsilon^\top_{ju} F_u ] &=& \+E[ \tr( \varepsilon_{it} \varepsilon^\top_{ju} F_u F_t^\top) ]  = \tr(\+E[\varepsilon_{it} \varepsilon^\top_{ju}] \+E[F_u F_t^\top])  \\
					&\leq& \norm{\+E[\varepsilon_{it} \varepsilon^\top_{ju}] }  \norm{\+E[F_u F_t^\top]} \leq \+E\left[\norm{\varepsilon_{it} \varepsilon^\top_{ju}}\right]  \+E[\norm{F_t} \norm{F_u} ]
				\end{eqnarray*}
				by the convexity of $\norm{\cdot}$ and the boundness of $E\left[\norm{\varepsilon_{it} \varepsilon^\top_{ju}}\right]$ from Assumption \ref{ass_err}.3 with $e_{it}$ replaced by $\varepsilon_{it}$.
				
				Assumption \ref{ass_err}.5 holds with $e_{it}$ replaced by $\varepsilon_{it}F_t$ because
				\begin{eqnarray*}
					&& \+E \left[N^{-1/2} \sum_{i=1}^N [F^\top_t \varepsilon_{it} \varepsilon^\top_{iu} F_u - \+E(F^\top_t\varepsilon_{it} \varepsilon^\top_{iu}F_u)] \right]^4 \\
					&=& \+E \left[ \+E \left[ \left( N^{-1/2} \sum_{i=1}^N [F^\top_t \varepsilon_{it} \varepsilon^\top_{iu} F_u - \+E(F^\top_t\varepsilon_{it} \varepsilon^\top_{iu}F_u)]\right)^4 |F_t, F_u \right] \right]\\
					&=& \+E \left[ \+E \left[ \left( F^\top_t \left(N^{-1/2} \sum_{i=1}^N [ \varepsilon_{it} \varepsilon^\top_{iu}  - \+E(\varepsilon_{it} \varepsilon^\top_{iu})] \right) F_u\right)^4 |F_t, F_u \right] \right]\\
					&\leq&  \+E \left[\norm{F_t}^4 \norm{F_u}^4\right] \+E \left[ \norm{ N^{-1/2} \sum_{i=1}^N [ \varepsilon_{it} \varepsilon^\top_{iu} - \+E(\varepsilon_{it} \varepsilon^\top_{iu})] }^4\right],
				\end{eqnarray*}
				where $\+E \left[ \norm{ N^{-1/2} \sum_{i=1}^N [ \varepsilon_{it} \varepsilon^\top_{iu} - \+E(\varepsilon_{it} \varepsilon^\top_{iu})] }^4\right]$ is bounded by Assumption \ref{ass_err}.5 with $e_{it}$ replaced by $\varepsilon_{it}$.
				
				Assumption \ref{ass_err}.6 holds with $e_{it}$ replaced by $\varepsilon_{it}F_t$ because
				\begin{eqnarray*}
					\norm{\+E[F_u F^\top_u \varepsilon_u \varepsilon_t F_t/N ]} &\leq&  \+E\left[ \norm{F_u F^\top_u \varepsilon_u \varepsilon_t F_t/N }\right] \leq  \+E\left[ \norm{F_u F^\top_u} \norm{\varepsilon_u \varepsilon^\top_t/N}  \norm{F_t}  \right]  \\
					&=& \+E[ \norm{F_u}^2\norm{F_t} ] \+E\left[ \norm{\varepsilon_u \varepsilon^\top_t/N  }\right],     
				\end{eqnarray*}
				where $\+E\left[ \norm{\varepsilon_u \varepsilon^\top_t/N  }\right]$ is bounded by Assumption \ref{ass_err}.2 with $e_{it}$ replaced by $\varepsilon_{it}$. Also,
				\begin{eqnarray*}
					\norm{\+E[F_u F^\top_u \varepsilon_u \varepsilon_t F_t/N |S_t, S_u]} &\leq&  \+E\left[ \norm{F_u F^\top_u \varepsilon_u \varepsilon_t F_t/N } |S_t, S_u\right]  \\
					&\leq& \+E\left[ \norm{F_u F^\top_u} \norm{\varepsilon_u \varepsilon^\top_t/N}  \norm{F_t} |S_t, S_u \right]  \\
					&=& \+E[ \norm{F_u}^2\norm{F_t} |S_t, S_u] \+E\left[ \norm{\varepsilon_u \varepsilon^\top_t/N  }\right],     
				\end{eqnarray*}
				where $\+E[ \norm{F_u}^2\norm{F_t} |S_t, S_u]$ is bounded by Assumption \ref{ass_factor}.

				\item Second is to show under Assumption \ref{assumption_noisy}.1, Assumption \ref{ass_err} holds with $e_{it}$ replaced by $\widetilde e_{it}$. 
				
				Assumption \ref{ass_err}.1 holds with $e_{it}$ replaced by $\widetilde e_{it}$ because
				\[\+E[\widetilde e^8_{it}] \leq 2^7 (\+E[(F^\top_t \varepsilon_{it})^8 ]+\+E[e^8_{it}]) \leq M\]
				
				Note that $\varepsilon_{it}$ is independent of $e_{ju}$ for all $i,j,t$ and $u$. Thus, \\ $\+E[F^\top_t \varepsilon_{it} e_{ju}] = \+E[F^\top_t \+E[ \varepsilon_{it}|F_t, e_{ju}]e_{ju}] = 0$. Employing this, Assumptions \ref{ass_err}.2, \ref{ass_err}.3, \ref{ass_err}.4 and \ref{ass_err}.6 hold. 
				
				Next is to verify Assumption \ref{ass_err}.5 holds with $e_{it}$ replaced by either $\widetilde e_{it}$. Denote $\psi_{it} = F^\top_t \varepsilon_{it}$, $y_1=\sum_{i=1}^N (\psi_{it}\psi_{iu}- \+E[\psi_{it}\psi_{iu}])$, $y_2=\sum_{i=1}^N (e_{it}e_{iu}- \+E[e_{it}e_{iu}])$, $y_3=\sum_{i=1}^N\psi_{it}e_{iu}$ and $y_4 = \sum_{i=1}^Ne_{it}\psi_{iu}$. 
				\begin{eqnarray*}
					&& \+E \left[N^{-1/2} \sum_{i=1}^N \left[(\psi_{it}+e_{it})(\psi_{iu}+e_{iu}) - \+E[(\psi_{it}+e_{it})(\psi_{iu}+e_{iu})]\right] \right]^4 \\
					&=& \+E \left[N^{-1/2} \sum_{i=1}^N \left[\psi_{it}\psi_{iu}- \+E[\psi_{it}\psi_{iu}] + e_{it}e_{iu}- \+E[e_{it}e_{iu}] + \psi_{it}e_{iu} + e_{it}\psi_{iu}\right] \right]^4 \\
					&\leq& 4^3/N^2 \cdot\left( \+E[y_1^4] + \+E[y_2^4] + \+E[y_3^4] + \+E[y_4^4]\right),
				\end{eqnarray*}
				where 
				\[N^{-2} \+E[y_3^4] = N^{-2} \left( \sum_{i=1}^N \+E[\psi_{it}^4] \+E[e_{iu}^4] + \sum_{i\neq j}\+E[\psi_{it}^2] \+E[e_{iu}^2] \+E[\psi_{jt}^2] \+E[e_{ju}^2] \right) \leq M \]
				and similarly $N^{-2} \+E[y_4^4] \leq M$ given $\+E[\psi_{it}^8]$ and $\+E[e_{it}^8]$ are bounded.

			\end{enumerate}
			\item 
			\begin{enumerate}
				\item Assumptions \ref{ass_mom}.2 and \ref{doublesum}.3 hold by the norm inequality $\norm{x+y} \leq \norm{x}+\norm{y}$.
				\item Assumptions \ref{ass_mom}.1, \ref{ass_mom}.5, \ref{doublesum}.1 and \ref{doublesum}.2, hold by $\psi_{it} = F^\top_t \varepsilon_{it}$ to be uncorrelated with $e_{ju}$, $\+E \norm{y_1 + y_2 + y_3 + y_4}^2\leq 4\left( \+E\norm{y_1}^2 +\+E\norm{y_2}^2+\+E\norm{y_3}^2 + \+E\norm{y_4}^2\right)$ for any matrices/vectors $y_1, y_2, y_3$ and $y_4$, and the same proof method as that to verify Assumption \ref{ass_err}.5 holds with $e_{it}$ replaced by either $\widetilde e_{it}$. 
				\item Since Assumptions \ref{ass_mom}.3, \ref{ass_mom}.4 and \ref{doublesum}.4 hold with $e_{it}$ replaced by $\psi_{it} = \varepsilon_{it}^\top F_t$, denote the asymptotic variances as 
				$$ \frac{1}{\sqrt{N}}\sum_{i = 1}^N \Lambda_i(s) F_t^\top \varepsilon_{it}  \ab  \xrightarrow{d} N(0, \Gamma_{\psi,t}^s),$$
				$$\frac{\sqrt{Th}}{T(s)} \sum_{t = 1}^T K_s(S_t) F_t F_t^\top \varepsilon_{it} \xrightarrow{d} N(0, \Phi_{\psi,i}^s),$$
				and 
				$$\sqrt{NTh} (B_{\psi}-0) \xrightarrow{d} N(0, \Sigma_{B_{\psi}, B_{\psi}}),$$
				where $B_\psi = \begin{bmatrix}
				\tvec\left( \mu_{\psi,1,1} \right) \\ \tvec\left( \mu_{\psi,1,2} \right) \\ \tvec\left( \mu_{\psi,2,1} \right) \\ \tvec\left( \mu_{\psi,2,2} \right) 
				\end{bmatrix} $ and $\mu_{\psi, l,l'} = \frac{1}{N T(s_l)}\sum_{i=1}^N \sum_{j=1}^{T} K_{s_l}(S_t) F_t F_t^\top \varepsilon_{it} \lambda^\T_{l'i}$.
				Since $\+E [\psi_{it}e_{it}] = 0$, we have
				\[ \frac{1}{\sqrt{N}}\sum_{i = 1}^N \Lambda_i(s) \widetilde e_{it}  \ab  \xrightarrow{d} N(0, \widetilde \Gamma_{t}^s),\]
				where $\widetilde \Gamma_{t}^s = \Gamma_{\psi,t}^s+\Gamma_{t}^s$, 
				\[\frac{\sqrt{Th}}{T(s)} \sum_{t = 1}^T K_s(S_t) F_t \widetilde e_{it} \xrightarrow{d} N(0, \widetilde \Phi_{i}^s),\]
				where $\widetilde \Phi_{i}^s =  \Phi_{\psi,i}^s + \Phi_{i}^s$,
				\[\sqrt{NTh} (\widetilde B-0) \xrightarrow{d} N(0, \widetilde \Sigma_{B, B}),\]
				where $\widetilde B = \begin{bmatrix}
				\tvec\left(\widetilde \mu_{1,1} \right) \\ \tvec\left( \widetilde \mu_{1,2} \right) \\ \tvec\left( \widetilde \mu_{2,1} \right) \\ \tvec\left( \widetilde \mu_{2,2} \right) 
				\end{bmatrix} $ and $\widetilde \mu_{l,l'} = \frac{1}{N T(s_l)}\sum_{i=1}^N \sum_{j=1}^{T} K_{s_l}(S_t) F_t \widetilde e_{it} \lambda^\T_{l'i}$ and $\widetilde \Sigma_{B, B} = \Sigma_{B_{\psi}, B_{\psi}} + \Sigma_{B, B}$.
			\end{enumerate}
		\end{enumerate}
	\end{proof}

\end{small}

\end{document}